\documentclass[12pt,a4paper,twoside,openright]{report}

\usepackage{cite}
\usepackage{geometry}

\usepackage{amsmath,amsthm,amssymb,amsfonts}
\usepackage{graphicx}
\usepackage{setspace}
\usepackage{mathrsfs}
\usepackage{cases}

\usepackage{todonotes}

\usepackage{amsthm}
\theoremstyle{plain}
	\newtheorem{thm}{Theorem}[chapter]
	\newtheorem{lem}[thm]{Lemma}
	\newtheorem{prop}[thm]{Proposition}
	\newtheorem{cor}[thm]{Corollary}
\theoremstyle{definition}
	\newtheorem{defn}[thm]{Definition}
	
	\newtheorem{exmp}[thm]{Example}
\theoremstyle{remark}
	\newtheorem*{rem}{Remark}

\newcommand{\ii}{{\rm i}}
\newcommand{\ee}{{\rm e}}
\newcommand{\vol}{{\rm vol}}
\newcommand{\tr}{{\rm tr}}
\newcommand{\Riem}{{\rm Riem}}
\newcommand{\Ric}{{\rm Ric}}
\newcommand{\R}{{\rm R}}
\newcommand{\Q}{{\rm Q}}
\newcommand{\Sol}{{\rm Sol}}
\newcommand{\csch}{{\rm csch}}
\newcommand{\sech}{{\rm sech}}
\newcommand{\arctanh}{{\rm arctanh}}

\newcommand{\supp}{{\rm supp}}
\newcommand{\x}{\mathsf{x}}
\newcommand{\W}{\mathcal{W}}

\begin{document}

\title{\textbf{Asymptotics in the time-dependent Hawking and Unruh effects}}

\author{{Benito A. Ju\'arez Aubry, MSc.} \vspace*{1cm} \\ 
	{\includegraphics[height=8cm]{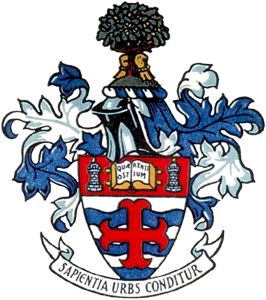}} \vspace*{1cm} \\
}
\date{Thesis submitted to the University of Nottingham \\
for the Degree of Doctor of Philosophy \vspace*{0.5cm} \\
March 2016}

\maketitle

\clearpage
\thispagestyle{empty}

\pagenumbering{roman}

\begin{abstract}
In this thesis, we study the Hawking and Unruh effects in time-dependent situations, as registered by localised spacetimes observers in several asymptotic situations.

In $1+1$ dimensions, we develop the Unruh-DeWitt detector model that is coupled to the proper time derivative of a real scalar field. This detector is insensitive to the well-known massless scalar field infrared ambiguity and has the correct massive-to-massless field limit. We then consider three scenarios of interest for the Hawking effect. The first one is an inertial detector in an exponentially receding mirror spacetime, which traces the onset of an energy flux from the mirror, with the expected Planckian late time asymptotics. The second one is the transition rate of a detector that falls in from infinity in Schwarzschild spacetime, gradually losing thermality. We find that the detector's transition rate diverges near the singularity proportionally to $r^{-3/2}$. The third one is the characterisation of the strength of divergence of the transition rate and of the (smeared) renormalised local energy density along a trajectory that approaches the future Cauchy horizon of a $(1+1)$-dimensional spacetime that generalises the non-extremal Reissner-Nordstr\"om spacetime and shares its causal structure. In both cases, the strength of the divergence as a function of the proper time is as on approaching the Schwarzschild singularity. We then comment on the limitations of our $(1+1)$-dimensional analysis as a model for the full $3+1$ treatment.

In $3+1$ dimensions, we revisit the Unruh effect and study the onset of the Unruh temperature. We treat an Unruh-DeWitt detector coupled to a massless scalar through a smooth switching function of compact support and prove that, while the Kubo-Martin-Schwinger (KMS) condition and the detailed balance of the response are equivalent in the limit of long interaction time, this equivalence is not uniform in the detector's energy gap. That is, we prove that the infinite-time and large-energy limits do not commute. We then ask and answer the question of how long one needs to wait to detect the Unruh temperature up to a prescribed large energy scale. We show that, under technical conditions on the switching function, in this large energy gap regime an adiabatically switched detector following a Rindler orbit will thermalise in a time scale that is polynomially large in the energy. We then consider an interaction between the detector and the field that switches on, interacts constantly for a long time, and then switches off. We show that a polynomially fast thermalisation cannot occur if the constant interaction time is polynomially large in the energy, with the switching tails fixed. Thus, we conclude that the details of the switching are relevant when estimating thermalisation time-scales.
\end{abstract}

\vspace*{\fill}
\begin{flushright}
\textit{A} Montse.
\end{flushright}
\vspace*{\fill}

\chapter*{Acknowledgements}

I thank my supervisor, for whom I have deep admiration, Dr. Jorma Louko, for taking me as a PhD student, for his guidance and advice throughout this work. It is hard for me to think of a supervisor who is more dedicated to his students than Dr. Louko. Many thanks for always pointing at the errors in my often flawed arguments. If I have at all become a stronger scientist, it is thanks to this. I also thank my supervisor for encouraging me to continue in academia in an evermore competitive environment.

Thanks are due to my collaborators, Prof.~J.~Fernando Barbero~G., Prof.~Christopher~J.~Fewster, Juan Margalef-Bentabol and Prof.~Eduardo~J.~S.~Villase\~nor, from whom I have learned many beautiful things about physics and mathematics while doing research together.

I thank Prof. Adrian Ottewill and Dr. Silke Weinfurtner for reading this work and suggesting many points that improved this thesis.

I would like to thank the members of the Mathematical Physics Group in the School of Mathematical Sciences with whom I have had the opportunity to discuss physics and mathematics: I thank Prof. John W. Barrett and Prof. Ivette Fuentes for assessing my research progress. I also enjoyed many insightful discussions with Prof. Kirill Krasnov and Dr. Thomas Sotiriou, and I acknowledge them for this. It is a pleasure for me to thank all of my colleagues researching gravity, postdocs and students, for our stimulating discussions and our reading groups together. I thank especially Manuel B\"arenz,  Tupac Bravo, Marco Cofano, Dr.~Hugo~R.~C. Ferreira, James Gaunt, Yannick Herfray, Jan Kohlrus, Dr.~Carlos~Scarinci and Vladimir Toussaint for sharing their knowledge with me. I hope that, but I am not sure if, I was able to share some valuable knowledge with them.

I also thank the staff in the School of Mathematical Sciences for making my stay as a student a smooth one.

On a personal note, I thank every single one of my friends who made my stay in Nottingham the happiest time. Thanks to the climbing crew for our great times bouldering, our trips and dinners together. I thank my friends in the Postgraduate New Theatre for our amazing time on and off stage. \textit{Muchas gracias} to my Mexican friends in Nottingham for bringing home to the UK. Thanks to my B50 officemates for our (big) breaks from research and thanks to all of my mathematician friends. Of course, I thank my friends who do not belong to any of the aforementioned categories. A big thank you for our time discussing life, philosophy, science and mathematics, as well as the forbidden topics: politics, religion and signature conventions.

\textit{J'aimerais remercier tr\`es fort ma copine,} Caroline\textit{, pour notre temps ensemble, qui continue \`a \^etre extraordinaire. Merci!}

\textit{Finalmente, agradezco muy cari\~nosamente a mi familia y a mis amigos de toda la vida. Much\'isimas gracias a mis padres,} Benito \textit{y} Flor\textit{, por todos los momentos felices, pero en especial por su apoyo en los momentos dif\'iciles. A mi hermana,} Montse\textit{, le debo el haber estudiado f\'isica. ?`Qu\'e m\'as puedo decir? No tengo palabras para agradecerle.}

\vfill
This work was supported by Consejo Nacional de Ciencia y Tecnolog\'ia, Mexico (CONACYT) REF 216072/311506. Early stages of this work were supported in part by Sistema Estatal de Becas del Estado de Veracruz.

\tableofcontents
\listoffigures
\listoftables

\clearpage
\thispagestyle{empty}

\pagenumbering{arabic}
\chapter{Introduction}
\label{ch:intro}

General relativity and quantum field theory are two of the greatest achievements of physics. On the one hand, general relativity describes the classical interactions between the gravitational field and classical matter at macroscopic scales, low energies and non-extreme curvature. While the first victory of general relativity, the correct calculation of the precession of the perihelion of Mercury, may seem modest today, the recent detection of gravitational waves by the collision of two black holes\cite{Abbott:2016blz} reminds us that Einstein's theory is not only elegant, but truly formidable. On the other hand, quantum field theory describes the fundamental interactions between matter at high energies, but below the Planck scale at $10^{19}$ GeV. The success of quantum field theory is manifest in the standard model of particle physics, which, in light of the recent discovery of the Higgs boson\cite{Englert:1964et, Higgs:1964ia, Guralnik:1964eu, Aad:2012tfa, Chatrchyan:2012xdj}, seems to be correct up to the Higgs scale at $10^2$ GeV. (It can be argued that phenomena such as dark matter and dark energy are still not fully understood, but the $\Lambda$CDM model remains phenomenologically viable, and neutrino physics phenomena are correctly explained by the see-saw mechanism extension of the standard model, within our experimental bounds.)

Yet, general relativity and quantum field theory remain disunited: Despite many efforts, the gravitational field has not been quantised. There are many reasons why this is so: First and most important, there are no experiments available at the energy scale at which quantum gravitational effects are expected to occur. While the experiments at the LHC have reached energies up to $10^4$ GeV, there are $15$ orders of magnitude before the Planck energy is reached, a feat that may not be attainable in the foreseeable future, if at all. Second, from a mathematical viewpoint, even at a perturbative level the theory is non-renormalisable, due to the physical dimension of the gravitational coupling constant, $G_\text{N}$, and non-perturbative techniques face difficulties due to the non-linearity of the Einstein field equations. Third, from a physical standpoint, the quantisation of gravity is the quantisation of the spacetime metric, which determines the notion of distances, time lapses and causality, thus, serious conceptual considerations have to be taken into account. Approaches in the direction of quantum gravity include string theory\cite{Green:1987sp, Becker:2007zj}, canonical quantum gravity\cite{DeWitt:1967yk, Ashtekar:1991hf, Thiemann:2007zz}, path integral approaches\cite{DeWitt:1967ub, Barrett:1999qw, Engle:2007wy} and many others (\textit{e.g.} \cite{Reuter:1996cp, Ambjorn:2012jv, Delfino:2012aj}), but no line of attack has succeeded thus far.

The point of view of \textit{quantum field theory in curved spacetimes}\cite{Birrell:1982ix, Bogoliubov:1959, Fulling:1989nb, Roman:1969, Wald:1995yp} is to take one step back and describe quantum matter, \textit{e.g.}, the quarks, leptons, gauge bosons, the Higgs and composite fields of the standard model, propagating in spacetime with a classical metric that encodes the gravitational field. We take this point of view in this thesis.

The surprising feature about quantum field theory in curved spacetime is that one can learn many things about the physics and mathematics of quantum gravity and even of flat spacetime quantum field theory. Arguably the most important physical insights in quantum gravity come from this semiclassical treatment. \textit{The Hawking effect}, \textit{i.e.}, the process whereby black holes radiate away their mass in the form of thermal radiation at infinity\cite{Hawking:1974sw, Hawking:1976ra}, indicates that there is a deep connection between gravity, thermodynamics and quantum physics. This connection expresses itself transparently in the laws of black hole mechanics\cite{Wald:1995yp, Bardeen:1973gs}: Because the horizon area of a black hole is proportional to its entropy, one can conjecture the existence of gravitational microscopic degrees of freedom at the horizon\cite{Ashtekar:1997yu, Agullo:2008yv, BarberoG.:2011zr}. A related phenomenon, which was discovered as a by-product of the study of the quantum phenomena of black holes, is the \textit{Unruh effect}, whereby a linearly uniformly accelerated observer in Minkowski spacetime will perceive a temperature proportional to their acceleration, due to the interaction with the quantum matter in the spacetime\cite{Unruh:1976db, DeWitt:1979}.

The most important mathematical insights for quantum gravity and flat spacetime quantum field theory come from the programme to formulate precisely quantum field theory in curved spacetimes. The attempt to generalise the quantum field theoretic flat spacetime formulation (see \textit{e.g.} \cite{Haag:1992hx}) led to the conclusion that concepts such as Poincar\'e invariance and particles are meaningless in general situations. There is simply no way to choose a preferred Hilbert space for the theory. Already in flat spacetime, a theory restricted to a Rindler wedge is unitarily inequivalent to the Poincar\'e invariant flat spacetime theory, based on the Minkowski vacuum state. Instead, quantum field theory, in flat spacetimes or otherwise, must be based on physical principles such as locality, covariance and causality \cite{Brunetti:2001dx, Hollands:2001nf, Fewster:2011pe}.\footnote{Field theory is under good control for globally hyperbolic spacetimes, but there is a large arena to explore for boundary value problems, especially in the context of non-trivial boundaries, or the imposition of boundary values as theory constraints. See, \textit{e.g.}, \cite{G.:2013zca, G.:2015uda} for classical aspects and \cite{G.:2015yxa} for quantisation in this context.} Thus, there is no reason to expect a theory of quantum gravity to be formulated directly in terms of a Hilbert space. It is very likely that such structure will appear only \textit{a posteriori}.

In addition to the Hawking and Unruh effects, other important discoveries and advances in the context of quantum field theory in curved spacetimes, or at least inspired by the theory, include the cosmological creation of particles\cite{Parker:1969au, Parker:1971pt, Zeldovich:1971mw}, the creation of particles by moving boundaries\cite{Davies:1976hi, Davies:1977yv}, analogue gravity \cite{Unruh:1980cg, Barcelo:2005ln, Weinfurtner:2010nu}, lightcone fluctuations in spacetime\cite{Ford:1994cr}, the detailed understanding of the Casimir effect\cite{Kay:1978zr, Dappiaggi:2014gea} and quantum energy inequalities\cite{Ford:1978qya, Fewster:1999gj}.

In this thesis, we revisit the Hawking and Unruh effects in selected contexts, which we describe below. Before that, we would like to point out that in a large portion of our work, we shall take the point of view that the Hawking and Unruh effects are detected by observers equipped with \textit{particle detectors}. To support our point of view, the statement that the Hawking effect is a process whereby a black hole \textit{radiates away particles} is ambiguous. The notion of a particle is a global one, because it makes reference to excitations of a fixed vacuum state. But, as we have mentioned, states are not unique (up to unitarity) in general curved spacetimes.\footnote{In Minkowski space one can use the Poincar\'e symmetry to select a prefer vacuum and fields are representations of the Poincar\'e group.} One can make sense of the concept of particles locally only if one interprets a particle as the absorption or emission transition in a localised test apparatus with internal states coupled to quantum fields. Such a test apparatus is called a particle detector. The dictum \textit{a particle is what a particle detector detects}\cite{Unruh:dictum} summarises this point of view.

Particle detectors are a powerful tool for analysing the experiences of observers along their spacetime worldlines. Perhaps the most prominent example is the Unruh-DeWitt detector \cite{DeWitt:1979}, a model upon which our work is based. While a large part of the literature has focused on stationary situations \cite{Takagi:1986kn, Crispino:2007eb}, particle detectors can be defined on non-stationary situations and analysed within first order perturbation theory \cite{Davies:2002bg, Schlicht:2003iy, Langlois:2005nf, Langlois:2005if, Louko:2006zv, Satz:2006kb, Louko:2007mu, Hodgkinson:2011pc, Hodgkinson:2012mr, Barbado:2012fy, Hodgkinson:2014iua, Hodgkinson:2013tsa, Ng:2014kha} and by non-perturbative techniques \cite{Raval:1995mb, Lin:2006jw, Ostapchuk:2011ud, Brown:2012pw, Bruschi:2012rx}. A review with further references can be found in \cite{Hu:2012jr}.

The aim of this work is to realise the Hawking and Unruh effects as phenomena experienced by local observers, and characterise how these effects emerge or disappear in certain regimes. We now discuss the objectives, contents and results of this work:

In Chapter \ref{ch:chap2}, we give an introduction to the quantisation of fields in curved spacetimes. The purpose of this chapter is, on the one hand, to make the thesis reasonably self-contained and, on the other hand, to introduce all the relevant notation and technical results that we shall exploit throughout the rest of the thesis. By no means is this an exhaustive review. There are many excellent introductions to the subject that deepen and expand the material that we present, such as the classic texts \cite{Birrell:1982ix, Bogoliubov:1959, Fulling:1989nb, Roman:1969, Wald:1995yp}. A reader well acquainted with quantum field theory in curved spacetimes can safely skip the reading of this chapter, with the possible exception of Section \ref{2sec:Detectors}, which motivates the part of our work in lower dimensions in this thesis.

Chapters \ref{ch:DeCo} and \ref{ch:BH} deal with quantum effects in $1+1$ dimensions. We shall work with a minimally coupled massless scalar field theory, which has the advantage of being conformal in $1+1$-dimensional spacetimes, giving full analytic control in many interesting situations. As simple as $(1+1)$-dimensional scalars are, the Wightman two-point function of the theory, which is crucial for defining (Hadamard) vacuum or thermal states, contains an infrared ambiguity.\footnote{This ambiguity is absent for Dirac fields, see \textit{e.g.} \cite{Louko:2016lfs} in the context of detectors.} While this is not a problem for uniquely defining the stress-energy tensor in a given state \cite{Fulling:1987, Kay:2000fi}, it is so for defining the vacuum polarisation. Hence, the transition probability of the Unruh-DeWitt detector, which is proportional to the (smeared) power spectrum of the vacuum polarisation (up to perturbative leading order), is ambiguous. To solve this problem, in Chapter \ref{ch:DeCo} we introduce a derivative-coupling detector \cite{Grove:1986fy, Davies:2002bg, Raval:1995mb, Raine:1991kc, Wang:2013lex}, in dimensions greater than or equal to two, which couples linearly to the derivative of the field, instead of the field itself as in the Unruh-DeWitt model, and that has an unambiguously defined transition probability. One way to interpret the coupling of such a detector can be understood by making an analogy with quantum electrodynamics. In this context, the derivative-coupling detector would not be sensitive to the gauge potential, but only to the electric field. Another way to interpret this coupling comes from the work of Grove \cite{Grove:1986fy}, who realised that such derivative-coupling detector would be sensitive to the energy flux of the quantum field to which it couples.

In Chapter \ref{ch:DeCo}, based on ref. \cite{Juarez-Aubry:2014jba}, we provide an ultraviolet-divergence-free and infrared-unambiguous formula for the response function\footnote{The term response function and transition probability are usually used interchangeably because the perturbative leading order in the transition probability of the detector is proportional to the response function, with a proportionality factor independent of the field state, spacetime and detector trajectory.} of the derivative-coupling detector along arbitrary worldlines in $1+1$ dimensions. We study the sharp-switching limit of the detector response and obtain a general formula for the sharp instantaneous transition rate along non-stationary and stationary worldlines. We then verify that the rate of the detector satisfies many desirable properties in stationary situations. Namely, we verify that the massive-to-massless limit of the field is continuous in the Minkowski vacuum and in Minkowski half-space, that the detector thermalises in a heat bath at the bath temperature, in accordance with the detailed balance version of the KMS condition and, finally, that the detector responds at the Unruh temperature, due to the Unruh effect, along linearly uniformly accelerated worldlines.

In Chapter \ref{ch:BH}, with the confidence that the derivative coupling detector has the correct massless limit and is sensitive to thermal phenomena, we study quantum effects in several black hole (and black hole motivated) spacetimes. The work that we present here is based on \cite{Juarez-Aubry:2014jba} and partly on the conference paper \cite{Juarez-Aubry:2015dla}.\footnote{At present time, a full paper that extends \cite{Juarez-Aubry:2015dla} and includes results that have first appeared in this thesis is in preparation.} 

First, in Section \ref{sec:4Classical}, we introduce the $1+1$-dimensional spacetimes that we shall be working on. Namely, a \textit{receding mirror spacetime} \textit{\`a la} Davies-Fulling\cite{Davies:1976hi, Davies:1977yv} in $1+1$ dimensions that mimics the gravitational effects of a spherically symmetric collapsing star at late times, and which asymptotes the Minkowski half-space with Dirichlet boundary conditions at early times; the $1+1$ Schwarzschild black hole and a conformal class of spacetimes that generalise the $1+1$ non-extremal Reissner-Nordstr\"om black hole, in the sense that these spacetimes share the causal structure of Reissner-Nordstr\"om, which includes the presence of Cauchy horizons. 

Second, in Section \ref{sec:4Quantisation}, we quantise the free (non-interacting) massless, minimally coupled scalar field in the aforementioned spacetimes. In the receding mirror spacetime, we obtain a state that vanishes at the mirror and that contains, in addition to the ultraviolet divergences, appropriately regularised divergences whenever two spacetime points are connected by a null ray reflected on the mirror. In the Schwarzschild and in our generalised Reissner-Nordstr\"om spacetimes, we define Unruh and Hartle-Hawking-Israel (HHI) vacua in the appropriate globally hyperbolic submanifolds of the black hole spacetimes. 

Third, in Section \ref{sec:4QuantumEffects}, we study quantum effects on these spacetimes. In the receding mirror spacetime, we compute the rate of detectors, which are inertial with respect to the mirror in the asymptotic past, at the early and late times, and find that the right-moving modes of the scalar field are detected by the derivative-coupling detector at the would-be Hawking temperature that emerges at late stages of the star collapse. The transient time behaviour is analysed numerically. In the Schwarzschild spacetime, we study the transition rate of static observers, as well as the transition rate along infalling geodesics. The transition rate is computed in the near-infinity and near-singularity regimes analytically for the Unruh and HHI vacua, and the interpolating behaviour is computed numerically, showing a loss of thermality along the infalling trajectory and a divergence of the transition rate proportional to $r^{-3/2}$ as the radial coordinate, $r$, approaches zero. In the HHI vacuum, we also consider detectors that are switched on in the white hole region and travel towards the black hole singularity, and their behaviour is analysed numerically. In our class of generalised Reissner-Nordstr\"om spacetimes, we compute the rate of divergence of the transition rate as a left-travelling geodesic observer equipped with a detector approaches the future Cauchy horizon in the Unruh and HHI vacua, and find that the strength of the divergence is as for the detector approaching the Schwarzschild singularity as a function of the proper time. We then compute the (sharply and smoothly smeared renormalised local energy density, obtained from the renormalised stress-energy tensor, along this trajectory and find agreement with the detector rate, in the sense that the energy density becomes large along the trajectory as the Cauchy horizon is approached.

Fourth, in Section \ref{sec:4Rindler}, we comment on the validity of using $1+1$-dimensional models to describe $3+1$ phenomena. We consider a $1+1$ observer equipped with a derivative-coupling detector in the Rindler vacuum and compute the divergence of the detector's transition rate as the observer approaches the future horizon of the Rindler spacetime along an orbit of the Minkowskian time translation Killing vector. We find that the transition rate diverges in a non-integrable fashion, which is stronger than the situation for the $3+1$ non-derivative Unruh-DeWitt detector. This leads us to conjecture that the transition rate divergent behaviour near a $3+1$ Reissner-Nordstr\"om Cauchy horizon may be integrable for a sharply switched detector.

In Chapter \ref{ch:Unruh}, we address the question of how long one needs to wait to detect the Unruh effect by looking at the response of particle detectors that interact with a field for a finite time, such that after a long interaction time has elapsed the detailed balance form of the KMS condition holds up to a prescribed precision at the Unruh temperature. The work that we present in this chapter is based in part on \cite{Fewster:2015dqb}.\footnote{At present time, a full paper that extends \cite{Fewster:2015dqb} and includes results that have first appeared in this thesis is in preparation.} We consider a finite-time interacting Unruh-DeWitt detector that interacts with a massless $3+1$ scalar field in the Minkowski vacuum state for a finite amount of time through an interaction Hamiltonian whose coupling strength is controlled by a smooth function of compact support. 

In Section \ref{sec5:KMS} we first show rigorously that, under general conditions, the asymptotic limit of the response function as the interaction time goes to infinity, prescribed by a time scale, $\lambda$, which rescales the detector-field interaction, satisfies the detailed balance condition if and only if the state is KMS. For the switching function, we consider both an adiabatic time scaling, which represents a long and slow switching, and a scaling that leaves the switch-on and switch-off tails of the switching function intact, but scales the constant-strength interaction time between the field and the detector. Second, we prove that, under mild assumptions, the thermal limit of the response function cannot be uniform in the detector transition energy gap. Third, we prove that in the case of the Unruh effect the above hypotheses hold and, hence, the thermal limit cannot be uniform in the detector energy.

Thus, the question of how long one waits to detect the Unruh temperature can only be answered in terms of the energy scale of the detector, \textit{i.e.}, how long one needs to wait to detect the Unruh temperature up to a prescribed large energy scale, $E$. In Section \ref{sec:slowswitching} we show that, in the case of the adiabatic scaling, one can find a class of switching functions for which the thermalisation time scale is polynomial in the large energy, \textit{i.e.}, one needs to wait a polynomially long time to detect the Unruh temperature. We then show, in Section \ref{sec:plateau}, that for the constant-interaction rescaling, there is no switching function for which the detector thermalises in a long time scale which is polynomial in the large energy. The moral is that the details of the thermalisation are in the switching of the interaction.

We believe that the estimates that we present in Chapter \ref{ch:Unruh} are relevant in the design of experiments seeking to  verify the Unruh effect by measuring the Unruh temperature with precision up to a prescribed energy scale in the detector response spectrum.

Finally, we summarise our results and make our concluding remarks in Chap. \ref{ch:conc}. 

Appendices \ref{app:DeCo}, \ref{app:BH} and \ref{app:Unruh} contain several technical computations and results that we use throughout this thesis.

We set the metric signature such that the contraction with two timelike vector fields is negative and the contraction with two spacelike vector fields is positive. We work in units in which $\hbar = c =1$, unless explicitly indicated. Complex conjugation is denoted by an overline and Hermite conjugation is denoted by an asterisk. 

Throughout this work we use the {\rm Big O} and {\rm small o notation} for asymptotic quantities. We illustrate this notation: Let $f:x \mapsto f(x)$ be a real-valued real function. We say that $f(x) = O(x)$ as $x \to 0$ if $x^{-1}f(x)$ is bounded as $x \to 0$. $f(x) = O^\infty(x)$ as $x\to 0$ if $f(x)$ falls off faster than any positive power of $x$ as $x\to 0$. $f(x) = O(1)$ if $f(x)$ is bounded in the limit under consideration. $f(x) = o(x)$ as $x \to 0$ if $x^{-1}f(x)\to 0$ as $x\to0$. Finally, $f(x) = o(1)$ if $f(x)$ vanishes in the limit under consideration.

\chapter[Elements of QFT in curved spacetimes]{Elements of quantum field theory in curved spacetimes}
\label{ch:chap2}

The aim of this chapter two-fold. First, we review quantisation in field theory. Second, we introduce the relevant technology in the mathematical theory of quantum fields. There exists an extensive literature on the subject of quantum field theory in curved spacetimes that expands and deepens the exposition that we present here. Classic texts include \cite{Birrell:1982ix, Bogoliubov:1959, Fulling:1989nb, Roman:1969, Wald:1995yp}. Recent progress on the field has been achieved from the so-called algebraic point of view. Relevant literature reviews on the topic include \cite{Bar:2009zzb, Brunetti:2015vmh, Hollands:2014eia} and references therein. For a discussion from a mathematical point of view, see for example \cite{Bar:2007zz, Derezinski:2013dra}.

In Section \ref{2sec:Classical} we review classical field theory, the theory that describes the macroscopic aspects of our universe. We introduce the procedure known as \textit{quantisation} in Section \ref{2sec:Quantum}. We take the opportunity to emphasise that the quantisation map acts on fields and, while the language of quantum field theory is appropriate to describe particle physics, particles are not the fundamental objects in the theory. Instead, quantum fields are. We explain in Section \ref{2sec:Detectors} what the relationship between quantum fields and particles is, by introducing \textit{particle detectors}.

\begin{defn}
A \textit{spacetime} is a pair $(M,g)$ consisting of a connected, real, smooth manifold, $M$, equipped with a Lorentzian metric, $g$, with signature $(-++ \ldots +)$.
\end{defn}

In calculations, we shall often make use of abstract index notation, by which we write, \textit{e.g.}, a vector field in the tangent bundle of the spacetime, $v \in TM$, as $v^a$ and a cotangent vector, $\omega \in T^*(M)$, as $\omega_a$. The metric tensor reads $g_{ab}$ in this notation.

We further demand that these spacetimes be orientable and time-orientable according to the following:

\begin{defn}
An $n$-dimensional spacetime is called an \textit{oriented spacetime}, denoted by $(M,g,\epsilon)$, if there exists a nowhere vanishing top form $\omega \in \Omega^n(M)$ with the orientation $\epsilon$ inherited from $\omega$.\footnote{We have denoted by $\Omega^p(M)$ the space of $p$-forms on the spacetime $M$.}
\end{defn}

\begin{defn}
A spacetime is called a \textit{time-oriented spacetime}, denoted by $(M,g,t)$, if there exists an everywhere timelike vector field $t \in TM$, i.e., $g(t,t) = g_{ab}t^a t^b < 0$.
\end{defn}

\begin{defn}
Let $K$ be a subset of the time-oriented spacetime $(M,g,t)$,
\begin{itemize}
\item $I^{+/-}(K)$ is the chronological future/past of $K$, i.e., the set of points $p \in M$ that can be reached by a future/past directed timelike curve that starts at a point in $K$.
\item $J^{+/-}(K)$ is the causal future/past of $K$, i.e., as above in the case of causal curves.
\item $I(K) = I^{-}(K) \cup I^{+}(K)$ and $J(K) = J^{-}(K) \cup J^{+}(K)$.
\end{itemize}
\end{defn}

We present some notation that will be useful in our discussion:

\begin{defn}
We denote:
\begin{itemize}
\item $C^\infty(M,V)$: the space of $V$-valued, smooth functions on a manifold $M$. 
\item $C_0^\infty(M, V)$: the space of $V$-valued, smooth functions of compact support on a manifold $M$.
\item $C_{\text{sc}}^\infty(M, V)$: the space of $V$-valued, smooth functions of space-like compact support, i.e., whose support is contained in $J(K)$ where $K \subset M$ is a compact subset of $M$.
\item $C_{\text{f/p}}^\infty(M,V)$: the space of $V$-valued, smooth functions of future/past compact support, i.e., if $f \in C_{\text{f/p}}^\infty(M,V)$, for any $p \in M$, $\supp (f) \cap J^{+/-}(p)$ is compact.
\end{itemize}
\end{defn}

In classical field theory, an all-important property for time-oriented spacetimes is \textit{global hyperbolicity}.

\begin{defn}
Let $(M,g,t)$ be a time-oriented spacetime. A subset $\Sigma \subset M$ is called an \textit{achronal set} if $I^+(\Sigma) \cap \Sigma = \emptyset$. Moreover, such an achronal set $\Sigma$ is called a \textit{Cauchy hypersurface} if for any $p \in M$, every past and future inextendible causal curve through $p$ intersects $\Sigma$.
\end{defn}

\begin{thm}
Let $(M,g,t)$ be a time-oriented spacetime. The following statements are equivalent:
\begin{enumerate}
\item $(M,g,t)$ is globally hyperbolic.
\item There exists a Cauchy hypersurface $\Sigma \subset M$, i.e., $J(\Sigma) = M$.
\item $(M,g,t)$ is isometric to $(\mathbb{R} \times \Sigma,g,t)$ with metric $g = −\beta dt^2 + g_t$, where $\beta \in C^\infty(M, \mathbb{R})$ is a smooth negative function, $g_t$ is a Riemannian metric on $\Sigma$ depending smoothly on $t \in \mathbb{R}$ and each $\{t\} \times \Sigma$ is a smooth spacelike Cauchy hypersurface in $M$.
\end{enumerate}
\label{thm:GlobHyp}
\end{thm}

The proof is due to A. Bernal and M. S\'anchez. See e.g. \cite{Bar:2009zzb}, Chapter 2, Theorem 1. The proof extends the work by Geroch \cite{Geroch:1970uw}, who showed that a globally hyperbolic manifold is homeomorphic to $\mathbb{R} \times \Sigma$. This result is strengthened to show that the $M$ and $\mathbb{R} \times \Sigma$ are diffeomorphic and, furthermore, that $(M,g)$ and $(\mathbb{R} \times \Sigma, −\beta dt^2 + g_t)$ are isometric.

Spacetimes which satisfy this property support fields with well-posed initial value problems. Throughout this chapter we will present the quantisation of fields in globally hyperbolic spacetimes. The quantisation of fields on non-globally hyperbolic manifolds has been studied in the context of boundary value problems. Notable examples include the Casimir effect \cite{Kay:1978zr, Dappiaggi:2014gea} and quantisation in Davies-Fulling mirror spacetimes, first introduced in \cite{Davies:1976hi, Davies:1977yv}. We shall encounter quantum fields in non-globally hyperbolic spacetimes in Chapters \ref{ch:DeCo} and \ref{ch:BH} and we shall be able to also perform the field quantisation in this context.

$n$-dimensional globally hyperbolic spacetimes, with $n \geq 2$,  can be understood as the dynamical evolution of $n-1$ Riemannian metrics. Such interpretation is provided by the ADM formalism, to which we turn our attention. We use abstract index notation:

Suppose one selects a time function $t \in C^\infty(M)$. Cauchy surfaces $\Sigma_t \subset M$ are defined by constant $t$ surfaces, and are normal to $g^{ab} \nabla_b t|_{t = \text{const}}$. We normalise this vector field to unity and call it $n^a$.  The metric induced on $\Sigma_t$ is defined by $h_{ab} = g_{ab} + n_a n_b$.

The 3-metric $h_{ab}$ defines a projector on the surface, $P_{(\Sigma_t)}{}^a{}_b \doteq g^{ac} h_{cb}$. The projector orthogonal to the surface is defined by $n^a$, $P_{(n)}{}^a{}_b \doteq - n^a n_b$. In this way, any vector field can be decomposed as $v^a = P_{(n)}{}^a{}_b v^b + P_{(\Sigma_t)}{}^a{}_b v^b$. 

These projectors define the important \textit{Lapse function}, $N$, and \textit{Shift vector}, $N^a$. We select a global timelike vector field $t^a \in TM$, which satisfies the condition $t^a \nabla_a t = 1$. The integral curves of $t^a$ represent the flow of time in $M$. The vector field $t^a$ is decomposed in its normal and tangential components to the Cauchy surface $\Sigma_t$ as $t^a = N n^a + N^a$, where
\begin{subequations}
\begin{align}
N n^a & \doteq P_{(n)}{}^a{}_b t^b =  \left(-n_b t^b\right) n^a, \\
N^a & \doteq P_{(\Sigma_t)}{}^a{}_b t^b = g^{ac} h_{cb} t^b.
\end{align}
\label{2:LapseShift}
\end{subequations}

Finally, let us discuss the situation for non-globally hyperbolic spacetimes. Many of the most important solutions in General Relativity are not globally hyperbolic. Notably, most members of the Kerr-Newman family are not globally hyperbolic. 

By Theorem \ref{thm:GlobHyp}, non-globally hyperbolic spacetimes contain no Cauchy hypersurfaces. Instead, to any achronal surface there will be an associated manifold subset called the \textit{domain of dependence}, but this domain of dependence cannot be the whole manifold. Let us make this geometrical notion more precise.

\begin{defn}
Let $(M,g,t)$ be a time-oriented spacetime. Let $K \subset M$ be a closed, achronal set. The \textit{future domain of dependence of} $K$, denoted by $D^+(K)$, is the set of points $p \in M$, such that every past-inextendible causal curve through $p$ crosses $K$.

The \textit{past domain of dependence of} $K$, denoted by $D^-(K)$, is the set of points $p \in M$, such that every future-inextendible causal curve through $p$ crosses $K$.

The \textit{domain of dependence of} $K$ is $D(K) \doteq D^+(K) \cup D^-(K)$.
\end{defn}

The boundary of the domain of dependence of an achronal surface is called a \textit{Cauchy horizon}, according to the following

\begin{defn}
Let $K \subset M$ be a closed, achronal set in $(M,g,t)$. The \textit{future Cauchy horizon of} $K$ is the achronal surface $H^+(K) \doteq \overline{D^+(K)} \setminus I^-[D^+(K)]$, where the overline denotes the set closure. 

The \textit{past Cauchy horizon of} $K$ is the achronal surface $H^-(K) \doteq \overline{D^-(K)} \setminus I^+[D^-(K)]$.

The \textit{Cauchy horizon of} $K$ is $H(K) \doteq H^+(K) \cup H^-(K)$.
\end{defn}

\begin{rem}
An achronal surface $\Sigma \subset M$ in the spacetime $(M,g,t)$ is a Cauchy hypersurface if and only if $H(\Sigma) = \emptyset$. Conversely, every achronal surface on a non-globally hyperbolic spacetime will contain a Cauchy horizon.
\end{rem}

We now proceed to introduce classical and quantum field theory in globally hyperbolic spacetimes, where the Cauchy problem is well-posed. In non-globally hyperbolic spacetimes, the same treatment can be applied inside the domain of dependence of an achronal surface with prescribed initial data, with the caveat that the dynamics cannot be extended past the Cauchy horizon.

\section{Classical field theory}
\label{2sec:Classical}

Classical field theory explains gravity, matter and their interactions at macroscopic scales. In this section, we shall review the kinematical and dynamical aspects of field theory. The kinematics describes the observables, classical field maps, understood as sections on vector bundles. The dynamics enters in the form of the equations of motion of the theory.

One usually describes classical field theory by writing an action functional on field configuration variables and the extremisation of this functional is equivalent to satisfying the equations of motion of the field theory. Prominent examples are the Einstein-Hilbert action, which yields the Einstein field equations in vacuum upon variation, the standard model, which describes the electromagnetic and nuclear forces, and the sum of the Einstein-Hilbert action\footnote{The coupling between fermions and gravity is best achieved using the \textit{vielbein} formulation of general relativity.} and the standard model action, which describes all of the fundamental interactions in nature. 

Action functionals are useful for physicists because they allow a theorist to postulate a field theory based on symmetries, which are believed to be fundamental in nature, but actions are not essential. A physical experiment is only sensitive to the equations of motion of the theory. For example, if one wishes to describe the theory of a free scalar field, $\phi$, with mass $m$ on a fixed Lorentzian spacetime, $(M,g)$, with Riemann curvature tensor, $\Riem$, one can postulate the action functional
\begin{equation}
S_\phi[\phi] = -\frac{1}{2} \int_M d\vol(\mathsf{x}) \left[g^{ab}(\mathsf{x}) \nabla_a \phi(\mathsf{x}) \nabla_b \phi(\mathsf{x}) + (m^2 +  \xi \R(\mathsf{x})) \phi(\mathsf{x})^2 \right],
\label{2:KGaction}
\end{equation}
subject to regularity conditions on the field, where the Ricci scalar, $\R = \tr(\Ric)= g^{ab}\Ric_{ab}$, is the trace of the Ricci tensor, built from the contraction of the Riemann tensor $\Ric_{ab} = \Riem^c{}_{acb}$, and with $\xi \in \mathbb{R}$. One can ensure that the theory is invariant under diffeomorphisms by using as an integration measure the generally covariant volume element defined locally by $d\vol(x) = dx^4 [-\det(g)(x)]^{1/2}$.

\begin{rem}
The orientation of our spacetime is provided by $d^4x = \epsilon dx^1 \wedge dx^2 \wedge dx^3 \wedge dx^4$. We choose $\epsilon = 1$ as an orientation.
\end{rem}

The variation of this action produces the Klein-Gordon equation for a free scalar field with mass $m$,
\begin{equation}
\frac{\delta S}{\delta \phi}= (\Box - m^2 - \xi \R) \phi = 0.
\label{2:KG}
\end{equation}

In abstract index notation, $\Box = g^{ab}\nabla_a \nabla_b$. Eq. \eqref{2:KG} dictates the dynamics of the field, and one can simply start from this equation without specifying an action.

If one wishes to study the dynamics between the field and gravity, one simply promotes $g$ to a dynamical field in the Klein-Gordon action \eqref{2:KGaction} and adds up the Einstein-Hilbert action. The action functional (with a cosmological constant, $\Lambda$, term) is
\begin{align}
S[g, \phi] & = \frac{1}{2}\int_M d\vol(\mathsf{x}) \left[\frac{1}{8 \pi G_\text{N}}(\text{R}(\mathsf{x})-2\Lambda)  \right. \nonumber \\
&  - g^{ab}(\mathsf{x}) \nabla_a \phi(\mathsf{x}) \nabla_b \phi(\mathsf{x}) - (m^2 +  \xi \R(\mathsf{x})) \phi(\mathsf{x})^2 \Big],
\end{align}
where $G_\text{N}$ is Newton's gravitational constant. Upon variation, one obtains the eleven partial differential equations that determine the metric and the field,
\begin{subequations}
\begin{align}
\frac{1}{(- \det(g))^{-1/2}} \frac{\delta S}{\delta g^{ab}} & = \frac{1}{16 \pi G_\text{N}}(G_{ab} + \Lambda g_{ab}) - \frac{1}{2}T_{ab} = 0, \\
\frac{1}{(- \det(g))^{-1/2}} \frac{\delta S}{\delta \phi} & = (\Box - m^2 - \xi \R) \phi = 0,
\label{2:g-KG}
\end{align}
\end{subequations}
\\
where $G \doteq \Ric - (1/2)g \R$ is the celebrated Einstein tensor and $T$ is the energy momentum tensor of $\phi$, $T \doteq 2(-\det(g))^{-1/2} \left(\delta S_\phi/\delta g^{-1}\right)$.

From now on we consider the metric field $g$ to be a background field and we further suppose that there is no back-reaction on the metric. In other words, we switch off gravity, $G_{\text{N}} \to 0$. (Because $G_\text{N}$ is a dimensionful quantity, the limit must be taken using the appropriate scales of the system. For example, if the coupled matter has mass $m$, then $G_\text{N}m^2$ is the dimensionless quantity that runs to zero.) In 
this way, we can focus on the essential structures of classical field theory by using a very simple model and, in turn, make the relation between the classical and quantum mathematical structures transparent.

\subsection{Covariant formulation}

The kinematics of classical field theory on a manifold $M$ can be formulated both covariantly and canonically. In the covariant approach, the relevant space of field configurations is the covariant configuration space, which consists of fields supported in spacetime regions. In the canonical approach, the relevant space of field configurations is the canonical configuration space, which consists of field configurations on a Cauchy surface of the spacetime equipped with a Riemannian metric. Additionally, one specifies the momenta of these canonical field configurations to define a canonical phase space. While the covariant configuration space and the canonical phase space are \textit{a priori} not the same, they lead to the same space of solutions, once the dynamics are imposed. Elements in the covariant and canonical space of solutions can be identified in a one-to-one correspondence.

We denote the covariant configuration space by $\Q$. The dynamics impose a condition on this configuration space and, when the field equations are satisfied, we say that the field is on-shell. The space of on-shell field configurations is the solution space of the field theory, denoted by $\Sol \subset \Q$.

Intuitively, the configuration space of fields is the space of smooth vector-valued functions over the spacetime manifold, $\Q = C^\infty(M, V)$.\footnote{While one needs not impose \textit{a priori} such strong regularity conditions, this does not affect the discussion at this level. We will not be worried about this problem at present, and we will trust our intuition in choosing $\Q = C^\infty(M, V)$.} In addition, we introduce the space of real valued functionals on these smooth functions, which are elements of the dual space $\Q^* = (C^\infty(M, V))^*$. Functionals will become relevant in the discussion of the field observables of the theory.

These functional spaces have a geometric interpretation in terms of sections on a vector bundle, which we denote $(F, \pi, M, V)$, with additional structure, depending on the spin and gauge symmetries of the field. 

In this chapter, it suffices for our purposes to introduce this geometric technology in the context of real scalar fields on an arbitrary fixed spacetime background.

Let us start by providing a useful definition:

\begin{defn}
A (smooth) vector bundle $(F, \pi, M, V)$ is a fibre bundle consisting of:
\begin{enumerate}
\item The total space, $F$, a smooth manifold.
\item The base manifold, $M$, a smooth manifold.
\item The fibre, $V$, a vector space and a smooth manifold.
\item The projection $\pi: F \to M$, which satisfies that, for $\mathsf{x} \in M$, $\pi^{-1}(\mathsf{x}) = V_\mathsf{x}$, the fibre over $\mathsf{x}$, is isomorphic to the fibre $V$.
\item The structure group, $GL(V)$, which is the general linear group acting on the left on $V$.
\end{enumerate}
Moreover, we have that
\begin{enumerate}
\item Local trivialisations, $\{\varphi_U\}$, satisfy that, for any coordinate chart $U \subset M$, $\varphi_U(U \times V) \to \pi^{-1}(U)$ is a diffeomorphism for which $U \ni \mathsf{x} = \pi \circ \varphi_U(\mathsf{x},v)$.
\item Transition functions, $G_{ij} \in GL(V)$, satisfy that, for any two coordinate charts $U_i \subset M \supset U_j$, with $U_i \cap U_j \neq \emptyset$, for every $\mathsf{x} \in U_i \cap U_j$, we have that $\varphi_i(\mathsf{x},v) = \varphi_i(\mathsf{x},G_{ij}v)$.
\end{enumerate}
\end{defn}

The relevant vector bundle for a real scalar field on a manifold $M$ is $(F, \pi, M, \mathbb{R})$, with structure group $GL(1,\mathbb{R}) \cong \mathbb{R}$. This vector bundle is called the line bundle. The next useful definition that we need is

\begin{defn}
A \textit{smooth section} on $M$ is the smooth map $\sigma: M  \to F$, which satisfies $\pi \circ \sigma = \text{id}_M$. The space of smooth sections on $M$ is denoted by $\Gamma^\infty(M,F)$. A local section on $U \subset M$ is the smooth map $\Gamma^\infty(U,F) \ni \sigma: U  \to F$.
\end{defn}

One can also define the spaces of smooth sections in terms of their support on $M$.

\begin{defn}
We denote:
\begin{itemize}
\item $\Gamma_0^\infty(M,F)$: the space of smooth sections of compact support on a manifold $M$.
\item $\Gamma_{\text{sc}}^\infty(M, F)$: the space of smooth sections of space-like compact support, i.e., whose support is contained in $J(K)$ where $K \subset M$ is a compact subset of $M$.
\item $\Gamma_{\text{f/p}}^\infty(M,F)$: the space of smooth sections of future/past compact support, i.e., if $f \in \Gamma_{\text{f/p}}^\infty(M,F)$, for any $p \in M$, $\supp(f) \cap J^{+/-}(p)$ is compact.
\end{itemize}
\end{defn}

In view of this definition, the configuration space of the scalar field, $Q = C^\infty(M, \mathbb{R})$, is isomorphic to $\Gamma^\infty(M, F)$, since a smooth function $\phi: M \to \mathbb{R}: \textsf{x} \mapsto \phi(\textsf{x})$ can be identified with a section $\sigma_\phi: M \to F: \textsf{x} \mapsto (\textsf{x}, \phi(\textsf{x}))$ and vice-versa.

The functionals, $Q^* \supset f:Q \to \mathbb{R}: \phi \mapsto \phi(f)$, are sections on $M$ valued in the dual vector bundle, i.e.,

\begin{defn}
The dual bundle of the fibre bundle $(F, \pi, M, V)$ is the fibre bundle $(F^*, \pi^*, M, V^*)$, in which the fibre $V^*$ is the dual space of $V$.
\end{defn}

In view of this definition, a functional $Q^* \ni f:\phi \mapsto \phi(f)$ is identified with the map $\Gamma(M,F^*) \ni \sigma_f: \sigma_\phi \mapsto \langle \sigma_f, \sigma_\phi \rangle$, where $\langle \,  , \, \rangle$ is a bilinear form that defines the pairing $\langle \,  , \, \rangle: ((C^\infty(M))^*, C^\infty(M)) \to \mathbb{R}$.

In the case of linear scalar field theory, in which the configuration space is comprised of field configurations of free, non-interacting fields, the space of functionals is linear, \textit{i.e.}, for $a \in \mathbb{R}$, we have that $\phi(f+a g) = \phi(f) + a\phi(g)$, and we can identify the space of linear functionals with $C_0^\infty(\mathbb{R})$, \textit{i.e.}, $\Gamma(M,F^*) \cong C_0^\infty(\mathbb{R})$ and field configurations, $\sigma_\phi: M \to F: \textsf{x} \mapsto (\textsf{x}, \phi(\textsf{x}))$, are mapped to the real numbers by functionals, $\sigma_f: M \to F^*: \textsf{x} \mapsto (\textsf{x}, f(\textsf{x}))$ with the bilinear form
\begin{equation}
\sigma_\phi \mapsto \langle \sigma_f, \sigma_\phi \rangle = \int_M \! d\vol \,  f(\textsf{x}) \phi(\textsf{x}) = \phi(f).
\label{2:CovClassObs}
\end{equation}

This concludes the discussion on the kinematics of the classical field theory. To implement the dynamics, we need to impose equations of motion and obtain the space of solutions of a real scalar field. The space of solutions, $\Sol$, contains so-called on-shell configurations of the field. The equation of motion that we are interested in solving is the Klein-Gordon equation,
\begin{equation}
P  \phi \doteq (\Box - m^2 - \xi \R) \phi = 0,
\label{2:Pdef}
\end{equation}
subject to appropriate initial conditions.

An important property of the Klein-Gordon operator in curved spacetimes is that it is a \textit{normally hyperbolic operator}.

\begin{defn}
A second order linear partial differential operator, $P: \Gamma(M,F) \to \Gamma(M,F)$, is called \textit{normally hyperbolic} if there exists a local trivialisation on $(M,g)$ such that, for $\sigma_\phi \in \Gamma(M,F)$,
\begin{equation}
P \phi = - \left( g^{ab} \partial_a \partial_b + A^a \partial_a + B \right) \phi,
\end{equation}
where $A^{a}$ and $B$ are smooth coefficients.
\end{defn}

Normally hyperbolic operators have the important characteristic that they are, morally speaking, invertible, although not in a strict sense. Let us make this statement more precise by defining the \textit{advanced} and \textit{retarded Green operators}.

\begin{defn}
Let $M$ be globally hyperbolic spacetime and $P: \Gamma(M,F) \to \Gamma(M,F)$ a normally hyperbolic operator. A linear operator, $E^{-/+}: \Gamma_0(M,F) \to \Gamma_{f/p}$, is called an {\rm advanced}/{\rm retarded} Green operator for $P$ if it satisfies
\begin{itemize}
\item $P E^{-/+} = Id_{\Gamma_0(M,F)}$
\item $E^{-/+} P|_{\Gamma_0(M,F)} = Id_{\Gamma_0(M,F)}$
\item $\supp \left(E^{-/+} \sigma_f \right) \subset J^{-/+}( \supp (\sigma_f) )$, \text{ for all } $\sigma_f \in \Gamma_0(M,F)$.
\end{itemize}
\end{defn}

We turn our attention to the inhomogeneous Klein-Gordon equation
\begin{equation}
P \phi = f.
\label{2:InhomKG}
\end{equation}

The following theorem allows us to solve eq. \eqref{2:InhomKG}:

\begin{thm}
Let $M$ be globally hyperbolic spacetime and $P: \Gamma(M,F) \to \Gamma(M,F)$ a normally hyperbolic operator. There exist unique advanced and retarded fundamental solutions to eq. \eqref{2:InhomKG}, $\phi = E^- f$ and $\phi = E^+ f$. We say that $P$ is Green hyperbolic.
\end{thm}

The proof is standard. See for example \cite{Bar:2007zz}, Chapter 3.

We now form the \textit{advanced-minus-retarded fundamental solution} to the homogeneous Klein-Gordon equation, by setting $E \doteq E^--E^+$, and solve the Klein-Gordon equation with $\phi = E f$, $f \in C_0^\infty(M,\mathbb{R})$. Moreover, every smooth, space-like compact solution is of this form, by the properties of the support of the fundamental solutions. Thus, the space of solutions is given by the linear surjection $E: C_0^\infty(M,\mathbb{R}) \to \Sol$. In other words,

\begin{equation}
\Sol =  \left\{ \phi = E f,\text{ such that } f \in C_0^\infty(M,\mathbb{R}) \right\}
\end{equation}

We have implicitly identified smooth sections and functions in the space of solutions by $E\sigma_f = \sigma_{Ef}$. The space of solutions is a symplectic space. Let $\sigma_\phi = E \sigma_f$ and $\sigma_\psi = E \sigma_g$. We define the symplectic structure, $\Omega: \Sol \times \Sol \to \mathbb{R}$ by
\begin{equation}
\Omega(\sigma_\phi, \sigma_\psi) \doteq \langle \sigma_{f}, \sigma_{E g} \rangle = \int_M \! d\vol \left(\x \right) \, f(\x) \int_M \! d\vol \left(\x' \right) \, E\left(\x,\x'\right) g(\x') \doteq E(f,g).
\label{2:CovSymp}
\end{equation}

Eq. \eqref{2:CovSymp} implies that the observables in the covariant picture, $\Sol \to \mathbb{R}$, are precisely defined by linear functionals of the form:
\begin{equation}
\sigma_f: \phi = Eg \mapsto E(f,g).
\label{2:CovObs}
\end{equation}

The symplectic structure defines a Poisson bracket for the observables, $\{ \cdot, \cdot \}: \Sol^* \times \Sol^* \to \mathbb{R}$, by
\begin{equation}
\left\{ \langle \sigma_f, \cdot \rangle, \langle \sigma_g, \cdot \rangle \right\} = - E(f,g).
\label{2:CovPoisson}
\end{equation}

The observables equipped with the Poisson bracket form the classical algebra of observables, or \textit{Poisson algebra}, $(\Sol^*, \{\cdot, \cdot \})$.

\begin{rem}
The Dirac equation is a first order linear differential equation, hence the Dirac operator is not normally hyperbolic. Nevertheless, it is a Green hyperbolic operator and we can form the fundamental solution as above.
\end{rem}

This concludes our discussion on the classical dynamics of the Klein-Gordon field from the covariant point of view.

\subsection{Canonical formulation}

Let us start by discussing the kinematical features of the theory of classical fields from the canonical point of view. In the canonical approach we define a \textit{Lagrangian density}, whose integral yields the action functional, because it allows one to obtain the canonical momenta in the canonical phase space of the theory. We consider the action \eqref{2:KGaction} for concreteness. For a free real scalar field \eqref{2:KG} on a fixed, oriented globally hyperbolic, smooth spacetime $(M,g,\epsilon,t)$, the covariant \textit{Lagrangian density}, $\mathscr{L}: Q \to Q$, is defined by
\begin{equation}
\mathscr{L}(\phi) \doteq -\frac{1}{2}\sqrt{-\det(g)} \left[g^{ab} \nabla_a \phi \nabla_b \phi + (m^2 +  \xi \R) \phi^2\right].
\end{equation}

On a globally hyperbolic spacetime, we seek to define the \textit{canonical configuration space} on a Cauchy surface. In abstract index notation, we choose a time function $t \in C^\infty(M)$ and a global timelike vector field $t^a \in TM$, and define the canonical configuration space, $Q_\Sigma$, on a Cauchy surface, $\Sigma_t$, defined by $t = \text{const}$ and normal to the unit timelike vector field, $n^a$. The surface is equipped with a Riemannian metric $h_{ab} \,$  induced by the metric of $(M,g)$, and hence $(\Sigma_t, h)$ is a Riemannian manifold. The covariant Lagrangian density induces a Lagrangian density on every slice $\Sigma_t \subset M$. Using the ADM decomposition $g_{ab} = -n_a n_b + h_{ab} = -N^{-2}(t_a - N_a)(t_b - N_b) + h_{ab}$, this Lagrangian density is defined by
\begin{align}
\mathscr{L}\left(\phi, \dot{\phi}\right) & = \sqrt{\det(h)}  \, N \left[ -\frac{1}{N^2}\dot{\phi}^2 + \frac{2}{N^2} \dot{\phi} \left(N^b \nabla_b \phi \right) \right. \nonumber \\
& \left. - \frac{1}{N^2} \left(N^a \nabla_a \phi \right)^2 + h^{ab}  \nabla_a \phi \nabla_b \phi + \left(m^2 +  \xi \R \right) \phi^2 \right]
\end{align}
where we have defined the field velocity, $\dot{\phi}$, by $t^a \nabla_a \phi$. Because the action of the theory is given by the integral of the Lagrangian density,
\begin{equation}
S_\phi \left[\phi, \dot{\phi}\right] = \int_\mathbb{R} \! dt \, \int_{\Sigma_t} \! d^3\x \, \mathscr{L}\left(\phi, \dot{\phi}\right),
\end{equation}
it suffices to consider $C_0^\infty \left( \Sigma_t \right)$ as the canonical configuration space to ensure that the integral exists. The \textit{velocity phase space} of the theory is $C_0^\infty\left( \Sigma_t \right) \times C_0^\infty\left( \Sigma_t \right)$. 
\begin{defn}

Let $S$ be an action functional with configuration $\phi$ and velocity $\dot{\phi}$, the canonical momentum density is defined by $\pi \doteq \delta S/\delta \dot{\phi}$.
\end{defn}

The canonical momentum of $S_\phi$ is
\begin{equation}
\pi = \frac{\delta S_\phi}{\delta \dot{\phi}} = \frac{\sqrt{h}}{N} \left(\dot{\phi} - N^a \nabla_a \phi \right) = \sqrt{h} n^a \nabla_a \phi.
\end{equation}

The canonical phase space of the theory is $C_0^\infty\left( \Sigma_t \right) \times C_0^\infty\left( \Sigma_t \right)$ consisting of couples of fields and momenta, $\left(\phi, \pi \right)$, on $\Sigma_t$.

This concludes our discussion on the kinematics in the canonical approach.

We proceed to implement the dynamics. The space of solutions, $\Sol$, contains so-called on-shell configurations of the field. For a free real scalar field \eqref{2:KG} on a fixed, oriented, globally hyperbolic, smooth spacetime $(M,g,\epsilon,t)$, we consider the initial value problem of \eqref{2:g-KG} on a Cauchy surface $\Sigma_t \subset M \cong \mathbb{R} \times \Sigma_t$ with unit normal vector field $n \in TM$.

Let $\varphi, \pi \in C_0^\infty(\Sigma_t)$ be initial data on $\Sigma_t$. The solution to the initial value problem,
\begin{subequations}
\begin{align}
&P  \phi \doteq (\Box - m^2 - \xi \R) \phi = 0,  & \text{on } M, & \\
&\phi  = \varphi, & \text{on } \Sigma_t, &\\
&\partial_n \phi  = \pi, & \text{on } \Sigma_t, &
\end{align}
\label{2:IVP}
\end{subequations}
\\
specifies the field values on all of $M$. Following \cite{Wald:1995yp} (Theorem 4.1.2),

\begin{thm}
Let $(M,g,\epsilon,t)$ be globally hyperbolic with smooth spacelike Cauchy surface $\Sigma_t$. Then the initial value problem \eqref{2:IVP} is well-posed, i.e., there exists a unique solution that depends only upon data on $\Sigma_t$, is smooth and varies continuously with the initial data with respect to a suitable Sobolev topology defined on the initial data.
\end{thm}

The content of the theorem is that all smooth solutions are in one-to-one correspondence with a set of initial value data $(\varphi, \pi)$ on a Cauchy surface, $\Sigma_t$. We call $S: C_0^\infty(\Sigma_t) \times C_0^\infty(\Sigma_t) \to C_\text{sc}^\infty(M)$ the isomorphism that relates the solutions in the canonical and covariant picture: $S(\varphi, \pi) = \phi \in \Sol$.

The space of solutions carries the symplectic structure, $\Omega: \Sol \times \Sol \to \mathbb{R}$, defined by
\begin{equation}
\Omega((\varphi_1,\pi_1),(\varphi_2,\pi_2)) \doteq \int_{\Sigma_t} \! d^3\x \, (\pi_1 \varphi_2 - \pi_2 \varphi_1).
\label{2:CanSymp}
\end{equation}

The Poisson structure of the theory is defined from the symplectic structure. The observables are maps of the form $\Omega((f,g), \cdot): \Sol \to \mathbb{R}$. The canonical position and momentum observables are
\begin{subequations}
\begin{align}
\Omega((0,f \sqrt{h}), \cdot): (\varphi, \pi) & \mapsto \int_{\Sigma_t} \! d\vol_{\Sigma_t}(\textsf{x}) \, f(\textsf{x})  \varphi(\textsf{x}) = \varphi(f), \\
\Omega((g,0), \cdot): (\varphi, \pi) & \mapsto \int_{\Sigma_t} \! d\vol_{\Sigma_t}(\textsf{x}) \, g(\textsf{x})  n^a \nabla_a \phi(\textsf{x}) = \pi(g).
\end{align}
\label{2:CanClassObs}
\end{subequations}
\\
and the Poisson bracket, $\{ \cdot, \cdot \}: \Sol^* \times \Sol^* \to \mathbb{R}$,
\begin{equation}
\left\{\Omega((\varphi_1,\pi_1), \cdot), \Omega((\varphi_2,\pi_2), \cdot) \right\} \doteq - \Omega((\varphi_1,\pi_1),(\varphi_2,\pi_2)),
\end{equation}
yields the familiar expression
\begin{equation}
\{\varphi(f), \pi(g)\} = \int_{\Sigma_t} d\vol_{\Sigma_t} fg.
\end{equation}

We now verify that the canonical and covariant pictures give rise to the same space of solutions and are, hence, equivalent. We show the equivalence of eq. \eqref{2:CovSymp} and \eqref{2:CanSymp}.

\begin{thm}
Let $(M,g,\epsilon,t)$ be globally hyperbolic. Let $\phi_1 = Ef$ be the solution to the Klein-Gordon equation with initial data $(\varphi_1, \pi_1)$ on the Cauchy surface $\Sigma \subset M$ and $\phi_2 = Eg$ be the solution with initial data $(\varphi_2, \pi_2)$. Then,
\begin{equation}
\Omega((\varphi_1,\pi_1), (\varphi_2,\pi_2)) = \int_\Sigma \! d^3\x \, (\pi_1 \varphi_2 - \pi_2 \varphi_1) = E(f,g) = \Omega(\sigma_{\phi_1}, \sigma_{\phi_2}).
\end{equation}
\end{thm}
\begin{proof}
Let us choose a Cauchy surface $\Sigma$ such that the support of $f \in C_0^\infty$ is $\supp(f) \subset J^-(\Sigma)$. By the integration theorem of Stokes,
\begin{align}
\Omega((\varphi_1,\pi_1), (\varphi_2,\pi_2)) & = \int_\Sigma \! d^3\x \, (\pi_1 \varphi_2 - \pi_2 \varphi_1)  \nonumber \\
& = \int_\Sigma \! d^3\x \sqrt{h} n^a \, (\varphi_2 \nabla_a \varphi_1  - \varphi_1 \nabla_a \varphi_2 ) \nonumber \\
& = \int_\Sigma \! d\vol_\Sigma^a \, (\varphi_2 \nabla_a \varphi_1  - \varphi_1 \nabla_a \varphi_2 ) \nonumber \\
& = \int_{J^-(\Sigma)} \! d\vol \nabla^a \, (\phi_2 \nabla_a \phi_1  - \phi_1 \nabla_a \phi_2 ) \nonumber \\
& = \int_{J^-(\Sigma)} \! d\vol \, ( \phi_2 \Box \phi_1  - \phi_1 \Box \phi_2 ),
\end{align}
where we have written the normal volume element on $\Sigma$ as $d\vol_\Sigma^a = d^3\x \sqrt{h} n^a$ before applying Stokes' theorem. Writing $\phi_1 = E^-f$, and using the field equation, $P \phi_2 = \left(\Box - m^2 - \xi R\right) \phi_2 = 0$, 
\begin{align}
\Omega((\varphi_1,\pi_1), (\varphi_2,\pi_2)) & = \int_{J^-(\Sigma)} \! d\vol \, \left( \phi_2 P \left(E^-f\right)  - \left(E^-f\right) P \phi_2 \right) \nonumber \\
& = \int_{J^-(\Sigma)} \! d\vol  \, \phi_2  ( P E^-f) = \int_{J^-(\Sigma)} \! d\vol  \, f (Eg) \nonumber \\
& = E(f,g) = \Omega(\sigma_{\phi_1}, \sigma_{\phi_2}).
\end{align}

\end{proof}

This concludes our discussion on the classical field theory.

\section{Quantum field theory}
\label{2sec:Quantum}

In this section, we describe how to obtain a quantum field theory from a classical theory in curved spacetimes. The procedure by which this is achieved is called \textit{quantisation}, whereby we describe the microscopic structure of fields. We review the quantisation procedure for the real, free scalar field of Section \ref{2sec:Classical}. In this way, the main steps in the quantisation procedure are transparent. We emphasise the distinction between the kinematics and the dynamics of the quantum field, both in the covariant and canonical pictures.

We have analysed the classical dynamics of a linear field theory. Such theories are called \textit{free} because the equations of motion contain no field interaction (or self-interaction) terms. A word on free theories is due: Strictly speaking, free fields interact only with the gravitational field, but in many cases the assumption that a field is free is in good agreement with nature, when interactions can be neglected. Moreover, they are very important as they are the starting point for perturbative interaction theories. For example, in quantum electrodynamics, the interacting theory takes place in the Fock space $\mathscr{F}(\mathcal{H}_\text{int})$, which we define below, but in scattering processes, the asymptotic states of the theory lie in the tensor product state space of the free photon and free electron fields, $\mathscr{F}_s(\mathcal{H}_{\gamma, \text{ in}}) \otimes \mathscr{F}_a (\mathcal{H}_{\text{e, in}})$ in the asymptotic past and $\mathscr{F}_s(\mathcal{H}_{\gamma, \text{ out}}) \otimes \mathscr{F}_a(\mathcal{H}_{\text{e, out}})$ in the asymptotic future. There exist many processes, such as the formation of bound states, that cannot be described in these terms, but fields do interact very weakly in asymptotic regions in many cases.

We discuss the quantum field theory of the real, free scalar field covariantly in Subsection \ref{2:subsecQCov}, and we review the canonical approach in Subsection \ref{2:subsecQCan}.

\subsection{Covariant formulation}
\label{2:subsecQCov}

We shall start our discussion by constructing the algebra of quantum observables. Then, we shall introduce states as maps from the algebra of observables to the complex numbers. Finally, we shall comment on the \textit{locality} and \textit{covariance} properties of relativistic fields.

\subsubsection{Quantum observables}

The quantisation procedure of Dirac consists of promoting the Poisson algebra of observables, $(\Sol^*, \{\cdot, \cdot \})$, to a quantum $^*$\textit{-algebra}, with a commutator, $[\cdot, \cdot]$, promoted from the Poisson bracket, and a set of additional properties, defined axiomatically, depending on the details of the field. 

\begin{defn}
An \textit{algebra} $\mathscr{A}$ over a field $F$ is a vector space over $F$ equipped with a bilinear form $\cdot: \mathscr{A} \times \mathscr{A} \to \mathscr{A}$, which defines an associative multiplication. Moreover, an algebra $\mathscr{A}$ is said to be \textit{unital} if there exists an identity element $I \in \mathscr{A}$. Furthermore, it is said to be a $^*$\textit{-algebra}, if there exists an involution $^*: \mathscr{A} \to \mathscr{A}$.\footnote{More generally, an algebra can be defined over a ring, but this is not necessary for our purposes.}
\end{defn}

For a Klein-Gordon field, in the covariant approach, we have that the classical observables of the theory are maps from the solution space to the real numbers, e.g., $f: \Sol \to \mathbb{R}: \phi = Eg \mapsto E(f,g)$, for $f, g \in C_0^\infty(M,\mathbb{R})$. The observables of the theory are equipped with a Poisson bracket $\{f,g\} = -E(f,g)$. See eq. \eqref{2:CovObs} and \eqref{2:CovPoisson}. 

In order to accommodate the involution in the quantum algebra of observables, we will extend the classical space of observables, $\Sol^* = C_0^\infty(M,\mathbb{R})$, to admit complex-valued test functions, by extending the map $f: \Sol \to \mathbb{C}: \phi = Eg \mapsto E(f,g)$ by complex linearity (and anti-linearity). $\Sol^*_\mathbb{C} = C_0^\infty(M,\mathbb{C})$ is the complex extension. We are ready to define the quantum field theory of a Klein-Gordon field:

\begin{defn}
Let $(M,g,\epsilon,t)$ be a globally hyperbolic spacetime. A Klein-Gordon quantum field is a map $\Sol^*_\mathbb{C} \to \mathscr{A}(M): f \mapsto \Phi(f)$, where $\mathscr{A}(M)$ is the neutral Klein-Gordon field algebra. This algebra is generated by elements $\Phi(f)$ satisfying the quantisation axioms:

\begin{enumerate}
\item Linearity: $f \mapsto \Phi(f)$ is a linear map.
\item Hermiticity: $\Phi^*(f) = \Phi(\bar{f})$.
\item Field equation: $\Phi(Pf) = \Phi((\Box - m^2 - \xi \R) f) = 0$.
\item CCR: $\left[\Phi(f), \Phi(g) \right] \doteq \Phi(f) \Phi(g) - \Phi(g) \Phi(f)  = -\ii \hbar E(f,g)I$.
\end{enumerate}
\label{2:KGaxioms}
\end{defn}

In axiom 4 CCR stands for canonical commutator relations and $I$ is the algebra unit. From this point onward, we shall set $\hbar = 1$. The CCR encode the bosonic character of the field. For fermions, the commutator in axiom 4 is replaced by the anti-commutator, usually denoted in the literature by $\{\cdot, \cdot\}$ (not to be confused with the Poisson bracket) or $[\cdot, \cdot]_+$, replacing the negative sign in the definition by a positive sign. For a charged field, axiom 2 needs to be modified. The field dynamics are encoded in axiom 3.

We would like to point out that in the axiomatic approaches to quantum field theory, several axioms may be added to strengthen the definition that we have provided. Of particular importance, the \textit{locally covariant quantum field theory} axioms further require
\begin{enumerate}
\item Isotony: Let $N \subset M$, then $\mathscr{A}(N) \subset \mathscr{A}(M)$
\item Einstein causality: Let $O_1$ and $O_2$ be causally disjoint subsets of $M$, then $[\mathscr{A}(O_1), \mathscr{A}(O_2)] = \{0\}$.
\item Time-slice: Let $\Sigma \subset N \subset M$ be a Cauchy surface, then $\mathscr{A}(N) \cong \mathscr{A}(M)$.
\end{enumerate}

We shall not dwell into the extensive subject of axiomatic quantum field theory, but some commentary is due. The Wightman axioms postulate the first attempt at rigorously defining a quantum field theory in flat spacetime. The subsequent Haag-Kastler axioms first introduced the formalisation of fields as operator algebras, and emphasised the role of locality in quantum field theory \cite{Haag:1992hx}. The axiomatic approach was extended to curved spacetimes by Brunetti, Fredenhagen and Verch, emphasising the locality and covariance of the theory \cite{Brunetti:2001dx}. The flavour of this last construction is much in the spirit of \textit{category theory} \cite{MacLane:1998}, where quantisation is understood as a \textit{functor} between relevant categories and smeared fields are \textit{natural transformations}. Recently, Hollands and Wald took an axiomatic \textit{operator product expansion} point of view for defining quantum field theories in curved spacetimes \cite{Hollands:2008vx}.

The hope, which we shall realise below, is that one can construct field states on which representations of the algebra of observables act as operators. These states, in turn, provide the probabilistic interpretation of quantum field theory. We turn now to this question.

\subsubsection*{Field states}

There is a prodecure, due to Gelfand, Naimark and Segal, known as the \textit{GNS construction}, by which one can construct a Hilbert space and linear operators acting on it out of a $^*$-algebra equipped with an \textit{algebraic state}, by which we mean the following:

\begin{defn}
Let $\mathscr{A}$ be a unital $^*$-algebra. An \textit{algebraic state on} $\mathscr{A}$ is a linear map $\omega: \mathscr{A} \to \mathbb{C}$, satisfying
\begin{enumerate}
\item Normalisation: $\omega(I) = 1$.
\item Positivity: $\omega(A A^*) \geq 0$ for any $A \in \mathscr{A}$.
\end{enumerate}

Further, we require that the state be determined fully by the knowledge of all the field \textit{n-point functions}: Let $(A_1, A_2, \ldots, A_n) \subset \mathscr{A}$ be a collection of $n$ elements in $\mathscr{A}$. Then the state $\omega$ is determined by the knowledge of all the n-point functions of the form $\omega(A_1  A_2  \cdots  A_n)$ for all $n \in \mathbb{N}$.
\end{defn}

\begin{rem}
States that are fully determined by the two-point functions are called {\rm quasi-free} or Gaussian states. Vacuum states and thermal states are quasi-free.
\end{rem}

The idea is to construct a norm for the field algebra, $\mathscr{A}(M)$, out of which the kinematic Hilbert space of the theory can be constructed. This norm should turn the field algebra into a  \textit{Banach $^*$-algebra}.

\begin{defn}
An algebra $\mathscr{A}$ is called a \textit{Banach algebra} if it is an algebra over a normed vector space, with norm $|| \cdot ||$, such that for any $A, B \in \mathscr{A}$, $||A B|| \leq ||A|| ||B||$. Moreover, it is called a \textit{Banach $^*$-algebra} if it has an involution, such that $||A|| = ||A^*||$.
 
A $^*$-algebra is called a \textit{C$^*$-algebra} if it is a Banach $^*$-algebra and it satisfies that, for any $A \in \mathscr{A}$, $||A A^*|| = ||A||^2 = ||A^*||^2$.
\end{defn}

It is clear that $\omega$ does not provide the necessary structure to endow the field algebra with a norm. Indeed, we could realise $\omega$ as an inner product if and only if the strict equality in the positivity condition above held only for $A = 0$. Alas, this is not the case. Yet, algebraic states add structure in the right direction. For example, the following lemma can be shown: 

\begin{lem}
Let $\mathscr{A}$ be a $^*$-algebra and $\omega: \mathscr{A} \to \mathbb{C}$ an algebraic state. Then, $\omega(A B^*) = \overline{\omega(B^* A)}$ and the Cauchy-Schwarz inequality holds, $|\omega(A^*B)|^2 \leq \omega(A^*A) \omega (B^* B)$.
\end{lem}

The proof makes use of the positivity of the state and one proceeds with the standard projective decomposition of the vector $A$ into orthogonal and parallel components to $B$ to attain the Cauchy-Schwarz inequality.

The next mathematical structure that we need to introduce is the notion of an \textit{algebraic representation}.

\begin{defn}
Let $\mathscr{A}$ be a unital $^*$-algebra, a \textit{representation} of $\mathscr{A}$ is a map $\pi: \mathscr{A} \to \mathcal{L}(\mathcal{H})$, where $\mathcal{H}$ is a Hilbert space and $\pi$ is a $^*$\textit{-homomorphism} between the $^*$-algebra and the space of \textit{linear operators on} $\mathcal{H}$, \textit{i.e.}, a linear map that satisfies, for any $A, B \in \mathscr{A}$, $\pi(AB) = \pi(A) \pi(B)$ and $\pi\left(A^*\right) = \pi(A)^*$. 

Moreover, if $\mathscr{A}$ is a unital C$^*$-algebra, the representation maps elements of the algebra to \textit{bounded linear operators}, \textit{i.e.}, $\pi: \mathscr{A} \to \mathcal{BL}(\mathcal{H})$, such that, for any $v \in \mathcal{H}$ and any $A \in \mathscr{A}$, $||\pi(A)v||_\mathcal{H} \leq C||v||_\mathcal{H}$, where $C$ is a constant.
\label{2:Repn}
\end{defn}

Morally speaking, the GNS construction endows a $^*$-algebra with a Banach norm, and by Definition \ref{2:Repn}, a representation $\pi \mathscr{A}: \mathcal{H} \to \mathcal{H}$ realises the construction as operators on the Hilbert space of the theory. Let us state this precisely:

\begin{thm}[Gelfand-Naimark-Segal construction]
Let $\mathscr{A}$ be a $^*$-algebra and $\omega: \mathscr{A} \to \mathbb{C}$ an algebraic state. Then, there exists a quadruple $(\mathcal{H}_\omega, \mathcal{D}_\omega, |0\rangle_\omega, \pi_\omega)$, such that $\mathcal{H}_\omega$ is a Hilbert space, $\mathcal{D}_\omega \subset \mathcal{H}_\omega$ is a dense subspace, $\pi_\omega: \mathscr{A} \to \mathcal{L}(\mathcal{H}_\omega)$ is a representation of the algebra elements as linear operators on the Hilbert space, $|0\rangle_\omega$ is a \textit{cyclic vector}, \textit{i.e.}, a vector such that $\{ \pi_\omega(A) |0\rangle_\omega, \text{ for all } A \in \mathscr{A} \} = \mathcal{D}_\omega$, and $\omega(A) = \langle 0 | \pi_\omega(A) |0\rangle_\omega$ for every $A \in \mathscr{A}(M)$.
\label{2:GNS}
\end{thm}
\begin{proof}
We start by endowing the algebra $\mathscr{A}$ with an inner product. Define the pairing $(A,B)_\omega \doteq \omega(A^* B)$ for any $A, B \in \mathscr{A}$. We claim that the set defined by
\begin{equation}
\mathscr{I}_\omega(M) \doteq \left\{A \in \mathscr{A}(M) \text{ such that } (A,A)_\omega = 0 \right\}
\end{equation}
is a \textit{left ideal of} $\mathscr{A}$, \textit{i.e.}, it satisfies that, for any $A \in \mathscr{A}$ and any $B \in \mathscr{I}_\omega$, the product $AB \in \mathscr{I}_\omega$. This is verified by the Cauchy-Schwarz inequality,
\begin{equation}
(AB,AB)_\omega = \omega\left(B^*A^*AB \right) \leq \omega\left(B B^*\right)\omega\left(\left(A^*AB\right)^*A^*AB \right) = 0.
\end{equation}

Moreover, $\mathscr{I}_\omega$ form a vector space, \textit{i.e.}, it is closed under scalar multiplication and vector addition, \text{i.e.}, for $A, B \in \mathscr{I}_\omega$, and $a,b \in \mathbb{C}$, $(a A + b B, a A + b B)_\omega = 0$. Again, this is guaranteed by the Cauchy-Schwarz inequality. Similarly, $\mathscr{I}_\omega^*$ is a right ideal and a vector space.

Next, we we define the vector space quotient $\mathcal{D}_\omega = \mathscr{A}/\mathscr{I}_\omega$ with elements $[A]$ defined by the equivalence relation $A \sim A'$ if $A = A' + B$ with $B \in \mathscr{I}_\omega$. We construct the inner product space $(\mathcal{D}_\omega, ([A],[B])_\mathcal{D})$ by endowing the quotient space $\mathcal{D}_\omega$ with the scalar product $([A],[B])_\mathcal{D} = (A,B)_{\omega}$. Finally, the Hilbert space completion\footnote{See, for example \cite{Reed:1981}.} of $([A],[B])_\mathcal{D} = (A,B)_{\omega}$ yields the kinematic Hilbert space of the theory, $\mathcal{H}_\omega$.

The algebra elements act on the Hilbert space $\mathcal{H}_\omega$ through the map $\pi_\omega: \mathscr{A} \to \mathcal{L}(\mathcal{H}_\omega)$,
\begin{equation}
\pi\left( A \right) \left[B \right] = \left[AB\right],
\end{equation}
which is a $^*$-homomorphism by the associative property of the algebra. Thus, $\pi$ is a representation of $\mathscr{A}$ into a set of (generally unbounded) linear operators that are defined on $\mathcal{D}_\omega \subset \mathcal{H}_\omega$. Finally, we can choose the cyclic vector using the unit element of the algebra, $|0\rangle_\omega = [I]$, such that
\begin{equation}
\omega(A) = ([I],[A])_\omega = ([I],\pi(A)[I])_\omega \doteq \langle 0 | \pi_\omega(A) 0 \rangle_\omega.
\end{equation}

This concludes the GNS construction with quadruple $(\mathcal{H}_\omega, \mathcal{D}_\omega, |0\rangle_\omega, \pi_\omega)$.
\end{proof}

To perform the GNS construction of the quantum field, we can take advantage of the Weyl formulation of the Klein-Gordon field.

\begin{defn}
The Klein-Gordon Weyl algebra,  $\mathscr{W}(M)$, is the unique, unital $^*$-algebra generated by elements $f \mapsto W(Ef)$, with $f \in \Sol^*_\mathbb{C}$, together with the relations
\begin{enumerate}
\item $W(0) = I$.
\item $W^*(Ef) = W(-Ef)$.
\item $W(Ef) W\left(Eg\right) = \exp\left(\frac{\ii}{2} E(f,g)\right) W\left( E(f+g) \right)$ for any $Ef, Eg \in \Sol_\mathbb{C}$.
\end{enumerate}
\label{2:Weyl}
\end{defn}

While the Weyl algebra does not yield the physical intuition that the field algebra does, it has the advantage that, upon choosing an algebraic state, it can be realised as a space of bounded linear operators acting on a Hilbert space using the GNS construction, which simplify the domain issues in the Hilbert space.

We define a state $\omega: W(Ef) \mapsto \mathbb{C}$. Then
\begin{equation}
(W(Ef),W(Eg))_\omega = \omega(W^*(Ef),W(Eg)) = \exp\left(-\frac{\ii}{2} E(f,g)\right) \omega \left( W\left( E(-f+g) \right) \right)
\end{equation}
already defines an inner product because the Weyl algebra is unitary in the sense that $(W(Ef),W(Ef))_\omega = 1$ for all $f \in \Sol^*_\mathbb{C}$. The Hilbert space, $\mathcal{H}_\omega$,is the completion of the Banach $^*$-algebra $(\mathscr{W}(M), (\cdot,\cdot)_\omega)$. A representation $\pi: \mathscr{W}(M) \to \mathcal{BL}(\mathcal{H}_\omega)$ is given by the left action of algebra elements. It is bounded because the Weyl algebra is unitary. The cyclic vector $|0\rangle_\omega = I$ defines
\begin{equation}
\langle 0| \pi(W(Ef)) | 0\rangle_\omega = \omega(W(Ef)).
\end{equation}

The two-point function takes the neat form
\begin{equation}
\langle 0| \pi(W(Ef))\pi(W(Eg)) | 0\rangle_\omega = \exp\left(-\frac{\ii}{2} E(f,g)\right) \omega \left( W\left( E(-f+g) \right) \right).
\end{equation}

Finally, we obtain a representation, $\pi'$, on the field algebra by using the relation
\begin{equation}
\pi(W(Ef)) = \exp\left(-\ii \pi'\Phi(f) \right).
\end{equation}

We emphasise that the GNS construction starts from considering the algebra of observables associated to a spacetime region and a linear functional from this algebra to the complex numbers. In Minkowski space, for example, selecting the algebra of observables of the whole spacetime and imposing Poincar\'e covariance, one arrives to the Minkowski vacuum state. If one considers simply, say, the right wedge of the Minkowski spacetime, this is a globally hyperbolic spacetime on its own right, and the GNS construction can be carried out starting from the algebra of observables confined to this wedge, whereby one obtains the boost-invariant Rindler vacuum by a similar procedure.

\subsection{Canonical formulation}
\label{2:subsecQCan}

The canonical formulation of quantum field theories has the advantage that the approach is more constructive. We start by defining the Klein-Gordon canonical quantum field theory by applying the Dirac quantisation to the canonical Poisson algebra of observables, and we then proceed to construct the field states.

\subsubsection*{Quantum observables}

In the canonical approach, the algebra of observables is constructed from the quantum field configurations and the momenta on a fixed spacelike slice. This leads to the so-called equal-time quantum field algebra.

\begin{defn}
Let $(M,g,\epsilon,t)$ be a globally hyperbolic spacetime and $\Sigma_t \subset M$ a Cauchy hypersurface of $M$ defined by constant $t$. The Klein-Gordon quantum observables are maps $\Sol^*_\mathbb{C} \to \mathscr{A}(\Sigma_t)$, where $\mathscr{A}(\Sigma_t)$ is the algebra generated by elements of the form $\varPhi(f)$ and $\Pi(g)$ satisfying the canonical quantisation axioms:

\begin{enumerate}
\item Linearity: $f \mapsto \varPhi(f)$ and $g \mapsto \Pi(g)$ are linear maps.
\item Hermiticity: $\varPhi^*(f) = \varPhi(\bar{f})$ and $\Pi^*(g) = \Pi(\bar{g})$.
\item Canonical commutation relations: 
\begin{subequations}
\begin{align}
\left[\varPhi(f), \Pi(g) \right]& = \ii \left(\int_{\Sigma_t} d\vol_{\Sigma_t} fg\right) I &\text{ on } \Sigma_t, \\
\left[\varPhi(f), \varPhi(g) \right]& = 0 &\text{ on } \Sigma_t, \\
\left[\Pi(f), \Pi(g) \right] & = 0 &\text{ on } \Sigma_t.
\end{align}
\end{subequations}
\end{enumerate}
\label{2:KGCanonicalaxioms}
\end{defn}

The usual \textit{creation} and \textit{annihilation operators} are defined in terms of the field and momenta operators. For the moment, they should be regarded as an abstract CCR algebra, and they will be realised as operator-valued distributions once we define the Hilbert space of the theory. Let $a \doteq (-\ii \varPhi + \Pi)/\sqrt{2}$ and $a^* \doteq (\ii \varPhi + \Pi)/\sqrt{2}$ be the annihilation and creation operators respectively, then the CCR reads
\begin{subequations}
\begin{align}
\left[a(f), a^*(g) \right]& = \left(\int_{\Sigma_t} d\vol_{\Sigma_t} fg\right) I &\text{ on } \Sigma_t, \\
\left[a(f), a(g) \right] & = 0 &\text{ on } \Sigma_t, \\
\left[a^*(f), a^*(g) \right] & = 0 &\text{ on } \Sigma_t.
\end{align}
\label{2:CCRaa*}

\end{subequations}

We proceed to construct the states of the theory.

\subsubsection*{Field states}

As we have seen above, the classical space of solutions, $\Sol$ is an infinite-dimensional symplectic manifold with symplectic structure, $\Omega$, defined by eq. \eqref{2:CanSymp}. The idea is, as usual, to complexify the space of solutions, and construct a positive inner product, by selecting a ``space of positive frequency", that leads to the Hilbert space of the theory.\footnote{Introducing a complex structure and selecting positive frequency solutions is not the only way in which the solution space can be endowed with an inner product. There exist inner products that may not be attainable from our construction. See, for example, the discussion in \cite{Wald:1995yp}.}

A complexification of Sol is attained by introducing a complex structure $J:\Sol \times \Sol \to \Sol \times \Sol$, and identifying the complexified space of solutions as $\Sol_\mathbb{C} \cong \Sol \times \Sol$. This procedure is non-unique. Indeed, the choice of a positive frequency space of solutions in the mode-sum picture comes with the choice of a complex structure, whereby the Lie-derivation of a plane-wave field mode with respect to a timelike vector field induces the action of a complex structure on the field mode.\footnote{See Chapter \ref{ch:BH}.} A selection of the complex structure can be made following physically-motivated criteria, for example, the energy minimisation condition, proposed by Ashtekar and Magnon, is suitable for stationary spacetimes where energy along spacelike slices is conserved \cite{Ashtekar:1975zn}. In general situations, for non-stationary spacetimes, the complex structure need not stem from a notion of positive frequency. See \textit{e.g.} \cite{Deutsch:1984jw} also in the context of energy minimisation.

We now introduce the \textit{holomorphic} and \textit{antiholomorphic} decomposition of the complexified space of solutions. The complex structure, on $\Sol_\mathbb{C}$, $J:\Sol_\mathbb{C} \to \Sol_\mathbb{C}$ decomposes $\Sol_\mathbb{C}$ into its holomorphic and antiholomorphic parts as follows: We perform a partition of unity , $I = I_+ + I_-$, where $I_+ \doteq (I-\ii J)/2$ and $I_- \doteq (I + \ii J)/2$. Then, the space of solutions decomposes into $\Sol_\mathbb{C} \cong \Sol_+ \times \Sol_-$, where $\Sol_+ \doteq \ker (J-\ii I)$ and $\Sol_- \doteq \ker (J + \ii I)$. 

The complex conjugation map $\phi \mapsto \overline{\phi}$ is a bijection between $\Sol_+$ and $\Sol_-$, hence, it is customary to identify $\Sol_- = \overline{\Sol_+}$. Under this identification, the complex structure defines an anti-involution in $\Sol_\mathbb{C}$, $J(\phi, \bar{\phi}) = (\, \ii \phi, \, \overline{\ii \phi} \,)$, for $\phi \in \Sol_+$ and $\overline{\phi} \in \Sol_-$.

The structure that we have introduced defines a complexified \textit{pseudo-K\"ahler space}.

\begin{defn}
A \textit{pseudo-K\"ahler space} is the quadruple $(V, \Omega, J, \nu)$, where
\begin{enumerate}
\item $V$ is a real vector space,
\item $\Omega$ is a symplectic form,
\item $J$ is an anti-involution,
\item $\nu(u,v) \doteq \Omega (u, Jv)$ is a non-degenerate symmetric form, for $u, v \in V$
\end{enumerate}

If $\nu$ is positive definite, then $(V, \Omega, J, \nu)$ is a \textit{K\"ahler space}.
\end{defn}

\begin{rem}
Item 4 in the definition above is understood weakly in infinite-dimensional spaces, such as field theory, because $\Omega$ is weakly non-degenerate.
\end{rem}

\begin{defn}
A \textit{complexified (pseudo-) K\"ahler} space is the quadruple $(V_\mathbb{C}, \Omega_\mathbb{C}, J, \nu_\mathbb{C})$, consisting of the complex vector space $V_\mathbb{C} \cong V \times \overline{V}$, the complex linear extension of $\Omega$, a complex structure $J$ and the complex linear extension of $\nu$, together with
\begin{enumerate}
\item a charged symplectic form, $(u,v) \mapsto \Omega_\mathbb{C}(\bar{u}, v)$, for $u, v \in V_\mathbb{C}$,
\item a Hermitian form $(u,v) \mapsto \nu_\mathbb{C}(\bar{u}, v) \doteq \Omega_\mathbb{C}(\bar{u}, Jv)$, for $u, v \in V_\mathbb{C}$.
\end{enumerate}
\end{defn}

Our complexified pseudo-K\"ahler space of solutions consists of the quadruple $(\Sol_\mathbb{C}, \Omega_\mathbb{C}, J, \nu_\mathbb{C})$, where the charged symplectic structure and the Hermitian form are defined by
\begin{subequations}
\begin{align}
\Omega_\mathbb{C}(\overline{(\varphi_1,\varpi_1)},(\varphi_2,\varpi_2)) & = \int_{\Sigma_t} \! d^3\x \left(\overline{\pi_1} \varphi_2 - \pi_2 \overline{\varphi_1} \right), \\
\Omega_\mathbb{C}(\overline{(\varphi_1,\pi_1)},J[(\varphi_2,\pi_2)]) & = \int_{\Sigma_t} \! d^3\x \left(\overline{\pi_1} (J\varphi_2) - (J\pi_2) \overline{\varphi_1} \right),
\end{align}
\end{subequations}
\\
respectively. 
An inner product can be defined on the subspace $\Sol_+$ from the Hermitian form defined above. Let $\phi_1, \phi_2 \in \Sol_+$, then it can be verified (by integration by parts) that
\begin{equation}
(\phi_1,\phi_2)_+ \doteq -\Omega_\mathbb{C}(\overline{(\varphi_1,\pi_1)},J[(\varphi_2,\pi_2)]) = -\ii \int_{\Sigma_t} \! d^3\x \left(\overline{\pi_1} \varphi_2 - \pi_2 \overline{\varphi_1} \right)
\label{2:Inner+}
\end{equation}
is a symmetric, positive-definite bilinear form. We call $(\Sol_+, (\cdot, \cdot)_+)$ the positive frequency space of solutions.\footnote{In full notation, $(\Sol_+, \Omega_\mathbb{C}, J, \nu_\mathbb{C})$ is a complex K\"ahler space, and we have defined the inner product \eqref{2:Inner+} by $\nu_\mathbb{C}$.}

The one-particle Hilbert space of the theory, $\mathcal{H}$ is the Hilbert space completion of the inner product space $(\Sol_+, (\cdot, \cdot)_+)$. The Hilbert space of the theory is the \textit{symmetric Fock space}, $\mathscr{F}_\text{s}(\mathcal{H}) = \oplus_{n=0}^\infty \mathcal{H}^{\odot n}$, where $\odot$ is a symmetrised tensor product and $\mathcal{H}^0 \doteq \mathbb{C}$ by convention.

The creation and annihilation operators, obeying the relations \ref{2:CCRaa*}, are realised as operator-valued distributions on $\mathscr{F}_\text{s}(\mathcal{H})$. Let $\Psi \in \mathscr{F}_\text{s}(\mathcal{H})$ be an element of the symmetric Fock space, represented in abstract index notation by
\begin{equation}
\Psi = (\psi, \psi^{a_1}, \psi^{(a_1 a_2)}, \ldots, \psi^{(a_1 a_2 \ldots a_n)}, \ldots),
\end{equation}
where all, but finitely many of the wavefunctions $\psi$ vanish.
then the action of the creation and annihilation operators is
\begin{subequations}
\begin{align}
a(\overline{f}) \Psi & = (\overline{f}_a \psi^a, \sqrt{2} \, \overline{f}_a \psi^{(a a_1)}, \sqrt{3} \, \overline{f}_a \psi^{(a a_1 a_2)}, \ldots), \\
a^*(f) \Psi & = (0, f^{a_1}\psi, \sqrt{2} \, f^{(a_1}\psi^{a_2)}, \sqrt{3} \, f^{(a_1}\psi^{a_2 a_3)}, \ldots).
\end{align}
\end{subequations}
\\
where the contractions are defined by the inner product \eqref{2:Inner+}, $\overline{f}_a \psi^a = (f, \psi)_+$. The vacuum state of the theory is defined by $a(\bar{f}) |0\rangle = 0$.

In the case of fermions, the construction differs in that the canonical space of solutions carries an inner product from the start, that can be completed into the Hilbert space of the theory. The Fock space is then constructed by taking anti-symmetric tensor products of the Hilbert space, yielding the \textit{antisymmetric Fock space} $\mathscr{F}_\text{a}(\mathcal{H})$.

\subsection{Physical states}
\label{2:subsecStates}

Physical states satisfy an important property known as the Hadamard condition, which ensures that the $n$-point functions of the observables have the correct ultraviolet behaviour. The idea is that, at sufficiently short scales, states should \textit{resemble the Minkowski state} in a precise way. In order to state the condition, we introduce Synge's world function, the half-square-geodesic distance between two points on a spacetime.

\begin{defn}
Let $(M,g)$ be a spacetime and let $\mathsf{x}_1, \mathsf{x}_2 \in M$. We call $N(\mathsf{x}_2)$ a \textit{geodesically convex neighbourhood of} $\mathsf{x}_2$, defined by the set of points connected to $\mathsf{x}_2$ by a unique geodesic. Let $\mathsf{x}_1 \in N(\mathsf{x}_2)$ and $\lambda$ be an affine parameter along the geodesic connecting $\mathsf{x}_1$ and $\mathsf{x}_2$, we define the Synge world function by
\begin{equation}
\sigma(\mathsf{x}_2,\mathsf{x}_1) = \frac{1}{2} (\lambda_2 - \lambda_1) \int_{\lambda_1}^{\lambda_2} \! d \lambda \, g_{ab}(\mathsf{x}) \frac{d\mathsf{x}^a}{d\lambda} \frac{d\mathsf{x}^b}{d\lambda},
\label{2:Synge}
\end{equation}
where $\mathsf{x}_1 = \mathsf{x}(\lambda_1)$ and $\mathsf{x}_2 = \mathsf{x}(\lambda_2)$.
\end{defn}

\begin{rem}
Along geodesic trajectories, $g_{ab}(\mathsf{x}) (d\mathsf{x}^a/d\lambda) (d\mathsf{x}^b/d\lambda)$ is a constant of motion. In particular, it is equal to $-1$ if the geodesic is timelike with proper time affine parameter ($\lambda = \tau$), $0$ if the geodesic is null and $1$ if the geodesic is spacelike with proper distance affine parameter ($\lambda = s$). Then,
\begin{subequations}
\begin{align}
\sigma\left(\mathsf{x}(\tau),\mathsf{x}(\tau')\right)  & = -\frac{1}{2} \left(\tau - \tau'\right)^2 & \text{ for timelike geodesics,}\\
\sigma\left(\mathsf{x}(\tau),\mathsf{x}(\tau')\right) & = 0 & \text{ for null geodesics,}\\
\sigma\left(\mathsf{x}(\tau),\mathsf{x}(\tau')\right) & = + \frac{1}{2} \left(s - s'\right)^2 & \text{ for spacelike geodesic.}
\end{align}
\label{2:sigma}
\end{subequations}
\end{rem}

We are ready to define the Hadamard property. We will restrict our attention to quasi-free states, \textit{i.e.}, states that are fully defined by the two-point function.

\begin{defn}
Let $\Phi$ be a scalar field on $(M,g)$, a spacetime of dimension $n$. A quasi-free field state $|\phi\rangle$ is called \textit{Hadamard} if its two-point function is of the form
\begin{equation}
\langle \phi | \Phi(\mathsf{x})\Phi(\mathsf{x}') | \phi \rangle = H(\mathsf{x},\mathsf{x}') + W(\mathsf{x},\mathsf{x}')
\end{equation}
for $\mathsf{x}' \in N(\mathsf{x})$ as $\mathsf{x} \to \mathsf{x}'$, where $W$ is a regular, state-dependent biscalar and $H$ is a state-independent bidistribution, called the \textit{Hadamard parametrix}, and is defined by
\begin{subequations}
\begin{align}
H(\mathsf{x},\mathsf{x}') & = - \frac{1}{4 \pi} \left[U(\mathsf{x},\mathsf{x}') \ln\left(\sigma_\epsilon(\mu^2 \mathsf{x},\mathsf{x}') \right) \right],  \text{ if }n = 2, \\
H(\mathsf{x},\mathsf{x}') & = - \frac{\Gamma(n/2-1)}{2(2 \pi)^{n/2}} \left[\frac{U(\mathsf{x},\mathsf{x}')}{\left(\sigma_\epsilon(\mathsf{x},\mathsf{x}') \right)^{n/2-1}}  + V(\mathsf{x},\mathsf{x}') \ln\left(\mu^2 \sigma_\epsilon(\mathsf{x},\mathsf{x}') \right) \right], \nonumber \\ &  \text{ if }n \neq 2 \text{ is even and} \\
H(\mathsf{x},\mathsf{x}') & = - \frac{\Gamma(n/2-1)}{2(2 \pi)^{n/2}}\left[\frac{U(\mathsf{x},\mathsf{x}')}{\left(\sigma_\epsilon(\mathsf{x},\mathsf{x}') \right)^{n/2-1}} \right],  \text{ if }n \text{ is odd.} 
\end{align}
\label{2:Hparametrix}
\end{subequations}

Here, $\sigma_\epsilon$ is the regularised Synge world function, which replaces \eqref{2:sigma}, along timelike geodesics according to the $\ii \epsilon$ prescription
\begin{equation}
\sigma_\epsilon\left( \mathsf{x}(\tau),\mathsf{x}(\tau') \right)   = -\frac{1}{2} \left(\tau - \tau' - \ii \epsilon \right)^2
\label{2:regsynge}
\end{equation}
and eq. \eqref{2:Hparametrix} is understood distributionally as $\epsilon \to 0_+$. The biscalars $U$ and $V$ are regular and have a (non-unique\footnote{See, \textit{e.g.}, \cite{Ottewill:2009uj}.}) asymptotic expansion
\begin{subequations}
\begin{align}
U & = \sum_{i=0}^\infty U_j(\x,\x') \sigma^j(\x,\x') \\
V & = \sum_{i=0}^\infty V_j(\x,\x') \sigma^j(\x,\x')
\end{align}
\end{subequations}
\\
with coefficients given by a covariant Taylor expansion,
\begin{subequations}
\begin{align}
U_j(\x,\x') & = u_(j)(\x) + \sum_{k=1}^\infty \frac{(-1)^k}{k!} u_{(j)a_1 a_2 \ldots a_k}(\x) \sigma^{;a_1}(\x,\x') \ldots \sigma^{;a_k}(\x,\x') \\
V_j(\x,\x') & = v_(j)(\x) + \sum_{k=1}^\infty \frac{(-1)^k}{k!} v_{(j)a_1 a_2 \ldots a_k}(\x) \sigma^{;a_1}(\x,\x') \ldots \sigma^{;a_k}(\x,\x'),
\end{align}
\end{subequations}
\\
given by recursion relations that guarantee that the equations of motion be on-shell, supplemented with 
the boundary condition $U_0(\mathsf{x},\mathsf{x}) = 1$. Above, we have used the semicolon notation to denote covariant differentiation. 
\label{2:defHadamard}
\end{defn}

The asymptotic expansions of the biscalars $U$,  $V$ and $W$ together with their boundary conditions have been worked out in detail by D\'ecanini and Folacci \cite{Decanini:2005eg}.

In the language of Green bi-distributions, informally known as Green functions, the two-point function is usually called the \textit{Wightman function}, which we denote by $\mathcal{W}\left(\x, \x'\right) \doteq \langle \phi | \Phi(\mathsf{x})\Phi(\mathsf{x}') | \phi \rangle$. An excellent reference for the use of Green functions in quantum field theory is \cite{Fulling:1989nb}.

A reformulation of the Hadamard condition was worked out by Radzikowski using microlocal analysis, and dubbed a wave front set spectral condition \cite{Radzikowski:1996pa}. Wald and Hollands used the microlocal {\rm wave front spectral condition} to prove the existence and uniqueness of local covariant time ordered products of quantum fields in curved space-time \cite{Hollands:2001nf}. For our purposes, it suffices to use Definition \ref{2:defHadamard}.

\subsection{Thermal states}

We now introduce an important class of quasi-free \textit{stationary} states that, in addition to being Hadamard, satisfy a special property at the level of the two-point function, called the KMS condition, introduced by Haag \cite{Haag:1967sg}, based on the pioneering studies of Kubo, Martin and Schwinger, who studied the mathematical properties of statistical systems in equilibrium \cite{Kubo:1957mj, Martin:1959jp}.

Let us make this statement precise.

\begin{defn}
Let $(M,g)$ be a stationary, globally hyperbolic $n$-dimensional spacetime with global timelike Killing vector $\xi$. The spacetime can be foliated with respect to $\xi$ by introducing the time function $\xi^a \nabla_a t = 1$. We denote points on the $(n-1)$-dimensional hypersurface normal to $g^{ab} \nabla_b t |_{t = t_1}$, where $t_1$ is constant, by $\x_1 \in \Sigma_{t_1}$. Let $|\psi \rangle$ be a quasi-free state of positive frequency with respect to $\xi$. 

The state is $|\psi \rangle$ is \textit{stationary with respect to} $\xi$ if the Wightman function of $|\psi \rangle$ satisfies the distributional relation $\mathcal{W}_\psi \left(\left(t_1,\x_1 \right), \left(t_2,\x_2 \right)\right) = \mathcal{W}_\psi \left(\left(t_1-t_2,\x_1 \right), \left(0,\x_2 \right)\right)$. The notation $\mathcal{W}_\psi \left(t_1-t_2; \x_1, \x_2 \right) \doteq \mathcal{W}_\psi \left(\left(t_1-t_2,\x_1 \right), \left(0,\x_2 \right)\right)$ is standard.
\end{defn}

We are ready to give the definition of a KMS state. From now on we set Boltzmann's constant to be $k_\text{B}=1$.

\begin{defn}
\label{2:defKMS}
Let $(M,g)$ be a stationary, globally hyperbolic spacetime. Let $|\psi \rangle$ be a stationary, quasi-free state with Wightman function (mapping pairs of spacetime points to) $\mathcal{W}_\psi(s; \x_1,x_2)$, with $s = t_1 - t_2$. Furthermore, let $\mathcal{W}_\psi$ be holomorphic on the complex strip $S = \left\{s: -\beta < Im(s) < 0 \right\}$ of the complex $s$-plane. We call $|\psi \rangle$ a \textit{thermal KMS state temperature} $T = 1/\beta$, where $\beta > 0$, if it satisfies the \textit{KMS condition}, 
\begin{equation}
\mathcal{W}_\psi(s; \x_1,x_2) = \mathcal{W}_\psi(-(s- \ii \beta); \x_1,x_2).
\label{2:KMS}
\end{equation}
\end{defn}

The KMS condition describes thermal states in the sense that it is satisfied by statistical states that are defined by thermal partition functions. In quantum field theory we take eq. \eqref{2:KMS} as the definition of the KMS condition for quasi-free, stationary states. 

\subsection{Conformal states}
\label{2sec:CFT}

Conformal transformations have a rich group-theoretic structure, as a subgroup of the diffeomorphism group, and a richer algebraic structure, in terms of vector fields generating local conformal transformations. The situation is especially important in $2$ (Euclidean) or $1+1$ (Lorentzian) dimensions, where all spacetimes, say with zero cosmological constant, are conformally related. We do not attempt to give a detailed account of this immense field of research, but we should mention that in $1+1$ dimensions, the group of orientation-preserving conformal diffeomorphisms is isomorphic to the product of two copies of the group of orientation-preserving diffeomorphisms on the circle, $\text{Conf}(\mathbb{R}^{1,1}) \cong \text{Diff}_+(\mathbb{S}) \times \text{Diff}_+(\mathbb{S})$. As a result, the Lie algebra generating $1+1$ conformal transformations must be infinite dimensional, as vector fields on the circle generate $\text{Diff}_+(\mathbb{S})$.  We refer the reader to \cite{Schottenloher:2008zz, Blumenhagen:2009zz} for more a more extense and precise account of the group- and algebraic-theoretic issues of conformal field theories.

Classical and quantum field theories that enjoy invariance under \textit{conformal transformations} are called \textit{conformal field theories}. They are important because much of the analysis is significantly simpler in conformal field theories, and many toy models with conformal invariance serve to provide physical intuition in more general situations.\footnote{Indeed, in this thesis we work out a variety of quantum phenomena with the aid of conformal methods in $1+1$ dimensions.} 

\begin{defn}
Let $(M,g)$ and $(M,\tilde{g})$ be oriented spacetimes, and let $\psi: M \to M$ be a conformal transformation with \textit{conformal factor} $\Omega$, \textit{i.e.}, such that $\psi^* \tilde{g} = \Omega^2 g$. Let $\phi$ be a classical on-shell field on $(M,g)$. We say that the field theory solved by $\phi$ is a \textit{conformal field theory} if $\psi^* \tilde{\phi} = \Omega^s \phi$, $s \in \mathbb{R}$, and $\tilde{\phi}$ is a solution to the field theory on spacetime $(M,\tilde{g})$. The real number $s$ is called the \textit{conformal weight of} $\phi$.
\end{defn}

Many important physical theories are conformally invariant on their own right. In $3+1$ dimensions, Maxwell's equations are an example. In $1+1$ dimensions, the decomposition $\text{Conf}(\mathbb{R}^{1,1}) \cong \text{Diff}_+(\mathbb{S}) \times \text{Diff}_+(\mathbb{S})$ gives rise to the decoupling of \textit{left}- and \textit{right}-\textit{movers} in field theory, which provides analytic control in many interesting situations.\footnote{Another relevant example in $1+1$ dimensions is string theory. The dynamics of the string is is a $1+1$ conformal field theory. This, in turn, allows string theorists to exploit a great deal of mathematical technology stemming from conformal field theory for stringy applications.}

A relevant example for this work is that of a $(1+1)$-dimensional minimally coupled, massless scalar field theory. Let $(M,g)$ be an $n$-dimensional spacetime, and $\phi$ a Klein-Gordon field on $M$ defined by the theory $P_\text{cf} \phi = (\Box - \xi_\text{cf}\R) \phi = 0$, where $\xi_\text{cf} \doteq (n-2)/[4(n-1)]$. The choice of $\xi_\text{cf}$ is called \textit{conformal coupling} because it renders the field theory into a conformal field theory. Indeed, under a conformal transformation, $\psi$, $P_\text{cf}( \tilde{\phi}) = 0$, with $\psi^*\tilde{\phi} = \Omega^{(2-n)/2} \phi$ if $P_\text{cf} \phi = 0$. This follows from the conformal transformation of the Ricci scalar. See, for example, the classic monograph of Birrell and Davies \cite{Birrell:1982ix} for the detailed form of the Ricci scalar transformation. It follows that in $1+1$ dimensions, the minimally coupled massless field is conformally coupled.

Because the on-shell solutions of a conformal field theory are related under conformal transformations, so is the quantum theory, as per the constructions that we have presented above. In particular, the quantum states of the theory after a conformal transformation are related to the states of the untransformed theory.

Let $\tilde{a} = (-\ii \tilde{\Phi} + \tilde{\Pi})/\sqrt{2})$ be the annihilation operator in the quantum field theory that comes from the quantisation of $P_\text{cf}(\Omega^{(2-n)/2} \phi) = 0$. This annihilation operator defines a vacuum state $\tilde{a} |\tilde{0}\rangle = 0$. We similarly define a vacuum state $|0\rangle$ by the action of the annihilation operator $a = (-\ii \Phi + \Pi)/\sqrt{2})$ coming from the theory $P_\text{cf} \phi = 0$.

The relation between the two vacua can be obtained from the Hilbert space construction that we have defined above. The annihilation operators are related by $\tilde{a}(f) = a\left(\psi^* f\right) = a \left(\Omega^{(2-n)/2} f\right)$. This relation is expressed succintly for quasi-free states, in terms of field operators

\begin{defn}
Let $|\tilde{0}\rangle$ and $|0 \rangle$ be two vacuum states, and let their two-point functions be related by
\begin{equation}
\langle \tilde{0}| \tilde{\Phi}(f) \tilde{\Phi}(g)  |\tilde{0}\rangle = \langle 0| \Phi\left(\Omega^{(2-n)/2} f\right) \Phi\left(\Omega^{(2-n)/2} g\right)  |0 \rangle,
\end{equation} 
then, we call $|\tilde{0}\rangle$ a conformal vacuum with respect to $|0 \rangle$.
\end{defn}


\section{Particles in quantum field theory and detectors}
\label{2sec:Detectors}

An important result of (finite-dimensional) quantum mechanics is that we can choose to describe a quantum system using the triple $(\mathcal{H}_1, \mathscr{A}_1, \pi_1)$, for example, by choosing a state $\omega_1$ and carrying out the GNS construction, or we can choose the triple $(\mathcal{H}_2, \mathscr{A}_2, \pi_2)$, and the two descriptions are equivalent in a precise sense. This is the content of the Stone-von Neumann theorem:

\begin{thm}[Stone-von Neumann]
Let $(V,\Omega)$ be a finite dimensional symplectic space and let $(\mathcal{H}_1, \mathscr{A}, \pi_1)$ and $(\mathcal{H}_2, \mathscr{A}, \pi_2)$ be irreducible, strongly continuous,\footnote{This point is technical, but not central for our discussion.} unitary representations of the Weyl relations. Then there exists a unitary map $U: \mathcal{H}_1 \to \mathcal{H}_2$, such that for any $\pi_1(A(v)):\mathcal{H}_1 \to \mathcal{H}_1$, $U \pi_1(A)(v) U^{-1}:\mathcal{H}_2 \to \mathcal{H}_2$, \textit{i.e.}, $(\mathcal{H}_1, \mathscr{A}, \pi_1)$ and $(\mathcal{H}_2, \mathscr{A}, \pi_2)$ are unitarily equivalent.
\label{2:S-vN}
\end{thm}

The field theories that we have described above do not satisfy the hypotheses of the Stone-von Neumann theorem. Elements of the space of solutions are labelled by space-time points. Hence, the space of solutions is infinite dimensional. This means, from the covariant point of view, that different algebraic state choices will lead to unitarily inequivalent theories. From the canonical point of view, this non-unique choice is traced back to the selection of a time function, that encompasses the notion of positive energy, which is, in turn, equivalent to the choice of a complex structure in the space of solutions.

The failure of the Stone-von Neumann theorem in infinite-dimensional systems brings in an ambiguity to the choice of a preferred Hilbert space in quantum field theory. While one may be guided to select a space of states, out of the infinitely many inequivalent choices, based on the symmetries of spacetime, in generic, curved spacetimes, spacetime isometries provide no guideline, in addition to the Hadamard condition. As a consequence, particles cannot be fundamental objects in the description of nature.

On this line, the correct definition of particles is operational and hence only makes sense for interacting systems. More precisely, one can interact with a field through a measuring apparatus that is coupled to the field. The system consisting of the apparatus and the field will evolve unitarily and the state transitions in the measurement apparatus are identified with the absorption or emission of field quanta. Because such measurement apparatus detect particles, they go by the name of \textit{particle detectors}.

There is an extensive literature on particle detector models. See, for example, \cite{Hu:2012jr} and references therein. A simple, yet powerful detector model is the so-called Unruh-DeWitt detector \cite{Unruh:1976db, DeWitt:1979}. A large part of this thesis is devoted to the analysis of this model and modifications thereof. For the moment, let us concentrate on situations where the field has a well-defined Wightman function. In the Unruh-DeWitt model one considers a Klein-Gordon quantum field $\Phi$ coupled to a two-level point-like particle detector on a spacetime $(M,g)$. In the regime that we are considering there is no back-reaction from the field or the detector on the spacetime. The Hamiltonian of the system is given by $H = H_{\Phi} \otimes I_{\text{D}} + I_{\Phi} \otimes H_{\text{D}} + H_{\text{int}}$, where $H_{\Phi}$ is the Hamiltonian operator of the scalar field, $H_{\text{D}}$ is the detector Hamiltonian, and the interaction Hamiltonian is given by
\begin{equation}
H_{\text{int}}(\tau) = c \chi(\tau) \Phi(\mathsf{x}(\tau)) \otimes \mu(\tau),
\label{Hint}
\end{equation}
where $c$ is a coupling constant, $\mu$ is the monopole moment operator of the detector and $\chi \in C_0^\infty(\mathbb{R})$ is a smooth switching function of compact support that controls the interaction of the field and the detector along the worldline of the detector. $\chi$ is a function of the detector proper time, $\tau$, and vanishes for times less than $\tau_i$ and greater than $\tau_f$, \textit{i.e.}, $\text{supp}(\chi) \subset [\tau_i, \tau_f]$.

The space of states of the field is given by the symmetric Fock space $\mathscr{F}_\text{s}(\mathcal{H}_{\Phi})$. The detector Hilbert space, $\mathcal{H}_\text{D}$, consists of a two-dimensional space spanned by the energy eigenvectors as follows: The detector Hamiltonian can be written in terms of the creation and annihilation operators of the detector, $d^*$ and $d$, and the detector energy gap, $E$. It is given by $H_{\text{D}} = E d^* d$. Then the detector Hilbert space is spanned by the energy eigenvectors $\lbrace |0 \rangle, |1 \rangle \rbrace$ with eigenvalues $H_{\text{D}} |0\rangle = 0 |0\rangle$ and $H_{\text{D}} |1 \rangle = E |1 \rangle$. For $E > 0$ we call $|0 \rangle$ the ground state and $|1 \rangle$ the excited state. The situation is the opposite for $E<0$. The space of states of the coupled system is $\mathscr{F}_\text{s}(\mathcal{H}_\Phi) \otimes \mathcal{H}_\text{D}$. Notice that when $E < 0$ the Hamiltonian is negative, but bounded from below.

The key idea is to compute the transition probability of the detector evolving from an initial state $| i \rangle \in \mathcal{H}_\text{D}$ at proper time $\tau_i$ to a final state $| f \rangle$ at time $\tau_f$. Such transition is interpreted in terms of field quanta. For example, if $E > 0$ and initially $| i \rangle = | 0 \rangle$, finding at a later time $\tau_f$ that $| f \rangle = | 1 \rangle$ is interpreted as the absorption of a particle of energy $E$ \cite{Wald:1995yp}. 

If one supposes that the system is weakly coupled, where the coupling constant $c$ is smaller than any other scale in the problem, and in the initial state $ |\phi_i \rangle \otimes | 0 \rangle \in \mathscr{F}_\text{s}(\mathcal{H}_\Phi) \otimes \mathcal{H}_\text{D}$, where $|\phi_i \rangle$ is a Hadamard state, it is possible to compute, perturbatively in $c$, the transition probability of finding the detector in the state $| 1 \rangle$ at proper time $\tau_f$. 

The evolution of the detector-field model in the interaction picture is given by the pull-back of the Schr\"odinger equation along the detector worldline
\begin{equation}
H_\text{int}(\tau)\Big[| \phi \rangle \otimes | d \rangle \Big]  = \ii \partial_\tau \Big[ | \phi \rangle \otimes | d \rangle \Big].
\label{2:Schr}
\end{equation}
subject to the initial conditions $| \phi \rangle \otimes | d \rangle = |\phi_i \rangle \otimes | 0 \rangle$ at $\tau = \tau_i$.

The solution to eq. \eqref{2:Schr} is given in terms of a unitary operator,
\begin{equation}
U(\tau,\tau_i): |\phi_i \rangle \otimes | 0 \rangle \mapsto |\phi \rangle \otimes | d \rangle = U(\tau,\tau_i) |\phi_i \rangle \otimes | 0 \rangle.
\label{2:UnitSchr}
\end{equation} 

Combining eq. \eqref{2:UnitSchr} and \eqref{2:Schr}, one obtains the Tomonaga-Schwinger equation for the unitary operator $U$,
\begin{equation}
H_\text{int}(\tau)U(\tau,\tau_i)  = \ii \partial_\tau U(\tau,\tau_i),
\label{2:Tomo-Schw}
\end{equation}
which is understood at the level of the operator algebra on the states, subject to the initial conditions $U(\tau_i,\tau_i) = I = I_\Phi \otimes I_\text{D}$. The pertubative solution is obtained by transforming eq. \eqref{2:Tomo-Schw} into an integral eq. and iterating, to obtain
\begin{align}
U(\tau,\tau_i) & = I -\ii \int_{\tau_i}^\tau \! d\tau' H_\text{int}(\tau')U(\tau',\tau_i) \nonumber \\
& = I + \frac{(-\ii)^k}{k!} \sum_{k = 1}^\infty \int_{\tau_i}^\tau \! d\tau_1 d\tau_2 \ldots d\tau_k \, T\left[H_\text{int}(\tau_1) \cdots  H_\text{int}(\tau_k)\right],
\end{align}
where $T$ is the time ordering operator, \textit{i.e.}, the operators under time ordering are arranged from past to future acting on the left. The \textit{S-matrix} is defined as $S(\tau_f, \tau_i) \doteq U(\tau_f, \tau_i)$,
\begin{equation}
S(\tau_f,\tau_i) = I_\Phi \otimes I_\text{D} - \ii c \int_{\tau_i}^{\tau_f} d\tau \, \chi(\tau) \Phi(\tau) \otimes \mu(\tau) + O\left(c^2\right).
\end{equation}

The scattering amplitude of a transition in the detector subsystem is given by the trace over the final states of the field, which are not being measured by an observer equipped with a particle detector,
\begin{equation}
\mathcal{P}\left(| 1 \rangle \, \Big| \, |\phi_i \rangle \otimes | 0 \rangle \right)(E,\tau_i, \tau_f) = \sum_n \left| \langle \phi_n | \otimes \langle 1 | S(\tau_f,\tau_i) |\phi_i \rangle \otimes | 0 \rangle \right|^2.
\end{equation}

To leading order in $c$, 
\begin{align}
| \langle \phi_n | \otimes & \langle 1 | T \left(\ee^{-\ii H_\text{int}[\chi]}\right) |\phi_i \rangle \otimes | 0 \rangle |^2  \nonumber \\ 
& = \left| \langle \phi_n | \otimes \langle 1 | \Big( I_\Phi \otimes I_\text{D} - \ii H_\text{int}[\chi] + O\left( c^2 \right)\Big)  |\phi_i \rangle \otimes | 0 \rangle \right|^2 \nonumber \\
 & =  \langle \phi_n| \otimes \langle 1  | \left( -\ii c \int \! d\tau'' \, \chi(\tau'') \Phi(\tau'') \otimes \mu(\tau'') + O\left(c^2\right) \right) | \phi_i \rangle \otimes |0 \rangle \nonumber \\
 & \times \langle \phi_i| \otimes \langle 0  | \left( \ii c \int \! d\tau' \, \chi(\tau') \Phi(\tau') \otimes \mu(\tau') + O\left(c^2\right) \right) | \phi_n \rangle \otimes |1 \rangle.
\end{align}

We have used the fact that observables are self-adjoint operators. Writing $\mu$ in the interaction representation, $\mu(\tau_f) = \ee^{\ii H_D (\tau_f-\tau_i)} \mu(\tau_i) \ee^{-\ii H_D (\tau_f-\tau_i)}$ and using the relation
\begin{equation}
\sum_n \langle \phi_i |\Phi(\tau') | \phi_n \rangle \langle \phi_n |\Phi(\tau'') | \phi_i \rangle = \langle \phi_i | \Phi(\tau') \Phi(\tau'') | \phi_i \rangle,
\end{equation}
one obtains that to leading order in $c$
\begin{align}
& \mathcal{P}\left(| 1 \rangle \, \Big| \, |\phi_i \rangle \otimes | 0 \rangle \right) (E,\tau_i, \tau_f) = c^2 \left| \langle 1 | \mu(\tau_i) | 0 \rangle\right|^2 \nonumber \\
& \times \left( \int \! d\tau'' \int \! d \tau' \ee^{-\ii E(\tau'-\tau'')}
 \chi(\tau') \chi(\tau'') \langle \phi_i | \Phi(\tau') \Phi(\tau'') | \phi_i \rangle \right).
\label{Probability}
\end{align}

It follows from eq. \eqref{Probability} that all the dependence in the initial field state, trajectory of the detector and the spacetime comes from the term inside the parentheses, while the rest of the expression on the right hand side of eq. \eqref{Probability} contains only the internal details of the detector. In most of our work we shall set the state to be initially in an appropriate Hadamard vacuum state. In this case, the field dependence comes only through the pull-back of the Wightman function of the field along the detector worldline, $\langle \phi_i | \Phi(\tau') \Phi(\tau'') | \phi_i \rangle = \mathcal{W}\left(\mathsf{x}\left(\tau'\right), \mathsf{x}\left(\tau''\right)\right)$. Because vacuum states are quasi-free states, the Wightman function fully determines the state. This leads to defining the response function of the detector, $\mathcal{F}$, which is equal to the double Fourier transform of the Wightman function of the field weighted by the switching function and describes the (smeared) power spectrum of the field vacuum noise. It is given by
\begin{equation}
\mathcal{F}(E,\tau_f,\tau_i) \doteq \int_{-\infty}^\infty \! d\tau' \, \int_{-\infty}^\infty \! d\tau'' \, \chi(\tau') \chi(\tau'') \, \ee^{-\ii E (\tau'-\tau'')} \mathcal{W}\left(\mathsf{x}\left(\tau'\right), \mathsf{x}\left(\tau''\right)\right).
\label{ResponseFn}
\end{equation}

Because $\mathcal{P}\left(| 1 \rangle \, \Big| \, |\phi_i \rangle \otimes | 0 \rangle \right)(E,\tau_i, \tau_f) = c^2 \left| \langle 1 | \mu(\tau_i) | 0 \rangle\right|^2 \mathcal{F}(E,\tau_f,\tau_i)$, it is customary to use the terms response and transition probability interchangeably in the jargon of weakly coupled particle detectors. For smooth, timelike trajectories, the pull-back of the Wightman function is a distribution in $\mathbb{R} \times \mathbb{R}$ and eq. \eqref{ResponseFn} is well-defined with $\chi$ playing the role of a test function \cite{Fewster:1999gj, Junker:2001gx, Hormander:1990, Hormander:1994}. Given a family of functions $\mathcal{W}_\epsilon$ that converges to the distribution $\mathcal{W}$ as $\epsilon \rightarrow 0_+$, $\mathcal{F}$ is evaluated by replacing $\mathcal{W}_\epsilon$ in \eqref{ResponseFn} and then taking the limit $\epsilon \rightarrow 0_+$.

\subsection{KMS states and the detailed balance condition}

From the point of view of a local observer along their spacetime worldline, the KMS condition can be realised by making use of particle detectors. The detailed balance form of the KMS condition associates the temperature measured by a stationary particle detector, with the response of such detector and the energy gap,

\begin{defn}
A function $\mathcal{G}: \mathbb{R} \to \mathbb{R}^+$ is said to obey the \textit{detailed balance condition} if
\begin{equation}
\beta = \frac{1}{E} \ln \left( \frac{\mathcal{G}(-E)}{\mathcal{G}(E)}\right).
\label{2:DBC}
\end{equation}
\end{defn}

In the language of detectors, a series of arguments relate $\mathcal{G}$ to the response of a stationary detector in a stationary state \cite{Birrell:1982ix, Wald:1995yp, Unruh:1976db, DeWitt:1979, Takagi:1986kn} that measures a temperature $T = \beta^{-1}$. We give a heuristic relation between the detailed balance condition and the KMS condition, which should not be taken as a proof, but rather as a motivation. We will make the relation precise and provide a proof in Chapter \ref{ch:Unruh}.

For a stationary detector, the response is given by
\begin{equation}
\mathcal{F}(E,\tau_f,\tau_i) \doteq \int_{-\infty}^\infty \! ds \, \int_{-\infty}^\infty \! d\tau \, \chi(\tau) \chi(\tau-s) \, \ee^{-\ii E s} \mathcal{W}\left(s\right).
\label{ResponseFnStat}
\end{equation}

If one wishes to make the interaction of such detector constant at all times, $\chi \rightarrow 1$, the response function diverges. This divergence can be written formally as 
\begin{equation}
\mathcal{F}(E) =  \int_{-a}^a \! d\tau \, \int_{-\infty}^\infty \! ds \, \ee^{-\ii E s} \mathcal{W}(s)
\label{Feternal}
\end{equation}
keeping in mind that $a \to \infty$. Yet, heuristically, the ratio $\mathcal{F}(-E)/\mathcal{F}(E)$ is finite and equal to\footnote{Alternatively, one can use the \textit{instantaneous transition rate} $\dot{\mathcal{F}}$, where the dot stands for a proper time derivative, which is finite, to obtain the relation on the right hand side of the equation above. We shall do this in Chapters \ref{ch:DeCo} and \ref{ch:BH}.}
\begin{align}
\frac{\mathcal{F}(-E)}{\mathcal{F}(E)} =  \left( \int_{-\infty}^\infty \! dr \, \ee^{-\ii E r} \mathcal{W}(r -\ii \epsilon) \right)^{-1} \int_{-\infty}^\infty \! ds \, \ee^{\ii E s} \mathcal{W}(s -\ii \epsilon).
\end{align}

If we assume that $\mathcal{W}$ is KMS, then the contour of integration on the numerator can be pushed down by $\ii (\beta - 2 \epsilon)$ inside the complex analytic strip $-\beta < \text{Im}(s) < 0$,
\begin{equation}
\int_{-\infty}^\infty \! ds \, \ee^{\ii E s} \mathcal{W}(s -\ii \epsilon) = \int_{-\infty}^\infty \! ds \, \ee^{\ii E [s - \ii (\beta - 2 \epsilon)]} \mathcal{W}(s - \ii \beta + \ii \epsilon).
\end{equation}

Changing variables to $q = -s$ and using the KMS condition $\mathcal{W}(q- \ii \epsilon) = \mathcal{W}(-q-\ii \beta +\ii \epsilon)$, eq. \eqref{2:DBC} follows. Following the same logic, if \eqref{2:DBC} holds, then $\mathcal{W}$ must satisfy the KMS condition.

We will make this heuristic argumentation precise in Chapter \ref{ch:Unruh}. We feel, however, that it is important to discuss the detailed balance form of the KMS condition already at this stage, as we will make use of it in Chapters \ref{ch:DeCo} and \ref{ch:BH}.

\subsection{Particles in $d \geq 2$ dimensions}

There are serious difficulties in extending the Unruh-DeWitt model to $1+1$ dimensions for a massless scalar field. Any Hadamard massless scalar state in $1+1$ dimensions contains infrared ambiguities. A way to see this is by taking the na\"ive massless limit of the Wightman function of a massive $1+1$ Klein-Gordon field in $1+1$ Minkowski spacetime. Let us start our discussion by analysing the ultraviolet behaviour of Hadamard states in $1+1$ dimensions.

The two-point function Hadamard form is, as per Def. \ref{2:defHadamard},
\begin{equation}
\langle \phi | \Phi(\mathsf{x})\Phi(\mathsf{x}') | \phi \rangle = - \frac{1}{4 \pi} \left[V(\mathsf{x},\mathsf{x}') \ln\left(\sigma_\epsilon(\mathsf{x},\mathsf{x}') \right) \right] + W(\mathsf{x},\mathsf{x}'),
\label{2:Had1+1vV-M}
\end{equation}
where, the bidistribution $V$ has the $1+1$ asymptotic expansion
\begin{equation}
V(\mathsf{x}, \mathsf{x}') = \sum_{n=0}^\infty U_n(\mathsf{x}, \mathsf{x}') \sigma^n(\mu \mathsf{x}, \mathsf{x}')
\end{equation}
with the condition
\begin{equation}
U_0(\mathsf{x}, \mathsf{x}') = \Delta^{1/2}(\mathsf{x}, \mathsf{x}') = \left[1 + O(\sigma)\right]^{1/2},
\end{equation}
where $\Delta$ is the \textit{van Vleck-Morette determinant}. Redefining the regular piece in \eqref{2:Had1+1vV-M}, the short-distance behaviour of the two-point function is
\begin{equation}
\langle \phi | \Phi(\mathsf{x})\Phi(\mathsf{x}') | \phi \rangle = - \frac{1}{4 \pi} \left[ \ln\left(\sigma_\epsilon(\mathsf{x},\mathsf{x}') \right) \right] + W(\mathsf{x},\mathsf{x}').
\label{2:Had1+1}
\end{equation}

Next, we obtain the $1+1$ massive Klein-Gordon Minkowski Wightman function. A standard strategy to do so is to use the formal plane-wave integral representation of the quantum field, in terms of annihilation and creation operators. This yields,
\begin{equation}
\mathcal{W}(t-t'; x - x') = \langle 0| \Phi(t,x) \Phi(t',x') |0\rangle = \int_{-\infty}^\infty \! \frac{d k}{4 \pi \omega_k} \, \exp\left(-\ii \omega_k (\Delta t-\ii \epsilon) + \ii k \Delta x\right),
\end{equation}
where $\omega_k = \left(k^2 + m^2 \right)^{1/2}$ and $(t,x)$ are Minkowskian coordinates. The expression above can be massaged and integrated into a modified Bessel function of the second kind, for example, following \cite{Hodgkinson:2013tsa},
\begin{equation}
\mathcal{W}(\Delta t; \Delta x) = \frac{1}{2\pi} K_0\left[m\sqrt{\left(\Delta x\right)^2 - \left(\Delta t -\ii \epsilon\right)^2} \right].
\label{2:K0Wightman}
\end{equation}

The $\ii \epsilon$ prescribes the analytic continuation to the timelike case $|\Delta t| > |\Delta x|$,
\begin{equation}
\mathcal{W}(\Delta t; \Delta x) = \frac{1}{2\pi} K_0\left[\ii \, \text{sgn}(\Delta t) m\sqrt{\left(\Delta t \right)^2 - \left(\Delta x\right)^2} \right],
\end{equation}
where sgn denotes the signum function,
\begin{align}
\text{sgn}(\tau) \doteq & 
    \begin{cases}
      -1, & \text{if}\ \tau < 0, \\
      0, & \text{if}\ \tau = 0, \\
      1, & \text{if}\ \tau > 0.
    \end{cases}    
\end{align}

Choosing the branch of the Bessel function following the analytic continuation formulas in \cite{NIST},
\begin{align}
K_\nu(z) & = 
    \begin{cases}
      + \frac{\ii \pi}{2} \ee^{+\ii \pi \nu/2} H_\nu^{(1)} \left(z \ee^{+\ii \pi /2} \right), & \text{if}\ - \pi < \text{Arg}(z) \leq \frac{\pi}{2}, \\
      - \frac{\ii \pi}{2} \ee^{-\ii \pi \nu/2} H_\nu^{(2)} \left(z \ee^{-\ii \pi /2} \right), & \text{if}\ - \frac{\pi}{2} \leq \text{Arg}(z) < \pi,
    \end{cases}
\label{BesselsAnalytic}   
\end{align}
where $H_\nu^{1}$ and $H_\nu^{2}$ are Hankel functions, we have for timelike separations.
\begin{align}
\mathcal{W}(\Delta t; \Delta x) = & 
    \begin{cases}
       -\frac{\ii}{4} H_0^{(2)}\left( m\sqrt{\left(\Delta t \right)^2 - \left(\Delta x\right)^2} \right), & \text{ if } + \Delta \tau > |\Delta x|, \\[10pt]
       +\frac{\ii}{4} H_0^{(1)}\left( m\sqrt{\left(\Delta t \right)^2 - \left(\Delta x\right)^2} \right), & \text{ if }  - \Delta \tau > |\Delta x|.
    \end{cases}
\label{2:K0WightmanTimelike} 
\end{align}

On the one hand, eq. \eqref{2:K0Wightman} and \eqref{2:K0WightmanTimelike} define a Hadamard state in $1+1$ dimensions in the coincidence limit, with the appropriate logarithmic ultraviolet behaviour. On the other hand, viewing, \textit{e.g.}, \eqref{2:K0Wightman} as a function of $m$, the limit $m \to 0$ is not defined, as it also diverges logarithmically at arbitrary (spacelike) separations. An additive constant is needed to make sense of this limit. Let us introduce the dimensionful constant $\mu > 0$. The limit
\begin{equation}
\lim_{m \to 0} \mathcal{W}(\Delta t; \Delta x) - \frac{1}{2\pi} \ln [m \ee^\gamma/(2\mu)] = -\frac{1}{2\pi} \ln\left[\mu \sqrt{\left(\Delta x\right)^2 - \left(\Delta t -\ii \epsilon\right)^2} \right]
\label{2:lnWightman}
\end{equation}
exists, but cannot be extended to $\mu = 0$. Above, $\gamma$ is Euler's constant. Because we have added a constant piece, the limit continues to be Hadamard, as is explicit in eq. \eqref{2:lnWightman}. Nevertheless, the massless limit possesses an additive ambiguity, and justifying the presence of this \textit{ad hoc} constant is problematic.

A way to go around this issue is to restrict the matter available for detection and work only with massive fields. This, however, spoils the advantages of conformal field theory in $1+1$ dimensions. We argue that a better approach is to introduce a coupling between the field and the detector in such a way that the interaction between the detector and massless fields in $1+1$ dimensions is well-defined and unambiguous. This would, in turn, allow one to perform calculations of interest with analytic control. Indeed, one can argue that no physical process should register this ambiguity. 

Let $(M,g)$ be a spacetime with dimensions $d = 1 + n$, $n \in \mathbb{N}$. A local and covariant element of the algebra of observables of $\Phi$ is given by an appropriate smearing of the operator $\nabla_v \Phi$, where $v \in TM$ is a test vector field and $\nabla_v$ is the covariant derivative along the vector field $v$. Along the worldline of the detector, a natural coupling between the monopole moment of the detector and $\nabla_v \Phi$ comes by setting $v = d \mathsf{x}/d \tau \doteq \dot{\mathsf{x}}$. In other words, the choice for this test vector along the worldline of the detector is merely given by the vector field which is everywhere tangent to the path of the detector parametrised by the proper time. Any other vector field choice would require the detector to have a spatial physical size. The derivative coupling is then given by
\begin{equation}
H_\text{int} = c \chi(\tau)  \nabla_{\dot{\mathsf{x}}} \Phi(\mathsf{x}(\tau)) \otimes \mu(\tau).
\label{DeCo}
\end{equation}

An alternative expression is $H_\text{int} = c \chi(\tau) \dot{\Phi}(\mathsf{x}(\tau)) \otimes \mu(\tau)$. This alternative representation has the advantage of making it transparent that the detector is sensitive to \textit{changes of the field along the worldline} rather than the value of the field itself.

The derivative-detector coupling has first been introduced in \cite{Grove:1986fy} in order to examine the relation between energy fluxes and particle production in receding mirror spacetimes. In \cite{Grove:1986fy}, Grove argues that the particle response of a derivative coupling detector is related to the (positive) energy flux coming from a receding boundary. We will turn to this issue in Chapter \ref{ch:BH}. 

In our context, the hope, that will be realised in Chapter \ref{ch:DeCo}, is that the derivative-coupling detector is insensitive to the infrared ambiguity. In this way, we admit couplings with massless fields in $1+1$ dimensions, and one is able to use the full machinery of conformal field theory. This allows us to examine quantum field theoretic effects analytically in $(1+1)$-dimensional models of many otherwise intractable situations in full $(3+1)$-dimensional spacetimes. We present several applications of this model in $(1+1)$-dimensional black hole spacetimes in Chapter \ref{ch:BH}. 

\chapter{Particles in $(1+1)$-dimensional spacetimes}
\label{ch:DeCo}

Let us begin this chapter by recapitulating our discussion on the nature of particles in quantum field theory. In curved spacetimes, given the unitarily inequivalent representations of the algebras of observables, there is no preferred vacuum state to which one can associate a canonical one-particle Hilbert space. Different vacua will lead to different notions of particles. A solution to this problem is to make the notion of particle operational, by defining particles as transitions in \textit{particle detectors}. The simplest particle detector models are so-called Unruh-DeWitt pointlike detectors. A difficulty with this model is that it is ambiguous for massless $1+1$ scalar fields. In $1+1$ dimensions we have argued that a derivative-coupling detector with interaction Hamiltonian
\begin{equation}
H_\text{int} = c \chi(\tau) \dot{\Phi}(\mathsf{x}(\tau)) \otimes \mu(\tau),
\label{3:DeCoHint}
\end{equation}
and which is sensitive to the changes of the field values along the detector worldline, is a more suitable definition.

The transition probability can be computed as in the Unruh-DeWitt detector model, described in Chapter~\ref{ch:chap2}, and one arrives at eq. \eqref{Probability} but making the replacement $\mathcal{W}(\tau',\tau'') \rightarrow \mathcal{A}(\tau',\tau'')$, where
\begin{equation}
\mathcal{A}(\tau',\tau'') \doteq \partial_{\tau'} \partial_{\tau''} \mathcal{W}(\tau',\tau''),
\label{A}
\end{equation}
in the response function formula. One obtains
\begin{equation}
\mathcal{F}(E,\tau_f,\tau_i) \doteq \int_{-\infty}^\infty \! d\tau' \, \int_{-\infty}^\infty \! d\tau'' \, \chi(\tau') \chi(\tau'') \, \ee^{-\ii E (\tau'-\tau'')} \mathcal{A}\left(\tau',\tau''\right),
\label{3:DeCoResponse}
\end{equation}
where the derivatives of the pullback of the Wightman function are understood distributionally.

In this chapter, we develop this model \eqref{3:DeCoHint} for $(1+1)$-dimensional spacetimes. First, in Section~\ref{sec:3Coindicence} we write the response in a way that the singularities at the diagonal, $\tau' = \tau''$, are replaced by locally integrable functions. Second, in Section~\ref{sec:3rate} we take the sharp-switching limit of formula \eqref{3:DeCoResponse}. The response is ill-defined in the sharp-switching limit, but the instantaneous transition probability per unit time $\dot{\mathcal{F}}$ is finite and all the distributional singularities at the coincidence limit remain under control represented by locally integrable functions. This renders the transition rate formula into a powerful tool for studying many situations in a variety of spacetimes supporting fields in suitable quantum states and along different detector trajectories, inertial or otherwise. Third, in Section~\ref{sec:3checks} we verify that the detector rate satisfies the following desirable properties: The sharp transition rate reduces to the intuitive formal expression that yields the spontaneous emission of particles in Minkowski spacetime \ref{subsec:DeCoStatic}, the rate has the finite massless limit \ref{subsec:DeComassless} and the detector thermalises in a heat bath and along linearly uniformly trajectories at the Unruh temperature \ref{subsec:DeCotherm}.

\section[Isolating the coincidence limit]{The response function: isolating the coincidence limit}
\label{sec:3Coindicence}

Physical states in quantum field theory satisfy the Hadamard property. This allows one to know in detail the singularity structure of $\mathcal{W}(\tau',\tau'')$ as $\tau' \rightarrow \tau''$. In $1+1$ dimensions, it suffices to write $\mathcal{W} = (\mathcal{W} - \mathcal{W}_\text{sing}) + \mathcal{W}_\text{sing}$, where $\mathcal{W}_\text{sing}$ is the locally integrable function
\begin{align}
\mathcal{W}_\text{sing}(\tau',\tau'') \doteq & 
    \begin{cases}
      - \frac{\ii}{4}\text{sgn}(\tau'-\tau'') - \frac{1}{2\pi}\ln\left|\tau'-\tau''\right|, & \text{if}\ \tau' \neq \tau'', \\
      0, & \text{if}\ \tau'=\tau'',
    \end{cases}
\label{Wsing}   
\end{align}
and sgn denotes the signum function. In higher dimensions, $\mathcal{W}_\text{sing}$ needs to approximate the Hadamard parametrix with a higher precision. The parameters $\tau'$ and $\tau''$ in eq. \eqref{Wsing} refer to the proper time along an arbitrary worldline, while the parameters $\tau'$ and $\tau''$ in the Hadamard parametrix refer to the proper time along a geodesic curve. For a $1+1$ derivative-coupling detector eq. \eqref{Wsing} captures the ultraviolet behaviour along any (geodesic or otherwise) worldline. For the four-dimensional situation, see \cite{Louko:2007mu} and for higher dimensions, see \cite{Hodgkinson:2011pc}.

Correspondingly, we define

\begin{subequations}
\begin{align}
\mathcal{F}_\text{reg}(E,\tau_f,\tau_i) & \doteq \int_{-\infty}^\infty  \! d\tau' \, \int_{-\infty}^\infty \! d\tau'' \, \chi(\tau') \chi(\tau'') \, \ee^{-\ii E (\tau'-\tau'')} \nonumber \\
& \times \partial_{\tau'}\partial_{\tau''} \left[\mathcal{W}\left(\tau',\tau''\right) - \mathcal{W}_\text{sing}\left(\tau',\tau''\right) \right], \label{Freg} \\
\mathcal{F}_\text{sing}(E,\tau_f,\tau_i) & \doteq \int_{-\infty}^\infty  \! d\tau' \, \int_{-\infty}^\infty \! d\tau'' \, \chi(\tau') \chi(\tau'') \, \ee^{-\ii E (\tau'-\tau'')} \partial_{\tau'}\partial_{\tau''}\mathcal{W}_\text{sing}\left(\tau',\tau''\right), \label{Fsing}
\end{align}
\end{subequations}
\\
where the derivatives are understood in the distributional sense and such that $\mathcal{F} = \mathcal{F}_\text{reg} + \mathcal{F}_\text{sing}$.

The Hadamard short distance form of the Wightman function, shown in eq. \eqref{2:Had1+1},  implies that both $\mathcal{W}(\tau',\tau'')$ and $\partial_{\tau'} \partial_{\tau''} \left[\mathcal{W}\left(\tau',\tau''\right) - \mathcal{W}_\text{sing}\left(\tau',\tau''\right) \right]$ are represented in a neighbourhood of $\tau' = \tau''$ by locally integrable functions, \textit{i.e.}, functions that are integrable in an open neighbourhood around any point lying on $\tau' = \tau$ in $\mathbb{R}^2$. This renders the integral \eqref{Freg} free of distributional contributions at $\tau' = \tau''$ and the integral can be split over the subdomains $\tau' > \tau''$ and $\tau' < \tau''$.
\begin{align}
\mathcal{F}_\text{reg}(E,\tau_f,\tau_i) & = \int_{-\infty}^\infty \! d\tau' \, \int_{-\infty}^{\tau'} \! d\tau'' \, \chi(\tau') \chi(\tau'') \, \ee^{-\ii E (\tau'-\tau'')} \nonumber \\
& \times \partial_{\tau'}\partial_{\tau''} \left[\mathcal{W}\left(\tau',\tau''\right) - \mathcal{W}_\text{sing}\left(\tau',\tau''\right) \right] \nonumber \\
& + \int_{-\infty}^\infty \! d\tau'' \, \int_{-\infty}^{\tau''} \! d\tau' \, \chi(\tau') \chi(\tau'') \,  \Big[ \ee^{-\ii E (\tau'-\tau'')} \nonumber \\
& \times \partial_{\tau'}\partial_{\tau''} \left[\mathcal{W}\left(\tau',\tau''\right) - \mathcal{W}_\text{sing}\left(\tau',\tau''\right) \right] \Big]^{\text{c.c.}}.
\end{align}

Letting $u \in \mathbb{R}$ and $0<s<\infty$, we perform the change of variables $\tau'' = u$ and $\tau' = u-s$ in the subdomain $\tau'<\tau''$, while in the subdomain $\tau' >\tau''$ we set $\tau' = u$ and $\tau'' = u-s$,
\begin{align}
\mathcal{F}_\text{reg}(E,\tau_f,\tau_i) & = 2 \int_{-\infty}^\infty \! du \, \int_{0}^{\infty} \! ds \, \chi(u) \chi(u-s) \, \text{Re} \left[\ee^{-\ii E s} \left(\mathcal{A}(u,u-s) + \frac{1}{2 \pi s^2} \right) \right],
\label{FregFinal}
\end{align}
where we have used the property $\mathcal{W}(\tau',\tau'') = \overline{\mathcal{W}(\tau'',\tau')}$ and $\mathcal{W}_\text{sing}(\tau',\tau'') = \overline{\mathcal{W}_\text{sing}(\tau'',\tau')}$, which follows from the definition of $\mathcal{W}_\text{sing}$ in \eqref{Wsing}, and used the definition of $\mathcal{A}$ as given in eq. \eqref{A}.

To evaluate $\mathcal{F}_\text{sing}$, we integrate \eqref{Fsing} by parts and we perform the split $\mathcal{F}_\text{sing} = \mathcal{F}_\text{sing,1} + \mathcal{F}_\text{sing,2}$, where
\begin{subequations}
\begin{align}
\mathcal{F}_\text{sing,1}(E,\tau_f,\tau_i) & \doteq -\frac{\ii}{4} \int_{-\infty}^\infty  \! d\tau' \, \int_{-\infty}^\infty \! d\tau'' \, \partial_{\tau'} \partial_{\tau''} \left(\chi(\tau') \chi(\tau'') \, \ee^{-\ii E (\tau'-\tau'')} \right)\text{sgn}(\tau'-\tau'') \label{Fsing1}\\
\mathcal{F}_\text{sing,2}(E,\tau_f,\tau_i) & \doteq  - \frac{1}{2\pi} \int_{-\infty}^\infty  \! d\tau' \, \int_{-\infty}^\infty \! d\tau'' \, \partial_{\tau'} \partial_{\tau''} \left(\chi(\tau') \chi(\tau'') \, \ee^{-\ii E (\tau'-\tau'')} \right)\ln\left|\tau'-\tau''\right| \label{Fsing2}
\end{align}
\end{subequations}

The right hand side of eq. \eqref{Fsing1} can be integrated immediately over $\tau''$ and one obtains

\begin{align}
\mathcal{F}_\text{sing,1}(E,\tau_f,\tau_i) & = -\frac{\ii}{2} \int_{-\infty}^\infty  \! d\tau' \, \partial_{\tau'} \left( \chi(\tau') \ee^{-\ii E \tau'} \right)  \chi(\tau') \ee^{\ii E \tau'} \nonumber \\
& = -\frac{\ii}{2} \int_{-\infty}^\infty  \! d\tau' \, \left[\chi'(\tau') - \ii E \chi'(\tau') \right]  \chi(\tau') = -\frac{E}{2} \int_{-\infty}^\infty  \! du \, [\chi(u)]^2,
\label{Fsing1Final}
\end{align}
where we have noted that
\begin{equation}
\int_{-\infty}^\infty \! du \, \chi'(u) \chi(u) = \frac{1}{2} \int_{-\infty}^\infty \! du \, \frac{d}{du}[\chi(u)]^2 = 0.
\end{equation}

For $\mathcal{F}_\text{sing,2}$ we follow a strategy similar to our analysis of $\mathcal{F}_\text{reg}$. We split the integrals in the subintervals $\tau' < \tau''$ and $\tau' > \tau''$. We perform the change of variables $\tau'' = u$ and $\tau' = u-s$ in the subdomain $\tau'<\tau''$ and $\tau' = u$ and we set $\tau'' = u-s$ in the subdomain $\tau' >\tau''$, where $u \in \mathbb{R}$ and $0<s<\infty$. Eq. \eqref{Fsing2} becomes
\begin{align}
\mathcal{F}_\text{sing,2}(E,\tau_f,\tau_i) = - \int_0^\infty \! \frac{ds}{\pi} \, \ln s \int_{-\infty}^\infty \! du \, \text{Re} \left[ \partial_u \left(\ee^{-\ii E u} \chi(u) \right) \partial_{u-s} \left( \ee^{\ii E (u-s)} \chi(u-s) \right)\right],
\end{align}
where we have integrated by parts the distributional derivatives acting on $\mathcal{W}_\text{sing}$. Now yet again performing an integration by parts in the $u$-derivative,
\begin{align}
\mathcal{F}_\text{sing,2}(E,\tau_f,\tau_i) & = \frac{1}{\pi} \int_0^\infty \! ds \, \ln s \int_{-\infty}^\infty \! du \, \text{Re} \left[ \left(\ee^{-\ii E u} \chi(u) \right) \partial^2_{u-s} \left( \ee^{\ii E (u-s)} \chi(u-s) \right)\right] \nonumber \\
& = \frac{1}{\pi} \int_0^\infty \! ds \, \ln s \frac{d^2}{ds^2} \int_{-\infty}^\infty \! du \, \text{Re} \left[ \ee^{- \ii E s} \chi(u) \chi(u-s)  \right] \nonumber \\
& = - \frac{1}{\pi} \int_0^\infty \! \frac{ds}{s} \, \frac{d}{ds} \left(\cos(Es) \int_{-\infty}^\infty \! du \,  \chi(u) \chi(u-s)\right),
\label{Fsing2IBP}
\end{align}
where in the last equality we have performed an integration by parts in the $s$-variable. The boundary term in $s = 0$ vanishes because $\cos(Es) \int_{-\infty}^\infty \! du \,  \chi(u) \chi(u-s)$ is even and, hence, has odd derivative. This also makes the integral over $s$ convergent at small $s$ in the last expression of \eqref{Fsing2IBP}. Writing $\chi(u) \chi(u-s) = [\chi(u)]^2 -\chi(u)[\chi(u)-\chi(u-s)]$ gives
\begin{align}
\mathcal{F}_\text{sing,2}(E,\tau_f,\tau_i) & =  \frac{E}{\pi} \int_0^\infty \! ds \, \frac{\sin(Es)}{s} \int_{-\infty}^\infty \! du \,  [\chi(u)]^2 \nonumber \\
& + \frac{1}{\pi} \int_0^\infty \! \frac{ds}{s} \, \frac{d}{ds} \left(\cos(Es) \int_{-\infty}^\infty \! du \,  \chi(u)[\chi(u) - \chi(u-s)]\right).
\label{Fsing2Final}
\end{align}

The $s$-integral in the first term on the right hand side of eq. \eqref{Fsing2Final} is standard, $\int_0^\infty \! ds \, \sin(Es)/s = (\pi/2)\text{sgn}(E)$ and it is clear that the second term is convergent at small $s$ by integrating by parts over $s$. We can now combine \eqref{Fsing1Final} and \eqref{Fsing2Final} into
\begin{align}
\mathcal{F}_\text{sing}(E,\tau_f,\tau_i) & = -E \Theta(-E) \int_{-\infty}^\infty \! du \,  [\chi(u)]^2 \nonumber \\
& +\frac{1}{\pi} \int_0^\infty \! ds \, \frac{\cos(Es)}{s^2} \int_{-\infty}^\infty \! du \,  \chi(u)[\chi(u) - \chi(u-s)].
\label{FsingFinal}
\end{align}

An alternative expression for $\mathcal{F}_{\text{sing}}$ can be obtained from \eqref{Fsing1Final} by writing $\cos(Es) = 1 -[1-\cos(Es)]$ and using the standard integral
\begin{equation}
\int_0^\infty \! ds \, \frac{1-\cos(Es)}{s^2} = \frac{\pi |E|}{2},
\label{AltRep}
\end{equation}
whereby one obtains
\begin{align}
\mathcal{F}_\text{sing}(E,\tau_f,\tau_i) & = -E \int_{-\infty}^\infty \! du \,  [\chi(u)]^2 + \frac{1}{\pi} \int_0^\infty \!  \frac{ds}{s^2}  \,\int_{-\infty}^\infty \! du \,  \chi(u)[\chi(u) - \chi(u-s)] \nonumber \\
& +\frac{1}{\pi} \int_0^\infty \! ds \, \int_{-\infty}^\infty \! du \,  \chi(u)\chi(u-s) \frac{1- \cos(Es)}{s^2}.
\label{FsingAlt}
\end{align}

It is clear that in the full response,
\begin{align}
\mathcal{F}(E,\tau_f,\tau_i) & = -E \Theta(-E) \int_{-\infty}^\infty \! du \,  [\chi(u)]^2 \nonumber \\
& +\frac{1}{\pi} \int_0^\infty \! ds \, \frac{\cos(Es)}{s^2} \int_{-\infty}^\infty \! du \,  \chi(u)[\chi(u) - \chi(u-s)] \nonumber \\
& +  2 \int_{-\infty}^\infty \! du \, \int_{0}^{\infty} \! ds \, \chi(u) \chi(u-s) \, \text{Re} \left[\ee^{-\ii E s} \left(\mathcal{A}(u,u-s) + \frac{1}{2 \pi s^2} \right) \right]
\label{F1+1}
\end{align}
or alternatively
\begin{align}
\mathcal{F}(E,\tau_f,\tau_i) & = -E \int_{-\infty}^\infty \! du \,  [\chi(u)]^2 + \frac{1}{\pi} \int_0^\infty \!  \frac{ds}{s^2}  \,\int_{-\infty}^\infty \! du \,  \chi(u)[\chi(u) - \chi(u-s)] \nonumber \\
& +\frac{1}{\pi} \int_0^\infty \! ds \, \int_{-\infty}^\infty \! du \,  \chi(u)\chi(u-s) \frac{1- \cos(Es)}{s^2} \nonumber \\
& + 2 \int_{-\infty}^\infty \! du \, \int_{0}^{\infty} \! ds \, \chi(u) \chi(u-s) \, \text{Re} \left[\ee^{-\ii E s} \left(\mathcal{A}(u,u-s) + \frac{1}{2 \pi s^2} \right) \right],
\label{F1+1alt}
\end{align} 
all the dependency on the state of the field is in \eqref{FregFinal}, while \eqref{FsingFinal} (or alternatively \eqref{FsingAlt}) contain purely switching details.


\section{The sharp switching limit and the transition rate}
\label{sec:3rate}

In this section, we consider the sharp switching limit of the derivative-coupling detector in $(1+1)$ dimensions. We consider a family of switching functions given by 
\begin{equation}
\chi(u) = h_1\left(\frac{u - \tau_0 + \delta}{\delta}\right)\times h_2 \left(\frac{-u + \tau + \delta}{\delta}\right),
\label{chi}
\end{equation}
where $\delta$, $\tau_0$ and $\tau$ are real positive parameters such that $\tau > \tau_0$ and where $h_1(x)$ and $h_2(x)$ are smooth, non-negative functions that satisfy $h_1 = h_2 = 0$ for $x < 0$ and $h_1 = h_2 = 1$ for $x >1$. The support of the switching function is compact, $\text{supp}(\chi) = [\tau_i,\tau_f]$ with $\tau_i \doteq \tau_0 - \delta$ and $\tau_f \doteq \tau + \delta$. The detector switches on smoothly during a time interval $\delta$ according to the profile of $h_1$, has a constant detector-field interaction for a duration $\Delta \tau \doteq \tau-\tau_0$ and switches off smoothly according to the profile of $h_2$ during a time interval $\delta$.

The sharp switching limit is obtained as $\delta \rightarrow 0$. At this point, we can exploit the similarity between the $(1+1)$ derivative coupling response \eqref{F1+1alt} with the $(3+1)$ Unruh-DeWitt response in \cite{Louko:2007mu}. Provided that $\mathcal{A}(\tau',\tau'')$ is free of distributional singularities located at  $(\tau',\tau'') = (\tau_0,\tau)$, in which case the singularities are located at the endpoints of the integration interval, one can follow the analysis leading to equations (4.4) and (4.5) in \cite{Louko:2007mu} to find that, for the $(1+1)$ derivative coupling response,
\begin{align}
\mathcal{F}(E,\tau,\tau_0) & = -\frac{E \Delta \tau}{2} + 2 \int_{\tau_0}^\tau  \! du \, \int_0^{u - \tau_0} \! ds \, \text{Re} \left[ \ee^{-\ii E s} \mathcal{A}(u,u-s) + \frac{1}{2 \pi s^2} \right] \nonumber \\
& + \frac{1}{\pi} \ln\left( \frac{\Delta \tau}{\delta} \right) + C + O(\delta),
\label{Fsharp}
\end{align}
where $C$ is a constant that depends only on $h_1$ and $h_2$.

The sharp response function \eqref{Fsharp} diverges logarithmically as $\delta \rightarrow 0$. This divergence depends purely on the switching and is independent of the quantum state and the detector worldline. The transition rate $\dot{\mathcal{F}}(E,\tau,\tau_0) \doteq d\mathcal{F}(E,\tau,\tau_0)/d\tau$ remains finite as $\delta \rightarrow 0$, and in this limit is given by
\begin{align}
\dot{\mathcal{F}}(E,\tau,\tau_0) & = -\frac{E}{2} + 2 \int_0^{\Delta \tau} \! ds \, \text{Re} \left[ \ee^{-\ii E s} \mathcal{A}(\tau,\tau-s) + \frac{1}{2 \pi s^2} \right] + \frac{1}{\pi \Delta \tau}. 
\label{Fdot}
\end{align}

Writing $1 = \cos(Es) +[1-\cos(Es)]$, using the sine integral in the notation of \cite{NIST},
\begin{subequations}
\begin{align}
\text{Si}(z) & \doteq \int_0^z \! dt \, \frac{\sin(z)}{z}, \\
\text{si}(z) & \doteq - \int_z^\infty \! dt \, \frac{\sin(z)}{z} = \text{Si}(z) -\frac{\pi}{2},
\end{align}
\end{subequations}
\\
and the oddness of the sine integral, we can express \eqref{Fdot} with the alternative representation
\begin{align}
\dot{\mathcal{F}}(E,\tau,\tau_0) & = -E \Theta(-E) + \frac{1}{\pi} \left[ \frac{\cos(E\Delta \tau)}{\Delta \tau} + |E| \text{si}(|E|\Delta \tau) \right]\nonumber \\
& + 2 \int_0^{\Delta \tau} \! ds \, \text{Re} \left[ \ee^{-\ii E s} \left( \mathcal{A}(\tau,\tau-s) + \frac{1}{2 \pi s^2} \right)\right].
\label{FdotAlt}
\end{align}

In practical situations, it is useful to consider when the switch-on has taken place in the asymptotic past. As $\tau_0 \rightarrow -\infty$, we have that
\begin{align}
\lim_{\tau_0 \rightarrow -\infty} \dot{\mathcal{F}}(E,\tau, \tau_0) & = -E \Theta(-E) + 2 \int_0^{\infty} \! ds \, \text{Re} \left[ \ee^{-\ii E s} \left( \mathcal{A}(\tau,\tau-s) + \frac{1}{2 \pi s^2} \right)\right].
\label{FdotAsymp}
\end{align}

From this point onward, we shall use the notation $\dot{\mathcal{F}}(E,\tau) \doteq \lim_{\tau_0 \rightarrow -\infty} \dot{\mathcal{F}}(E,\tau, \tau_0)$, whereby suppressing the $\tau_0$ dependence in the arguments of $\mathcal{F}$ indicates that we are considering a (sharp) switch-on in the asymptotic past.

A final remark is that in this long interaction time limit, $\tau_0 \rightarrow -\infty$ the transition rate formula \eqref{FdotAsymp} is valid even when $\mathcal{A(\tau',\tau'')}$ has isolated singularities located at $\tau' \neq \tau''$ and $\mathcal{A}$ cannot be represented by an integrable function in a neighbourhood around such singular points. This can occur, for example, when one is dealing with quantum field theory in manifolds with boundary. Suppose one has a point in the detector timelike worldline parametrised by a proper time $\tau_f$ which can be traced back along a null ray that is reflected at the boundary to a point $\tau_p \in J^-(\tau_f)$ in the worldline. In this case, $\mathcal{A}(\tau_f,\tau_p)$ will be singular. An example of this situation will be encountered below in Chapter~\ref{ch:BH} in the context of a receding mirror spacetime. Along a null curve, a similar situation can occur when photons reintersect their own worldline, for example, when photons orbit a black hole.

\section{The transition rate: stationary case}
\label{sec:3checks}

In this section, we would like to verify a series of reasonable checks on the derivative-coupling detector in stationary situations for detectors that have been switched on sharply in the asymptotic past, $\delta \rightarrow 0$ and $\tau_0 \rightarrow -\infty$.

First, we verify that the transition rate formula in this limit is given precisely by the formula obtained by formally setting $\chi \rightarrow 1$ and factoring out the infinite detection time \cite{Birrell:1982ix}, which yields the spontaneous emission spectrum in Minkowski space. Second, we verify the massive-to-massless consistency limit. This is important because in $1+1$ dimensions the Wightman function has an infrared ambiguous additive constant for massless fields. The hope, that will be realised in this work, is that, because in the derivative coupling model the response depends only on derivatives of the Wightman function, $\mathcal{A}(\tau',\tau'') = \partial_{\tau'} \partial_{\tau''} \mathcal{W}(\tau',\tau'')$, one can verify that the massless limit of the response function satisfies the consistency relation $\lim_{m \rightarrow 0} \dot{\mathcal{F}}_m(E) = \dot{\mathcal{F}}_0(E)$. Finally, we shall demonstrate that the derivative coupling detector thermalises. To this end, we present two examples. We consider a detector submerged in a heat bath and the uniformly accelerated detector interacting with the Minkowski vacuum.

\subsection{Stationary formula}
\label{subsec:DeCoStatic}

Starting out from eq. \eqref{FdotAsymp} and using $A(\tau', \tau'') = \overline{A(\tau'',\tau')}$,
\begin{align}
\dot{\mathcal{F}}(E,\tau) = -E \Theta(-E) + 2 \int_0^{\infty} \! ds \, & \left[ \ee^{-\ii E s} \left( \mathcal{A}(\tau,\tau-s) + \frac{1}{2 \pi s^2} \right) \right. \nonumber \\
& + \left. \ee^{\ii E s} \left( \mathcal{A}(\tau-s,\tau) + \frac{1}{2 \pi s^2} \right)\right].
\end{align}

This formula holds in general. Now, using the translation invariance of the stationary Wightman function to write $\mathcal{A}(\tau,\tau-s) = \mathcal{A}(s,0)$ and $\mathcal{A}(\tau-s,\tau) = \mathcal{A}(-s,0)$, we can extend the limits of integration of eq. \eqref{FdotAsymp} to
\begin{equation}
\dot{\mathcal{F}}(E,\tau) = -E \Theta(-E) + \int_{-\infty}^{\infty} \! ds \, \ee^{-\ii E s} \left( \mathcal{A}(s,0) + \frac{1}{2 \pi s^2} \right),
\label{StaticEpsilon}
\end{equation}
where each term inside the integral has distributional singularities.

We can now evaluate the second term in eq. \eqref{StaticEpsilon}, by deforming the integration contour into the complex lower half-plane near $s = 0$. From the residue theorem and Jordan's lemma,
\begin{subequations}
\begin{align}
\oint_{\gamma_+} \! dz \, \frac{\ee^{-\ii E z}}{z^2} & = \oint_{\gamma_+} \! dz \, \frac{-\ii E}{z} = 2 \pi E,  \hspace{0.2cm} E < 0, \\
\oint_{\gamma_-} \! dz \, \frac{\ee^{-\ii E z}}{z^2} & = 0,  \hspace{0.2cm} E \geq 0,
\end{align}
\end{subequations}
\\
where the closed path $\gamma_+$ closes following an arc at infinity on the upper complex plane enclosing the point $z = 0$ and $\gamma_-$ closes following an arc at infinity on the lower complex plane. We obtain the standard formula
\begin{equation}
\dot{\mathcal{F}}(E,\tau) = \int_{-\infty}^{\infty} \! ds \, \ee^{-\ii E s} \mathcal{A}(s-\ii \epsilon,0),
\label{DeCoStat}
\end{equation}
which is usually obtained by formally setting $\chi = 1$.

\subsection{Massless limit}
\label{subsec:DeComassless}

In this subsection, we shall show that the transition rate is continuous in the mass of the field, $m \geq 0$, for a static detector in Minkowski space. Because of the importance of boundary conditions in field theory, we perform this analysis also for the Minkowski half-space with Dirichlet and Neumann boundary conditions.

Let $(M,g)$ be the $(1+1)$-dimensional Minkowski spacetime. In the standard Minkowski coordinates, $(t,x)$, the spacetime metric tensor reads $g = -dt^2 + dx^2$. The Minkowski half-space is the submanifold $\tilde{M}$ of $M$ specified by $x > 0$, equipped with the metric $\tilde{g}$ inherited from $g$, i.e. $\tilde{g} = g|_{\tilde{M}}$.

In the Minkowski spacetime, the field of mass $m \geq 0$ is in the Minkowski vacuum state $|0\rangle$ which is of positive frequency with respect to the global timelike Killing vector $\xi = \partial_t$. In the Minkowski half-space, the $m \geq 0$ field is in the vacuum state $|\tilde{0}\rangle$ which is also of positive frequency with respect to $\xi$. We consider a detector following the static worldline
\begin{equation}
\mathsf{x}(\tau) = (\tau, d),
\end{equation}
where $d > 0$ is a positive constant devoid of geometric significance in $M$. In $\tilde{M}$ $d$ represents the distance of the detector from the boundary at $x = 0$. The static detector is switched on in the asymptotic past, such that formula \eqref{DeCoStat} gives its transition rate. We now show that $\dot{\mathcal{F}}$ is continuous in the limit $m \rightarrow 0$.


The classical equation of motion on $(\tilde{M},g)$ for a massive scalar field is given by $P f \doteq \left(\Box - m^2\right) \phi = 0$ subject to the boundary conditions
\begin{subequations}
\begin{align}
\text{Dirichlet:}&  \hspace*{0.2cm} \phi = 0, \hspace*{0.2cm} \text{at}  \hspace*{0.1cm} x = 0, \label{Dirichlet}\\
\text{Neumann:}&  \hspace*{0.2cm} \partial_x \phi = 0, \hspace*{0.2cm} \text{at}  \hspace*{0.1cm} x = 0. \label{Neumann}
\end{align} 
\label{BdyHalfSpace}
\end{subequations}

The quantum field, understood as an operator valued distribution, satisfies the equations of motion via the condition
\begin{equation}
\Phi(P f) = \int_{\tilde{M}} d\text{vol}_g [\left(\Box - m^2\right) f] \Phi = 0,
\end{equation}
with $f \in C_0^\infty(\tilde{M})$ satisfying eq. \eqref{Dirichlet} or \eqref{Neumann} for Dirichlet or Neumann boundary conditions respectively. In practice, the best way to obtain the state in the Minkowski half-space is to use the method of images with respect to the Minkowski massive Wightman function,
\begin{equation}
\langle 0 | \Phi(\mathsf{x}) \Phi(\mathsf{x}') | 0 \rangle = \frac{1}{2 \pi} K_0 \left( m \sqrt{(\Delta x)^2 - (\Delta t - \ii \epsilon)^2} \right),
\label{MWm}
\end{equation}
where $K_0$ is the modified Bessel function of the second kind, and the expression is understood as a distribution as $\epsilon \rightarrow 0^+$.

The Wightman function in $(\tilde{M},\tilde{g})$ is equal to
\begin{equation}
\langle \tilde{0} | \Phi(\mathsf{x}) \Phi(\mathsf{x}') | \tilde{0} \rangle = \langle 0 | \Phi(\mathsf{x}) \Phi(\mathsf{x}') | 0 \rangle + \frac{\eta}{2 \pi} K_0 \left( m \sqrt{(x + x')^2 - (\Delta t - \ii \epsilon)^2} \right),
\label{HSWm}
\end{equation}
where $\eta = -1$ for Dirichlet and $\eta = +1$ for Neumann boundary conditions. The image term in \eqref{HSWm} is most efficiently obtained by studying the fundamental solutions of the wave equation. The causal propagator, $E = E^+-E^-$, satisfies $PEf = 0$ for $f \in C_0^\infty(\tilde{M})$. It follows that every solution of $P\phi = 0$ subject to the boundary conditions is of the form $\phi = Ef$, i.e.,
\begin{equation}
P_\mathsf{x} \phi(\mathsf{x}) = \int_{\tilde{M}} d\text{vol}(\mathsf{x}') \left(\Box - m^2\right) E(\mathsf{x},\mathsf{x}') f(\mathsf{x}') = 0.
\end{equation}

Denoting $E_M$ the Causal propagator in Minkowski, the Dirichlet boundary conditions,
\begin{equation}
\phi(t,0) = \int_{\tilde{M}} \! dt'\, dx' \, E((t,0),(t',x')) f(t',x') = 0,
\end{equation}
imply that $E((t,0),(t',x')) = 0$. It follows that $E((t,x),(t',x')) = E_M((t,x),(t',x')) - E_M((t,-x),(t',x'))$, which accounts for the image term in eq. \eqref{HSWm} with $\eta = -1$.

Similarly, the Neumann boundary conditions,
\begin{equation}
\partial_x \phi(t,0) = \int_{\tilde{M}} \! dt'\, dx' \, \partial_x E((t,0),(t',x')) f(t',x') = 0,
\end{equation}
imply that $\partial_x E((t,0),(t',x')) = 0$. It follows that $E((t,x),(t',x')) = E_M((t,x),(t',x')) + E_M((t,-x),(t',x'))$, which accounts for the image term in eq. \eqref{HSWm} with $\eta = +1$.

As we have discussed in Chapter~\ref{ch:chap2}, the na\"ive massless limit of the Minkowski Wightman function diverges. However, the massless limit of $\langle 0 | \Phi(\mathsf{x})\Phi(\mathsf{x}') | 0 \rangle - (2\pi)^{-1} \ln\left[m e^\gamma/(2\mu)\right]$, where $\gamma$ is Euler's constant and $\mu$ is a positive constant with inverse length dimensions, remains finite as $m \rightarrow 0$ \cite{NIST}. We define the massless Minkowski Wightman function as

\begin{align}
\langle 0 | \Phi(\mathsf{x}) \Phi(\mathsf{x}') |0 \rangle_{0,\mu} & \doteq \lim_{m \rightarrow 0} \langle 0 | \Phi(\mathsf{x}) \Phi(\mathsf{x}') | 0 \rangle - (2\pi)^{-1} \ln\left[m e^\gamma/(2\mu)\right] \nonumber \\
& = -\frac{1}{2 \pi} \ln \left[ \mu \sqrt{(\Delta x)^2 - (\Delta t - \ii \epsilon)^2} \right].
\label{MW0}
\end{align}

Similarly, we define the massless half-space Wightman function as

\begin{align}
\langle \tilde{0} | \Phi(\mathsf{x}) \Phi(\mathsf{x}') |\tilde{0} \rangle_{0,\mu} \doteq \langle 0 | \Phi(\mathsf{x}) \Phi(\mathsf{x}') |0 \rangle_0 -\frac{\eta}{2 \pi} \ln \left[ \mu \sqrt{(x+x')^2 - (\Delta t - \ii \epsilon)^2} \right].
\label{HSW0}
\end{align}

Notice that the way in which we have defined the massless Wightman functions makes the massless limit discontinuous. What we wish to argue is that this mass discontinuity is irrelevant for an observer who interacts with the matter field using a particle detector coupled via \eqref{3:DeCoHint}. To support this argument, we proceed to show that the instantaneous rate along a static worldline in the Minkowski half-space is continuous in the massless limit. More precisely

\begin{thm}
Let $\dot{\mathcal{F}}_m$ be the instantaneous transition rate of a sharply switched static \eqref{DeCoStat} detector following the trajectory $(t,x) = (\tau,d)$ in the Minkowski half-space and coupled through \eqref{DeCo} in the sharp switching limit, to a massive Klein-Gordon field of mass $m$ subject to Dirichlet or Neumann boundary conditions \eqref{BdyHalfSpace} in the state \eqref{HSWm}
\begin{equation}
\langle \tilde{0} | \Phi(\mathsf{x}) \Phi(\mathsf{x}') | \tilde{0} \rangle = \langle 0 | \Phi(\mathsf{x}) \Phi(\mathsf{x}') | 0 \rangle + \frac{\eta}{2 \pi} K_0 \left( m \sqrt{(x + x')^2 - (\Delta t - \ii \epsilon)^2} \right).
\end{equation}

Concordantly, let $\dot{\mathcal{F}}_0$ be the instantaneous transition rate of such detector coupled instead to a massless scalar field in the state \eqref{HSW0}
\begin{align}
\langle \tilde{0} | \Phi(\mathsf{x}),\Phi(\mathsf{x}') |\tilde{0} \rangle_{0,\mu} \doteq \langle 0 | \Phi(\mathsf{x}),\Phi(\mathsf{x}') |0 \rangle_0 -\frac{\eta}{2 \pi} \ln \left[ \mu \sqrt{(\Delta x)^2 - (\Delta t - \ii \epsilon)^2} \right].
\end{align}

It then holds that $\lim_{m \rightarrow 0} \dot{\mathcal{F}}_m = \dot{\mathcal{F}}_0$.
\label{thm:DecoMassless}
\end{thm}
\begin{proof}
The massive Wightman function,
\begin{align}
\langle \tilde{0} | \Phi(\mathsf{x}) \Phi(\mathsf{x}') | \tilde{0} \rangle & = \frac{1}{2 \pi} K_0 \left( m \sqrt{(\Delta x)^2 - (\Delta t - \ii \epsilon)^2} \right) \nonumber \\
& + \frac{\eta}{2 \pi} K_0 \left( m \sqrt{(x + x')^2 - (\Delta t - \ii \epsilon)^2} \right),
\end{align}
can be pulled back along the static detector trajectory and differentiated twice to yield
\begin{align}
\mathcal{A}(s,0) = -\frac{d^2}{ds^2} \left[ \frac{1}{2 \pi} K_0 \left( \ii m \sqrt{ (s - \ii \epsilon)^2} \right) + \frac{\eta}{2 \pi} K_0 \left( m \sqrt{4d^2 - (s - \ii \epsilon)^2} \right) \right].
\label{Am}
\end{align}

Using the representation \eqref{FdotAsymp} of the static transition rate, and inserting eq. \eqref{Am} into this equation, we see that the formula for the sharp, static transition rate splits into the integral of a regular function and an integral with distributional contributions, which we denote
\begin{subequations}
\begin{align}
\dot{\mathcal{F}}^\text{reg}_m(E) & \doteq -E \Theta(-E) + 2 \int_0^{\infty} \! ds \, \text{Re} \left[ \ee^{-\ii E s} \left( \frac{m^2}{2 \pi} K_0'' \left( \ii m s \right) + \frac{1}{2 \pi s^2} \right)\right] \\
\dot{\mathcal{F}}^\text{dist}_m(E) & \doteq \frac{\eta}{2 \pi} \int_{-\infty}^\infty \! ds \, \ee^{-\ii E s}  \frac{d^2}{ds^2} K_0 \left( m \sqrt{4d^2 - (s - \ii \epsilon)^2} \right) 
\end{align}
\end{subequations}
\\
such that $\dot{\mathcal{F}}_m = \dot{\mathcal{F}}^\text{reg}_m + \dot{\mathcal{F}}^\text{dist}_m$. A prime denotes a derivative with respect to the argument. Following the steps in Appendix~\ref{sec:app3:HalfSpace}, we find that
\begin{equation}
\dot{\mathcal{F}}_m(E) = \left[\frac{E^2}{\sqrt{E^2 - m^2}} + \frac{\eta E^2}{\sqrt{E^2 - m^2}} \cos\left(2d \sqrt{E^2-m^2} \right)\right] \Theta(-E - m).
\label{3:Fm}
\end{equation}

In the massless case, inserting
\begin{equation}
\mathcal{A}_0(s,0) = -\frac{d^2}{ds^2} \left[-\frac{1}{2 \pi} \ln \left[ \mu \sqrt{- (s - \ii \epsilon)^2} \right]  -\frac{\eta}{2 \pi} \ln \left[ \mu \sqrt{4d^2 - (s - \ii \epsilon)^2} \right]\right].
\label{A0}
\end{equation}
into \eqref{StaticEpsilon} one obtains
\begin{equation}
\dot{\mathcal{F}}_0(E) = -E \Theta(-E) - \frac{\eta}{2 \pi} \int_{-\infty}^\infty \! ds \, \ee^{-\ii E s} \frac{d^2}{ds^2}\ln \left[ \mu \sqrt{4d^2 - (s - \ii \epsilon)^2} \right].
\end{equation}

This integral is evaluated by elementary complex integration methods. We detail this in Appendix~\ref{sec:app3:HalfSpacem0}. One obtains
\begin{equation}
\dot{\mathcal{F}}_0(E) = \left[-E - \eta E \cos(2d E)\right] \Theta(-E).
\label{3:F0}
\end{equation}

It follows from eq. \eqref{3:Fm} and \eqref{3:F0} that
\begin{equation}
\lim_{m \rightarrow 0} \dot{\mathcal{F}}_m(E) = \dot{\mathcal{F}}_0(E).
\end{equation}
\end{proof}

\begin{cor}
Let instead $\dot{\mathcal{F}}_m$ be the instantaneous transition rate of a static detector coupled to a massive Klein-Gordon field in the Minkowski vacuum state, and let  $\dot{\mathcal{F}}_0$ be the instantaneous transition rate of such detector coupled to a massless Klein Gordon field. Then $\lim_{m \rightarrow 0} \dot{\mathcal{F}}_m = \dot{\mathcal{F}}_0$.
\label{cor:DecoMassless}
\end{cor}
\begin{proof}
Setting $\eta = 0$,
\begin{subequations}
\begin{align}
\dot{\mathcal{F}}_m(E) & = \frac{E^2}{\sqrt{E^2 - m^2}}  \Theta(-E - m), \\
\dot{\mathcal{F}}_0(E) & = -E \Theta(-E),
\end{align}
\end{subequations}
\\
and the result follows.
\end{proof}

\subsection{Thermalisation}
\label{subsec:DeCotherm}

We now verify that the derivative coupling detector in $1 + 1$ dimensions responds to a thermal state by the detailed balance condition. In Minkowski spacetime there are two situations of prime interest. First, an inertial detector immersed in a heat bath. Second, a uniformly accelerated detector that registers the Unruh effect.

\subsubsection{Inertial detector in a thermal bath}

We consider a massless $1+1$ field in Minkowski space in a thermal state with positive temperature $T$. The Wightman function can be obtained from the resummation of the formal non-convergent sum
\begin{equation}
-\frac{1}{4\pi} \sum_{n = -\infty}^\infty \ln \left[\mu\left((\Delta x)^2 - (\Delta t - \ii \epsilon + \ii n T)^2 \right) \right].
\label{3:formalSum}
\end{equation}

Differentiating \eqref{3:formalSum} with respect to $\Delta x$ term by term yields a convergent sum that converges to the elementary function defined by
\begin{equation}
g(T, \Delta x, \Delta t) \doteq \frac{T}{4} \left[\coth\left(\pi T(-\Delta x + \Delta t - \ii \epsilon) \right) - \coth\left( \pi T (\Delta x + \Delta t - \ii \epsilon) \right) \right].
\label{3:convergentSum}
\end{equation}

The function $g$ has a tower of infinite periodic singularities in the complex plane located at $\Delta x = \pm \Delta t - \ii(n/T \mp \epsilon)$, $n \in \mathbb{Z}$, with period $1/T$. We define the thermal Wightman function (up to the $1+1$ massless ambiguity) as the improper integral of this elementary function, fixing the arbitrary $\Delta t$-dependent constant in such a way that our definition is a Green function of the massless Klein-Gordon equation, and requiring evenness for $\Delta t$ for $(\Delta x)^2 - (\Delta t)^2 > 0$, \textit{i.e.},
\begin{align}
\langle T | \Phi(x) \Phi(x') | T \rangle  \doteq & \int \! d(\Delta x) \, g(T, \Delta x, \Delta t) + \text{const}(\Delta t) \nonumber \\
 = &-\frac{1}{4 \pi} \ln\left\{ \sinh\left[ \pi T (\Delta x + \Delta t - \ii \epsilon ) \right] \right\} \nonumber \\
& -\frac{1}{4 \pi} \ln\left\{ \sinh\left[ \pi T (\Delta x - \Delta t + \ii \epsilon ) \right] \right\}.
\label{3:ThermalWightman}
\end{align}

\begin{rem}
The thermal two-point function has periodic logarithmic singularities at $\Delta t = \Delta x + \ii(n/T \mp \epsilon)$ on the left-moving sector and at $\Delta t = - \Delta x - \ii(n/T \mp \epsilon)$ on the right-moving sector.
\end{rem}

We consider an inertial detector with worldline $\x(\tau) = (\tau \cosh \lambda, - \tau \sinh \lambda)$, where $\lambda \in \mathbb{R}$ is the detector's rapidity parameter with respect to the rest frame of the heat bath. We consider a switch-on in the asymptotic past. The $\mathcal{A}$ bi-distribution is obtained by differentiating twice the pullback of the Wightman function, and the transition rate is given by the stationary formula \eqref{DeCoStat}. The transition rate can be written in terms of the left- and right-moving contributions as $\dot{\mathcal{F}}(E) = \dot{\mathcal{F}}_{T_+}(E) + \dot{\mathcal{F}}_{T_-}(E)$, with
\begin{align}
\dot{\mathcal{F}}_{T_\pm}(E) \doteq -\frac{1}{16 \pi} \int_{-\infty}^{\infty} \! ds \, \ee^{-\ii E s}  \frac{(2 \pi T_\pm)^2}{\sinh^2 \left[\pi T_\pm\left(s-\ii \epsilon\right) \right]},
\label{3:DeCoThermal}
\end{align}
where $T_\pm \doteq \ee^{\pm \lambda} T$.

\begin{thm}
Each of the the left- and right-moving parts of the transition rate, defined by $\dot{\mathcal{F}}_{T_-}$ and $\dot{\mathcal{F}}_{T_+}$ in \eqref{3:DeCoThermal}, satisfies the detailed balance condition at temperature, $T_- \doteq \ee^{- \lambda} T$ and $T_+ \doteq \ee^{+ \lambda} T$ respectively.
\label{thm:DecoHeatBath}
\end{thm}
\begin{proof}
We deform the contour integral in \eqref{3:DeCoThermal} to $s = -\ii \pi /a + r$, where $r \in \mathbb{R}$ and, using formula 3.985 in \cite{Gradshteyn:2007}, we find
\begin{equation}
\dot{\mathcal{F}}_{T_\pm}(E) = \frac{E}{2\left(\ee^{E/T_\pm} - 1\right)}.
\label{DeCoHeatPlanckian}
\end{equation}

The detailed balance condition,
\begin{equation}
\frac{1}{T} = \frac{1}{E} \ln \left( \frac{\dot{\mathcal{F}}_\pm(-E)}{\dot{\mathcal{F}}_\pm(E)} \right) = \frac{1}{T_\pm}
\end{equation}
is satisfied at the correct temperatures.
\end{proof}

The content of the theorem is that the rate of response of a detector (or ensemble thereof) that is travelling inertially with respect to the bath will detect left-moving and right-moving particles at a temperature $T$, compensated by the appropriate Doppler shift due to the relative velocity of the detector with respect to the bath.   

A detector that is static with respect to the bath will thermalise exactly at the heat bath temperature, as can be seen by examining the total rate of the detector,
\begin{equation}
\dot{\mathcal{F}}(E) = \frac{E}{2} \left(\frac{1}{\ee^{E/T_+} - 1} + \frac{1}{\ee^{E/T_-} - 1}\right),
\label{3:ThermalBathRate}
\end{equation}
which simplifies in the special case $\lambda = 0$ to
\begin{equation}
\dot{\mathcal{F}}(E) = \frac{E}{\ee^{E/T} - 1}.
\label{3:DecoThermalStat}
\end{equation}

Eq. \eqref{3:DecoThermalStat} conforms fully with the $3+1$ expectation, where an Unurh-DeWitt detector traveling statically in a heat bath thermalises at the heat bath temperature.

\subsubsection{Unruh effect}

Let again $(M,g)$ be $1+1$ Minkowski spacetime and $\Phi$ be a massless scalar field in the Poincar\'e invariant Minkowski vacuum. The trajectory of a uniformly accelerated detector confined to the right Rindler wedge is the integral curve of the boost Killing vector field $\xi \doteq t \partial_x + x \partial_t$, given by the solution of the equation $\dot{\mathsf{x}}(\tau) = \xi(\mathsf{x}(\tau))$. The solution, parametrised by the proper time normalised to $g(\dot{\mathsf{x}},\dot{\mathsf{x}}) = -1$, is
\begin{equation}
\mathsf{x}(\tau) = \left(a^{-1} \sinh(a \tau), a^{-1} \cosh(a \tau)\right),
\end{equation}
where $a>0$ is the magnitude of the proper acceleration. By construction, the trajectory is stationary along $\xi$. 

With the detector switch-on pushed to the asymptotic past, the transition rate is independent of time and we can use the stationary formula \eqref{DeCoStat}. Differentiating twice the pullback of the Wightman function and using eq. \eqref{DeCoStat} we have that
\begin{equation}
\dot{\mathcal{F}}(E) = \int_{-\infty}^{\infty} \! ds \, \ee^{-\ii E s} \left(-\frac{a^2}{8 \pi \sinh^2 (a(s-\ii \epsilon)/2)} \right).
\label{DeCoUnruh}
\end{equation}

\begin{thm}
The transition rate \eqref{DeCoUnruh} satisfies the detailed balance condition at the Unruh temperature, $T_\text{U} \doteq a/2 \pi$.
\label{thm:DecoUnruh}
\end{thm}
\begin{proof}
As before, we deform the contour integral to $s = -\ii \pi /a + r$, where $r \in \mathbb{R}$ and, using formula 3.985 in \cite{Gradshteyn:2007}, we find
\begin{equation}
\dot{\mathcal{F}}(E) = \frac{E}{\ee^{2 \pi E/a} - 1},
\label{DeCoUnruhPlanckian}
\end{equation}
which satisfies the detailed balance condition at the Unruh temperature,
\begin{equation}
\frac{1}{T} = \frac{1}{E} \ln \left( \frac{\dot{\mathcal{F}}(-E)}{\dot{\mathcal{F}}(E)} \right) = \frac{2\pi}{a}.
\end{equation}
\end{proof}

The form of eq. \eqref{DeCoUnruhPlanckian} is the usual Planckian spectrum encountered in $3+1$ dimensions for a non-derivative detector coupled to a massless field along uniformly accelerated trajectories. Moreover, it is known that in even dimensions the spectrum arising from the Unruh effect is bosonic, whereas in odd dimensions it is fermionic \cite{Takagi:1986kn}, and result also conforms to this expectation.

With the confidence that the $1+1$ detector that we have introduced satisfies the conditions verified in Section~\ref{sec:3checks}, we are set to examine phenomena on curved spacetimes and non-stationary situations in Chapter~\ref{ch:BH}.

\chapter{Quantum aspects of black holes}
\label{ch:BH}

The study of quantum field theory in black hole spacetimes was pioneered by Hawking, with the discovery that black holes emit thermal radiation at a temperature proportional to their surface gravity \cite{Hawking:1974rv, Hawking:1974sw}. Understanding this radiation has been an effort at the forefront of a number of physical investigations, mainly because of the thermal character of radiation. Suppose a pure state is defined on a Cauchy hypersurface at early stages of the collapse of a star that will form a black hole. Once the black hole has formed, the state defined in the asymptotic future infinity appears to be thermal, because there exist correlations with field modes inside the black hole. The future infinity is not Cauchy in the whole spacetime and the modes inside the black hole region need to be traced. As a result of the collapse, the state will appear to have evolved from a pure state to a mixed state. This apparent pure to mixed evolution for an observer that only has access to the exterior region of the black hole should not be alarming. It simply indicates that there exist field modes in the interior black hole region to which one has no access from the exterior region \cite{Wald:1995yp}. 

Nevertheless, the conventional wisdom affirms that if the back-reaction of such black hole radiation onto the spacetime is considered, the black hole will radiate all of its mass and the information, initially encoded in a pure state, will be lost in the form of thermal radiation in the asymptotic future. This unresolved problem bears the name of the information paradox \cite{Hawking:1976ra}. The resolution of this paradox has stimulated much research. We simply mention non-exhaustively a few lines of attack. These include the proposal of brick-wall conditions at the horizon \cite{'tHooft:1996tq}, quantum gravity resolutions \cite{Ashtekar:2008jd}, non-unitary quantum dynamics  and firewalls \cite{Almheiri:2012rt}. More recently, it has been proposed that the information of infalling matter is stored at the black hole horizon as \textit{soft hair} \cite{Hawking:2016msc}. 

The objective of this chapter is to put forward several calculations in $1+1$ dimensions that allow analytical control in many interesting situations that occur in black hole spacetimes, especially in strongly time-dependent situations, in which analytic control in $3+1$ dimensions is more difficult. The message that we wish to convey is that many of the interesting features of full $(3+1)$-dimensional black holes survive in the $(1+1)$-dimensional setting, where the conformal symmetry of the theory allows one to compute analytically interesting quantities, such as the transition rate of a detector and the renormalised stress-energy tensor. An important remark is that the issue of back-reaction is a subtle one in $1+1$ dimensions, because the Einstein tensor vanishes identically, and the Einstein field equations are trivial. We do not deal with this issue, but one can certainly venture in this direction. 

In Section \ref{sec:4Classical}, we discuss classical aspects of black holes. First, we introduce a receding mirror spacetime \textit{\`a la} Davies-Fulling \cite{Carlitz:1986nh, Davies:1976hi, Davies:1977yv, Good:2013lca, Wang:2013lex} that mimics the process by which a spherically symmetric collapsing star forms a Schwarzschild black hole. We consider a mirror receding to the left in $1+1$ Minkowski spacetime and impose Dirichlet boundary conditions along the mirror trajectory, in such a way that the asymptotic past matches the Minkowski half-space. The motion of the mirror produces a shift in the classical modes that mimics the redshift of field modes in a collapsing star at late times. Second, we introduce the $1+1$ Schwarzschild spacetime and discuss the classical features of a massless scalar field on this spacetime. Third, we consider the $(1+1)$-dimensional Reissner-Nordstr\"om spacetime. This is especially interesting due to the appearance of future and past Cauchy horizons, a feature that is shared by the astrophysically relevant $(3+1)$-dimensional Kerr black hole. 

Sections \ref{sec:4Quantisation} and \ref{sec:4QuantumEffects} deal with quantum aspects of black holes. In Section \ref{sec:4Quantisation} we shall perform the quantisation of massless scalar fields in the spacetimes introduced in Section \ref{sec:4Classical}. We use conformal techniques that allow one to obtain the two-point function directly. This is the main advantage of massless $1+1$ dimensional field theories. In the case of the receding mirror spacetime, we choose a field state that resembles the Minkowski half-space vacuum in the asymptotic past. In the Schwarzschild and generalised Reissner-Nordstr\"om spacetimes, we are concerned with massless fields in the Unruh and Hartle-Hawking-Israel vacua.

In Section \ref{sec:4QuantumEffects} we study several quantum effects of black holes. The main result in the receding mirror spacetime is the thermalisation of the rate of inertial detectors in the asymptotic future, by a radiation flux coming from the receding mirror. In the $1+1$ Schwarzschild spacetime we verify that static detectors thermalise due to the Hawking radiation, and we verify the loss of thermality of inertial infalling detectors. Their asymptotic behaviour, both near infinity and as they approach the curvature sigularity, is computed analytically, while the interpolating behaviour is calculated numerically. We also consider the rate of an inertial detector that switches on near the white hole singularity and follows a ``straight up" trajectory across the bifurcation point of the Killing horizons into the black hole and analyse numerically the transition rate as a function of the detector's energy gap, $E$, and of the detector's position relative to the white hole and black hole singularities. In the $1+1$ Reissner-Nordstr\"om spacetime we calculate the strength at which the detector's transition rate diverges as the detector approaches the Cauchy horizon, as well as the renormalised local energy along the worldline of such detector.

Finally, in Section \ref{sec:4Rindler}, we discuss the validity and limitations of the $1+1$ models that we have examined in the light of our expectations in $3+1$ situations.

Throughout this chapter we set $G_\text{N}=1$, but continue to work in the regime of no back-reaction of the fields on the metric.

\section{Classical aspects of black holes in $1+1$ dimensions}
\label{sec:4Classical}

We now introduce the classical geometry of the aforementioned spacetimes. In Subsection \ref{subsec:4Mirror}, we present the receding mirror spacetime that we use to model the collapse of a star in $1+1$ dimensions. In Subsection \ref{subsec:4Schw} we introduce the $1+1$ Schwarzschild black hole. In Subsection \ref{subsec:4RN} we generalise the $(1+1)$-dimensional Reissner-Nordstr\"om black hole to a family of black hole spacetimes which have the same causal structure. We refer to this class of spacetimes as the generalised Reissner-Nordstr\"om black hole.

In this chapter, we deal with asymptotically flat spacetimes, that are, additionally, \textit{strongly asymptotically predictable spacetimes}\footnote{We follow Wald \cite{Wald:1984rg}, chap. 11 and 12, in our definition, but see also the classical monograph by Hawking and Ellis, chap. 9 \cite{Hawking:1974rv}.}.

\begin{defn}
An asymptotically flat spacetime $(M,g,t)$ with a conformal map to the spacetime $(\tilde{M}, \tilde{g},\tilde{t})$, $f: M \to f(M) \subset \tilde{M}$, $f^*\tilde{g}|_{f(M)} = \Omega^2  g$, $\tilde{t} = f_*t$, is \textit{strongly asymptotically predictable} if there exists an open region $\tilde{U} \subset \tilde{M}$, such that $\overline{M \cap J^-\left(\mathscr{I}^+\right)} \subset \tilde{U}$ and $(\tilde{U}, \tilde{g})$ is a globally hyperbolic spacetime.
\end{defn}

We are obliged to provide the following important definition:

\begin{defn}
Let $(M,g,t)$ be a strongly asymptotically predictable spacetime. $(M,g,t)$ is said to contain a \textit{black hole} if $M$ is not contained in the causal past of future null infinity, $J^-\left(\mathscr{I}^+\right)$. We call $B \doteq M \setminus J^-\left(\mathscr{I}^+\right)$ the black hole region of the spacetime and the event horizon is the intersection of the boundary of the causal past of future null infinity with $M$, i.e., $H \doteq \partial J^-\left(\mathscr{I}^+\right) \cap M$.
\end{defn}

This definition formalises the intuition that black holes are regions of spacetime from which classical signals cannot escape to infinity.

A final technical point that we need to clarify is that the spacetimes that we introduce in this chapter are \textit{geodesically convex}. This property, which strengthens the notion of geodesic neighbourhoods introduced in Def. \ref{2:Synge}, will allow us to perform asymptotic analysis unambiguously, in the spacetimes that we will define below, in the situations in which we are interested.

\begin{defn}
A spacetime $(M,g)$ is said to be a geodesically convex set if, given any two points in $M$, there is a unique geodesic between the two points lying in $M$.
\end{defn}

\subsection{The receding mirror spacetime}
\label{subsec:4Mirror}

Let $(M,g)$ be the Minkowski spacetime in $1+1$ dimensions. The spacetime metric can be written in terms of the global null coordinates $g = -dudv$, where $u \doteq t-x$ and $v \doteq t+x$. We consider a mirror moving in this spacetime with a trajectory given by
\begin{equation}
v = p(u) = -\frac{1}{\kappa} \ln\left(1+ \ee^{-\kappa u}\right),
\label{mirror}
\end{equation}
where $\kappa$ is a positive constant. Writing eq. \eqref{mirror} as $\ee^{-\kappa v} = 1 + \ee^{-\kappa u}$ allows one to parametrise the trajectory with the proper time $\tau \in (-\infty, 0)$ by
\begin{subequations}
\begin{align}
u = -\frac{2}{\kappa} \ln \left[ \sinh(-\kappa \tau/2)\right], \\
v = -\frac{2}{\kappa} \ln \left[ \cosh(-\kappa \tau/2)\right].
\end{align}
\label{4:uvmirror}
\end{subequations}

The proper velocity, $\dot{\mathsf{x}}$, and proper acceleration $\ddot{\mathsf{x}}$ are towards the negative $x$ coordinate. The dots represent derivatives with respect to the proper time. The magnitude of the acceleration, $|\ddot{\mathsf{x}}| = \sqrt{-g\left(\ddot{\mathsf{x}},\ddot{\mathsf{x}}\right)}$, is equal to $|\ddot{\mathsf{x}}| = (\kappa/2)  \csch(\kappa \tau/2) \sech(\kappa \tau/2) = \kappa \csch(\kappa \tau)$. At early times, $\tau \rightarrow - \infty$, the trajectory is asymptotically inertial, $|\ddot{\mathsf{x}}| \rightarrow 0$, with a proper acceleration vanishing exponentially fast. At late times, the mirror recedes to infinity as $\tau \rightarrow 0_-$ and the trajectory asymptotes the null line $v = 0$ from below. The proper acceleration diverges as $-1/\tau$. A spacetime diagram is shown in Fig. \ref{RecMir}.

\begin{figure}[!ht]
\centering
 \includegraphics[width=0.7\textwidth]{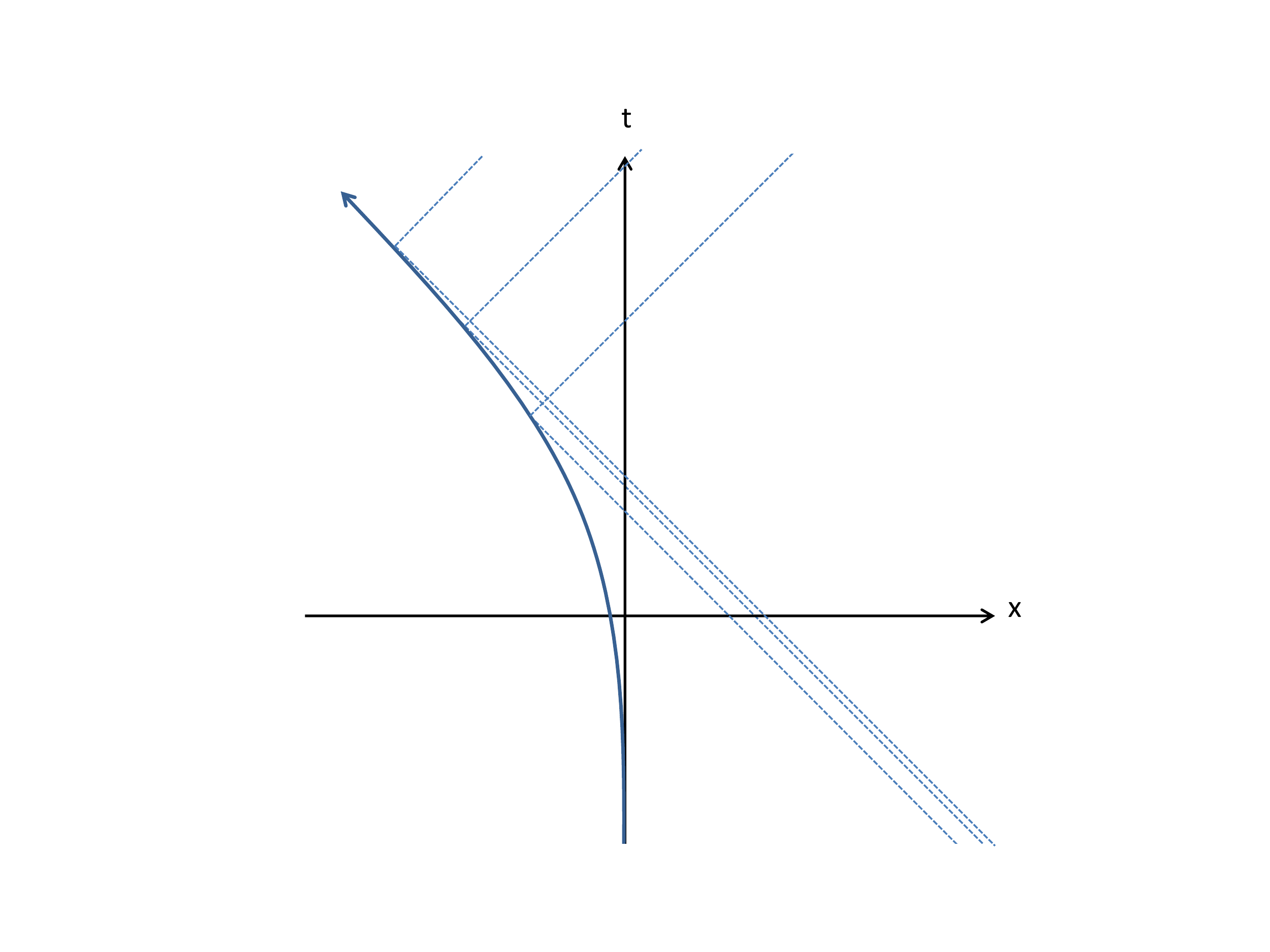}
\caption[Receding mirror spacetime]{Minkowski spacetime with a mirror receding according to eq. \eqref{mirror}. The receding mirror spacetime corresponds to the right hand side of the mirror. Dashed lines show
a selection of null geodesics that bounce off the mirror.}
\label{RecMir}
\end{figure}

We now consider the spacetime $\left(M',g\right)$ which consists of the Minkowski spacetime submanifold at the right of the mirror equipped with the inherited Minkowski metric. This is the receding mirror spacetime.

\subsection{The $1+1$ Schwarzschild black hole}
\label{subsec:4Schw}

The $(1+1)$-dimensional Schwarzschild spacetime is constructed by analogy with the $(3+1)$-dimensional counterpart. The $3+1$ Schwarzschild black hole consists of the real manifold $\mathbb{R}^2 \times S^2$ equiped with a metric $g_{3+1}$ which in the exterior region, $r > 2M$, is locally given by
\begin{equation}
g_{3+1}|_{r>2M}(t,r,\theta,\phi) = -F(r) dt^2 + F(r)^{-1} dr^2 + r^2 d\Omega^2(\theta,\phi),
\end{equation}
where $d\Omega^2$ is the Riemannian metric on $S^2$ and $F(r) = 1-2M/r$. $M > 0$ is a length parameter  proportional to the ADM mass\footnote{The ADM mass or ADM energy of a spacetime was introduced in \cite{Arnowitt:1959ah}. See also standard texts, such as \cite{Wald:1984rg}.} of the black hole . The $1+1$ Schwarzschild black hole is obtained by dropping the $S^2$-dependence in the $3+1$ metric. In other words, it consists of the manifold $(M,g)$ with $M = \mathbb{R}^2$ and metric in the exterior region given by
\begin{equation}
g|_{r>2M}(t,r) = -F(r)dt^2 + F(r)^{-1} dr^2.
\end{equation}

The exterior region is static with Killing vector field $\xi \doteq \partial_t$. An event horizon is located at $r = 2M$ and the spacetime is regular across the event horizon $r_+$ as is evident by introducing the Eddington-Finkelstein time coordinate $v \doteq t + r^*$, where $dr^* \doteq dr/F(r)$ is the tortoise coordinate. The surface gravity $\kappa$ with respect to $\xi = \partial_v$ satisfies the eigenvalue equation $\nabla_\xi \xi = \kappa \xi$ with solution $\kappa = 1/(4M)$. 

Using $\kappa$ we can perform the Kruskal extension of the Schwarzschild spacetime. We introduce null coordinates in the exterior region $(u, v) \doteq (t-r^*, t+r^*)$. We cover the region $r > 0$ with the Kruskal-Szekeres chart $(U,V) \doteq (-\exp(-\kappa u), \exp(\kappa v))$. We extend $(U,V)$ to cover four regions as indicated in Table \ref{Table:Kruskal}.

\begin{table}[h]
\begin{center}
\begin{tabular}{ l l c r }
  & Region & sgn$(U)$ & sgn$(V)$ \\
  \hline
  I: Exterior &$2M < r$ & $-1$ & $+1$ \\
  II: Black hole &$0< r < 2M$ & $+1$ & $+1$ \\
  III: Isometric exterior &$2M < r$ & $+1$ & $-1$ \\
  IV: White hole &$0< r < 2M$ & $-1$ & $-1$ 
\end{tabular}
\caption[Kruskal-Szekeres coordinates for the Schwarzschild extension]{Coordinate charts of the Kruskal extension of the Schwarzschild black hole.}
\label{Table:Kruskal}
\end{center}
\end{table}

Figure \ref{SchwConformal} depicts the four regions in a conformal diagram. The new manifold, $\mathbb{R}^2 \times S^2$, is the Kruskal extension of the Schwarzschild black hole. The spacetime contains a bifurcate Killing horizon, $H^- \cup H^+$, at $r = 2M$, generated by the Kruskal Killing vector $\xi \doteq (4M)^{-1}(-U \partial_U + V \partial_V)$. $\xi$ is timelike in the regions I and III and spacelike in regions II and IV, and takes the form $\xi = \partial_t$ in region I. We call $H^-$, the null surface generated by $\partial_U$, the past horizon and $H^+$, the null surface generated by $\partial_V$, the future horizon. We continue to denote the extension by the pair $(M,g)$. The metric is given locally by
\begin{equation}
g(U,V) = F(r)\left( \kappa ^2 U V \right)^{-1} dU dV
\label{4:SchwKruskalMetric}
\end{equation}
in the whole Kruskal manifold, except for $UV = 0$, where the metric is understood to be given by the limit of eq. \eqref{4:SchwKruskalMetric}. $H^+$ is located at $U = 0$ and $H^-$ is located at $V=0$.

\begin{figure}
\begin{center}
\includegraphics[height=8cm]{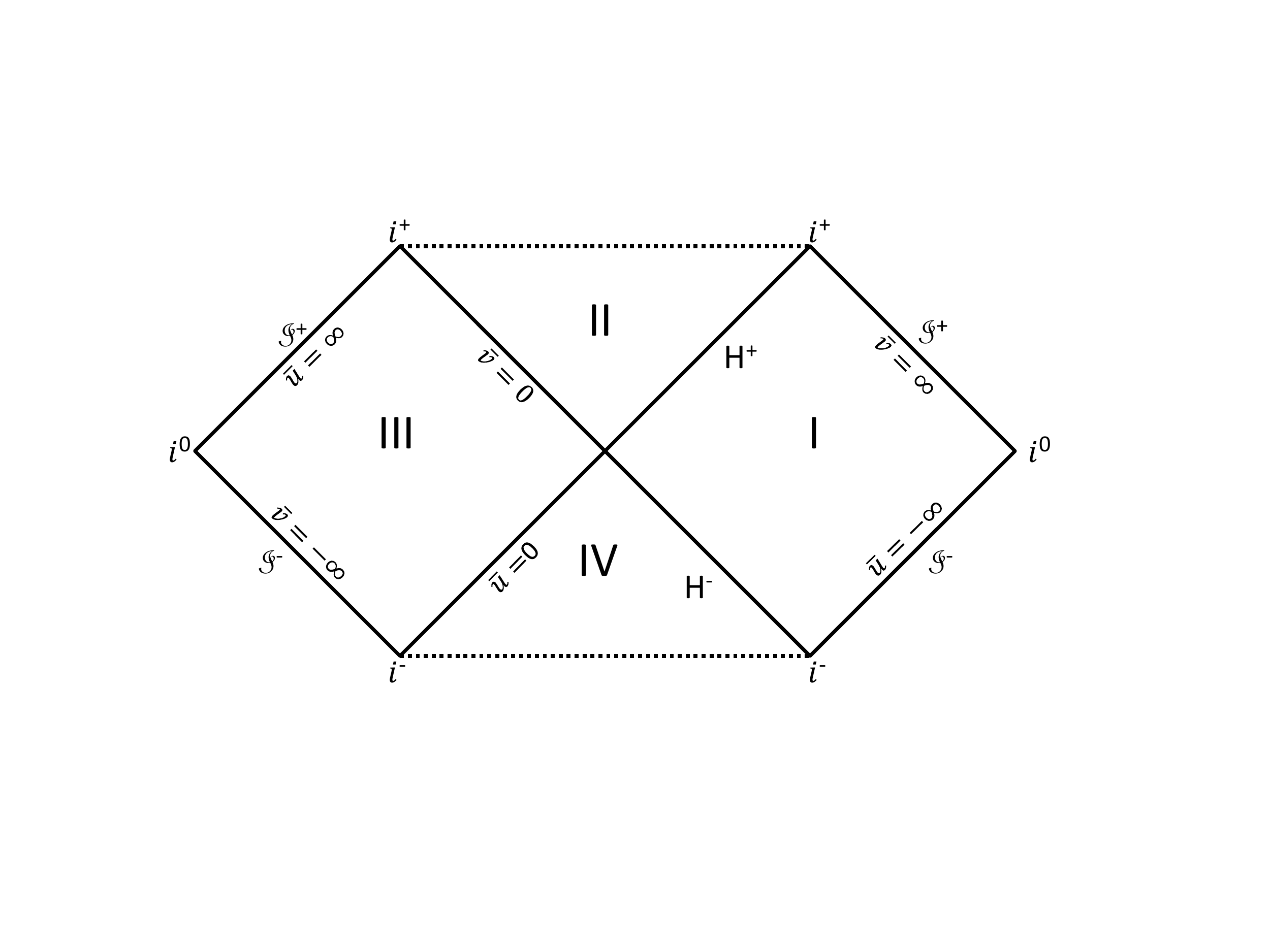}
\caption[Conformal diagram of the maximally extended Schwarzschild spacetime.]{Conformal diagram of the maximally extended Schwarzschild black hole. The bifurcate Killing horizon at $r=2M$ is $H^+ \cup H^-$. The singularities in regions II and IV are indicated by dotted lines. The exterior and interior black hole regions correspond to regions I and II respectively, while III is the isometric copy of region I and IV is the white hole region.}
\label{SchwConformal}
\end{center}
\end{figure}

\subsubsection{Geodesics}

A revision of the geodesics in the Schwarzschild spacetime will be useful for our computations below. If one starts out in region I and falls freely into region II along a timelike trajectory, it is well known that the geodesics, in Schwarzschild coordinates, are given by the integral curves of
\begin{subequations}
\begin{align}
\dot{t} &= \frac{\varepsilon}{F(r)}, \label{t-eq} \\
\dot{r} &= -\sqrt{\varepsilon^2-F(r)}. \label{r-eq}
\end{align}
\label{4:Schweqmotion} 
\end{subequations}

The dot stands for a derivative with respect to the proper time along the geodesic path, in region I. The sign of $\dot{r}$ indicates that the trajectory is falling into the black hole region. The solutions to \eqref{4:Schweqmotion} can be continued to region II using appropriate coordinates, \textit{e.g.}, Kruskal coordinates. 

We recall a convenient parametrisation of the geodesic solutions. Let $\varepsilon > 1$, \textit{i.e.}, the geodesic falls in from infinity with initial non-vanishing speed $\left(1-\varepsilon^{-2}\right)^{1/2}$ with respect to the Killing vector $\xi$. The geodesic can be parametrised as
\begin{subequations}
\begin{align}
\tau & = \frac{M}{\left(\varepsilon^2 - 1\right)^{3/2}}(\sinh \chi - \chi), \\
r & = \frac{M}{\left(\varepsilon^2 - 1\right)}(\cosh \chi - 1), \\
t & = \frac{M \varepsilon}{\left(\varepsilon^2 - 1 \right)^{3/2}} \left[\sinh\chi +  \left(2\varepsilon^2 - 3 \right)\chi \right] + 2M \ln \left( \frac{-\tanh(\chi/2) + \left(1 - \varepsilon^{-2} \right)^{1/2}}{-\tanh(\chi/2) - \left(1 - \varepsilon^{-2} \right)^{1/2}} \right), \label{4:SchwGeodesicsE>1t}
\end{align}
\label{4:SchwGeodesicsE>1}
\end{subequations}
\\ 
where $\chi \in (-\infty,0)$ and $\tau \in (-\infty,0)$, so that the trajectory starts in the asymptotic past as $\chi \to - \infty$ and approaches the singularity as $\chi \to 0$, when the proper time of the observer $\tau \to 0$. The horizon-crossing occurs at $\chi = \chi_h \doteq -2 \arctanh \left[\left( 1 - \varepsilon^{-2} \right)^{1/2} \right]$. Eq. \eqref{4:SchwGeodesicsE>1t} only holds in region I and the parametrisation in region II is obtained by introducing appropriate coordinates.

Let $\varepsilon = 1$, \textit{i.e.}, the geodesic falls into the black hole from infinity with initial vanishing speed with respect to $\chi$. The solution to the geodesic equations in region I \eqref{4:Schweqmotion} is
\begin{subequations}
\begin{align}
r & = 2M[-3 \tau/(4M)]^{2/3}, \\
t & = \tau - 4M[-3\tau/(4M)]^{1/3} + 2M \ln \left( \frac{[-3\tau/(4M)]^{1/3}+1}{[-3\tau/(4M)]^{1/3}-1} \right), \label{4:SchwGeodesicE=1t}
\end{align}
\label{4:SchwGeodesicsE=1}
\end{subequations}
\\
where $\tau \in (-\infty,0)$, the horizon-crossing occurs at $\tau = \tau_h \doteq -4M/3$ and the singularity is approached as $\tau \to 0$. The solution in region II is obtained by introducing appropriate coordinates in eq. \eqref{4:SchwGeodesicE=1t}.

Let $0< \varepsilon < 1$, \textit{i.e.}, the geodesic attains a finite maximum radial coordinate $2M < r < \infty$. The geodesic can be parametrised as
\begin{subequations}
\begin{align}
\tau & = \frac{M}{\left(1 - \varepsilon^2\right)^{3/2}}(\sin \eta + \eta), \label{4:SchwGeodesicsE<1tau}\\
r & = \frac{M}{\left(1 - \varepsilon^2\right)}(\cos \eta +1), \label{4:SchwGeodesicsE<1r}\\
t & = \frac{M \varepsilon}{\left(1 - \varepsilon^2 \right)^{3/2}} \left[\sin\eta +  \left(3 - 2\eta^2 \right)\eta \right] + 2M \ln \left( \frac{1 + \left(\varepsilon^{-2}-1 \right)^{1/2} \tan (\eta/2) }{1 - \left(\varepsilon^{-2}-1 \right)^{1/2} \tan (\eta/2)} \right), \label{4:SchwGeodesicsE<1t}
\end{align}
\label{4:SchwGeodesicsE<1}
\end{subequations}
\\
where $\eta \in (-\pi, \pi)$, so that the trajectory starts at the white hole singularity as $\eta \to -\pi$ and ends at the black hole singularity as $\eta \to \pi$. The trajectory reaches its maximum $r$-value, $r = r_\text{max} \doteq 2M/\left(1-\varepsilon^2\right)$ at proper time $\tau = 0$, and the total proper time elapsed between the singularities is $2\pi M \left(1-\varepsilon^2 \right)^{-3/2}$. The horizon-crossings occur at $\eta = \pm \eta_h$, where $\eta_h \doteq 2 \arctan \left[ \left(\varepsilon^{-2} -1 \right)^{1/2} \right]$. Eq. \eqref{4:SchwGeodesicsE<1t} only holds in region I and the parametrisation in region II is obtained by introducing appropriate coordinates.

Finally, let $\varepsilon = 0$, for the family of geodesics that stay in regions IV and II only, crossing through the bifurcation point $U = V = 0$, which are of the form
\begin{equation}
U = V = \sin(\eta/2) \exp \left[\frac{1}{2} \cos^2(\eta/2) \right],
\label{4:SchwGeodesicsWH}
\end{equation}
where $\eta \in (-\pi,\pi)$, and with $\tau$ and $r$ given by eq. \eqref{4:SchwGeodesicsE<1tau} and \eqref{4:SchwGeodesicsE<1r} respectively.

This concludes our discussion of the classical aspects of the $1+1$ Schwarzschild black hole.

\subsection{The $1+1$ generalised Reissner-Nordstr\"om black hole}
\label{subsec:4RN}

We introduce a class of $(1+1)$-dimensional spacetimes that generalise the Reissner-Nordstr\"om black hole. Let $(M,g)$ be a $(1+1)$-dimensional black hole spacetime with metric given in the exterior region by
\begin{equation}
g(t,r) = -F(r) dt^2 + F(r)^{-1} dr^2,
\end{equation}
where $r \in (0,\infty)$ is a radial coordinate and $t \in (-\infty,\infty)$ is a time coordinate and with $F$ a smooth, real-valued function\footnote{One does not need impose such regularity conditions on $F$, but for concreteness we consider only smooth functions.} satisfying in terms of the global $r$ coordinate
\begin{subequations}
  \begin{align}
    F(r)&
    \begin{cases}
      > 0, & \text{if}\ r>r_+, \\
      < 0, & \text{if}\ r_-<r<r_+, \\
      > 0, & \text{if}\ r<r_-.
    \end{cases}
    \\
    F'(r) & 
    \begin{cases}
      > 0, & \text{if}\ r=r_+, \\
      < 0, & \text{if}\ r=r_-,
    \end{cases}    
  \end{align}
  \label{FBH}
\end{subequations}
\\
with $0 < r_- < r_+$. The exterior region is static with Killing vector field $\xi \doteq \partial_t$. Given that $F$ is continuous, the spacetime contains two horizons, at $r = r_+$ and $r = r_-$. The spacetime is regular across the horizons $r= r_+$ and $r = r_-$. This is evident by introducing the Eddington-Finkelstein time coordinate $v \doteq t + r^*$, where $dr^* \doteq dr/F(r)$ is the tortoise coordinate. The surface gravities associated to $r_-$ and $r_+$, $\kappa_-$ and $\kappa_+$ respectively, are $F'(r_-)/2 = \kappa_- < 0 < \kappa_ + = F'(r_+)/2$. A primed quantity, as usual, indicates a derivative with respect to the argument. A useful piece of notation for computational purposes is to introduce the functions $f$ and $g$ defined by $(r-r_+)g(r) \doteq F(r) \doteq (r-r_-) f(r)$. Notice that $\kappa_- = f(r_-)/2$ and $\kappa_+ = g(r_+)/2$.

The details of the spacetime under consideration depend on the specifics of the function $F$. As an example, the $1+1$ non-extremal Reissner-Nordstr\"om black hole is a member of this family with $F(r)  = (r-r_+)(r-r_-)/r^2$ and $r_\pm  = M \pm (M^2-Q^2)^{1/2}$, where $M$ is associated with the total mass of the spacetime, $Q$ is associated with the total charge and $0 < Q^2 < M^2$. Specifying $F$ is not essential for our calculations and we shall deal with the whole class of spacetimes at once.

Let us now perform the Kruskal extension of this spacetime. We continue to call the extension $(M,g)$ with a slight abuse of notation, but we believe that no confusion may arise from this. Let us start by introducing null coordinates in the exterior region $(u, v) \doteq (t-r^*, t+r^*)$. One can then cover the region $r > r_-$ by introducing the Kruskal-Szekeres chart $(U,V) \doteq (-\exp(-\kappa_+ u), \exp(\kappa_+ v))$. The pair $(U,V)$ can be extended to cover four regions as indicated in Table \ref{exterior}.

\begin{table}[h]
\begin{center}
\begin{tabular}{ l l c r }
  & Region & sgn$(U)$ & sgn$(V)$ \\
  \hline
  I: Exterior &$r_+ < r$ & $-1$ & $+1$ \\
  II: Black hole &$r_-< r < r_+$ & $+1$ & $+1$ \\
  I': Isometric exterior &$r_+ < r$ & $+1$ & $-1$ \\
  II': White hole &$r_-< r < r_+$ & $-1$ & $-1$ 
\end{tabular}
\caption[Kruskal-Szekeres coordinates for $r_-<r<\infty$ in the generalised Reissner-Nordstr\"om extension.]{Kruskal-Szekeres coordinates in the region $r_-<r<\infty$ of the generalised $1+1$ Reissner-Nordstr\"om extension.}
\label{exterior}
\end{center}
\end{table}

The analytic extension contains the Killing vector $\xi \doteq \kappa_+ (-U \partial_U + V \partial_V)$. A bifurcate Killing horizon is located at $r_+$, $H^\text{P} \cup H^\text{F}$, where $H^\text{P}$ is generated by $\partial_U$ and $H^\text{F}$ by $\partial_V$ (see Fig. \ref{RNConformal}). We refer to the open subset of $M$ covered by the charts $(U,V)$ as $M_\text{H} \doteq \text{I} \cup H^\text{P} \cup \text{II} \cup \text{I'} \cup H^\text{F} \cup \text{II'}$. Notice that $(M_\text{H},g|_{M_\text{H}})$ is a spacetime on its own right, where the metric can be written locally as
\begin{equation}
g|_{M_\text{H}}(U,V) = F(r)\left( \kappa_+^2 U V \right)^{-1} dU dV.
\end{equation}

Again, the metric at $UV = 0$ is understood in the limiting sense. 

Moreover, it is globally hyperbolic and conformal to the $1+1$ Minkowski spacetime with a conformal factor locally given by

\begin{equation}
\Omega^2(U,V)|_{M_\text{H}} = - F(r)\left( \kappa_+^2 U V \right)^{-1}.
\label{Omega2}
\end{equation}

In order to cover the region $0 < r<r_-$ one can introduce null coordinates $(\tilde{u},\tilde{v}) \doteq (\tilde{r}^*-\tilde{t}, \tilde{r}^*+\tilde{t})$ in region II, where $\tilde{r}^*$ is the tortoise coordinate of the future increasing time coordinate $\tilde{r}$, $d\tilde{r}^* = d\tilde{r}/F(\tilde{r})$ and $\tilde{t}$ is a spatial coordinate that increases towards the right. Then, the Kruskal-Szekeres coordinates for $r_-<r<r_+$ are $(U_-,V_-) \doteq ( -\exp(-\kappa_- \tilde{u}), -\exp(-\kappa_- \tilde{v}))$. The pair $(U_-,V_-)$ can be used to cover the interior regions indicated in Table \ref{interior}.

\begin{table}[h]
\begin{center}
\begin{tabular}{ l l c r }
  & Region & sgn$(U_-)$ & sgn$(V_-)$ \\
  \hline
  II: Past &$r_- < r < r_+$ & $-1$ & $-1$ \\
  III: Left &$0 < r < r_-$ & $+1$ & $-1$ \\
  III': Right &$0 < r < r_-$ & $-1$ & $+1$ \\
  II{'}{'}{'}: Future &$r_- < r < r_+$ & $+1$ & $+1$ 
\end{tabular}
\caption[Kruskal-Szekeres coordinates for $0 < r <r_+$ in the generalised Reissner-Nordstr\"om extension.]{Kruskal-Szekeres coordinates in the region $0 < r <r_+$ of the generalised $1+1$ Reissner-Nordstr\"om extension.}
\label{interior}
\end{center}
\end{table}

The bifurcate Killing horizon at $r_-$, $H_-^\text{P} \cup H_-^\text{F}$, is generated by $\partial_{U_-}$ and $\partial_{V_-}$. The two sets of Kruskal coordinates relate neatly in region II by the transition functions $(U_-,V_-)=(-U^{\kappa_-/\kappa_+},-V^{\kappa_-/\kappa_+})$. In the case of Reissner-Nordstr\"om, one can find a left timelike curvature singularity at $U_-V_- = -1$ in region III and a right timelike curvature singularity at $U_-V_- = -1$ in region III' corresponding to the global coordinate $r = 0$. 

The future region II{'}{'}{'} in Table \ref{interior} defines a white hole region isometric to the region II' described in Table \ref{exterior}. The procedure now can be repeated \textit{ad infinitum} to obtain an infinite tower of copies of the regions here described into the future. An analogous procedure provides the extension into the past. See Figure \ref{RNConformal} for details.

\begin{figure}[ht]
\begin{center}
\includegraphics[height=8cm]{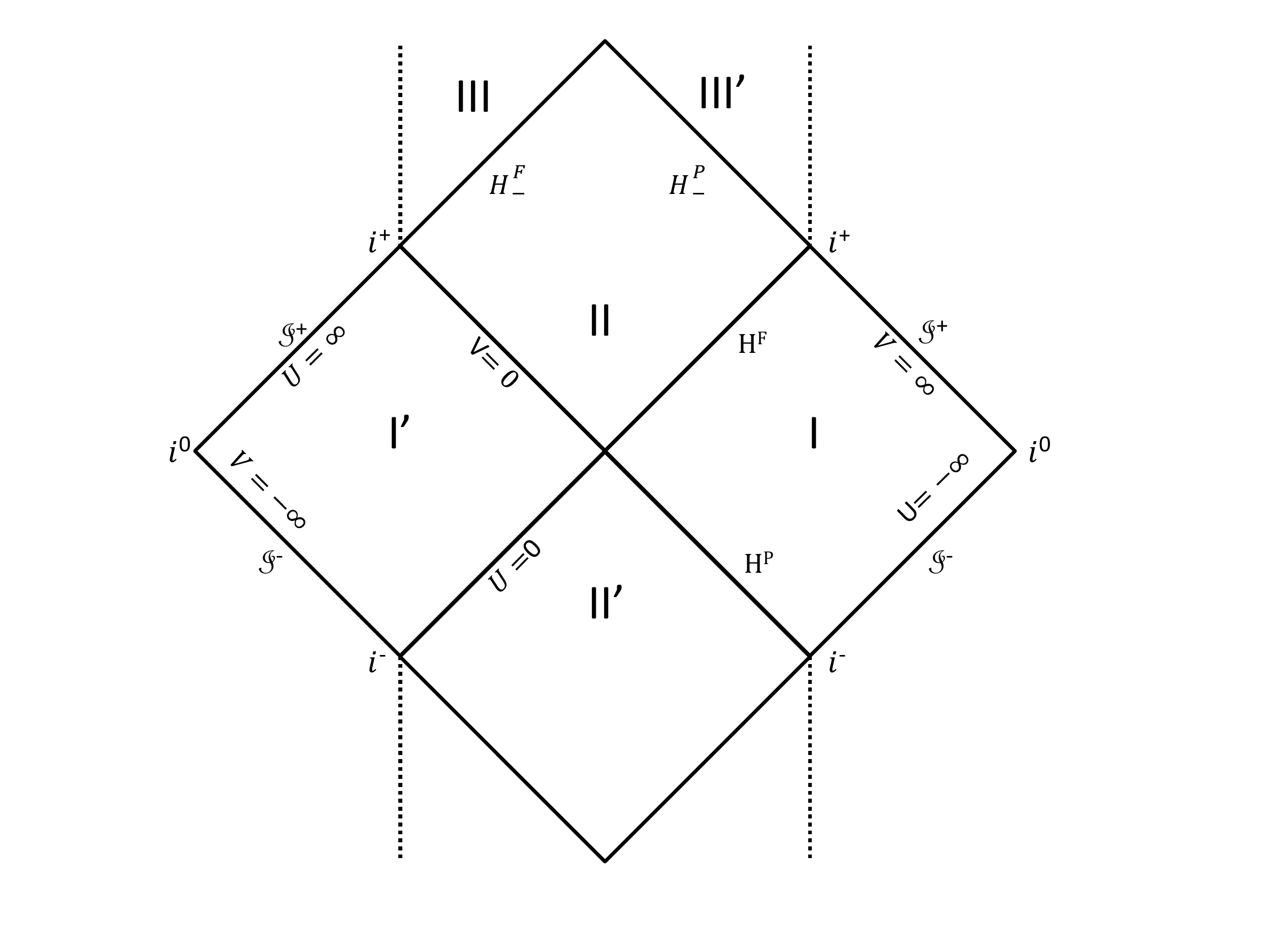}
\caption[Part of the conformal diagram of the maximally extended generalised $1+1$ non-extremal Reissner-Nordstr\"om spacetime]{Part of the conformal diagram of the maximally extended generalised $1+1$ Reissner-Nordstr\"om non-extremal black hole. The bifurcate Killing horizon at $r_+$ is $H^\text{P} \cup H^\text{F}$. The bifurcate Killing horizon at $r_-$ is $H^\text{P}_- \cup H^\text{F}_-$ and part of it is displayed in the future of the diagram. In the case of Reissner-Nordstr\"om, the singularities in regions III and III' are indicated by dotted lines, but the spacetime needs not contain singularities in general. The diagram extends to the past and future.}
\label{RNConformal}
\end{center}
\end{figure}

For the purposes of our analysis, we will make use of the future-directed, left-moving geodesic equations. If one starts out in region I and falls freely into region II and continues to region III across the Cauchy horizon, the geodesic in region I is given by the integral curve of
\begin{subequations}
\begin{align}
\dot{t} &= \frac{\varepsilon}{F(r)}, \label{4:RNt-eq} \\
\dot{r} &= -\sqrt{\varepsilon^2-F(r)}, \label{4:RNr-eq}
\end{align}
\label{4:RNeqmotion} 
\end{subequations}
\\
where the dot stands for a derivative with respect to the proper time along the geodesic path. Eq. \eqref{4:RNeqmotion} are extended to regions II and III using the appropriate coordinate charts. In particular, near the horizon, the chart $(\tilde{r}, \tilde{t})$, defined in region II is more convenient. We have that
\begin{subequations}
\begin{align}
\dot{\tilde{t}} &= -\frac{\varepsilon}{F(r)}, \label{4:RNt2-eq} \\
\dot{\tilde{r}} &= \sqrt{\varepsilon^2-F(r)}. \label{4:RNr2-eq}
\end{align}
\label{4:RNeqmotion2} 
\end{subequations}

The way in which we have defined the coordinates in region II is such that the time coordinate $\tilde{r}$ increases towards the future and $\tilde{r} = \text{const}$ represents spacelike hyperbolas, while $\tilde{t}$ increases towards the right, with $\tilde{t} = \text{const}$ defining timelike hyperbolas.

An interesting property of the Reissner-Nordstr\"om spacetime in $3+1$ dimensions and of the class of $1+1$ spacetimes that we describe here, is that they are not globally hyperbolic. Indeed, many of the most important solutions in General Relativity are not globally hyperbolic. Notably, most members of the Kerr-Newman family are among these solutions. Hence, understanding whether Cauchy horizons form by physical processes is a long standing research topic in mathematical relativity.

If we consider a spacelike surface, $\Sigma$, which connects $i^0$ in region I, on the right, with $i^0$ in region I', on the left, a future Cauchy horizon separates region II from regions III and III'. More precisely, the future Cauchy horizon, which we denote by $\mathcal{C}^\text{F}$, is located at $\mathcal{C}^{\text{F}} = J^-(P_-) \cap (H^{\text{F}}_- \cup H^{\text{P}}_-)$, where $P_-$ is the bifurcation point at $r_-$. The situation is similar in the past of region II', where there is a past Cauchy horizon $\mathcal{C}^\text{P}$, and the full Cauchy horizon is $\mathcal{C} = \mathcal{C}^\text{F} \cup \mathcal{C}^\text{P}$

Eq. \eqref{4:RNeqmotion2} detector that approaches the future Reissner-Nordstr\"om-like Cauchy horizon located at $r_-$ as it moves towards the left in region II. It is useful to define $\mathcal{C}^{\text{FL}} \doteq J^-(P_-) \cap H^{\text{F}}_-$ as the left portion of the future horizon, since we are interested in trajectories that cross this region of $\mathcal{C}^{\text{F}}$. See Fig. \ref{RNdetector}.

\begin{figure}[ht]
\begin{center}
\includegraphics[height=8cm]{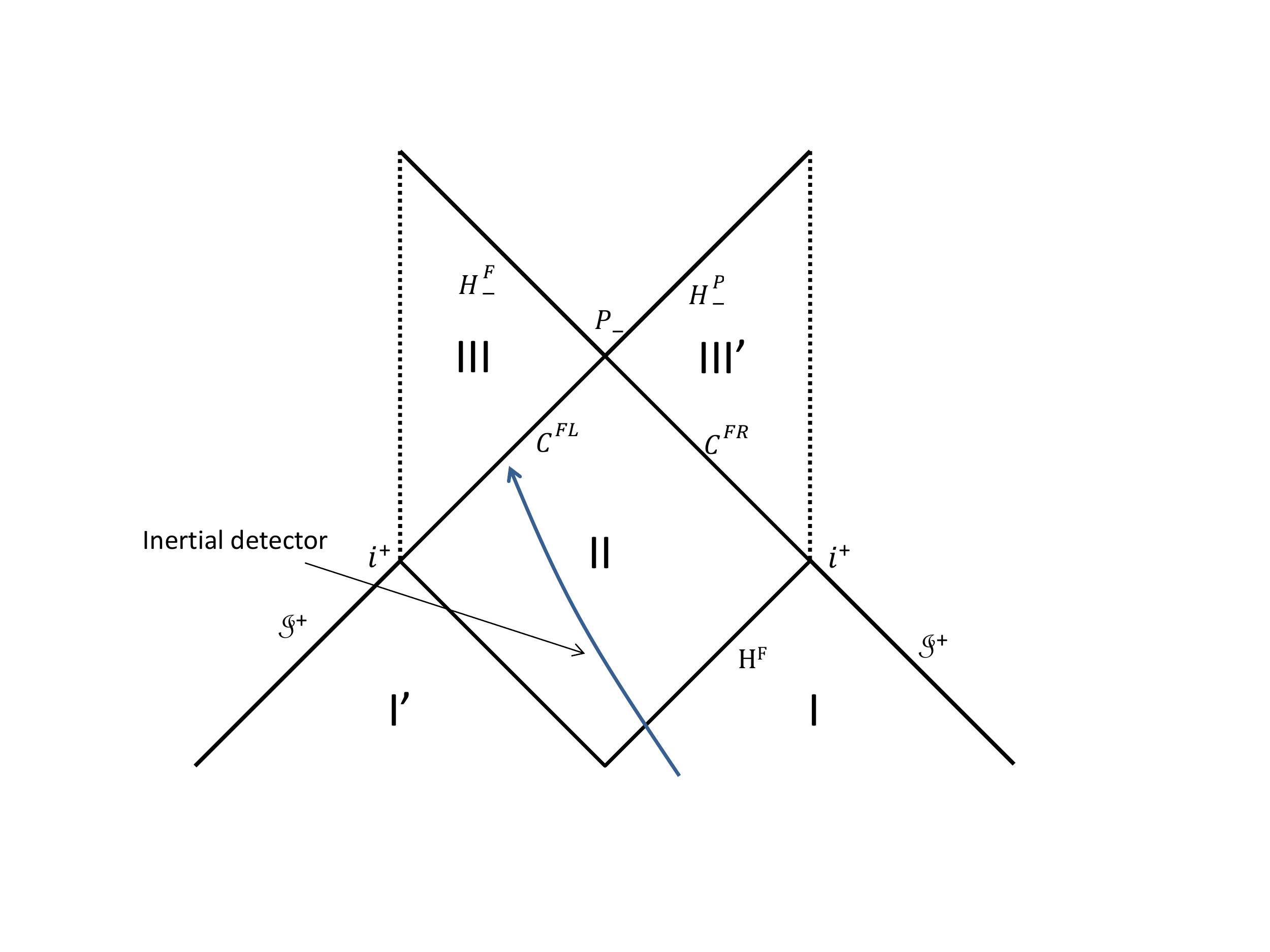}
\caption[Detector approaching the generalised $1+1$ Reissner-Nordstr\"om Cauchy horizon.]{Detector approaching the generalised $1+1$ Reissner-Nordstr\"om Cauchy horizon.}
\label{RNdetector}
\end{center}
\end{figure}

\section{Quantisation in $(1+1)$-dimensional black holes}
\label{sec:4Quantisation}

Quantum field theory is under full control for free fields on globally hyperbolic spacetimes. This means that when working on non-globally hyperbolic spacetimes one has to be careful in defining a globally hyperbolic region where a quantum field theory can be appropriately constructed. If a timelike boundary is present, imposing boundary conditions may suffice to supplement for the lack of global hyperbolicity. In the case of the receding mirror, the introduction of boundary conditions introduces distributional singularities to the two-point function of the theory away from the short distance limit, whenever two points are connected by a null ray that is reflected at the boundary. In the $1+1$ Reissner-Nordstr\"om-like spacetime we perform the quantisation in suitable globally hyperbolic submanifolds of the Kruskal extension.

In each of the spacetimes introduced in Section \ref{sec:4Classical}, we deal with a massless, minimally-coupled scalar field. In the case of the receding mirror spacetime, we construct the relevant state by imposing that the state in the asymptotic past coincide with the Minkowski half-space, subject to Dirichlet boundary conditions at the mirror.

For the Schwarzschild spacetime, we consider two relevant quasi-free states that can be constructed. The first relevant state in the Schwarzchild spacetime is the Hartle-Hawking-Israel (HHI) state \cite{Hartle:1976tp, Israel:1976ur}, which is regular in the whole Kruskal manifold. For a massless scalar field, the HHI state is obtained by imposing boundary conditions for the field modes at the bifurcate Killing horizon, such that in $H^+ \cup H^-$ the modes are given by plane waves with respect to the affine parameter of the horizon generators. This state was first rigorously constructed by Kay and Wald in the seminal paper\cite{Kay:1988mu}. The second relevant state is the Unruh state\cite{Unruh:1976db}. Let $H_{\text{ev}} \doteq (J^+(P) \cap H^+) \setminus P$, where $P$ is the bifurcation point of the bifurcate Killing horizon at $r=2M$. The Unruh state is defined to be regular in the globally hyperbolic region $M_\text{BH} \doteq \text{I} \cup H_{\text{ev}} \cup \text{II}$ of the Kruskal manifold, but not in the whole Kruskal extension. For a massless scalar field, one can construct this state by imposing plane wave solutions at the Cauchy surface $H^\text{-} \cup \mathscr{I}^-$ using the affine parameter at the past horizon $H^\text{P}$ and the advanced time at past null infinity. The Unruh state corresponds to a scalar field on a black hole that has formed from collapse. See \cite{Dappiaggi:2009fx} for a recent rigorous construction of the Unruh state. 

In the case of the Reissner-Nordstr\"om black hole we also have suitable HHI and Unruh vacua. We can define a Hartle-Hawking-Israel (HHI) state to be regular in the globally hyperbolic submanifold $M_\text{H} = \text{I} \cup H^\text{P} \cup \text{II} \cup \text{I'} \cup H^\text{F} \cup \text{II'}$, which we have introduced above, equipped with the induced metric. Contrary to the Schwarzschild case, the HHI state does not cover the full Kruskal extension of the manifold\footnote{Another HHI state can be defined to be regular in the region $\text{II} \cup H_-^\text{P} \cup \text{III} \cup \text{II''} \cup H_-^\text{F} \cup \text{III'}$ if appropriate boundary conditions are set at $U_- V_- = -1$ to compensate for the fact that this region is not globally hyperbolic. However, if we are concerned about an observer coming from region I, the choice in the main text is the relevant one.}. For the Unruh state, we let $H_{\text{ev}} \doteq (J^+(P_+) \cap H^\text{F}) \setminus P_+$, where $P_+$ is the bifurcation point of the bifurcate Killing horizon at $r_+$ in the Kruskal extension of the generalised Reissner-Nordstr\"om black hole. The Unruh state is defined to be regular in the globally hyperbolic region $M_\text{U} \doteq \text{I} \cup H_{\text{ev}} \cup \text{II}$, but not in the whole globally hyperbolic submanifold $M_\text{H}$. For a massless scalar field, one can construct this state by constructing plane wave solutions at the hypersurface $H^\text{P} \cup \mathscr{I}^-$, using the affine parameter at the past horizon and the advanced time at past null infinity. 

\subsection{Quantisation in the receding mirror spacetime}
\label{subsec:4MirrorQ}

We wish to solve the boundary value problem on $(M',g)$, where $M'$ is the manifold to the right of the moving mirror in Minkowski spacetime, given by $\Box  \phi = \partial_u \partial_v \phi = 0$ on $M'$, subject to Dirichlet boundary condition along the mirror trajectory. The solution is $\phi(u,v) = f(u) + g(v)$ subject to $f(u) = - g(v)$ at the mirror. A mode solution satisfying the boundary conditions is
\begin{equation}
\phi_k(u,v) = \alpha_k(v) \ee^{-\ii k v} + \beta_k(v) \ee^{\ii k v} -\left[ \alpha_k(p(u)) \ee^{-\ii k p(u)} + \beta_k(p(u)) \ee^{\ii k p(u)} \right].
\end{equation}
where $\alpha_k$ and $\beta_k$ are functions, which we can take to be smooth. We select the positive frequency subspace with respect to $\xi = \partial_t$ in the asymptotic past, by introducing the complex structure $J: \Sol \times \Sol \to \Sol \times \Sol$, such that $J^2 = -I$, and selected the positive frequency modes as the subspace $(\mathcal{L}_{\partial_t} + kJ) \phi_k = (\partial_t + \ii k) \phi_k = 0$, whereby we set $\beta_k = 0$.

Next, we normalise the modes by the \textit{Dirac orthonormality condition}. To this end, we introduce a wave-packet profile function in Fourier space. Let $f_{k} \in L^2(\mathbb{R}, \Theta(k) dk)$ be sharply peaked about $k = k_p$ and of unit-norm. Then, setting $\alpha_k = (4\pi k)^{-1/2}$, we write the solution to the boundary value problem as linear combinations of Dirac-orthonormal wave packets of the form,
\begin{equation}
\varphi = \phi_k (f_k) = \int_0^\infty \! dk \, f_k \phi_k,
\label{4:phiMirror}
\end{equation}
satisfying
\begin{align}
-\ii \left\{\phi_k\left(f_k\right), \phi_{k'}\left(f_{k'}\right) \right\} = \int_{\Sigma_t}  \! dx \, \left[ \dot{\phi}_k\left(f_k\right) \phi_{k'}^*\left(f_{k'}\right) - \phi_k\left(f_k\right) \dot{\phi}_{k'}^*\left(f_{k'}\right) \right] = 1,
\label{4:DiracOrthonormality}
\end{align}
where we have introduced the canonical momenta $\dot{\phi}$.

The field quantisation prescription gives
\begin{align}
\varPhi(f) & = \ii a\left(\overline{f \varphi} \right) - \ii a^* \left(f \varphi \right), \\
\Pi(g) & = \ii a\left(\overline{g \dot{\varphi}} \right) - \ii a^* \left(g \dot{\varphi} \right)
\end{align}
and the CCR follow from \eqref{4:DiracOrthonormality}.

The action of a creation operator on the vacuum state of the theory, $|0_\text{H}\rangle$, produces a wave packet as in expression \eqref{4:phiMirror} in the one-particle space, $|k\rangle \in \mathcal{H}$. The full Hilbert space of the theory is the symmetric Fock space $\mathscr{F}_s(\mathcal{H})$. The two-point function of the theory is given by
\begin{equation}
\mathcal{W}\left(\x, \x' \right) = \langle 0 | \Phi(\mathsf{x}) \Phi(\mathsf{x}') | 0 \rangle = -\frac{1}{4 \pi} \ln \left[\frac{(p(u)-p(u')-\ii \epsilon)(v-v' -\ii \epsilon)}{(v-p(u')-\ii \epsilon)(p(u)-v'-\ii \epsilon)} \right],
\label{4:WightmanRecMir}
\end{equation}
understood in the distributional sense. The $\ii \epsilon$ prescribes the ultraviolet behaviour of the two-point function.

\subsection{Quantisation in the $1+1$ Schwarzschild black hole}
\label{subsec:4SchwQ}

We proceed to obtain the two relevant states for the Schwarzschild black hole. The first one is the Hartle-Hawking-Israel (HHI) state, which is regular in the whole Kruskal extension of the Schwarzschild spacetime. The second one is the Unruh vacuum, regular in regions I, II and across the future horizon.

\subsubsection*{Hartle-Hawking-Israel state}

The Schwarzschild Kruskal extension is conformally flat in $1+1$ dimensions. Let $(M_\text{M}, g_\text{M})$ denote Minkowski spacetime. The map $\psi: M_{\text{M}} \to M$, $g = \Omega^2 f^*g_{\text{M}}$ is conformal with $\Omega^2(U,V) = - F(r)\left( \kappa_+^2 U V \right)^{-1}$, where $UV = 0$ is understood in the limiting sense. The Klein-Gordon equation for a minimally coupled $1+1$ scalar field is conformal, and given by
\begin{equation}
\Box\phi = \Omega^{-2}(U,V) \partial_U \partial_V \phi = 0,
\end{equation}
where $(U,V)$ are the Schwarzschild Kruskal coordinates.

The subspace of solutions of Dirac orthonormal left and right-moving wave packets, which are of positive frequency with respect to the generators of the bifurcate Killing horizon, contains elements of the form
\begin{subequations}
\begin{align}
\varphi^\text{L} & = \phi^\text{L}_k \left(f^{\text{L}}\right) = \int_{k_0}^\infty \frac{dk}{(4 \pi k)^{1/2}} f^{\text{L}}_k \ee^{-\ii k V}, \\
\varphi^\text{R} & = \phi^\text{R}_k \left(f^{\text{R}}\right) = \int_{k_0}^\infty \frac{dk}{(4 \pi k)^{1/2}} f^{\text{R}}_k \ee^{-\ii k U},
\end{align}
\label{4:WPSchwHHI}
\end{subequations}
\\
where $f^{\text{L}}_k$ and $f^{\text{R}}_k$ are sharply peaked around the values $k_\text{L}$ and $k_\text{R}$ respectively. We have introduced the infrared cut-off $k_0>0$ in eq. \eqref{4:WPSchwHHI} because the integral is not defined all the way down to $k = 0$. Linear combinations of $\varphi^\text{R}$ and $\varphi^\text{R}$ produce solutions, $\varphi$, with left and right-moving contributions of positive frequency with respect to $\partial_V$ and $\partial_U$ respectively.

The canonical quantisation produces operator valued distributions on the left-moving sector
\begin{subequations}
\begin{align}
\varPhi_\text{L}(f) & = \ii a\left(\overline{f_\text{L} \varphi} \right) - \ii a^* \left(f_\text{L} \varphi \right), \\
\Pi_\text{L}(g) & = \ii a\left(\overline{g_\text{L} \dot{\varphi}} \right) - \ii a^* \left(g_\text{L} \dot{\varphi} \right),
\end{align}
\label{4:HHILeftAlg}
\end{subequations}
\\
and on the right-moving sector
\begin{subequations}
\begin{align}
\varPhi_\text{R}(f) & = \ii a\left(\overline{f_\text{R} \varphi} \right) - \ii a^* \left(f_\text{R} \varphi \right), \\
\Pi_\text{R}(g) & = \ii a\left(\overline{g_\text{R} \dot{\varphi}} \right) - \ii a^* \left(g_\text{R} \dot{\varphi} \right),
\end{align}
\label{4:HHIRightAlg}
\end{subequations}
\\
generate the left and right CCR algebras. The full CCR algebra of the theory is the tensor product of the algebras generated by \eqref{4:HHILeftAlg} and \eqref{4:HHIRightAlg}. The creation and annihilation operators are defined as explained in Chapter \ref{ch:chap2} on each sector of the theory, $a_\text{L}$, $a^*_\text{L}$ for left-movers and $a_\text{R}$ and $a^*_\text{R}$ for right-movers. The annihilation operators extend as $a_\text{L} \otimes I$ and $I \otimes a_\text{R}$ in the full CCR algebra and the creation operators extend similarly. 

The Fock space of the theory is defined from the action of the (extended) creation and annihilation operators. The vacuum state of the theory is defined by the condition $a_\text{L} \otimes I |0_\text{H}\rangle = I \otimes a_\text{R} |0_\text{H}\rangle = 0$. The action on the vacuum of a creation operator, given by a linear combination of the extensions of $a^*_\text{L}$ and $a^*_\text{R}$, produces a wave packet with decoupled left and right-movers, as in \eqref{4:WPSchwHHI}, in the one-particle Hilbert space of the theory, $|k_\text{L} \rangle_\text{H} \otimes |k_\text{R} \rangle_\text{H} \in \mathcal{H}_\text{H} = \mathcal{H}_\text{L} \otimes \mathcal{H}_\text{R}$. The wave packets propagate in the whole of the Kruskal manifold. The full Hilbert space of the theory is the factorisable symmetric Fock space $\mathscr{F}_s(\mathcal{H}_\text{H}) \cong \mathscr{F}_s(\mathcal{H}_\text{L}) \oplus \mathscr{F}_s(\mathcal{H}_\text{R})$.

Finally, the two-point function of the state can be obtained by the conformal methods introduced in Chapter \ref{ch:chap2}, Subsection \ref{2:subsecStates}, whereby one obtains that
\begin{equation}
\mathcal{W}_\text{H}(\mathsf{x},\mathsf{x'}) = -\frac{1}{4 \pi} \ln \left[(\epsilon+i \Delta U) (\epsilon + i \Delta V) \right] + \text{div},
\label{4:WightmanSchwHHI}
\end{equation}
where div is the infrared ambiguity for $1+1$ massless scalars. Because the vacuum state is quasi-free, the two-point function determines the state. The $\ii \epsilon$ prescribes the ultraviolet behaviour and the branch of the logarithm is understood as in our discussion in Chapter \ref{ch:chap2}. (See eq. \eqref{2:lnWightman}.)

\subsubsection*{Unruh state}

For the Unruh state, we restrict our attention to the submanifold $M_\text{BH} = \text{I} \cup H_{\text{ev}} \cup \text{II}$ of the Schwarzschild Kruskal extension, which can be conformally embedded in Minkowski spacetime. Thus, the conformal techniques are at our disposal. The Klein-Gordon equation implies
\begin{equation}
\partial_U \partial_v \phi = 0
\end{equation}
and we take the subspace of plane wave solutions given by linear combinations of
\begin{subequations}
\begin{align}
\varphi^\text{R} & = \phi^\text{R}_k \left(f^{\text{R}}\right) = \int_{k_0}^\infty \frac{dk}{(4 \pi k)^{1/2}} f^{\text{R}}_k \ee^{-\ii (4M k) U}, \\
\varphi^\text{L} & = \phi^\text{L}_k \left(f^{\text{L}}\right) = \int_{k_0}^\infty \frac{dk}{(4 \pi k)^{1/2}} f^{\text{L}}_k \ee^{-\ii k v},
\end{align}
\label{4:WPSchwUnruh}
\end{subequations} 

Again, we have introduced the infrared cut-off $k_0>0$ in eq. \eqref{4:WPSchwUnruh}. Eq. \eqref{4:WPSchwUnruh} has left-travelling modes which are plane waves at past null infinity and a red-shifted right moving contribution. The shift in the right-moving mode is to be understood as the shift that a Minkowskian mode would suffer after crossing the centre of a collapsing star at late times. In the conformal diagram that we have presented there is no such collapse, but the choice of modes provides the correct mathematical description for intepreting these modes as coming from a star at a late stage of the collapse. See the discussion below eq. (2.29) in \cite{Unruh:1976db}.

The Dirac quantisation is as before and we obtain on the left-moving sector
\begin{subequations}
\begin{align}
\varPhi_\text{L}(f) & = \ii a\left(\overline{f_\text{L} \varphi} \right) - \ii a^* \left(f_\text{L} \varphi \right), \\
\Pi_\text{L}(g) & = \ii a\left(\overline{g_\text{L} \dot{\varphi}} \right) - \ii a^* \left(g_\text{L} \dot{\varphi} \right),
\end{align}
\label{4:UnruhLeftAlg}
\end{subequations}
\\
and on the right-moving sector
\begin{subequations}
\begin{align}
\varPhi_\text{R}(f) & = \ii a\left(\overline{f_\text{R} \varphi} \right) - \ii a^* \left(f_\text{R} \varphi \right), \\
\Pi_\text{R}(g) & = \ii a\left(\overline{g_\text{R} \dot{\varphi}} \right) - \ii a^* \left(g_\text{R} \dot{\varphi} \right),
\end{align}
\label{4:UnruhRightAlg}
\end{subequations}
\\
which generate the CCR algebra $\mathcal{A}(M_\text{BH})$ by the tensor product of the algebras generated by \eqref{4:UnruhLeftAlg} and \eqref{4:UnruhRightAlg}. 

As before, the vacuum state of the theory is defined by the condition $a_\text{L} \otimes I |0_\text{H}\rangle = I \otimes a_\text{R} |0_\text{U}\rangle = 0$, and the one-particle Hilbert pace by the action of creation operators on the vacuum, $  (\alpha \, a^*_\text{L}\left(f_\text{L} \varphi \right) \otimes I )(  \beta \, I \otimes a^*_\text{R} \left(f_\text{R} \varphi \right) )  |0_\text{U}\rangle = \alpha \, \beta \, |k_\text{L}\rangle_\text{U} \otimes |k_\text{R} \rangle_\text{U} \in \mathcal{H}_\text{U}$, with $\alpha, \beta \in \mathbb{C}$. The state produces a wave packet that has a left-moving contribution indistinguishable from a Minkowskian plane wave packet, and a shifted-frequency right-mover, as discussed above. The waves move inside $M_\text{U}$ and do not reach regions III and IV of the conformal diagram \ref{SchwConformal}. The full Hilbert space of the theory is the factorisable symmetric Fock space $\mathscr{F}_s(\mathcal{H}_\text{U}) \cong \mathscr{F}_s(\mathcal{H}_\text{L}) \oplus \mathscr{F}_s(\mathcal{H}_\text{R})$.

The two-point function that determines the state can be obtained by the conformal methods introduced in Chapter \ref{ch:chap2}, \ref{2:subsecStates}, whereby one obtains that
\begin{equation}
\mathcal{W}_\text{U}(\mathsf{x},\mathsf{x'}) = -\frac{1}{4 \pi} \ln \left[(\epsilon+i \Delta U) (\epsilon + i \Delta v) \right] + \text{div},
\label{4:WightmanSchwUnruh}
\end{equation}
where, again, div is the infrared ambiguity for $1+1$ massless scalars. The awkwardness in the dimensions in eq. \eqref{4:WightmanSchwUnruh} stems from this ambiguity.

\subsection{Quantisation in the $1+1$ generalised Reissner-Nordstr\"om black hole}
\label{subsec:4RNQ}

We proceed to obtain the HHI and Unruh states for the generalised Reissner-Nordstr\"om spacetimes discussed above.

\subsubsection*{Hartle-Hawking-Israel state}

The HHI state is regular in the submanifold $(M_\text{H},g|_\text{H})$, covered by the Kruskal Reissner-Nordstr\"om-like coordinates, $(U,V)$, defined in regions I, I', II and II', and across the bifurcate Killing horizon that separates these regions in Fig.~\ref{RNConformal}. The state is a conformal vacuum with respect to the Minkowski vacuum, with conformal factor $\Omega^2(U,V) = - F(r)\left( \kappa_+^2 U V \right)^{-1}$. Once more, the Klein-Gordon equation for a minimally coupled $1+1$ scalar field is conformal, and given by
\begin{equation}
\Box\phi = \Omega^{-2} \partial_U \partial_V \phi = 0.
\end{equation}

The subspace of solutions of Dirac orthonormal wave packets, with left-moving component of positive frequency with respect to the generator of the future horizon and right-moving component of positive frequency with respect to the generator of the past horizon, is generated by linear combinations of
\begin{subequations}
\begin{align}
\varphi^\text{R} & = \phi^\text{R}_k \left(f^{\text{R}}\right) = \int_{k_0}^\infty \frac{dk}{(4 \pi k)^{1/2}} f^{\text{R}}_k \ee^{-\ii (4M k) U}, \\
\varphi^\text{L} & = \phi^\text{L}_k \left(f^{\text{L}}\right) = \int_{k_0}^\infty \frac{dk}{(4 \pi k)^{1/2}} f^{\text{L}}_k \ee^{-\ii k V}
\end{align}
\label{4:WPRNHHI}
\end{subequations}

Here, $k_0>0$ is an infrared cut-off. The left-moving contribution of the wave-packages that propagate from region I into the future will eventually reach the future horizon, $\mathcal{C}^\text{F}$, and escape into region III, where the evolution of the wave will not be determined by data on any Cauchy surface, $\Sigma_\text{H}$, of the spacetime $(M_\text{H},g|_\text{H})$. This is so because $\Sigma_\text{H}$ is not a Cauchy surface of the whole Kruskal extension. Similarly, the right-moving contribution of the wave-packages that propagate from region I' into the future will eventually reach the future horizon, $\mathcal{C}^\text{F}$, and escape into region III'.

The canonical quantisation produces operator valued distributions, which map smooth functions with compact support on surfaces inside $M_\text{H}$ to the operators
\begin{subequations}
\begin{align}
\varPhi_\text{L}(f) & = \ii a\left(\overline{f_\text{L} \varphi} \right) - \ii a^* \left(f_\text{L} \varphi \right), \\
\Pi_\text{L}(g) & = \ii a\left(\overline{g_\text{L} \dot{\varphi}} \right) - \ii a^* \left(g_\text{L} \dot{\varphi} \right),
\end{align}
\label{4:HHILeftAlgRN}
\end{subequations}
\\
on the left moving sector, and 
\begin{subequations}
\begin{align}
\varPhi_\text{R}(f) & = \ii a\left(\overline{f_\text{R} \varphi} \right) - \ii a^* \left(f_\text{R} \varphi \right), \\
\Pi_\text{R}(g) & = \ii a\left(\overline{g_\text{R} \dot{\varphi}} \right) - \ii a^* \left(g_\text{R} \dot{\varphi} \right),
\end{align}
\label{4:HHIRightAlgRN}
\end{subequations}
\\
on the right-moving sector. The tensor product of the left and right CCR algebras, generated by the relations \eqref{4:HHILeftAlgRN} and \eqref{4:HHIRightAlgRN} respectively, yields the algebra $\mathcal{A}(M_\text{H})$.

The annihilation operators define the vacuum by $a_\text{L} \otimes I |0_\text{H}\rangle = I \otimes a_\text{R} |0_\text{H}\rangle = 0$. The action of a creation operator on the vacuum state of the theory, $|0_\text{H}\rangle$, produces a wave packet in the one-particle Hilbert space $\mathcal{H}_\text{H}$. The waves originate inside region $M_\text{H}$ and propagate towards the future. They eventually cross the Cauchy horizon and escape into regions III and III' of the conformal diagram \ref{RNConformal} in such a way that the dynamical evolution of the waves in the future of the Cauchy horizon is not determined by initial data on $M_\text{H}$. The full Hilbert space is the factorisable symmetric Fock space $\mathscr{F}_s(\mathcal{H}_\text{H}) \cong \mathscr{F}_s(\mathcal{H}_\text{L}) \oplus \mathscr{F}_s(\mathcal{H}_\text{R})$.

Finally, the two-point function is defined in the interior of $M_\text{H}$ and determines the HHI state. It is given by
\begin{equation}
\langle 0_\text{H} | \phi(\mathsf{x}) \phi(\mathsf{x'}) | 0_\text{H} \rangle = \mathcal{W}_\text{H}(\mathsf{x},\mathsf{x'}) = -\frac{1}{4 \pi} \ln \left[(\epsilon+i \Delta U) (\epsilon + i \Delta V) \right] + \text{div},
\label{4:RNHHI}
\end{equation}
where div is the infrared ambiguity for $1+1$ massless scalars.

\subsubsection*{Unruh state}

For the Unruh state, we restrict our attention to the submanifold $M_\text{U} = \text{I} \cup H_{\text{ev}} \cup \text{II}$ of the Kruskal extension in the Reissner-Nordstr\"om-like manifolds. The Klein-Gordon equation implies
\begin{equation}
\partial_U \partial_v \phi = 0
\end{equation}
and we take the subspace of plane wave solutions given by linear combinations of left-movers of positive frequency with respect to the advanced time in the asymptotic past, and of right-movers of positive frequency with respect to the generator of the past horizon,
\begin{subequations}
\begin{align}
\varphi^\text{R} & = \phi^\text{R}_k \left(f^{\text{R}}\right) = \int_{k_0}^\infty \frac{dk}{(4 \pi k)^{1/2}} f^{\text{R}}_k \ee^{-\ii k U}, \\
\varphi^\text{L} & = \phi^\text{L}_k \left(f^{\text{L}}\right) = \int_{k_0}^\infty \frac{dk}{(4 \pi k)^{1/2}} f^{\text{L}}_k \ee^{-\ii k v}
\end{align}
\label{4:WPRNUnruh}
\end{subequations}

Again, we have introduced the infrared cut-off $k_0>0$ in eq. \eqref{4:WPSchwUnruh}. The left-moving contribution of the wave packets that propagate from region I into the future will eventually reach the future horizon, $\mathcal{C}^\text{F}$, and escape into region III of Fig. \ref{RNConformal}, where the dynamical evolution of the wave will be undetermined by the initial conditions on $H^\text{P} \cup \mathscr{I}^-$. Similarly, right-moving wave packets in region II can escape into region III'.

The Dirac quantisation is as before and maps smooth functions with compact support on surfaces inside $M_\text{U}$ to the operators
\begin{subequations}
\begin{align}
\varPhi_\text{L}(f) & = \ii a\left(\overline{f_\text{L} \varphi} \right) - \ii a^* \left(f_\text{L} \varphi \right), \\
\Pi_\text{L}(g) & = \ii a\left(\overline{g_\text{L} \dot{\varphi}} \right) - \ii a^* \left(g_\text{L} \dot{\varphi} \right),
\end{align}
\label{4:UnruhLeftAlgRN}
\end{subequations}
\\
on the left moving sector, and 
\begin{subequations}
\begin{align}
\varPhi_\text{R}(f) & = \ii a\left(\overline{f_\text{R} \varphi} \right) - \ii a^* \left(f_\text{R} \varphi \right), \\
\Pi_\text{R}(g) & = \ii a\left(\overline{g_\text{R} \dot{\varphi}} \right) - \ii a^* \left(g_\text{R} \dot{\varphi} \right),
\end{align}
\label{4:UnruhRightAlgRN}
\end{subequations}
\\
on the right-moving sector. The tensor product of the left and right CCR algebras, generated by the relations \eqref{4:UnruhLeftAlgRN} and \eqref{4:UnruhRightAlgRN} respectively, yields the algebra $\mathcal{A}(M_\text{U})$.

The annihilation operators define the vacuum by $a_\text{L} \otimes I |0_\text{H}\rangle = I \otimes a_\text{R} |0_\text{U}\rangle = 0$. The action of a creation operator on the vacuum state of the theory, $|0_\text{U}\rangle$, produces a wave packet in the one-particle Hilbert space $\mathcal{H}_\text{U}$. The waves originate inside region $M_\text{U}$ and propagate towards the future, enventually crossing the Cauchy horizon, as discussed above. See Fig. \ref{RNConformal}. The full Hilbert space is the factorisable symmetric Fock space $\mathscr{F}_s(\mathcal{H}_\text{U}) \cong \mathscr{F}_s(\mathcal{H}_\text{L}) \oplus \mathscr{F}_s(\mathcal{H}_\text{R})$.

Finally, the (quasi-free) state, which is regular inside $M_\text{U}$, is obtained from conformal methods from the two-point function, as before,
\begin{equation}
\langle 0_\text{U} | \phi(\mathsf{x}) \phi(\mathsf{x'}) | 0_\text{U} \rangle = \mathcal{W}_\text{U}(\mathsf{x},\mathsf{x'}) = -\frac{1}{4 \pi} \ln \left[(\epsilon+i \Delta U) (\epsilon + i \Delta v) \right] + \text{div}.
\label{4:RNUnruh}
\end{equation}

As before, the awkwardness in the dimensionality in eq. \eqref{4:RNUnruh} stems from the infrared ambiguity. This concludes our quantum field theoretic constructions.

\section{Quantum effects of black holes}
\label{sec:4QuantumEffects}

In this section, we examine a series of quantum phenomena that occur in the spacetimes introduced above, in strongly time-dependent situations, such as along non-stationary trajectories and with non-stationary boundary conditions. A large part of our computations deal with the experience of observers, carrying particle detectors along their worldlines, which register particles by interacting with quantum fields in the states that we have constructed above. We also exemplify how the rate of particle detectors is in agreement with our expectations by analysing the local renormalised energy density, along the worldline of an observer, and comparing this to the transition rate, in the context of the generalised Reissner-Nordstr\"om black hole. Finally, we comment on the dependence of our results on the spacetime dimension.

\subsection{Effects in the receding mirror spacetime}
\label{subsec:4MirrorE}

An interesting question to ask is how thermality arises during the collapse of a star and eventual black hole formation for quantum fields in the vicinity of the collapsing star. We examine this question in our receding mirror model. We compute the instantaneous transition rate, $\dot{\mathcal{F}}$, of a derivative-coupling detector both in the asymptotic past and in the asymptotic future. We anticipate that the rate in the asymptotic past will be, to leading order, the rate of a detector in the Minkowski half-space. In the asymptotic future, we shall verify that right-moving modes coming from the mirror will be of thermal character to leading order. We then study the interpolating behaviour numerically to analyse the transient regime in the process of emergence of thermality.

The main result of Section \ref{subsec:4MirrorQ} provides the Wightman function in the receding mirror spacetime,
\begin{equation}
\langle 0 | \phi(\mathsf{x}) \phi(\mathsf{x}') | 0 \rangle = -\frac{1}{4 \pi} \ln \left[\frac{(p(u)-p(u')-\ii \epsilon)(v-v' -\ii \epsilon)}{(v-p(u')-\ii \epsilon)(p(u)-v'-\ii \epsilon)} \right],
\label{4:RecMirWightman}
\end{equation}
where $p(u) = -(1/\kappa) \ln\left(1+\ee^{-\kappa u}\right)$ and the null coordinates are $(u,v) = (t-x, t+x)$. The $\ii \epsilon$ prescription controls the ultraviolet behaviour. The $(1+1)$-dimensional infrared problems are cured by virtue of the Dirichlet boundary condition and eq. \eqref{4:RecMirWightman} is free of infrared divergences.

We now examine the transition rate of a detector, switched on in the asymptotic past, which follows a static trajectory along the integral of $\partial_t$.

\subsubsection*{Static with respect to mirror in the distant past}

We consider a static detector at fixed distance $d$ from the mirror in the asymptotic past, with trajectory $\mathsf{x}(\tau) = (\tau, d)$. The pullback of the twice-differentiated Wightman function along the worldline is
\begin{align}
\mathcal{A}\left(\tau',\tau''\right) = \partial_{\tau'} \partial_{\tau''} \mathcal{W}(\tau',\tau'') = & -\frac{1}{4 \pi} \left( \frac{p'\left(u'\right) p'\left(u''\right)}{\left[p\left(u'\right) - p \left(u''\right) -\ii \epsilon \right]^2} + \frac{1}{\left(v' - v'' - \ii \epsilon \right)^2} \right. \nonumber \\
& \left. -\frac{p'\left(u''\right)}{\left[v' - p\left(u''\right) -\ii \epsilon \right]^2} - \frac{p'\left(u'\right)}{p\left(u'\right)-v''- \ii \epsilon}\right).
\label{4:RecMirA}
\end{align}

As above, primes denote derivatives with respect to the argument. Using \eqref{4:RecMirA} and \eqref{FdotAsymp}, the transition rate is given by $\dot{\mathcal{F}}(E,\tau) = \dot{\mathcal{F}}_0(E,\tau) + \dot{\mathcal{F}}_1(E,\tau) + \dot{\mathcal{F}}_2(E,\tau)$, where
\begin{subequations}
\begin{align}
\dot{\mathcal{F}}_0(E,\tau) & \doteq -E \Theta (-E) + \frac{1}{2 \pi} \int_0^\infty \! ds \, \cos(Es) \left( - \frac{p'\left(\tau-d\right) p'\left(\tau-s-d\right)}{\left[p\left(\tau-d\right) - p \left(\tau-s-d\right) \right]^2} + \frac{1}{s^2}\right) \label{4:StaticRecMirF0} \\
\dot{\mathcal{F}}_1(E,\tau) & \doteq \frac{1}{2 \pi} \int_0^\infty \! ds \,  \frac{\cos(Es)  p'\left(\tau-s-d\right)}{\left[\tau + d - p \left(\tau-s-d\right)\right]^2} \label{4:StaticRecMirF1} \\
\dot{\mathcal{F}}_2(E,\tau) & \doteq \frac{1}{2 \pi} \int_0^\infty \! ds \, \text{Re} \left( \frac{\ee^{-\ii E s} p'(\tau-d)}{\left[p(\tau-d) - \tau +s - d - \ii \epsilon \right]^2} \right) \label{4:StaticRecMirF2}.
\end{align}
\end{subequations}

We have removed the regulator from eq. \eqref{4:StaticRecMirF0} and \eqref{4:StaticRecMirF1} because the integrand is free of singularities. In eq. \eqref{4:StaticRecMirF2} the $\ii \epsilon$-prescription is kept due to the singularities that occur whenever two timelike-separated points along the static trajectory are connected by a null geodesic that is reflected at the boundary. This represents no problem because the switch-on of the detector occurs in the asymptotic past, so two points connected by a null ray reflected on the mirror can never be endpoints of the worldline.

In Appendix \ref{app:BH}, we show that at early and late times along the detector trajectory, the response function of the detector is respectively
\begin{subequations}
\begin{align}
\dot{\mathcal{F}}(E,\tau) & = -E \left[1 - \cos(2dE) \right] \Theta(-E) + O\left(\ee^{\kappa \tau} \right), & \text{ as } \tau \rightarrow - \infty, \label{4:StaticRecMirPast} \\
\dot{\mathcal{F}}(E,\tau) & = -\frac{E}{2}  \Theta(-E) + \frac{E}{2 \left(\ee^{2 \pi E/\kappa} -1\right)} + \textit{o}(1), & \text{ as } \tau \rightarrow + \infty. \label{4:StaticRecMirFuture}
\end{align}
\label{4:StaticRecMir}
\end{subequations}

It is odd that as $d \to \infty$ one does not recover the Minkowski vacuum spectrum. We suspect that this is due to the $1+1$ ambiguity of massless fields.

\subsubsection*{Travelling towards mirror in the distant past}

We now bring our attention to an inertial detector moving towards the mirror, at speed $\tanh \lambda > 0$, along the trajectory $\mathsf{x}(\tau) = \left( \tau \cosh \lambda, -\tau \sinh \lambda \right)$. It can be verified from eq. \eqref{4:uvmirror} that the mirror trajectory and the detector's trajectory do not intersect and, thus, the detector never approaches the boundary of the spacetime.

Using \eqref{4:RecMirA} and \eqref{FdotAsymp}, the transition rate is given by $\dot{\mathcal{F}}(E,\tau) = \dot{\mathcal{F}}_0(E,\tau) + \dot{\mathcal{F}}_1(E,\tau) + \dot{\mathcal{F}}_2(E,\tau)$, where
\begin{subequations}
\begin{align}
\dot{\mathcal{F}}_0(E,\tau) & = -E \Theta (-E) + \frac{1}{2 \pi} \int_0^\infty \! ds \, \cos(Es) \left( \frac{\ee^{2\lambda} p'\left( \ee^\lambda \tau \right) p'\left(\ee^\lambda (\tau-s) \right)}{\left[p\left(\ee^\lambda \tau\right) - p \left(\ee^\lambda (\tau-s)\right) \right]^2} + \frac{1}{s^2}\right) \label{4:BoostRecMirF0} \\
\dot{\mathcal{F}}_1(E,\tau) & = \frac{1}{2 \pi} \int_0^\infty \! ds \,  \frac{\cos(Es)  p'\left(\ee^\lambda (\tau-s) \right)}{\left[\ee^{-\lambda} \tau + d - p \left(\ee^\lambda (\tau-s) \right)\right]^2} \label{4:BoostRecMirF1} \\
\dot{\mathcal{F}}_2(E,\tau) & = \frac{1}{2 \pi} \int_0^\infty \! ds \, \text{Re} \left( \frac{\ee^{-\ii E s} p'\left( \ee^\lambda \tau \right)}{\left[p\left(\ee^\lambda \tau \right) - \ee^{-\lambda} (\tau - s) - \ii \epsilon \right]^2} \right) \label{4:BoostRecMirF2}.
\end{align}
\end{subequations}

As before, only eq. \eqref{4:BoostRecMirF2} needs the $\ii \epsilon$ regulator. In Appendix \ref{app:BH} we show that the response is at early and late times
\begin{subequations}
\begin{align}
\dot{\mathcal{F}}(E,\tau) & = -E \left[1 - \ee^{2\lambda} \cos\left(2\tau \sinh \lambda \ee^\lambda E \right) \right] \Theta(-E) + O\left(\tau^{-1} \right), & \text{ as } \tau \rightarrow - \infty, \label{4:BoostRecMirPast} \\
\dot{\mathcal{F}}(E,\tau) & = -\frac{E}{2}  \Theta(-E) + \frac{E}{2 \left(\ee^{2 \pi \ee^{-\lambda} E/\kappa} -1\right)} + \textit{o}(1), & \text{ as } \tau \rightarrow + \infty. \label{4:BoostRecMirFuture}
\end{align}
\label{4:BoostRecMir}
\end{subequations}

\subsubsection*{Onset of thermality}

We comment on the asymptotic form of the transition rates \eqref{4:StaticRecMir} and \eqref{4:BoostRecMir}.

For the detector that follows a static worldline with respect to the mirror in the asymptotic past, we see that at early times the transition rate is the same as in the Minkowski half-space with Dirichlet boundary conditions \eqref{3:F0} to leading order, with an exponentially suppressed correction term. At late times, the contribution from the left-moving field modes yields a Minkowskian spectrum, while the right-travelling modes carry an energy flux to the right \cite{Birrell:1982ix, Davies:1976hi, Davies:1977yv, Carlitz:1986nh, Good:2013lca} producing a Planckian spectrum at temperature $T = \kappa/(2 \pi)$. This result is consistent with the one reported for the non-derivative Unruh-DeWitt detector \cite{Birrell:1982ix}.

For the detector that moves towards the mirror, the static distance is replaced by the time-dependent distance $-\tau \sinh \lambda$, the energy is blue-shifted by a Doppler factor $\ee^\lambda$, and the amplitude in the oscillatory contribution to the response carries a factor $\ee^{2 \lambda}$. In the asymptotic future, the temperature detected from the right-movers is Doppler blue-shifted by the factor $\ee^\lambda$.

The onset of thermality from early to late times is examined numerically in Figs. \ref{fig:mirror3d} and \ref{fig:mirrorasympt}, for the detector that is static in the asymptotic past. Numerical analyses indicate that the late time rate is reached via a ring-down oscillation with period $2 \pi/\kappa$, but we have not attempted to examine this oscillation analytically. The proper time along the detector worldline at which the transition rate oscillations become important is in agreement with an increase in the stress-energy radiation flux, which crosses the detector worldline, moving away from the mirror,
\begin{equation}
\langle T_{uu}(u) \rangle^{\text{ren}} = \frac{1}{12 \pi} [p'(u)]^{1/2} \partial_u^2 [p'(u)]^{-1/2}.
\end{equation} 

At late times, the radiation flux becomes constant and the transition rate oscillations become damped, leading to a constant transition rate at late times. Compare Figs. \ref{fig:mirror3d} and \ref{fig:mirrorasympt} with \ref{fig:EnergyFlux}

\begin{figure}[!ht]
\centering
\includegraphics[height=8cm]{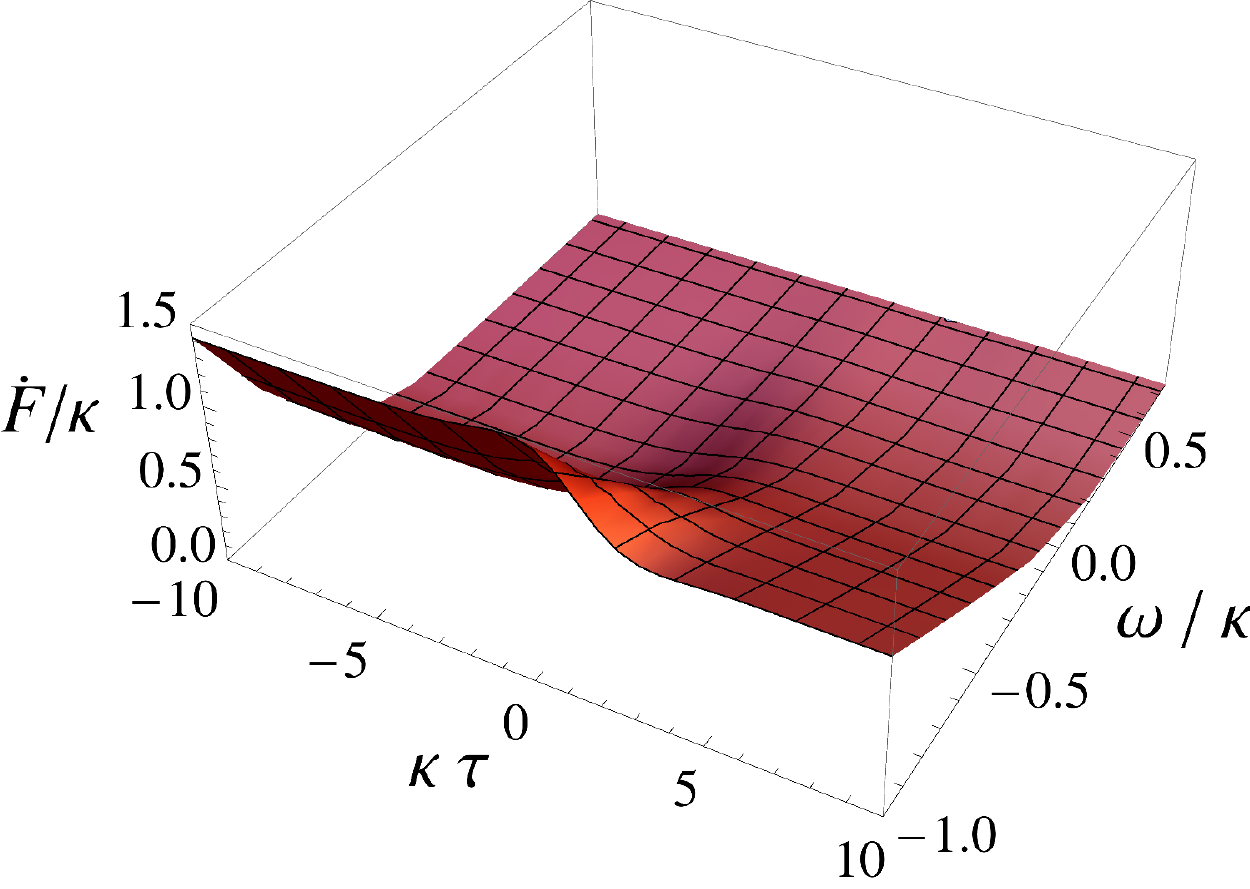}
\caption[Transition rate in the receding mirror spacetime.]{Perspective plot of the transition rate $\dot{\mathcal{F}}(E,\tau)$ for the detector that is static with respect to the mirror in the asymptotic past \eqref{4:StaticRecMir} with $d = 1/\kappa$ and $E = \omega$ (in units $\hbar = 1$).}
\label{fig:mirror3d}
\end{figure}

\begin{figure}[!ht]
\centering
\begin{tabular}{cc}
\includegraphics[width=0.45\textwidth]{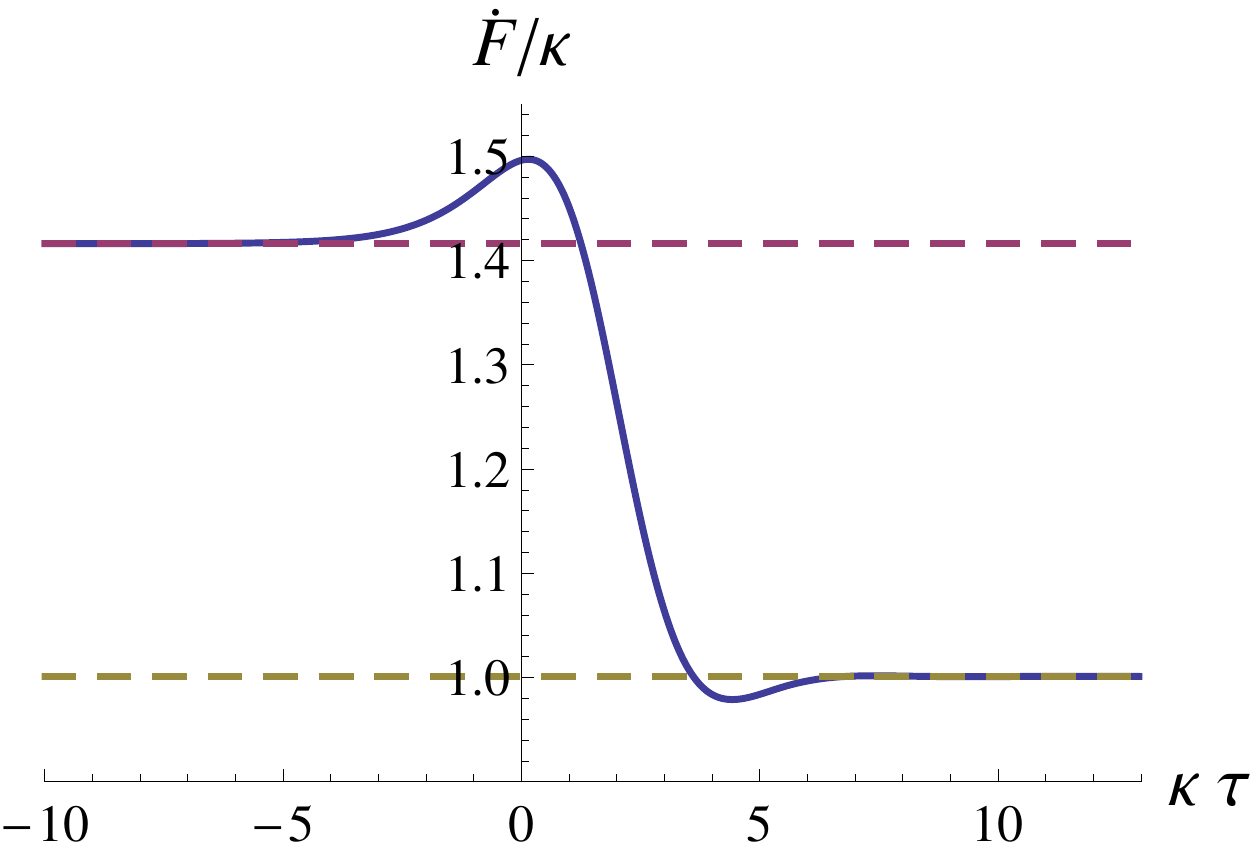}& 
\includegraphics[width=0.45\textwidth]{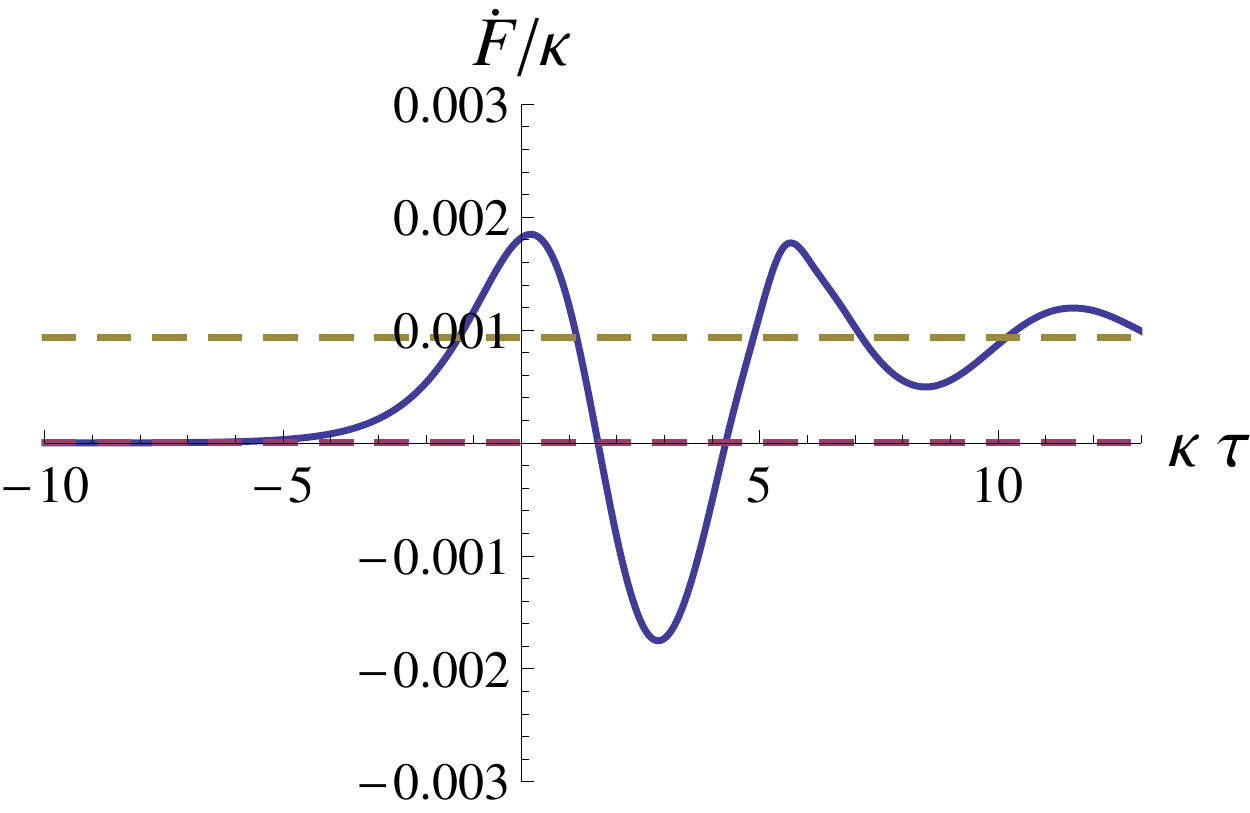} \\
(a) & (b)
\end{tabular}
\caption[Asymptotic early and late-time behaviour of the transition rate in the receding mirror spacetime.]{Cross-sections of the plot in Fig. \ref{fig:mirror3d} at (a) $E = - \kappa$ and (b) $E = + \kappa$. The asymptotic values are represented by the dashed lines.}
\label{fig:mirrorasympt}
\end{figure}

\begin{figure}[!ht]
\centering
\includegraphics[height=7cm]{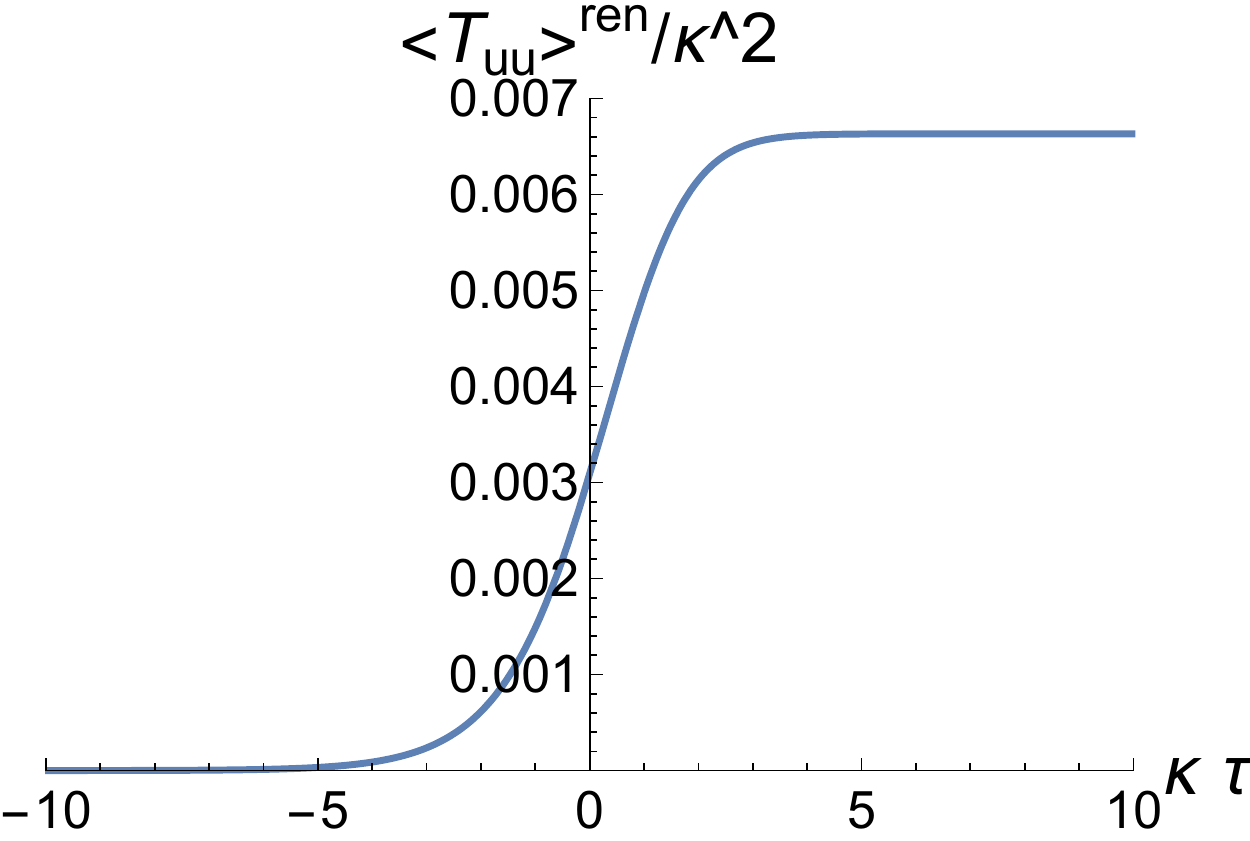}
\caption[Energy flux in the receding mirror spacetime.]{Right-moving energy flux crossing the trajectory \eqref{4:StaticRecMir} with $d = 1/\kappa$.}
\label{fig:EnergyFlux}
\end{figure}

\subsection{Effects in the $1+1$ Schwarzschild black hole}
\label{subsec:4SchwE}

We consider the effect of a massless scalar field in $(1+1)$ Schwarzschild spacetime in the Hartle-Hawking-Israel and Unruh vacua on particle detectors. We shall briefly discuss the situation for a static detector at fixed radius in the exterior region, switched on sharply in the asymptotic past, which registers an Unruh temperature due to the acceleration that it requires to remain static, which coincides with the Hawking temperature weighted by the appropriate Tolman factor. A key difference with $3+1$ black holes is that in $1+1$ dimensions there is no black hole mining: The high energy modes of the quantum field are radiated away, as there is no potential angular momentum barrier that reflects back high angular momentum modes, as in the $3+1$ situation.\footnote{The author thanks Prof. Adrian Ottewill for pointing this out.} We then proceed to inspect the interesting case of an inertial, geodesically-infalling detector. First, we deal with the asymptotic rate near the infinity. In this case, the detector registers an Unruh temperature proportional to the surface gravity of the black hole, as in the full $3+1$ case. In the case of the HHI state, the rate satisfies the KMS condition, while in the Unruh state case only the modes coming from the black hole are thermal. Next, we analyse the non-thermal, divergent behaviour near the singularity and characterise the divergence rate of the detector. The loss of thermality of the infalling detector is computed numerically and we find that the thermal character of the response is lost before the detector crosses the horizon. Finally, we examine numerically the response of the detector that is switched on in the white hole region, in the future of the white hole singularity, and falls geodesically to the black hole through the bifurcation point in the Schwarzschild spacetime. The rate diverges as the switch-on time approaches the white hole singularity.

\subsubsection*{Static detector}

Consider a detector following the trajectory $r = R >2M$ in region I. Calculating the derivatives of the pull-back of the Wightman function, we can use formula \eqref{DeCoStat} from Chapter \ref{ch:DeCo} with
\begin{subequations}
\begin{align}
\mathcal{A}_\text{H}\left(\tau',\tau''\right) & = \frac{1}{4\pi} \left[\frac{\dot{U}\left(\tau'\right) \dot{U}\left(\tau''\right)}{\left[\epsilon + \ii \left( U\left( \tau' \right) - U\left(  \tau'' \right) \right) \right]^2} + \frac{\dot{V}\left(\tau'\right) \dot{V}\left(\tau''\right)}{\left[\epsilon + \ii \left( V\left( \tau' \right) - V\left(  \tau'' \right) \right) \right]^2}\right], \label{4:SchwHHIA} \\
\mathcal{A}_\text{U}\left(\tau',\tau''\right) & = \frac{1}{4\pi} \left[\frac{\dot{U}\left(\tau'\right) \dot{U}\left(\tau''\right)}{\left[\epsilon + \ii \left( U\left( \tau' \right) - U\left(  \tau'' \right) \right) \right]^2} + \frac{\dot{v}\left(\tau'\right) \dot{v}\left(\tau''\right)}{\left[\epsilon + \ii \left( v\left( \tau' \right) - v\left(  \tau'' \right) \right) \right]^2}\right], \label{4:SchwUA}
\end{align}
\label{4:SchwA}
\end{subequations}
\\
in the HHI state and Unruh states respectively and along the static trajectories, $t'-t'' = (\tau'-\tau'')/(1-2M/R)^{1/2}$, $r' = r'' = R$. 

In Appendix \ref{app:BH} we show that
\begin{subequations}
\begin{align}
\dot{\mathcal{F}}_\text{H}(E) & = \frac{E}{\ee^{E/T_{\text{loc}}}-1}, \label{4:StaticHHI} \\
\dot{\mathcal{F}}_\text{U}(E) & = -\frac{E}{2}\Theta(-E) + \frac{E}{2\left(\ee^{E/T_{\text{loc}}}-1\right)}, \label{4:StaticUnruh}
\end{align}
\end{subequations}
\\
where $T_\text{loc} = 1/[(8\pi M)(1-2M/R)^{1/2}]$ is defined in the appendix. Thus, in the HHI state, the transition rate is thermal and satisfies the detailed balance condition at local temperature $T_\text{loc}$, which is proportional to the surface gravity of the black hole weighted by the local Tolman factor. The detector coupled to the field in the Unruh state perceives a temperature coming from the right movers, while the left moving modes detected are indistinguishable from those in the Minkowski vacuum state.


\subsubsection*{Inertial detector near the infinity}

We consider detectors falling geodesically from infinity with $\varepsilon >1$ and $\varepsilon = 1$, eq. \eqref{4:SchwGeodesicsE>1} and \eqref{4:SchwGeodesicsE=1} respectively, which have been switched on in the asymptotic past. We consider the cases of interactions with fields in the HHI and Unruh states. In this early time regime, we find in Appendix \ref{app:BH} that the instantaneous transition rates are given by

\begin{subequations}
\begin{align}
\dot{\mathcal{F}}_\text{H}(E) & = \frac{E}{2\left(\ee^{E/T_-}-1\right)} + \frac{E}{2\left(\ee^{E/T_+}-1\right)} + \textit{o}(1), \label{4:InftyHHI} \\
\dot{\mathcal{F}}_\text{U}(E) & = -\frac{E}{2}\Theta(-E) + \frac{E}{2\left(\ee^{E/T_+}-1\right)} + \textit{o}(1), \label{4:InftyUnruh}
\end{align}
\label{4:InfinitySchw}
\end{subequations}
\\
where $T_\pm = \ee^{\pm \lambda} / (8\pi M)$ is the Doppler-shifted local temperature measured by the infalling detectors, with $\lambda = \arctanh\left[\left(1- \varepsilon^{-2} \right)^{1/2}\right]$. In the case $\varepsilon = 1$, $T_+ = T_- = 1/(8\pi M)$, as there is no Doppler shift experienced.

The results \eqref{4:InfinitySchw} conform fully with our expectations. The HHI transition rate is in agreement with the rate of a detector immersed in a heat bath in Minkowski spacetime \eqref{3:ThermalBathRate}, and satisfies the detailed balance condition at the Hawking temperature, $T = 1/(8\pi M)$, at infinity. The left-moving and right-moving contributions are Doppler-shifted to account for the detector velocity at infinity. 

The Unruh transition rate registers the Hawking effect coming from the black hole, and contains a blue shift due to the ingoing velocity of the detector. Left-moving modes coming from the asymptotic infinity are indistinguishable from the Minkowski modes. The response of the detector coincides with the late time asymptotic behaviour in the receding mirror spacetime \eqref{4:BoostRecMirFuture}, confirming the interpretation that the Unruh state is the state of a quantum field produced by the collapse of a star.

\subsubsection*{Inertial detector near the singularity}

We consider the transition rate in the HHI and Unruh vacua as the detector approaches the black hole singularity. We allow for all non-negative values of $\varepsilon$ and consider an arbitrary switch-on time. For $0 \leq \varepsilon < 1$, the switch-on in the HHI vacuum takes place after the detector has emerged from the white hole singularity and the switch-on in the Unruh vacuum takes place after the detector has crossed the past horizon.

We find in Appendix \ref{app:BH} that the near-singularity behaviour is, to leading order
\begin{equation}
\dot{\mathcal{F}}(E,\tau) = \frac{1}{8\pi M} \left[\left(\frac{2M}{r(\tau)}\right)^{3/2} + \frac{1+\varepsilon^2}{2}\left(\frac{2M}{r(\tau)}\right)^{1/2} \right] + O(1).
\label{4:NearSing}
\end{equation}
for both the HHI and Unruh vacua. The leading divergence is generic. The dependence on the trajectory enters as a subleading divergence. The state-dependent part is $O(1)$. In terms of the proper time along the trajectory, reaching the singularity at $\tau = \tau_\text{sing}$, the leading term is $1/[6\pi(\tau_\text{sing} - \tau)]$.


\subsubsection*{Intermediate regime: loss of thermality}

We consider the transition rate as a detector falls freely into the black hole singularity from infinity, while interacting with fields in the HHI or Unruh vacua. Numerical evidence of how the loss of thermality takes place for the $\varepsilon = 1$ trajectory in the HHI and Unruh vacua is presented in Figures \ref{fig:HHomega} and \ref{fig:Unruhomega} respectively. 
Our numerics show that the Planckian form of the transition rate is lost already in region I, well before the detector crosses the horizon. 

\begin{figure}
\centering
\begin{tabular}{c}
\includegraphics[height=5cm]{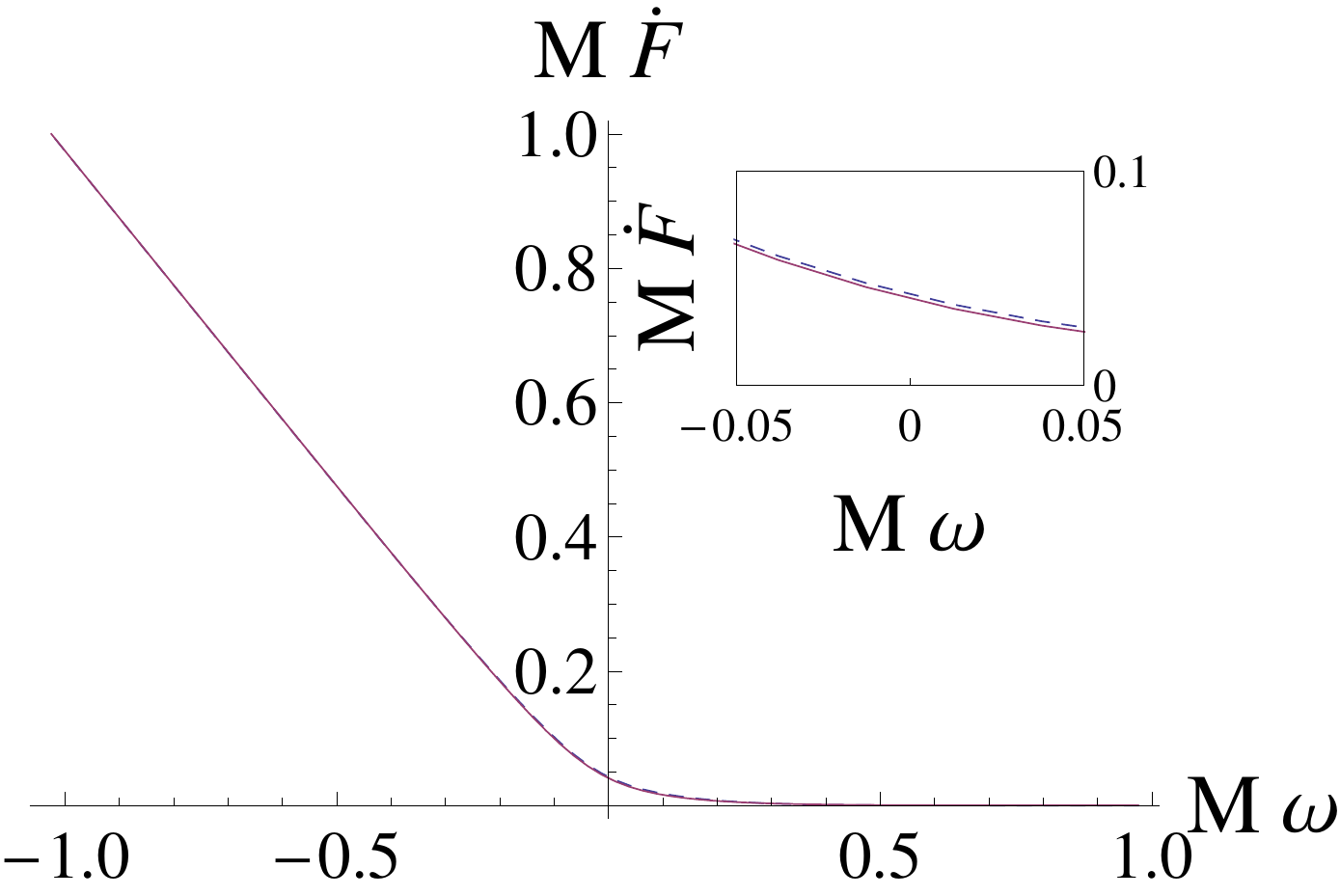} \\ (a) \\ \includegraphics[height=5cm]{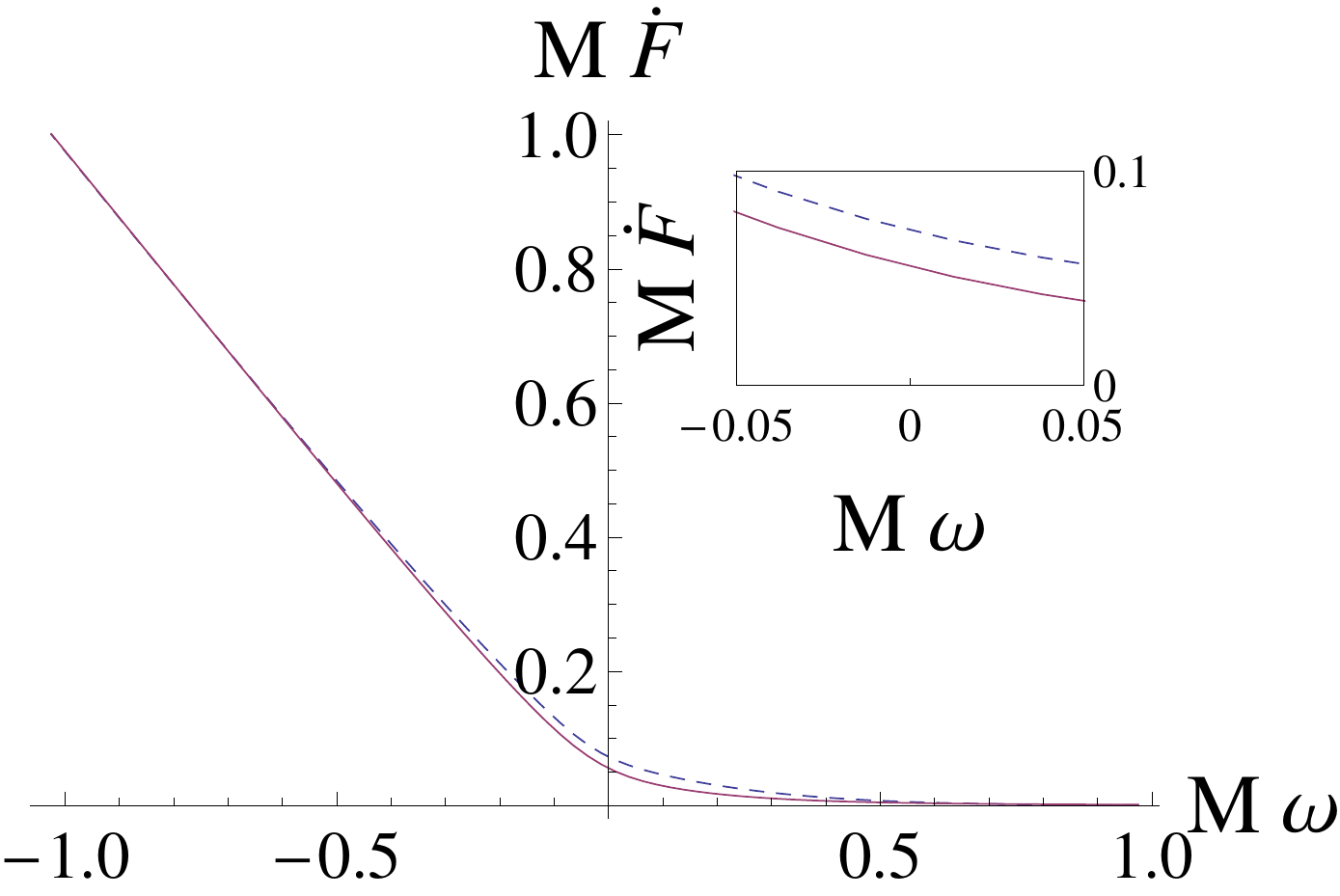} \\ (b) \\
\includegraphics[height=5cm]{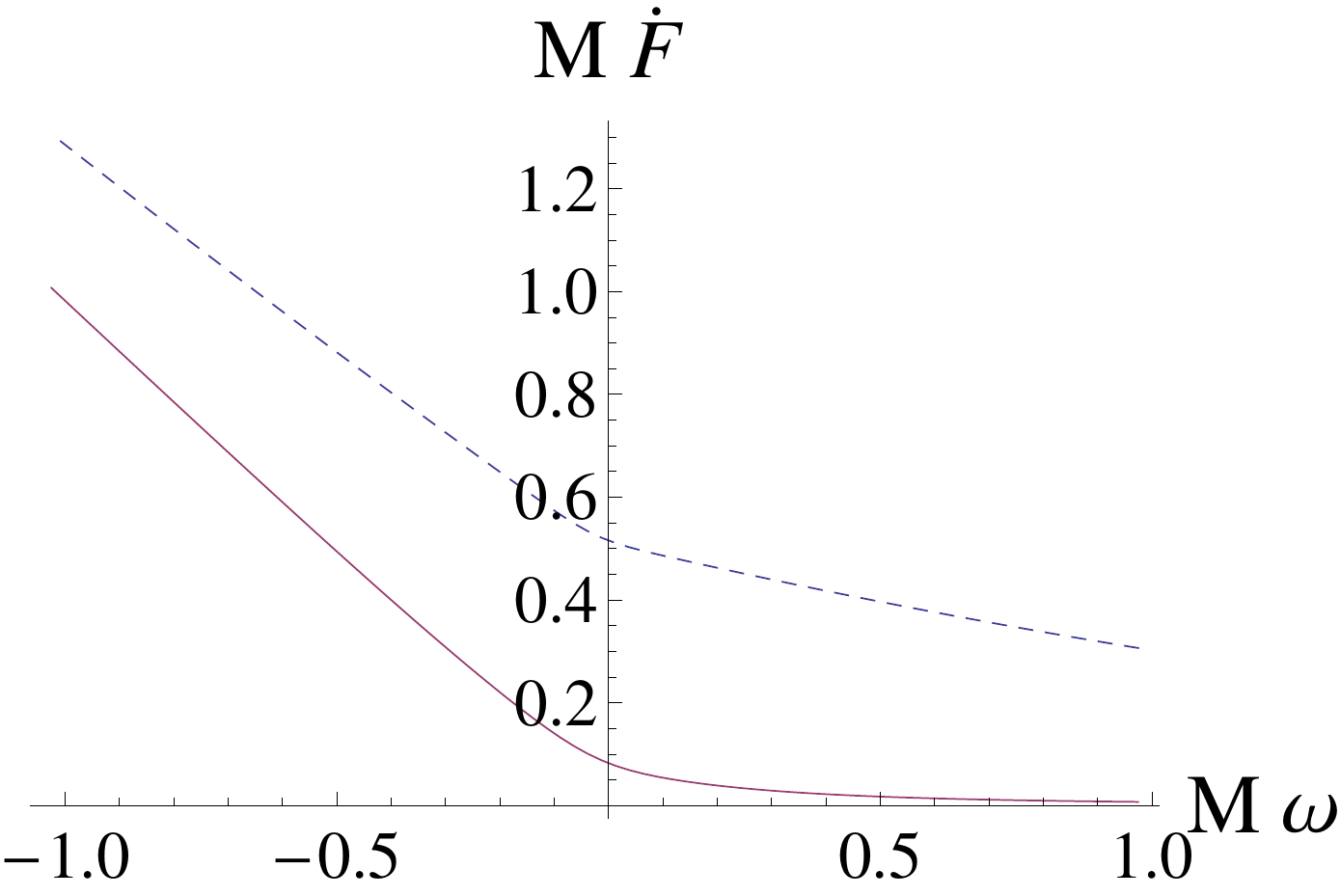} \\ (c)
\end{tabular}
\caption[Loss of thermality along infalling geodesic in the HHI state in Schwarzschild]{The solid (red) curve shows $M\dot{\mathcal{F}}$ as a function of $M\omega$, with $\omega = E$, setting $\hbar = 1$, the frequency of a detector transition, for the $\varepsilon = 1$ geodesic in the HHI vacuum state at times (a) $\tau = -10 M$, (b) $\tau = -3.5 M$ and (c) $\tau = -1.5 M$. The dashed (blue) curve shows $M$ times the Minkowski thermal bath rate \eqref{3:ThermalBathRate} at the local Hawking temperature $T_\text{loc} = \lambda_\text{loc}/\left[(8\pi M)\left(1-2M/r\right)^{1/2}\right]$ against $M\omega$, where $\lambda_\text{loc} = \arctanh \left[ \left(2M/r \right)^{1/2} \right]$ is the Doppler factor of a detector moving at rapidity $\left(2M/r \right)^{1/2}$ in Minkowski spacetime. The curve discrepancy shows how the Planckian character of the transition rate is lost as $\tau$ approaches the value at the horizon-crossing, where the solid curve remains finite, but the dashed curve increases towards $+\infty$.}
\label{fig:HHomega}
\end{figure}

\begin{figure}
\centering
\begin{tabular}{c}
\includegraphics[height=5cm]{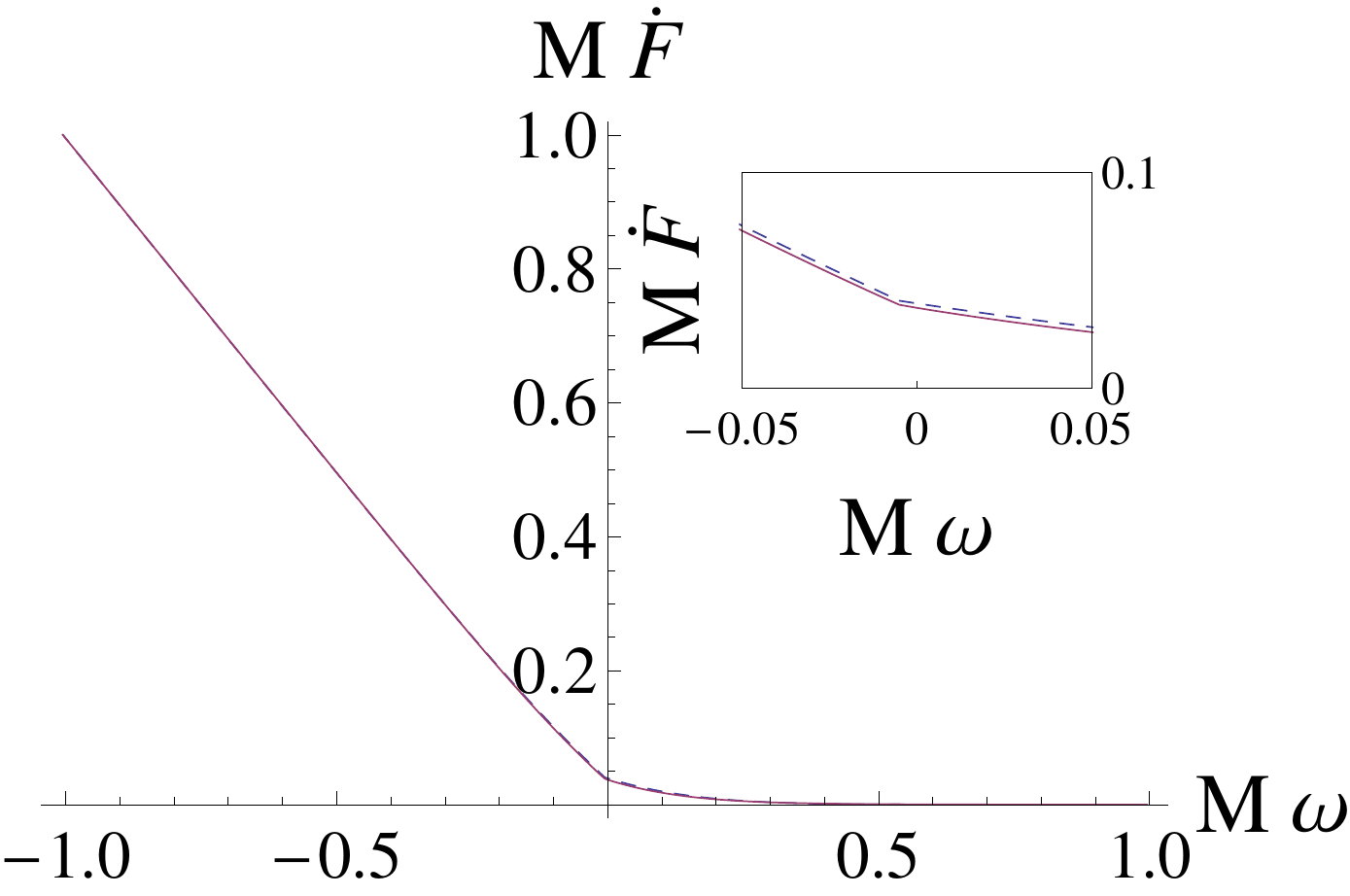} \\ (a) \\ 
\includegraphics[height=5cm]{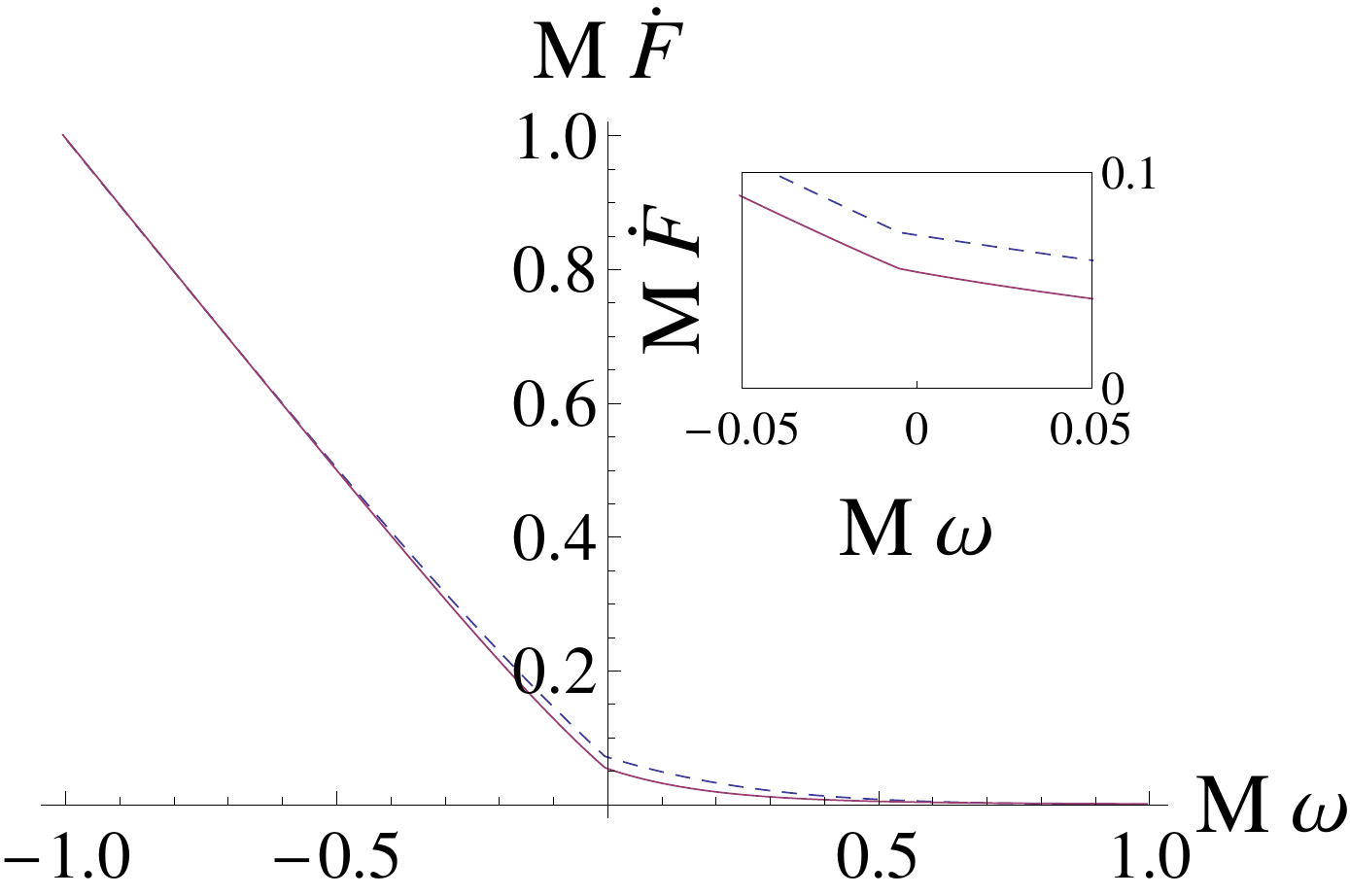} \\ (b) \\
\includegraphics[height=5cm]{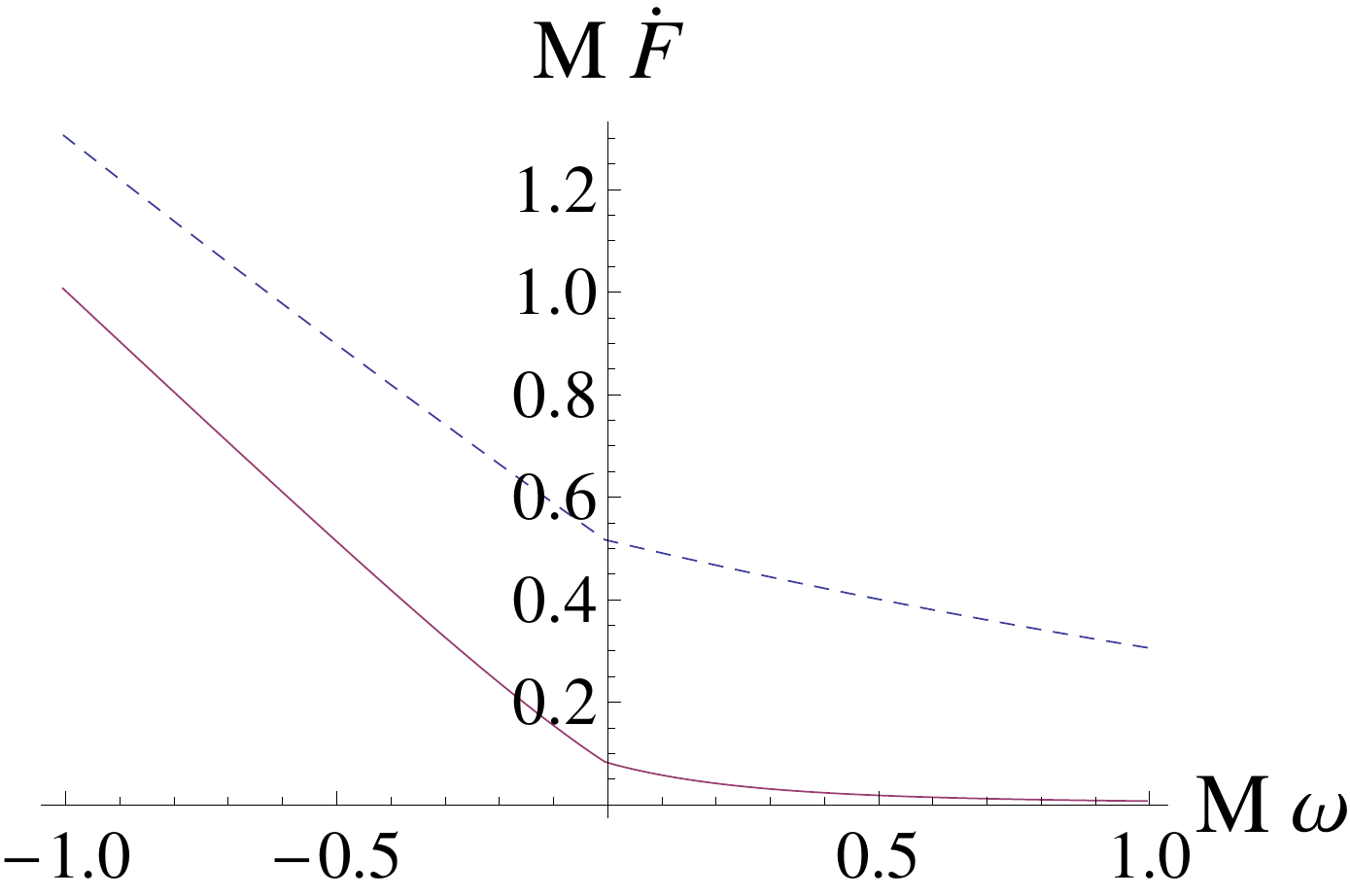} \\ (c)
\end{tabular}
\caption[Loss of thermality along infalling geodesic in the Unruh state in Schwarzschild]{The solid (red) curve is as in Fig. \ref{fig:HHomega}, this time in the Unruh vacuum. The dashed (blue) curve shows $M$ times the Minkowski thermal bath rate for right-movers at temperature $T_\text{loc}$ and time the Minkowski vacuum transition rate for left-movers. The curve discrepancy shows how the Planckian character of the transition rate is lost as $\tau$ approaches the value at the horizon-crossing, where the solid curve remains finite, but the dashed curve increases towards $+\infty$. Note that the dashed curve has discontinuous slope at $\omega = 0$.} 
\label{fig:Unruhomega}
\end{figure}

The numerical evidence also confirms that the magnitude of the transition rate increases as the detector falls into the black hole and becomes unbounded as the detector falls into the singularity. Again, for the $\epsilon = 1$ infalling detector, our numerics agree with the asymptotic divergent behaviour obtained for the transition rate near the singularity. This is shown in Fig. \ref{fig:HHIdiv} for the HHI vacuum. 

\begin{figure}
\centering
\begin{tabular}{cc}
\includegraphics[width=0.45\textwidth]{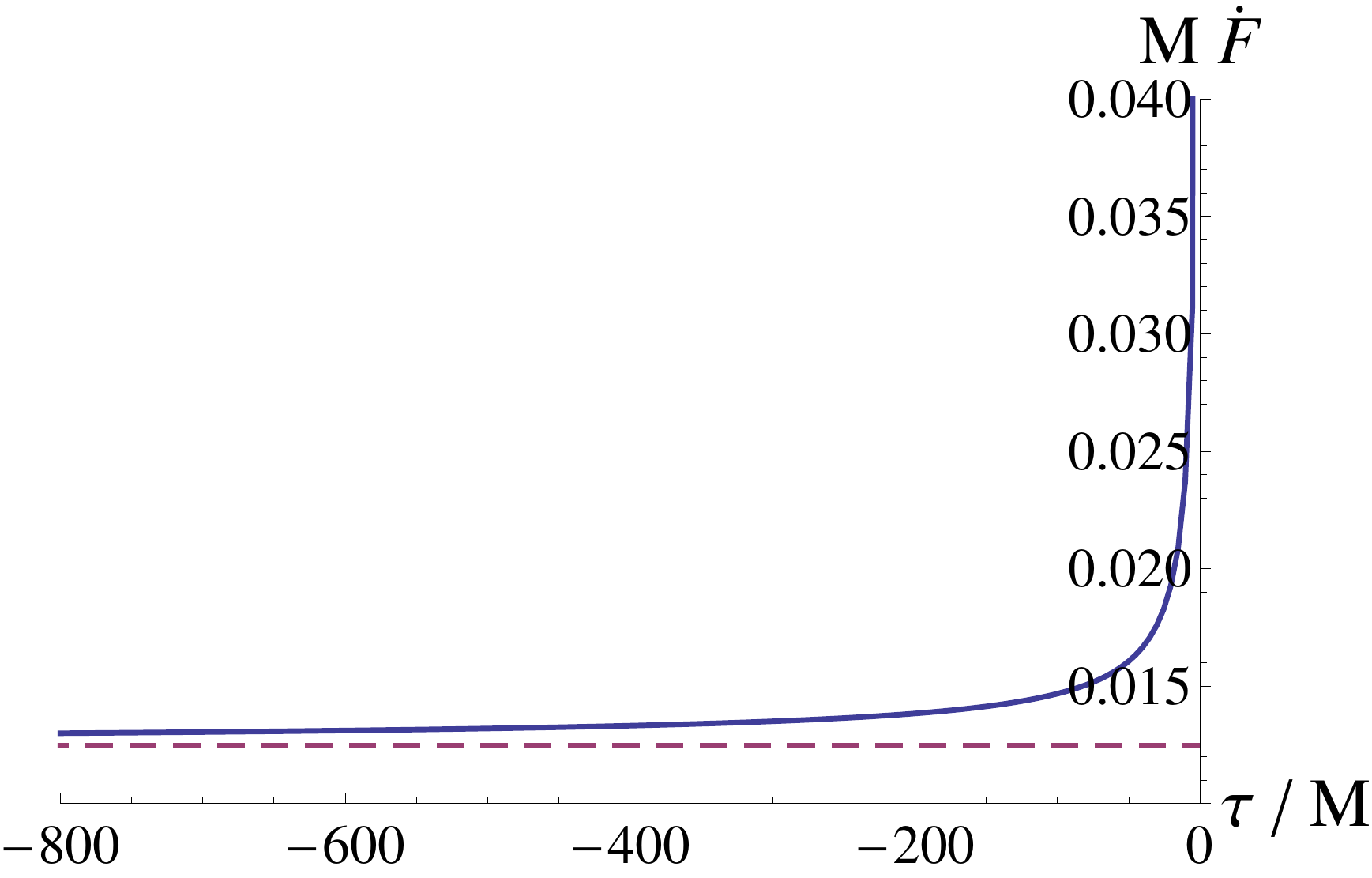}& 
\includegraphics[width=0.45\textwidth]{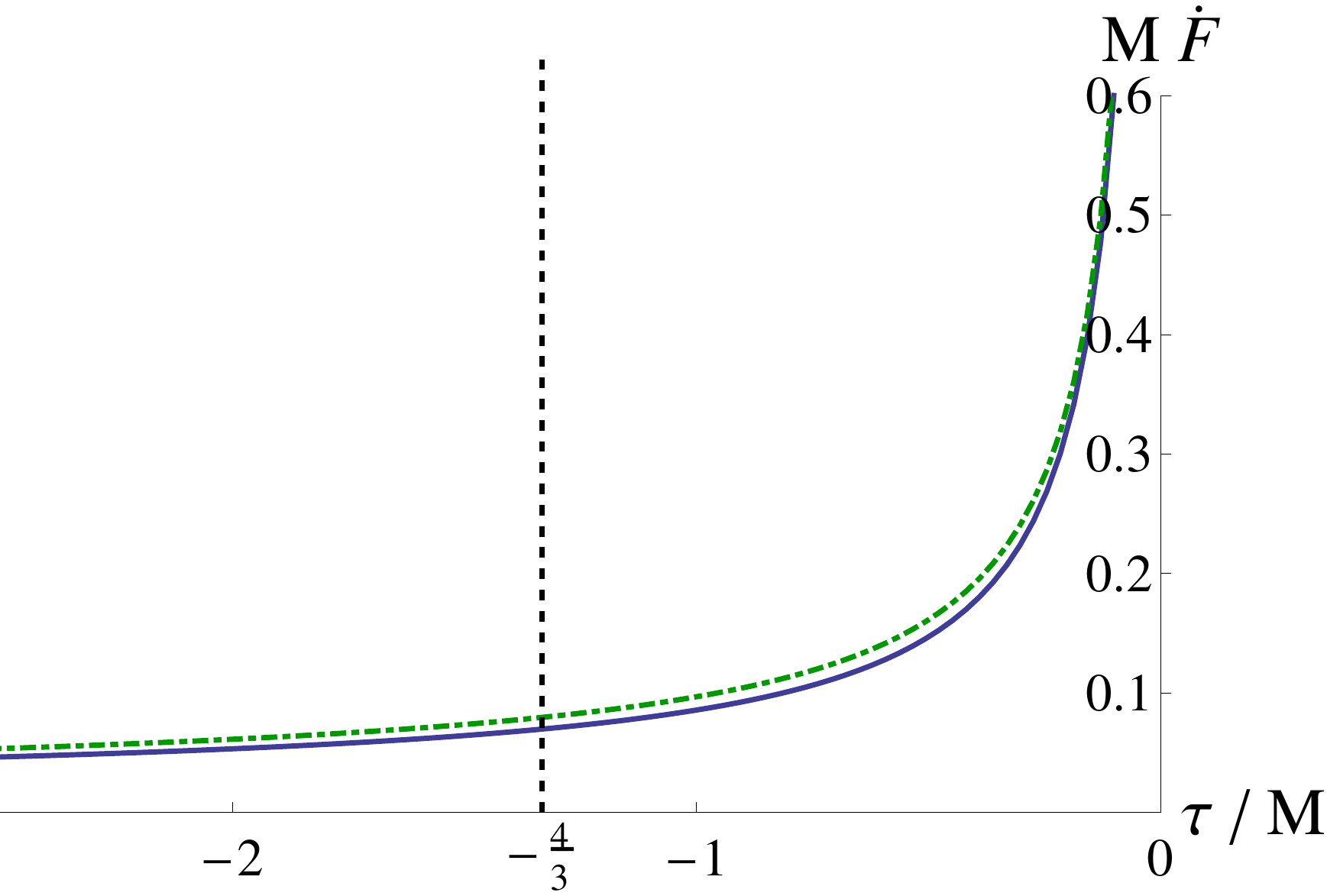} \\
(a) & (b)
\end{tabular}
\caption[Asymptotic behaviour of the transition rate near infinity and near the singularity in the HHI state in Schwarzschild]{(a) The solid (blue) curve shows $M \dot{\mathcal{F}}(1/(4\pi M))$ as a function of $\tau/M$ for the $\varepsilon = 1$ trajectory in the HHI vacuum. The dashed (red) line shows the asymptotic value $1/\left[ 4\pi \left( \ee^2 - 1 \right) \right]$ in the distant past. (b) The solid (blue) curve shows a close-up of (a) near the horizon-crossing at $\tau/M = -4/3$. The dash-dotted (green) curve shows the $\tau$-dependent terms included in the $\tau \to 0$ asymptotic expression \eqref{4:NearSing}.} 
\label{fig:HHIdiv}
\end{figure}

Finally, let the field be in the HHI state, and consider the $\varepsilon = 0$ trajectory that goes from the white hole region directly into the black hole region through the bifurcation point.  The geodesic equation is \eqref{4:SchwGeodesicsWH}, such that the parametrisation \eqref{4:SchwGeodesicsE<1} holds with $\varepsilon = 0$. We consider a detector switch-on at parameter $\eta = -9\pi/10$, near the white hole singularity at $\eta = -\pi$. Fig \ref{fig:HHIWhiteHole} shows a perspective plot of the transition rate as a function of the transition frequency and the proper time along the trajectory. The plot shows an inverse linear divergence as the proper time approaches the switch-on time (cf. eq. \eqref{Fdot}), as well as the divergence that occurs as the detector approaches the black hole singularity.

\begin{figure}
\begin{center}
\includegraphics[height=8cm]{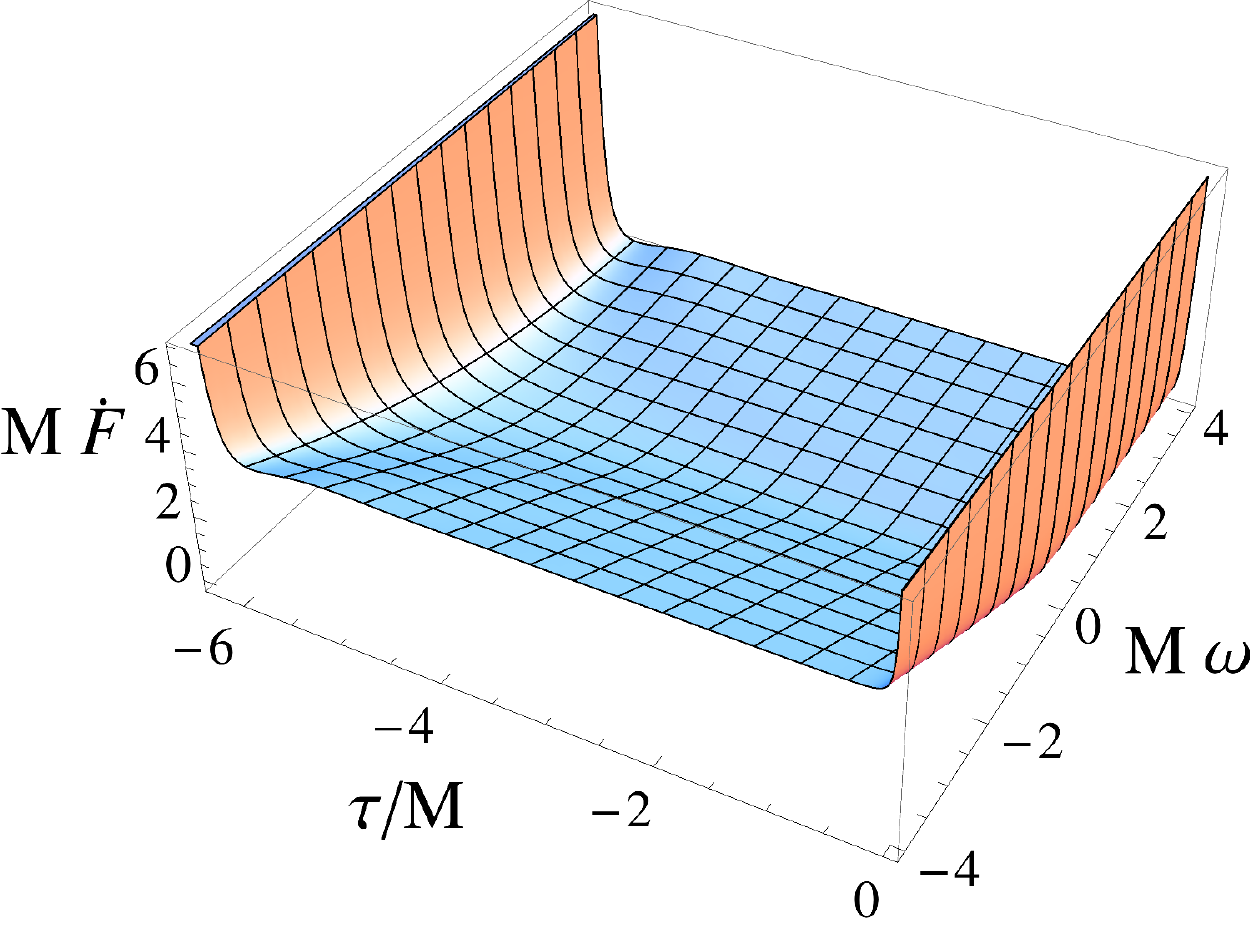}
\caption[Transition rate for $\varepsilon = 0$ geodesic in the HHI state in Schwarzschild]{Transition rate of a detector following the $\varepsilon = 0$ geodesic in the HHI vacuum. The white hole singularity is located at $\tau = -\pi M$ ($\eta = -9/10$), the switch-on occurs at $\tau \approx 3.136 M$ ($\eta = -9 \pi/10$). The transition rate diverges as the switch-on time is approached, as well as near the black hole singularity.}
\label{fig:HHIWhiteHole}
\end{center}
\end{figure}

A final comment is due regarding numerical applications of the $1+1$ derivative-coupling detector described in Chapter \ref{ch:DeCo}. The regulator-free formulas allow one to implement numerics directly at the level of the responses and rates. If one seeks to examine the thermal character of a response, it suffices to numerically compute the quantity
\begin{equation}
T_\text{KMS}(E,\tau,\tau_0) \doteq \frac{E}{\ln \left( \dot{\mathcal{F}}(-E,\tau,\tau_0)/\dot{\mathcal{F}}(E,\tau,\tau_0) \right)}.
\end{equation}

If $T_\text{KMS}$ is approximately constant for a range of values of $E$ at fixed $\tau$, the transition rate would approximately satisfy the KMS condition in such range of $E$-values, with, possibly $\tau$-dependent, approximate temperature $T_\text{KMS}$.

We have performed this test in the case of the $\varepsilon = 0$ trajectory in Fig. \ref{fig:HHIWhiteHole}, and we have found no range of values at which $T_\text{KMS}$ would be approximately independent of $E$ within our space of parameters.

\subsection{Effects in the $1+1$ generalised Reissner-Nordstr\"om black hole}
\label{subsec:4RNE}

As we have discussed before, causality is protected by global hyperbolicity. Thus, given that many important solutions in general relativity are non-globally hyperbolic, it is not surprising that a great deal of effort has been spent in understanding the stability of Cauchy horizons in General Relativity. 

Early attempts to understand this problem lead to the \textit{strong cosmic censorship conjecture} \cite{Simpson:1973ua}, which states that for generic initial data the spacetime is inextendible beyond the maximal Cauchy development. It turns out the strong cosmic censorship conjecture is incorrect. Given polynomially decaying initial data for a dynamical Maxwell-Einstein-scalar system settling to a Reissner-Nordstr\"om black hole, the spacetime is extendible with a $C^0(M,\mathbb{R})$-metric, falsifying the conjecture. However, not all the physical invariants are finite and, in particular, the Hawking mass goes to infinity. This is known as the mass inflation scenario \cite{Poisson:1989zz, Poisson:1990eh, Dafermos:2002ka, Dafermos:2003wr, Costa:2014zha}. On the other hand, Chandrasekhar and Hartle \cite{Chandrasekhar:1982} have shown that the (electromagnetic or gravitational) classical radiation felt by an observer diverges as the Reissner-Nordstr\"om horizon is approached.

Thus far, all of the above is concerned with classical aspects of the problem of causality. We wish to address the quantum experience of observers as they attempt to cross a Cauchy horizon. In particular, we consider the interaction of such observers with a scalar field, in the limit of no back-reaction, as the observers attempt to cross the horizon.

We analyse the experiences of an inertial observer who attempts to fall through the Cauchy horizon in our $(1+1)$-dimensional generalised nonextremal Reissner-Nordstr\"om black hole. First, we compute the transition rate measured by the derivative-coupling Unruh-DeWitt particle detector that is coupled to a scalar field in the HHI or in the Unruh state, as the observer carrying the detector approaches the future Cauchy horizon. Second, we compute the local near-horizon renormalised energy density along the infalling worldline. We shall find that both quantities diverge proportionally to $1/(\tau_h - \tau)$, where $\tau$ is the
observer's proper time and $\tau_h$ is the value at horizon-crossing, for the detector coupled to the field in the HHI and Unruh states, and for the energy density in the HHI vacuum. This divergent behaviour is, in both cases, independent of the initial conditions of the geodesic, and similar to the one encountered above for the detector approaching the $1+1$ Schwarzschild singularity \eqref{4:NearSing}.

\subsubsection*{Inertial detector near the Cauchy horizon in the HHI state}

Let us start by computing the transition rate of the geodesic detector (see \eqref{4:RNeqmotion} and \eqref{4:RNeqmotion2}) when the field is found in the HHI state. Using the pullback along the geodesic of the HHI Wightman function \eqref{4:RNHHI} in the transition rate formula \eqref{Fdot}, we obtain
\begin{align}
\dot{\mathcal{F}}_\text{H}(E, \tau, \tau_0) & = -\frac{E}{2} + \frac{1}{\pi \Delta \tau} + 2\int_0^{\Delta \tau} \! ds \, \left(\frac{1- \cos(Es)}{2 \pi s^2} \right) \nonumber \\
& + \int_{\tau_0}^\tau \! d\tau' \, \cos\left[E\left( \tau-\tau' \right)\right] \, \partial_{\tau'} \left( \partial_\tau \W\left(\tau,\tau'\right)  + \frac{1}{2 \pi \left( \tau - \tau' \right)}\right),
\label{4:HHILimPre}
\end{align}
where we have written $\mathcal{A}(\tau,\tau') = \partial_\tau \partial_{\tau'} \W (\tau,\tau')$.

All but the last terms on the right hand side of \eqref{4:HHILimPre} are $O(1)$ as $\tau' \to \tau$. Integrating the last term by parts one can write
\begin{align}
\dot{\mathcal{F}}_\text{H}(E, \tau, \tau_0) & = - 2 \cos\left(E \Delta \tau\right)\partial_\tau \mathcal{W}_\text{H}(\tau,\tau_0)\nonumber \\
& + \lim_{\tau' \rightarrow \tau }2 \cos\left(E(\tau-\tau')\right) \left[\partial_\tau \mathcal{W}_\text{H}(\tau,\tau')+\frac{1}{2 \pi (\tau-\tau')} \right] \nonumber \\
& - 2 \int_{\tau_0}^\tau \! d \tau' \, E \sin \left(E(\tau-\tau')\right) \left[\partial_\tau \mathcal{W}_\text{H}(\tau,\tau')+\frac{1}{2 \pi (\tau-\tau')} \right] + O(1),
\label{4:HHIlim}
\end{align}
\\
where $\tau_0$ is a fixed point along the trajectory inside region II and the order $O(1)$ is with respect to the quantity $(\tau_h - \tau)$, which is small as the detector approaches the Cauchy horizon at proper time $\tau_h$. We show in Appendix \ref{app:BH} that the contribution of the boundary terms is important and scales as $1/(\tau_h-\tau)$, near the Cauchy horizon,
\begin{align}
\dot{\mathcal{F}}_\text{H}(E, \tau, \tau_0) & = -\left[ \frac{1+( 2 \cos(E \Delta \tau)-1)(\kappa_+/\kappa_-)}{4 \pi} + O\left((\tau_--\tau)^{-\kappa_+/\kappa_-}\right) \right] (\tau_h-\tau)^{-1}
 \nonumber \\
& - 2 \int_{\tau_0}^\tau \! d \tau' \, E \sin \left(E(\tau-\tau')\right) \left[\partial_\tau \mathcal{W}_\text{H}(\tau,\tau')+\frac{1}{2 \pi (\tau-\tau')} \right] + O(1) \nonumber \\
& = -\left[ \frac{1+(2 \cos(E \Delta \tau)-1)(\kappa_+/\kappa_-)}{4 \pi} + O\left((\tau_--\tau)^{-\kappa_+/\kappa_-}\right) \right] (\tau_h-\tau)^{-1}
 \nonumber \\
& + \left[ \frac{E}{2\pi} \frac{\dot{U}(\tau)}{U(\tau)} \int_{\tau_0}^\tau \! d \tau' \, \frac{ \sin \left(E(\tau-\tau')\right)}{1-U(\tau')/U(\tau)} + O(1) \right].
\label{4:HHIlim2}
\end{align}

In the last line of eq. \eqref{4:HHIlim2}, the factor multiplying the integral term contributes to leading order as $\dot{U}(\tau)/U(\tau) = -(\kappa_+/\kappa_-)(\tau_h-\tau)^{-1} + O(1)$. The integral term on the last line of eq \eqref{4:HHIlim2}, which we denote
\begin{equation}
I(E, \tau,\tau_0) \doteq \int_{\tau_0}^\tau \! d \tau' \, \frac{ \sin \left(E(\tau-\tau')\right)}{1-U(\tau')/U(\tau)},
\label{4:RNIint}
\end{equation}
is estimated in Appendix \ref{app:BH} to be
\begin{align}
I(E, \tau, \tau_0) & = - \frac{1 - \cos(E \Delta \tau)}{E} + \mathit{o}(1).
\label{4:RNintResult}
\end{align}

Collecting the terms together, we find that
\begin{equation}
\dot{F}_\text{H}(E,\tau,\tau_0) = -\frac{1}{4 \pi} \left(1 +\frac{\kappa_+}{\kappa_-} +\textit{o}(1) \right) \frac{1}{\tau - \tau_h}.
\label{4:RNHHIFdotfinal}
\end{equation}

This concludes our computation for the detector interacting with the field in the HHI state. It can be read off that the independence on the geodesic parameter $\varepsilon>0$ makes the result independent of the initial details of the trajectory to leading order. This is a feature that is shared with a geodesic detector approaching the spacelike Schwarzschild curvature singularity in $1+1$ in the HHI vacuum. Moreover, the leading degree of divergence is of the same order \eqref{4:NearSing}.

This calculation suggests that something special occurs whenever $\kappa_+ = - \kappa_-$ as the coefficient in \eqref{4:RNHHIFdotfinal} becomes $\textit{o}(1)$. This observation can be strengthened. Indeed, a direct calculation shows that in this case $\dot{F}_\text{H}(E,\tau,\tau_0) = O(1)$ and the horizon can be crossed. Black holes with the property $\kappa_+ = -\kappa_-$ are called lukewarm black holes. To achieve this feature one relies of the fine tuning of the parameters of a de-Sitter Reissner-Nordstr\"om black hole with the cosmological constant \cite{Romans:1991nq}.

Comments are in order concerning the cases $\varepsilon < 0$ and $\varepsilon = 0$. It is clear that for $\varepsilon < 0$ the situation is the same as above, except in this case the divergence will arise from the right-moving sector of the Wightman function. The case $E = 0$ describes a straight-up trajectory from the white hole into the black hole region crossing through the bifurcation point of the bifurcate Killing horizon, located at Kruskal coordinates $ U = V = 0$. In this case, the rate of divergence is stronger, and in fact twice as $\eqref{4:RNHHIFdotfinal}$ cf. eq. \ref{4App:FdotUp}, as both the left-moving and right-moving sectors of the field contribute strongly as $\tau \rightarrow \tau_h$.

\subsubsection*{Inertial detector near the Cauchy horizon in the Unruh state}

We consider now the detector approaching $\mathcal{C}^\text{FL}$, the left portion of the future Killing horizon, along the left-moving inertial trajectory, as it interacts with a field in the Unruh state. If $\varepsilon \geq 1$, we consider the switch-on time to take place at an arbitrary initial time $\tau_i$. If $0< \varepsilon <1$, the detector is switched on after it has left region II'. If we consider $\tau_0$ as the proper time along the trajectory at a point inside region II, we can write
\begin{align}
\dot{F}_\text{U}(E, \tau, \tau_0) & = \dot{F}_\text{H}(E,\tau,\tau_0) \nonumber \\
& + 2 \int_0^{\Delta \tau} \! ds \, \text{Re} \Big[ \ee^{-\ii E s} \partial_\tau \partial_{\tau-s} \Big( \mathcal{W}_\text{U}(\tau, \tau-s) - \mathcal{W}_\text{H}(\tau, \tau-s) \Big) \Big].
\label{4:RNUnruhFdot}
\end{align}

The integral on the right hand side of eq. \eqref{4:RNUnruhFdot} is finite by virtue of the Hadamard property. Integrating by parts, the boundary values,
\begin{align}
\left[-2 \cos(Es) \partial_\tau \Big( \mathcal{W}_\text{U}(\tau, \tau-s) - \mathcal{W}_\text{H}(\tau, \tau-s) \Big)\right]_{0}^{\Delta \tau}
\end{align}
give a term of order $O(1)$. The lower evaluation is understood as a limit. Moreover, the integral
\begin{equation}
\partial_\tau \Big(\mathcal{W}_\text{U}(\tau, \tau-s) - \mathcal{W}_\text{H}(\tau, \tau-s)\Big) = -\frac{1}{4 \pi} \left(\frac{\dot{v}(\tau)}{v(\tau)-v(\tau-s)} - \frac{\dot{V}(\tau)}{V(\tau)-V(\tau-s)} \right)
\end{equation}
is $O(1)$ in $(\tau_h-\tau)$ near $\mathcal{C}^{\text{FL}}$. Hence, the transition rate is the same as in the Hartle-Hawking scenario, to leading order, and the transition rate diverges, to leading order, at the same rate as the future horizon is approached,
\begin{equation}
\dot{\mathcal{F}}(E,\tau,\tau_0) = -\frac{1}{4 \pi} \left(1 +\frac{\kappa_+}{\kappa_-} +\textit{o}(1) \right) \frac{1}{\tau - \tau_h}.
\label{4:RNDetectorRate}
\end{equation}

\subsubsection*{Local energy density near the Cauchy horizon}

Having examined the transition rate of a sharply switched, derivative-coupling detector in $1+1$ dimensions, let us turn to the question of what the renormalised energy density, $\langle \mathcal{\rho} \rangle^\text{ren}$, along the worldline of a geodesically-infalling observer is, as they approach the region $\mathcal{C}^{\text{FL}}$ of the future Cauchy horizon. 

\begin{defn}
Let $\gamma:\mathbb{R} \rightarrow M: \tau \mapsto \gamma(\tau)$, be a timelike path, where $\tau$ is the proper time along the worldline. The \textit{renormalised energy density} along the worldline $\gamma$ is the observable defined by
\begin{equation}
\langle \rho \rangle^\text{ren}[\chi  \dot{\gamma} \otimes  \dot{\gamma}] \doteq \int \! d \tau' \, \chi(\tau') \dot{\gamma}^a(\tau') \otimes \dot{\gamma}^b(\tau') \langle \mathcal{T}_{ab}(\tau') \rangle^\text{ren},
\label{4:rhodef}
\end{equation}
where $\chi$ is a test function of compact support and $\langle \mathcal{T}_{ab} \rangle^\text{ren}$ is the renormalised stress-energy tensor of the matter in the spacetime.
\end{defn}

Typically, $\chi$ is considered to be smooth, but we shall deal with both sharp and smooth test functions. If an observer is equipped with an apparatus that measures the stress tensor along their worldline, the switching profile controls the switch-on and switch-off of such an apparatus. This can be considered a toy model of a detector that couples to the stress-energy tensor of the matter field.

We are interested in computing $\langle \rho \rangle^\text{ren}$ in the massless scalar $1+1$ HHI state along a geodesically infalling worldline that crosses the Cauchy horizon across $\mathcal{C}^\text{FL}$, given by the integral curve of eq. \eqref{4:RNeqmotion} and \eqref{4:RNeqmotion2}, in regions I and II respectively. We shall show that, along this trajectory, there is no significant difference between the estimates that we obtain for sharp or smooth test functions. In both cases, the divergence will be inverse linear in proper time as the observer approaches $\mathcal{C}^\text{FL}$, with the coefficients differing by a constant factor.

We remind ourselves, from Section \ref{sec:4Quantisation}, that the region covered by the coordinates $(U,V)$ in the generalised Reissner-Nordstr\"om is a spacetime $(M_\text{H},g_\text{H})$, which is conformal to Minkowski spacetime, with a metric given by $g_\text{H} = -\Omega^2(U,V) dU dV$ with
\begin{equation}
\Omega^2(U,V) =  -F(r)\left(\kappa_+^2 U V\right)^{-1}
\end{equation}
and the HHI state is a conformal vacuum with respect to the Minkowski vacuum. Then, the stress-energy tensor can be obtained using the conformal techniques in $1+1$ dimensions, developed by Davies, Fulling and Unruh \cite{Davies:1976ei} and Wald \cite{Wald:1978ce}. The stress-energy tensor reads
\begin{equation}
\langle T_{ab}(\mathsf{x}) \rangle^{\text{ren}}_{\text{H}} = \Xi_{ab}(\mathsf{x})-\frac{\text{Ric}(\mathsf{x})}{48 \pi} g_{ab}, 
\label{4:Tren}
\end{equation}
where in the null Kruskal-Szekeres coordinates, $(U,V)$, we have that
\begin{subequations}
\begin{align}
\Xi_{UU} &= -(1/12\pi) \Omega \partial_U^2 \Omega^{-1}, \\
\Xi_{VV} &= -(1/12\pi) \Omega \partial_V^2 \Omega^{-1}, \\
\Xi_{UV} &=  \Xi_{VU} = 0,
\end{align}
\label{4:RNXi}
\end{subequations}
\\
where $\Ric$ is the Ricci scalar. 

\begin{rem}
The generalised $1+1$ Reissner-Nordstr\"om black hole is not Ricci-flat. The Ricci scalar, given in terms of the $r$-coordinate by $\Ric(r) = -F''(r)$, is non-zero, but it is regular across the Cauchy horizon.
\end{rem}

Using eq. \eqref{4:RNXi}, in the vicinity of the Cauchy horizon, the energy density along the inertial worldline takes the form
\begin{equation}
\langle \rho \rangle^\text{ren}[\chi  \dot{\gamma} \otimes  \dot{\gamma}] = \int \! d \tau' \, \chi(\tau') \left( \dot{U}(\tau')^2 \Xi_{UU}(\tau') + \dot{V}(\tau')^2 \Xi_{VV}(\tau') \right) + O(1).
\label{4:RNWorldlineTab}
\end{equation}

We consider the sharp function $\chi_0(\tau') = \Theta(\tau-\tau')\Theta(\tau'-\tau_0)$ with compact support, $\text{supp}(\chi_0) = [\tau_0,\tau]$, and let $\tau \rightarrow \tau_h$ with $\tau_h$ the Cauchy-horizon-crossing proper time along the inertial worldline. We show in Appendix \ref{app:BH} that the energy density along the worldline is
\begin{equation}
\langle \rho \rangle^\text{ren}[\chi_0  \dot{\gamma} \otimes  \dot{\gamma}] = -\frac{(1+ \kappa_+ / \kappa_-)^2}{12 \pi }\left[1 + O\left((\tau_h-\tau) \ln \left(1 - \frac{\tau}{\tau_h} \right)\right)\right]\frac{1}{\tau - \tau_h},
\label{4:RNSharpEnergy}
\end{equation}
which is inverse linear in the small quantity $(\tau_h - \tau)$, positive and independent of the initial details of the trajectory to leading order.

We now consider the smooth test function introduced in Chapter~\ref{ch:DeCo}, defined by eq. \eqref{chi}, as
\begin{equation}
\chi(\tau') = h_1\left(\frac{(\tau+\delta)-\tau'}{\delta} \right) h_2 \left(\frac{\tau'-(\tau_0 - \delta)}{\delta} \right).
\end{equation}

In this case, we take the switching time to be the finite time interval  $\delta \doteq \tau_h - \tau$. As before, $h_1$ and $h_2$ are non-negative smooth functions that satisfy $h_1(x) = h_2(x) = 0$ for $x \leq 0$ and $h_1(x) = h_2(x) = 1$ for $x \geq 1$. The support of the test function is $[\tau_0 - \delta, \tau_h]$. The switching time $\delta$ becomes short (but finite), as the switch-off time $\tau$ approaches the horizon-crossing time $\tau_h$. 

It is convenient to define the $C_0^\infty(\mathbb{R})$ test functions
\begin{subequations}
\begin{align}
\chi_{\text{off}}(\tau') & \doteq h_1\left(\frac{\tau_h - \tau'}{\tau_h-\tau}\right)\Theta(\tau'-\tau), \\
\chi_{\text{on}}(\tau') & \doteq h_2\left(\frac{\tau'-\tau_0-\tau_h+\tau}{\tau_h-\tau}\right)\Theta(\tau_0 -\tau'),
\end{align}
\end{subequations}
\\
and write $\chi = \chi_{\text{on}} + \chi_0 + \chi_{\text{off}}$. The energy density along the worldline is then
\begin{equation}
\langle \rho \rangle^\text{ren}[\chi  \dot{\gamma} \otimes  \dot{\gamma}] = \langle \rho \rangle^\text{ren}[\chi_{\text{on}} \, \dot{\gamma} \otimes  \dot{\gamma}] + \langle \rho \rangle^\text{ren}[\chi_0  \dot{\gamma} \otimes  \dot{\gamma}] + \langle \rho \rangle^\text{ren}[\chi_{\text{off}} \, \dot{\gamma} \otimes  \dot{\gamma}].
\end{equation}

The contribution of $\langle \rho \rangle^\text{ren}[\chi_{\text{on}} \, \dot{\gamma} \otimes  \dot{\gamma}] = O(1)$ and the contribution of $\langle \rho \rangle^\text{ren}[\chi_0 \, \dot{\gamma} \otimes  \dot{\gamma}]$ is given by eq. \eqref{4:RNSharpEnergy}. The contribution of $\langle \rho \rangle^\text{ren}[\chi_{\text{off}} \, \dot{\gamma} \otimes  \dot{\gamma}]$ is calculated in Appendix \ref{app:BH} and we obtain that
\begin{equation}
\langle \rho \rangle^\text{ren}[\chi  \dot{\gamma} \otimes  \dot{\gamma}] = -\frac{(1 + \kappa_+ / \kappa_-)^2}{12 \pi } \left[\int_0^1 \! dx \, \frac{h'_1(x)}{x} + O\left((\tau_h-\tau) \ln \left(1 - \frac{\tau}{\tau_h} \right)\right)\right]\frac{1}{\tau_h - \tau}.
\label{4:RNSmoothEnergy}
\end{equation}

As claimed, the inverse linear divergent behaviour is preserved. The leading coefficient depends on the details of the switching and, in particular, a rapid switching $h_1'(x)/x \gg 1$ will give rise to a large constant coefficient multiplying the leading divergence.

\section{Spacetime dimension and the response function}
\label{sec:4Rindler}

In Section \ref{sec:4QuantumEffects} we have studied several quantum effects that occur due to the curvature of black hole spacetimes, including the stimulation of particle detectors by black hole radiation in time-dependent situations, as well as the divergence of the rate of inertial detectors as they approach the Schwarzschild singularity \eqref{4:NearSing}, and the divergence rate of detectors and of the renormalised local energy density of inertial observers approaching the future Cauchy horizon in Reissner-Nordstr\"om-type spacetimes \eqref{4:RNDetectorRate} in the HHI and Unruh states.

We now discuss how the estimates that we have obtained in Section \ref{sec:4QuantumEffects} depend on the spacetime dimension: The logarithmic singularity structure of the twice-differentiated $1+1$ Hadamard two-point function, given by the Hadamard parametrix \eqref{2:Hparametrix}, diverges as $1/\left(\sigma_\epsilon(\mathsf{x},\mathsf{x}') \right)^{2}$. In other words, the singularity structure is of the same order of that of the $3+1$ Wightman function. This means that the analytic results that we obtain in $1+1$ dimensions does indeed contain \textit{qualitative} information about the full $3+1$ behaviour, and can be used as a reference for what to expect in more complicated situations. Yet, concluding that our analysis contains \textit{quantitative} information about the $3+1$ is erroneous, as we shall now demonstrate.

Consider a derivative-coupling detector following the static trajectory, given in Minkowskian coordinates by $(t,x) = (\tau, \tau_h)$, as it interacts with a field in the Rindler state in the $1+1$ dimensional right Rindler wedge, considered as a spacetime on its own right $(M^{1+1}_\text{R},g^{1+1})$, with metric $g = - \exp(2 a \xi) du_\text{R} dv_\text{R}$, where the Rindler null coordinates $(u_R,v_R) \doteq (\eta - \xi,\eta + \xi)$ are related to the Minkowskian coordinates by
\begin{subequations}
\begin{align}
\eta & = a^{-1} \text{arctanh}(t/x), \\
\xi & = a^{-1} \ln \left[ a(x^2-t^2)^{1/2} \right].
\end{align}
\end{subequations}

Here, $a$ is an arbitrary acceleration parameter.

Suppose the detector is switched on at time $\tau_0$. The trajectory of the switched-on detector will cross the future horizon of the Rindler wedge at proper time $\tau_h$, and the transition rate of such detector can be computed analytically close to the horizon, at proper time $\tau$, to leading order in the small quantity $(\tau_h - \tau)$. See Appendix \ref{app:BH}. The result is
\begin{equation}
\dot{\mathcal{F}}_{1+1}(E, \tau, \tau_0) = -\frac{1}{4 \pi} \left[1 + O\left(\left[\ln\left(1-\frac{\tau}{\tau_h}\right)\right]^{-1}\right) \right] \frac{1}{\tau_h-\tau}.
\label{4:1+1}
\end{equation}

A similar calculation can be performed in the $3+1$ (right) Rindler spacetime, $(M_\text{R},g)$, with a metric given locally by $g = - \exp(2 a \xi) du_\text{R} dv_\text{R} + dy^2 + dz^2$. A detector following the worldline defined by $(t,x,y,z) = (\tau,\tau_h,0,0)$ in Minkowskian coordinates, which has been switched on at proper time $\tau_0$ inside the Rindler wedge, and interacting with a field in the Rindler state, will cross the future Rindler horizon at time $\tau_h$. The transition rate close to the horizon is \cite{Louko:2007mu}
\begin{equation}
\dot{\mathcal{F}}_{3+1}(E, \tau, \tau_0) = \frac{1}{2 \pi^2 \tau_h } \left[\frac{1}{4}\ln\left(1 - \frac{\tau}{\tau_h} \right) + \frac{1}{2} \ln \left[-\ln\left(1-\frac{\tau}{\tau_h}\right) +  O(1)\right]	 \right].
\label{4:3+1}
\end{equation}

In view of equations \eqref{4:1+1} and \eqref{4:3+1} it is clear that the $1+1$ detector model that we are considering diverges more rapidly than the $3+1$ Unruh-DeWitt detector at the level of the transition rate. Moreover, at the level of the transition probability
\begin{subequations}
\begin{align}
\mathcal{F}_{1+1}(E, \tau, \tau_0)&= \frac{1}{4 \pi} \ln \left(1 - \frac{\tau}{\tau_h} \right) + O \left(\ln \left[-\ln\left(1-\frac{\tau}{\tau_h}\right)\right]\right), \\
\mathcal{F}_{3+1}(E, \tau, \tau_0) &= -\frac{1 - \tau/\tau_h}{2 \pi^2 \tau_h} \ln\left(1 - \frac{\tau}{\tau_h}\right) + O(\tau_h-\tau),
\end{align}
\end{subequations}
\\
$\mathcal{F}_{1+1}$ diverges logarithmically, while $\mathcal{F}_{3+1}$ vanishes as the future horizon is approached.

There exists the possibility that this behaviour may be an artefact  of the sharp switching of the detector that the computation of the transition probability smeared with a smooth switching function will resolve this awkwardness. A computation is, however, lacking, and this is future work that deserves attention. 

If the mismatch between the responses in $1+1$ and $3+1$ dimensions in the Rindler space persists with a smooth switching function,  one can formulate the following conjecture: that the behaviour of the instantaneous transition rate, when inertially approaching a Cauchy horizon in a non-extremal $(3+1)$-dimensional Reissner-Nordstr\"om, must diverge \textit{at most} inverse polynomial and, possibly, logarithmic in the proper time of the detector.

An obstruction to this conjecture may come from the presence of highly energetic modes due to the potential barrier on the radial equation in $3+1$-dimensional black holes, which is not present in $1+1$ dimensional models, which may affect the response of a detector in $3+1$ dimensions, but this effect cannot be registered in our $1+1$-dimensional analysis.

If the divergence in $3+1$ dimensions is indeed logarithmic or weaker, then the detector can cross the horizon with a finite transition probability. Nevertheless, our $1+1$ computation of the renormalised energy density along the detector worldline suggests that the region near the horizon is highly energetic. This is consistent with the transition rate computed in $1+1$ dimensions.

\chapter{Asymptotic time-scales in the Unruh effect}
\label{ch:Unruh}

A property that is shared by both the Unruh\footnote{Rigorous proofs of the existence of the Unruh effect in this setting are \cite{Bisognano:1976za, DeBievre:2006px}.} and Hawking phenomena is that thermal phenomena can be detected only in stationary situations and after an infinite interaction time has elapsed. In the first case, the detailed balance condition at the Unruh temperature, $T_\text{U}$, is satisfied along the stationary orbit of a uniformly accelerated detector that interacts with the field for an infinite amount of proper time, as in our discussion in Chapter \ref{ch:DeCo}. See eq. \eqref{DeCoUnruhPlanckian}. In the second case, the black hole radiation is measured by a detector at asymptotically late times along a stationary trajectory in the exterior region of the black hole, as in our discussion in Chapter \ref{ch:BH}. See eq. \eqref{4:InftyUnruh}. Nevertheless, as we have discussed in Chapter \ref{ch:chap2}, from an operational point of view, a detector should interact with a quantum field through a smooth switching function of compact support, $\chi$, for a finite amount of proper time. Indeed, it is the long term interactions that have lead us to work at the level of the transition rate, and not the response, in the models that we examined in Chapters \ref{ch:DeCo} and \ref{ch:BH}.

In this chapter, we work at the level of the response function with a switching function of compact support. We wish to understand the emergence of thermality, in the sense of the detailed balance condition \eqref{2:DBC}, whenever the interaction between a detector and a field only occurs for a finite amount of time. We address this question in the context of the Unruh effect in $3+1$ dimensions. As a prototype, we consider the idealised, point-like Unruh-DeWitt detector \cite{DeWitt:1979}, introduced in Chapter~\ref{ch:chap2}, following a Rindler orbit, which is weakly and linearly coupled to a massless scalar field through a smooth switching function of the proper time with compact support, $\chi \in C_0^\infty(\mathbb{R}, \mathbb{R})$, along the detector worldline.

The first objective of this chapter, which we accomplish in Section \ref{sec5:KMS}, is to understand rigorously the equivalence of the detailed balance of the response function and the KMS condition, when the interaction time between a field and a detector becomes long. To this end, we shall introduce a time scale in the problem, $\lambda$, that will rescale the interaction proper time, in such a way that the switching function that controls the interaction, as well as the response function, scale together with $\lambda$. The equivalence between the response detailed balance and the KMS condition is then obtained in the limit $\lambda \to \infty$. We then prove that this limit cannot be uniform in the detector transition energy gap $E$, under modest, physically motivated assumptions on the (Fourier transform of the) two-point function. In particular, as $E \to \infty$, the detailed balance condition cannot hold. This poses the question, how long does one need to wait to detect the KMS temperature in terms of the detailed balance condition up to the energy scale $E$?


This leads to the second objective, which we treat in Sections \ref{sec:slowswitching} and \ref{sec:plateau}, and which is to answer this question for the case of the Unruh-DeWitt detector following a Rindler trajectory. After all, the conventional wisdom states that such a detector will respond at the Unruh temperature after it has interacted with the field for an infinite amount of time. 

The message that we wish to convey is that the thermalisation time scale depends crucially on the details of the switching function of the detector. To illustrate this point, we deal with two situations: First, in Section \ref{sec:slowswitching}, we treat an adiabatic rescaling of the switching function, which represents a slow and long switching. The main result of this section is that there exist a class of switching functions, with sufficiently fast Fourier decay, for which the detector thermalises at the Unruh temperature in a time scale that is polynomial in the large energy gap $E$, \textit{i.e.}, for which the response is \textit{polynomially asymptotically thermal}. Second, in Section \ref{sec:plateau}, we consider a switching function that switches on for a fixed time, then interacts constantly for a long time, and then switches off. The main result of this section is that, if one only rescales the constant interaction time to infinity, at large energy gap, $E \to \infty$, the detector cannot thermalise in any time that is polynomial in $E$. Thus, the thermalisation time scale \textit{is} in the details of the switching.

\section{The KMS and the detailed balance conditions}
\label{sec5:KMS}

We motivated the introduction of the detailed balance condition in Chapter \ref{ch:chap2} as a condition that, if satisfied, allows a localised observer to measure the temperature of a KMS state with the aid of particle detectors. Nevertheless, we postponed a proof of the equivalence of the KMS condition and the detailed balance of the detector response, and relied on a heuristic argument to see that they are equivalent at a formal level. In this section, we prove that, indeed, under a set of technical conditions, the detailed balance of the detector response at the KMS temperature is equivalent to the KMS condition for the Wightman function.

We consider an stationary Unruh-DeWitt detector, coupled to a quantum field through the interaction Hamiltonian \eqref{Hint}. The interaction occurs for a finite amount of time, and the smooth switching function of compact support, $\chi$, with $\supp(\chi) = [-\tau_i, \tau_f]$, controls the interaction, such that the interaction is switched on at detector proper time $\tau_i$ and switched off at time $\tau_f$. The response function is given by eq. \eqref{ResponseFn}, which can be written along a stationary trajectory as
\begin{equation}
\mathcal{F}(E) = \frac{1}{2 \pi} \int_{-\infty}^\infty \! d \omega' \, \hat{\chi}(\omega') \hat{\chi}(-\omega') \int_{-\infty}^\infty \! ds \, \ee^{-\ii(\omega' + E)s} \mathcal{W}(s),
\label{Fstationary}
\end{equation}
where we  have written the (real) Fourier transform of the switching function as $\hat{\chi}(\omega) \doteq \mathcal{F}[\chi](\omega) = \int d \tau \, \ee^{-\ii \omega \tau} \chi(\tau)$. $\hat{\chi}$ exists by virtue of the compact support of $\chi$. The equation above is to be understood in a distributional sense, whenever the state is Hadamard, by an $\ii \epsilon$ prescription, as we have discussed above in this thesis. 

The first step that we take towards showing the equivalence between the detailed balance of the response and the KMS condition and the two-point function comes from our heuristic discussion in Chapter \ref{ch:chap2}. Namely, we look at the Fourier transform of the Wightman function.

\begin{prop}
Let $\W$ be a stationary Wightman function that has the analyticity properties of Definition~\ref{2:defKMS} and, thus, defines a KMS state. Suppose that the falloff of $\W(s)$ at large $|s|$ is so strong that the integral formula
\begin{align}
\widehat{\W}(\omega) \doteq
\int_{-\infty}^\infty \! d s \, \ee^{-\ii \omega s} \, \W(s), 
\label{5:What-def}
\end{align}
with the distributional singularities understood in the $s \to s - \ii\epsilon$ sense, defines the Fourier transform of $\W$ pointwise at each $\omega\in\mathbb{R}$. Finally, suppose that the integration contour in \eqref{5:What-def} can be deformed to any line of constant $\text{Im}(s) \in [-\beta,0]$, with the distributional singularities at $\text{Im}(s) =-\beta$ treated in the $s \to s + \ii\epsilon$ sense. Then $\widehat{\W}$ satisfies the detailed balance condition.
\label{5:bal-KMS-twoways}
\end{prop}
\begin{proof}
\begin{align}
\widehat{\W}(-\omega) - \ee^{\beta \omega} \, \widehat{\W}(\omega) & = \int_{-\infty}^\infty \! ds \, \ee^{\ii \omega s} \, \W(s - \ii\epsilon) - \int_{-\infty}^\infty \! dr \, \ee^{-\ii \omega (r+ \ii\beta)} \, \W(r-\ii\epsilon) \notag \\
& = \int_{-\infty}^\infty \! ds \, \ee^{\ii \omega s} \left[ \W(s - \ii\epsilon) - \W(-s-\ii\beta + \ii\epsilon) \right], 
\end{align}
where in the last equality we have deformed the contour of the $r$-integral to $r = - \ii\beta + 2\ii\epsilon- s$ with $s\in\mathbb{R}$, and the $\ii\epsilon$ prescribes the distributional singular behaviour.
\end{proof}

The next step is to scale the interaction time between the field and the detector by a real positive parameter $\lambda$, such that, as $\lambda \to \infty$, the interaction time becomes long. The way in which one introduces this scaling is not unique. We present two important choices:

The first choice is the \textit{adiabatic scaling}, whereby $\chi(\tau) \to \chi_\lambda(\tau) = \chi(\tau/\lambda)$, which represents a long and slow switching as $\lambda \rightarrow \infty$. Under this scaling, the adiabatically rescaled response function \eqref{Fstationary}, $\mathcal{F}(E) \to \widetilde{\mathcal{F}}_\lambda(E)$, is given by
\begin{equation}
\widetilde{\mathcal{F}}_\lambda \doteq \mathcal{F}_\lambda(E)/\lambda = \frac{1}{2 \pi} \int_{-\infty}^{\infty} \! d \omega \, \left|\hat{\chi}(\omega)\right|^2 \, \widehat{\W}(E+ \omega/\lambda),
\label{Flambda}
\end{equation}
where we have used the fact that $\chi$ is real-valued. 

The second choice is the \textit{constant interaction scaling}. In this case, we consider the switching functions which consist of a switch-on profile of duration $\Delta \tau_{\text{on}}$, a constant interaction of duration $\Delta \tau$ and a switch-off profile of duration $\Delta \tau_{\text{off}}$. The specifics of this class of switching functions allows us to treat the switching time and the total interaction times independently. The constant interaction time scaling sends the constant time $\Delta \tau \to \lambda \Delta\tau$. For simplicity, we consider $\Delta \tau_{\text{on}} = \Delta \tau_{\text{off}} = \Delta \tau_s$, and we consider switching functions $\chi$ that are constructed by integrating bump functions $\psi \in C_0^\infty(\mathbb{R}, \mathbb{R})$ with support $\text{supp}(\psi) = [0, \Delta \tau_s]$ in the following way
\begin{equation}
\chi(\tau) \doteq \int_{-\infty}^{\tau} d \tau' \left[ \psi(\tau') - \psi(\tau' - (\Delta \tau_s + \Delta \tau)) \right].
\label{chiconst}
\end{equation}

The switching function defined by \eqref{chiconst} has $\text{supp}(\chi) = [0, \Delta \tau + 2 \Delta \tau_s]$ and a switch-on profile of duration $\Delta \tau_s$, a constant interaction time of duration $\Delta \tau$ and a switch-off profile of duration $\Delta \tau_s$. It is clear that different bump functions, $\psi$, yield switching functions, $\chi$, with different switching profiles. The scaling $\Delta \tau \mapsto \lambda \Delta \tau$ produces the constant-interaction rescaled response function $\breve{\mathcal{F}}_\lambda \doteq \mathcal{F}_\lambda(E)/\lambda$ given by (see Lemma \ref{LemChi})
\begin{equation}
\breve{\mathcal{F}}_\lambda \doteq \mathcal{F}_\lambda(E)/\lambda = \frac{1}{2 \pi} \int_{-\infty}^\infty \! d \omega \,\,  \frac{2}{\lambda} \left[1-\cos\left(\omega\left(\lambda \Delta \tau + \Delta \tau_s \right)\right)\right] \left| \frac{\hat{\psi}(\omega)}{\omega} \right|^2 \widehat{\W}(E+ \omega).
\label{Flambda2}
\end{equation}

Let us collect our discussion in the following definition:

\begin{defn}
We call the response $\widetilde{\mathcal{F}}_\lambda$, defined by \eqref{Flambda}, the \textit{adiabatically rescaled response function}. We call the response $\breve{\mathcal{F}}_\lambda$, defined by \eqref{Flambda2}, the \textit{constant-interaction rescaled response function}.
\label{Def:Flambda}
\end{defn}

We are ready to formulate the equivalence between the KMS condition and the detailed balance of the response. Informally, we say that the detailed balance of the response at the KMS temperature is equivalent to the KMS condition if and only if, given a scaling $\lambda$, as discussed above, $\lim_{\lambda \to \infty} \mathcal{F}_\lambda/\lambda$ satisfies the detailed balance condition. Let us state this precisely in the form of two theorems:

\begin{thm}
Suppose $\W$ and $\widehat{\W}$ satisfy the technical assumptions of Proposition \ref{5:bal-KMS-twoways}, \textit{i.e.} $\W$ and $\widehat{\W}$ satisfy the KMS and detailed balance conditions respectively. Further, suppose that $|\widehat{\W}(\omega)|$ is bounded by a polynomial function of $\omega$. Then the function defined by the long time limit of \eqref{Flambda}, $\widetilde{\mathcal{F}}_\infty(E) \doteq \lim_{\lambda\to\infty} \widetilde{\mathcal{F}}_\lambda(E)$ is defined for each $E\in\mathbb{R}$. Further, 
$\widetilde{\mathcal{F}}_\infty$ satisfies satisfies the detailed balance condition.  
\label{5:KMSadiabatic}  
\end{thm}
\begin{proof}
By the polynomial bound on $|\widehat{\W}(\omega)|$, there exist constants $A >0$, $B>0$ and $n \in \mathbb{Z}_+$ such that for $\lambda\ge1$ we have $|\widehat{\W}(E + \omega/\lambda)| < A  + B (E + \omega/\lambda)^{2n} \le A + B (|E| + |\omega|/\lambda)^{2n} \le A + B (|E| + |\omega|)^{2n}$. As $\hat\chi(\omega)$ decays at $\omega \to\pm \infty$ faster than any inverse power of~$\omega$, the integrand in \eqref{Flambda} is hence bounded in absolute value for $\lambda\ge1$ by a $\lambda$-independent integrable function, and the $\lambda\to\infty$ limit in \eqref{Flambda} may be taken under the integral by dominated convergence. Then, for each $E\in\mathbb{R}$, 
\begin{align}
\widetilde{\mathcal{F}}_\lambda(E) \xrightarrow[\lambda\to\infty]{}
\left( \frac{1}{2 \pi} \int_{-\infty}^{\infty} \! d \omega \, 
\left|\hat{\chi}(\omega)\right|^2 \right) 
\widehat{\W}(E).
\end{align}

The term inside the parentheses is a constant and Proposition \ref{5:bal-KMS-twoways} completes the proof.
\end{proof}

\begin{thm}
\label{5:KMSconstant}
Suppose $\W$ and $\widehat{\W}$ satisfy the technical assumptions of Proposition \ref{5:bal-KMS-twoways}, \textit{i.e.} $\W$ and $\widehat{\W}$ satisfy the KMS and detailed balance conditions respectively. Further, suppose that $|\widehat{\W}(\omega)|$ is bounded by a polynomial function of $\omega$. Then the function defined by the long time limit of \eqref{Flambda2}, $\breve{\mathcal{F}}_\infty(E) \doteq \lim_{\lambda\to\infty} \breve{\mathcal{F}}_\lambda(E)$ is defined for each $E\in\mathbb{R}$. Further, 
$\breve{\mathcal{F}}_\infty$ satisfies satisfies the detailed balance condition. 
\end{thm}
\begin{proof}
Let $g: \mathbb{R} \rightarrow \mathbb{R}$ be defined by
\begin{equation}
g(\omega) \doteq \left| \hat{\psi}(\omega) \right|^2 \widehat{\W}(\omega + E).
\end{equation}

It follows from the decay properties of $\psi \in C_0^\infty(\mathbb{R}, \mathbb{R})$ that $g$ has rapid decay. Also, $g$ is bounded, hence we have $0 < C \doteq \sup\{|g(u)|: u \in \mathbb{R}\} < \infty$. Eq. \eqref{Flambda2} can be written as
\begin{equation}
\breve{\mathcal{F}}_\lambda(E) = \frac{1}{2 \pi} \int_{-\infty}^\infty \! d u \, 2(\Delta \tau + \Delta \tau_s/\lambda) \left(\frac{1-\cos\left(u \right)}{u^2}\right) g\left( \frac{u}{\lambda \Delta \tau + \Delta \tau_s} \right).
\end{equation}

For $\lambda \geq 1$, the integrand of $\mathcal{F}_\lambda(E)/\lambda$ can be bounded uniformly in $\lambda$ by the integrable function $2(\Delta \tau + \Delta \tau_s) C (1-\cos(u))/u^2$. 
An application of the dominated convergence theorem yields
\begin{equation}
\breve{\mathcal{F}}_\infty(E) = \left(\Delta \tau \left|\hat{\psi}(0) \right|^2\right) \, \widehat{\W}(E).
\end{equation}

The term inside the parentheses is a constant and Proposition \ref{5:bal-KMS-twoways} completes the proof.
\end{proof}

We collect the results that we have proven in the following definition:

\begin{defn}
If $\widetilde{\mathcal{F}}$ is the adiabatically rescaled response function of $\mathcal{F}$ and $\widetilde{\mathcal{F}}$ satisfies Theorem \ref{5:KMSadiabatic}, we say that $\mathcal{F}$ is \textit{adiabatically asymptotically thermal}. If $\breve{\mathcal{F}}$ is the constant-interaction rescaled response function of $\mathcal{F}$ and $\breve{\mathcal{F}}$ satisfies Theorem \ref{5:KMSconstant}, we say that $\mathcal{F}$ is \textit{constantly asymptotically thermal}.
\end{defn}

\subsection{Non-uniformity of the asymptotic thermal limit}

The next thing that we want to show is that the way in which the asymptotic thermality of the response is achieved cannot be uniform in $E$.

The logarithmic form of the right hand side of the detailed balance condition,
\begin{align}
\beta = \frac{1}{\omega} \ln \! \left( \frac{G(-\omega)}{G(\omega)}\right),
\label{5:DBC}
\end{align}
shows that if $G$ is positive and $G(\omega)$ is bounded above by a polynomial in~$\omega$, $G(\omega)$ can satisfy the detailed balance condition at $|\omega|\to\infty$ only if $G(\omega)$ bounded by an exponentially decreasing function of~$\omega$ at $\omega\to\infty$. The following proposition shows $\mathcal{F}$ \eqref{Fstationary} cannot decrease this fast under modest technical conditions on~$\widehat{\W}$, and hence neither can $\widetilde{\mathcal{F}}_\lambda$ ~\eqref{Flambda} or $\breve{\mathcal{F}}_\lambda$~\eqref{Flambda2}. 

\begin{prop}
Suppose $\widehat{\W}$ is positive, $\widehat{\W}(\omega)$ is bounded from above by a polynomial function of\/~$\omega$, and there exist positive constants $b$ and $C$ such that $\widehat{\W}(\omega) \ge C$ for $-b\le\omega\le b$. Then $\mathcal{F}(E)$ \eqref{Fstationary} cannot fall off exponentially fast as $E\to\infty$. 
\label{prop:key-falloff}
\end{prop}
\begin{proof}
We use proof by contradiction, 
showing that an exponential falloff of $\mathcal{F}(E)$ at $E\to\infty$ 
implies that $\chi$ is identically vanishing. 

By the bound on $\widehat{\W}$ from above, $\mathcal{F}(E)$ 
is well defined for all~$E$. By the bounds on $\widehat{\W}(\omega)$ from below, we have 
\begin{align}
\mathcal{F}(E) &= \frac{1}{2 \pi} \int_{-\infty}^\infty \! d v \, 
\left|\hat{\chi}(v-E) \right|^2 \, \widehat{\W}(v)
\notag
\\
& \ge
\frac{1}{2 \pi} \int_{-b}^b \! d v \, 
\left|\hat{\chi}(v-E)\right|^2 \, \widehat{\W}(v) 
\notag
\\
& \ge 
\frac{C}{2 \pi} \int_{-b}^b \! d v \, 
\left|\hat{\chi}(v-E)\right|^2
\ , 
\label{eq:est1}
\end{align}
where we have first changed the integration variable by 
$\omega = v -E$, then restricted the integration range to $-b \le v \le b$, 
and finally used the strictly positive bound below on~$\widehat{\W}(v)$. 

Suppose now that there are positive constants $A$ and $\gamma$ such that 
$\mathcal{F}(E) < A \, \ee^{-\gamma E}$. From \eqref{eq:est1} we then have 
\begin{align}
\int_{-b}^b \! d v \, 
\left|\hat{\chi}(v-E)\right|^2 < A' \, \ee^{-\gamma E}
\ , 
\label{eq:est2}
\end{align}
where $A' = 2 \pi A C^{-1} >0$. 
Inserting the factor $\ee^{\gamma(E-v)/2}$
under the integral in \eqref{eq:est2} gives 
\begin{align}
\int_{-b}^b \! d v \, 
\ee^{\gamma(E-v)/2}
\left|\hat{\chi}(v-E)\right|^2 < A'  \, \ee^{\gamma b/2} \, \ee^{-\gamma E/2}
\ . 
\label{eq:est3}
\end{align}
Setting in \eqref{eq:est3}
$E = (2k+1)b$, $k = 0,1,\ldots$, and summing over~$k$, 
we obtain 
\begin{align}
\int_{-\infty}^0 \! d \omega \, 
\ee^{-\gamma \omega/2}
\left|\hat{\chi}(\omega)\right|^2 < \infty 
\ , 
\label{eq:est4}
\end{align}
and since $\left|\hat{\chi}(\omega)\right|$ is even in~$\omega$, it follows that 
\begin{align}
\int_{-\infty}^\infty \! d \omega \, 
\ee^{\gamma |\omega|/2}
\left|\hat{\chi}(\omega)\right|^2 < \infty 
\ . 
\label{eq:est5}
\end{align}
This implies that $\ee^{\gamma|\omega|/4}\hat{\chi}(\omega)$ is in $L^2(\mathbb{R}, d\omega)$. 

Let now 
\begin{subequations}
\label{eq:est6}
\begin{align}
\Xi(\omega) & \doteq \overline{\hat{\chi}(\omega)}\left(1+ \ee^{\gamma|\omega|/4}\right)
\ , 
\\
\Phi_z(\omega) & \doteq \frac{\ee^{\ii z \omega}}{1+ \ee^{\gamma|\omega|/4}}
\ , 
\end{align}
\end{subequations}
\\
where $- \gamma/4 < \text{Im}(z) < \gamma/4$.
Both $\Xi$ and $\Phi_z$ are in $L^2(\mathbb{R}, d\omega)$. We may hence 
define in the strip $- \gamma/4 < \text{Im}(z) < \gamma/4$ the function 
\begin{align}
\tilde\chi(z) \doteq (2\pi)^{-1} (\Xi, \Phi_z)
\ , 
\label{eq:est7}
\end{align}
where $(\cdot , \cdot)$ denotes the inner product in $L^2(\mathbb{R}, d\omega)$. 
A~straightforward computation shows that $\tilde\chi(z)$ is holomorphic in the strip. 

Formulas \eqref{eq:est6} and \eqref{eq:est7} show that 
$\tilde\chi(\tau) = \chi(\tau)$ for $\tau \in \mathbb{R}$. 
Since $\chi$ is by assumption compactly supported on~$\mathbb{R}$, 
the holomorphicity of $\tilde\chi$ implies that $\tilde\chi$ vanishes everywhere. 
As $\chi$ is by assumption nonvanishing somewhere, this provides the sought contradiction. 
\end{proof}

The content of Proposition \ref{prop:key-falloff} is that $\mathcal{F}$, $\widetilde{\mathcal{F}}$ and $\breve{\mathcal{F}}$ cannot satisfy the detailed balance condition at high energies, $E \to \infty$.

\subsection{Example: The Unruh effect}

We now wish to apply the general theory, thus far developed, to the Unruh effect. We consider an observer, equipped with a particle detector, which follows a Rindler orbit in $3+1$ Minkowski spacetime, given in Minkowskian coordinates by 
\begin{equation}
\mathsf{x}(\tau) = \left(a^{-1} \sinh(a \tau), a^{-1} \cosh(a \tau), y_0, z_0 \right),
\label{5:RindlerTraj}
\end{equation}
where $a>0$ is the proper acceleration of the detector, and $y_0$ and $z_0$ are constants. The detector is coupled to a massless scalar field, initially in the Minkowski vacuum state, through the interaction Hamiltonian \eqref{Hint} for a finite amount of time, and the smooth switching function of compact support, $\chi$, with $\supp(\chi) = [-\tau_i, \tau_f]$, controls the interaction, such that the interaction is switched on at detector proper time $\tau_i$ and switched off at time $\tau_f$. 

The response of the detector can be obtained from the two-point function of the massless scalar field,
\begin{equation}
\mathcal{W}(\x, \x') = \langle 0 | \Phi(\x) \Phi(\x') | 0 \rangle = \frac{1}{2(2\pi)^2} \frac{1}{\sigma_\epsilon(\x,\x')},
\label{5:Wightman}
\end{equation}
where $\sigma_\epsilon$ is the regularised Synge world function, as in \eqref{2:Synge} and \eqref{2:regsynge}. Using eq. \eqref{5:RindlerTraj}, we find
\begin{align}
\mathcal{W}(s) = - \frac{a^2}{16 \pi^2 \sinh^2\left(a(s-\ii\epsilon)/2\right)}. 
\label{eq:rindler-wightman}
\end{align}

Using techniques as the ones that we have presented in $1+1$ dimension, in the context of the derivative coupling detector in Chapter \ref{ch:DeCo}, the Fourier transform of \eqref{eq:rindler-wightman} is
\begin{align}
\widehat{\W}(\omega) = \frac{\omega}{2\pi \, \left(\ee^{2\pi \omega /a} -1\right)}. 
\label{eq:Unruh-hatWstat}
\end{align}

The response function $\mathcal{F}$ is given by \eqref{Fstationary}. Theorems~\ref{5:KMSadiabatic} and \ref{5:KMSconstant} apply, showing respectively that $\widetilde{\mathcal{F}}_\infty(E)$ and $\breve{\mathcal{F}}_\infty(E)$ are multiples of $\widehat{\W}(E)$. Hence, the response function $\mathcal{F}$ is (both adiabatically and constantly) asymptotically thermal at the Unruh temperature, $T_\text{U} = a/(2\pi)$. This is the Unruh effect.

Nevertheless, the Unruh temperature is not achieved uniformly in $E$. To show that this is the case, we need to prove that $\mathcal{F}$ is positive and that $\mathcal{F}(E)$ is bounded above by a polynomial in~$E$. That $\mathcal{F}$ is positive follows from eq. \eqref{Fstationary} and \eqref{eq:Unruh-hatWstat}. That $\mathcal{F}(E)$ is bounded above by a polynomial follows from the following proposition\footnote{We remind that reader that $O^\infty(x)$ indicates the quantity becomes suppressed faster than any positive power of $x$ as $x \to 0$.}

\begin{prop}
Suppose that the spacetime dimension is four, that $\mathcal{W}$ is stationary and that $\mathcal{W}(s) + {(4 \pi^2 s^2)}^{-1}$ is smooth in~$s$. Then
\begin{align}
\mathcal{F}(E) = -\frac{E \Theta(-E)}{2 \pi} \int_{-\infty}^\infty \! d\tau \, [\chi(\tau)]^2 
+ O^\infty(1/E)
\label{eq:F-leadingterm}
\end{align}
as $|E| \to \infty$.
\label{prop:F-leadingterm}
\end{prop}

\begin{proof}
The idea is to use the $3+1$-dimensional version of eq. \eqref{F1+1}, which can be obtained from equation (3.16) in \cite{Louko:2007mu}. We write the response as
\begin{align}
\mathcal{F}(E) & = 
-\frac{E \Theta(-E)}{2 \pi} \int_{-\infty}^\infty \! d\tau \, [\chi(\tau)]^2 
+ \int_0^\infty \! ds \,\frac{\cos(Es)}{2 \pi^2 \,s^2} 
\int_{-\infty}^\infty \! d\tau \, \chi(\tau) [\chi(\tau) - \chi(\tau-s)] 
\nonumber 
\\
& \hspace{3ex}
+ 2 \int_{-\infty}^\infty \! d \tau \, \chi(\tau) \int_0^\infty \! ds \, \chi(\tau-s)
\text{Re} 
\left[ \ee^{-\ii E s} \left( \mathcal{W}(\tau, \tau-s) + \frac{1}{4 \pi^2 s^2} \right) \right].
\label{Louko-Satz2}
\end{align}

We use the stationarity of $\mathcal{W}$ and interchange the 
integrals in the last term. We note that 
$\mathcal{W(-s) = \overline{\mathcal{W}}(s)}$, 
the function 
$s \mapsto \int_{-\infty}^\infty \! d \tau \, \chi(\tau) \chi(\tau-s)$ 
is an even smooth function of compact support, and 
the function 
$s \mapsto s^{-2} \int_{-\infty}^\infty \! d\tau \, \chi(\tau) [\chi(\tau) - \chi(\tau-s)]$ 
is an even smooth function with falloff $O\bigl(s^{-2}\bigr)$ at $s \to \pm\infty$. 
This allows us to write 
\begin{align}
\mathcal{F}(E) & = 
-\frac{E \Theta(-E)}{2 \pi} \int_{-\infty}^\infty \! d\tau \, [\chi(\tau)]^2 
+  \int_{-\infty}^\infty \! ds \,\frac{\cos(Es)}{4 \pi^2 \, s^2} 
\int_{-\infty}^\infty \! d\tau \, \chi(\tau) [\chi(\tau) - \chi(\tau-s)] 
\nonumber 
\\
& \hspace{3ex}
+ \int_{-\infty}^\infty \! ds
\, \ee^{-\ii E s} \left( \mathcal{W}(s) + \frac{1}{4 \pi^2 s^2} \right)
\int_{-\infty}^\infty \! d \tau \, \chi(\tau) \chi(\tau-s)
\ . 
\label{Louko-Satz4}
\end{align}

In the last term in~\eqref{Louko-Satz4}, the leading coincidence limit singularity of $\mathcal{W}$ is cancelled by the additive term proportional to~$s^{-2}$ because of the Hadamard property of the state, discussed in Chap. \ref{ch:chap2}, and by assumption $\mathcal{W}$ has no other singularities. 
The result now follows by applying to the second and third term in 
\eqref{Louko-Satz4} the method of repeated integration by parts~\cite{Wong:2001}, 
integrating respectively the trigonometric factor and the exponential factor. 
\end{proof}

Thus, the Unruh effect is not achieved uniformly in $E$ for long interaction times. 

Given that the detailed balance condition fails as $E \rightarrow \infty$ because of Prop. \ref{prop:key-falloff}, is there a precise sense in which one can asymptotically approach the Unruh temperature at finite time for a given large energy gap $E$? We consider our two cases of interest, the adiabatic scaling and the constant interaction scaling. We anticipate that the answer depends crucially on the details of the switching. In the case of the adiabatic scaling, subject to certain conditions on the Fourier transform of the switching function that we describe below, we will find that once the time has scaled by a factor $\lambda$, which is polynomial in the large energy, the Unruh temperature will be asymptotically registered by the detector. In the case of the constant scaling, we show that the Unruh temperature cannot be achieved in a time scale that is polynomial in the large energy.

The strategy to see how the Unruh temperature emerges is to obtain bounds on how the detailed balance condition formula fails, as functions of the time scale and the energy, $\lambda$ and $E$, such that these bounds become suppressed as the time scale and energy become large.

\begin{defn}
Let $\mathcal{F}$ be the response of a stationary particle detector. Let $\lambda>0$ be an interaction time scale for the detector proper time, $\tau$, which rescales the response function as in Def. \ref{Def:Flambda}. The rescaled response function, $\mathcal{F}_\lambda/\lambda$, is called \textit{polynomially asymptotically thermal at temperature} $\beta^{-1}$ if, given $\lambda = \lambda(E)$ a polynomial function of the energy gap, such that $\lambda(E) \to \infty$ as $|E| \to \infty$, there exist positive functions $\mathcal{B}^{+}$ and $\mathcal{B}^{-}$ of the energy gap and the interaction time scale such that
\begin{equation}
E \beta - \mathcal{B}^{-}(E,\lambda(E)) \leq  \ln \left( \frac{\mathcal{F}_{\lambda(E)}(-E)}{\mathcal{F}_{\lambda(E)}(E)}\right) \leq E \beta + \mathcal{B}^{+}(E,\lambda(E))
\label{AsympDBC}
\end{equation}
and $\mathcal{B}^{\pm}/E \rightarrow 0$ polynomially fast as $|E| \rightarrow \infty$ and $\lambda(E) \rightarrow \infty$. 
\label{DefAsympKMS}
\end{defn}

Let us discuss the relevance of distinguishing non-polynomially from polynomially asymptotic thermality for the response function. Consider an experiment in which one sought to verify the Unruh effect up to a large energy gap scale, $E$, then the waiting time is proportional to the scale $\lambda$. If the asymptotic thermality condition were satisfied \textit{only} whenever $\lambda$ is an exponential function of $E$, then the waiting time for the detectors to thermalise at the Unruh temperature would be exponentially large and impractical for all purposes. The more interesting scenario is when the asymptotic thermality condition is guaranteed at a time scale that is polynomial in the energy gap. In other words, when the polynomially asymptotically thermal detailed balance condition is satisfied.

In the following sections, we will show that the polynomially asymptotic thermality condition is guaranteed for the adiabatic scaling, subject to conditions on the Fourier transform of the switching function, but cannot be satisfied if only the constant interaction part of a switching function is scaled.

\section{Approaching the Unruh temperature by adiabatic interaction}
\label{sec:slowswitching}

We have seen that a detector following a Rindler orbit detects the Unruh temperature as the interaction time between the detector and the field scales to infinity, at fixed energy $E$. If the scaling of the switching function is adiabatic, $\chi(\tau) \to \chi_\lambda(\tau) = \chi(\tau/\lambda)$, by Theorem \ref{5:KMSadiabatic} and eq. \eqref{eq:Unruh-hatWstat}, the asymptotic response, as $\lambda \to \infty$ at fixed $E$, is
\begin{align}
\mathcal{F}_{T_\text{U}} \doteq \widetilde{\mathcal{F}}_\infty(E) =
\left( \frac{1}{(2 \pi)^2} \int_{-\infty}^{\infty} \! d \omega \, 
\left|\hat{\chi}(\omega)\right|^2 \right) 
\frac{E}{\ee^{2\pi E /a} -1},
\label{FT}
\end{align}
which satisfies the detailed balance condition at a temperature proportional to the acceleration of the orbit, which is the Unruh temperature, $T_\text{U} = a/(2\pi)$. In this section, we denote $\widetilde{\mathcal{F}}_\infty$ by $\mathcal{F}_{T_\text{U}}$ to remind ourselves that this is the thermal limit (pointwise in $E$). We also know, by Proposition \ref{prop:key-falloff}, that this limit cannot be uniform in $E$. 

For a special class of switching functions whose Fourier transforms decay sufficiently fast, the functions $\mathcal{B}^{\pm}$ of Def. \ref{DefAsympKMS} exist, such that the asymptotic polynomially asymptotic thermality condition holds. Let us motivate the introduction of this class of switching functions. 

From this point onward we shall consider $E>0$. Lemma \ref{Lem2} gives the following estimate,
\begin{equation}
\frac{E}{T_\text{U}} - \mathcal{B}^{-}(E,\lambda) \leq \ln \left(\frac{\mathcal{F}_\lambda(-E)}{\mathcal{F}_\lambda(E)}\right) \leq \frac{E}{T_\text{U}} + \mathcal{B}^{+}(E,\lambda),
\end{equation}
where
\begin{subequations}
\begin{align}
\mathcal{B}^{-}(E,\lambda) & = \ln \left( 1 + \frac{ \left|\left| \omega \hat{\chi} \right|\right|_{L^2}}{(24 \pi a) \lambda^{2} \mathcal{F}_{T_\text{U}}(E)}\right) = \ln \left(1 + \frac{\left|\left| \omega \hat{\chi} \right|\right|_{L^2}}{\left|\left| \hat{\chi} \right|\right|_{L^2}}\frac{ \left( \ee^{2\pi E/a} -1 \right)}{(6 a E/\pi) \lambda^2} \right), \label{B-gen}\\
\mathcal{B}^{+}(E,\lambda) & = \mathcal{B}^{-}(-E,\lambda). \label{B+gen}
\end{align}
\label{Bgen}
\end{subequations}

Here $||\cdot||_{L^2}$ is the $L^2$ norm with respect to the Lebesgue measure. 
When $\lambda$ is a polynomially increasing function of $E$, $\mathcal{B}^{-}(E,\lambda) = \to 0$ as $E \to \infty$, but $\mathcal{B}^{+}(E,\lambda) \to T_\text{U}$.
This means that the bound that we have provided in \eqref{Bgen} is not sufficient to guarantee polynomial asymptotic thermality. Yet, the presence of the switching function in eq. \eqref{Bgen} indicates that $\chi$ must play a crucial role in the estimates. For this reason, let us introduce a class of switching functions for which an improved bound can be found.

\begin{defn}
Let $\psi \in C_0^\infty$. 
If the Fourier transform of $\psi$ satisfies 
$|\hat\psi(\omega)| \le C (B+|\omega|)^r \exp(-A |\omega|^q)$ 
for some positive constants $A$, $B$, $C$, $r$ and some constant $q \in (0,1)$, 
we say that $\psi$ has \emph{strong Fourier decay\/}. 
We denote by $F$ the linear subspace of $C_0^\infty$ 
with strong Fourier decay, 
and we refer to elements of $F$ as switching 
functions with strong Fourier decay. 
\label{RapidFourier}
\end{defn}

\begin{rem}
The constants $A$, $B$, $C$, $r$ and $q$ 
may differ from element to element in~$F$. 
The linear subspace property of $F$ 
follows by applying the triangle inequality to 
$|\alpha_1 \hat\psi_1 + \alpha_2 \hat\psi_2|$ 
for $\psi_1$ and $\psi_2$ in~$F$. 
\end{rem}

Now, we are ready to state the main theorem of this section, which states that there exist switching functions with strong Fourier decay such that the polynomial asymptotic thermality condition is satisfied for adiabatic switchings.

\begin{thm}
Let $\widetilde{\mathcal{F}}_\lambda$ be the response function given by eq. \eqref{Flambda} for a linearly uniformly accelerated detector, $a > 0$, with $\widehat{\W}$ defined by eq. \eqref{eq:Unruh-hatWstat}. If the switching function $\chi \in F$ has strong Fourier decay, with Fourier transform satisfying
\begin{equation}
|\hat{\chi}(\omega)| \leq C\kappa^{-1}(B+|\omega/\kappa|)^r \exp(-A |\omega/\kappa|^q),
\end{equation}
for positive constants $A$, $B$, $C$, $r$ and $0 < q < 1$ and where $\kappa$ is a dimensionful positive constant, the response \eqref{Flambda} is polynomially asymptotically thermal as $ E \rightarrow \infty$ whenever $\lambda$ is given by $\lambda(E) = \alpha(2 \pi E/a)^{1+p}$ with the real parameters $\alpha$ and $p$ satisfying $\alpha > (2 \pi \kappa/a) (2 A)^{-1/q} > 0$ and $p > q^{-1} - 1$.
\label{thm:Main1}
\end{thm}
\begin{proof}
Following Lemma \ref{Lem2}, the response function satisfies $0 \leq \mathcal{F}_{T_\text{U}}(E) \leq \widetilde{\mathcal{F}}_\lambda(E)$. The idea is to improve the result in Lemma \ref{Lem2} for positive $E$ in order to control $\mathcal{B}^-$ (see eq. \eqref{B-gen}). Let us define
\begin{equation}
\mathcal{G}(E) \doteq \frac{\exp(2 \pi E/a)}{2\pi E/a}\left(\widetilde{\mathcal{F}}_\lambda(E) - \mathcal{F}_{T_\text{U}}(E) \right),
\end{equation}
with the scale $\lambda(E) = \alpha (2\pi E/a)^{1+p}$, with the constants $\alpha$ and $p$ as above. 

The first step of the proof is to show that 
$\mathcal{G}(E)\geq 0$. To this end, we define the function $g: \mathbb{R}^2 \rightarrow \mathbb{R}: \, (u,v) \mapsto \ee^{u} f_s(u,v)$, with $f_s$ as defined in Lemma \ref{Lem1}. Symmetrising the integrals \eqref{Flambda} and \eqref{FT} it is clear that $\mathcal{G}$, given by
\begin{equation}
\mathcal{G}(E) = \frac{a}{(2 \pi)^4 (E/a)} \int_0^\infty \! d\omega \, \left| \hat{\chi}(\omega) \right|^2 g(2\pi E/a, 2 \pi \omega/(\lambda a))
\end{equation}
is a non-negative function for $E >0$.
 
The second step of the proof is to show that $G$ is polynomially bounded at large $E$. This ensures that as $E \to \infty$, the rescaled response $\widetilde{\mathcal{F}}_\lambda(E) \to \mathcal{F}_{T_\text{U}}(E)$ sufficiently fast, so as to satisfy the polynomially asymptotic thermality condition. We now address this issue. We define
\begin{subequations}
\begin{align}
\mathcal{G}_{(0,\kappa)}(E)& \doteq \frac{a}{(2 \pi)^4 (E/a)}  \int_0^\kappa \! d\omega \, \left| \hat{\chi}(\omega) \right|^2 g(2\pi E/a, 2 \pi \omega/(\lambda a)), \\
\mathcal{G}_{(\kappa, E \lambda)}(E)& \doteq \frac{a}{(2 \pi)^4 (E/a)}  \int_\kappa^{E \lambda} \! d\omega \, \left| \hat{\chi}(\omega) \right|^2 g(2\pi E/a, 2 \pi \omega/(\lambda a)), \\
\mathcal{G}_{(E \lambda,\infty)}(E)& \doteq \frac{a}{(2 \pi)^4 (E/a)}  \int_{E \lambda}^\infty \! d\omega \, \left| \hat{\chi}(\omega) \right|^2 g(2\pi E/a, 2 \pi \omega/(\lambda a)).
\end{align}
\label{5:Gs}
\end{subequations}
\\
such that $\mathcal{G} = \mathcal{G}_{(0,\kappa)} + \mathcal{G}_{(\kappa, E \lambda)} + \mathcal{G}_{(E \lambda,\infty)}$. 

The bounds are estimated in Appendix \ref{app:Unruh}, and collecting all the relevant estimates \eqref{boundpiece1}, \eqref{boundpiece2} and \eqref{boundpiece3}, we have that as $E \rightarrow \infty$
\begin{align}
\mathcal{G}(E) & \leq \frac{2 a ||\hat{\chi}||_{L^2} }{(2 \pi)^3} \left(\frac{2\pi E}{a}\right)^{-\left(1+p-q^{-1}\right)/2} \left[1 + O\left(\left(\frac{2\pi E}{a}\right)^{-\left(1+p-q^{-1}\right)}\right) \right] \nonumber \\
& + \frac{\kappa \left|\left|\hat{\chi} \right|\right|_{L^2}}{(2 \pi)^2 \alpha}  \left(\frac{2\pi E}{a}\right)^{-(1+p)} \left[ 1 + O\left(\left(\frac{2 \pi E}{a}\right)^{-(1+p)} \right)\right].
\label{Gestimate}
\end{align}

At $E$ sufficiently large, $\mathcal{G}$ is dominated by the first term. We define $\mathcal{G}_\text{est}$ as
\begin{equation}
\mathcal{G}(E)_\text{est} = \frac{4 a ||\hat{\chi}||_{L^2} }{(2 \pi)^3} \left(\frac{2\pi E}{a}\right)^{-\left(1+p-q^{-1}\right)/2} \geq \mathcal{G}(E),
\end{equation}
which vanishes polynomially as $E \to \infty$.

We have everything that we need to prove that the detailed balance condition \eqref{2:DBC} holds. We have from Lemma \ref{Lem2} and the estimate \eqref{Gestimate} that as $E \rightarrow \infty$ and for $\lambda = \alpha (2\pi E/a)^{1+p}$
\begin{subequations}
\begin{align}
\mathcal{F}_{T_\text{U}}(-E) & \leq \widetilde{\mathcal{F}}_\lambda(-E)  \leq \mathcal{F}_{T_\text{U}}(-E) + \frac{\left|\left| \omega \hat{\chi}(\omega) \right|\right|_{L^2}}{24 \pi a \alpha^2} \left(\frac{2 \pi E}{a}\right)^{-2(1+p)}, \\
\mathcal{F}_{T_\text{U}}(E) & \leq \widetilde{\mathcal{F}}_\lambda(E)  \leq \mathcal{F}_{T_\text{U}}(E) + \frac{2 \pi E}{a}\exp( -2\pi E /a) \mathcal{G}_\text{est}(E).
\end{align}
\label{EstimatesMainThm1}
\end{subequations}

It follows from equations \eqref{EstimatesMainThm1} that
\begin{subequations}
\begin{align}
\frac{\mathcal{F}_\lambda(-E)}{\mathcal{F}_\lambda(E)} & \geq \mathcal{F}_{T_\text{U}}(-E) \left(\mathcal{F}_{T_\text{U}}(E) + \cfrac{2\pi E/a}{ \exp( 2\pi E /a)} \, \mathcal{G}_\text{est}(E) \right)^{-1} \, , \\
\frac{\mathcal{F}_\lambda(-E)}{\mathcal{F}_\lambda(E)} & \leq \frac{\mathcal{F}_{T_\text{U}}(-E)}{\mathcal{F}_{T_\text{U}}(E)} + \cfrac{\left|\left| \omega \hat{\chi}(\omega) \right|\right|_{L^2} \, \left(2 \pi E/a \right)^{-2(1+p)}}{(24 \pi a \alpha^2) \mathcal{F}_{T_\text{U}}(E)} \, ,
\end{align}
\end{subequations}
\\
from which eq. \eqref{AsympDBC} follows at the Unruh temperature, $T_\text{U} = a/(2\pi)$, with
\begin{subequations}
\begin{align}
\mathcal{B}^{-}(E) & =  \ln \left(1 + \frac{\left(1 - \ee^{-2\pi E/a}\right) \mathcal{G}_\text{est}(E)}{(a/(2\pi)^3) \left| \left| \chi \right|\right|_{L^2}} \right)  \nonumber \\ 
& = \frac{4}{(2 \pi)^6} \left(\frac{2\pi E}{a} \right)^{-\left(1 + p - q^{-1}\right)/2} + O \left( \left(\frac{2\pi E}{a} \right)^{-\left(1 + p - q^{-1}\right)}\right) \\
\mathcal{B}^{+}(E) & = \ln \left(1 + \frac{\left|\left| \omega \hat{\chi} \right|\right|_{L^2}}{\left|\left| \hat{\chi} \right|\right|_{L^2}}\frac{ \left(1 - \ee^{-2\pi E/a} \right)}{3 \alpha^2 a^2/\pi^2} \left(\frac{2 \pi E}{a}\right)^{-(3+2p)} \right)  \nonumber \\
& = \frac{\left|\left| \omega \hat{\chi} \right|\right|_{L^2}}{\left|\left| \hat{\chi} \right|\right|_{L^2}}\frac{ \pi^2}{3 \alpha^2 a^2} \left(\frac{2 \pi E}{a}\right)^{-(3+2p)} + O\left( \left( \frac{2 \pi E}{a} \right)^{-2(3+2p)} \right). \label{B+thm}
\end{align}
\label{Bthm}
\end{subequations}
\end{proof}

Theorem \ref{thm:Main1} is the main result of this section. Notice that the correction terms $\mathcal{B}^{-}(E)$ and $\mathcal{B}^{+}(E)$ given in eq. \eqref{Bthm} become polynomially suppressed as $E \rightarrow \infty$. Let us summarise the physical relevance of this result. In Minkowski spacetime, given a field in the Minkowski state and a point-like probe along its timelike worldline that are initially non-interacting, one can set the probe trajectory at uniform acceleration $a$ at proper time $\tau < \tau_i$ and switch on the interaction at time $\tau_i$ smoothly and slowly, such that the probe acts like a particle detector with large energy gap $E$, and keep the interaction ongoing for a long time scale polynomially proportional to the detector gap energy. Once the detector has been switched off smoothly and slowly at time $\tau_f$, the transition probability of the detector will obey a thermal spectrum at the Unruh temperature.

To end this section, we provide an example of a switching function with rapid Fourier decay and present the estimates for the polynomially asymptotic thermality condition as an application of Theorem. \ref{thm:Main1}.

\begin{exmp}
The idea is to cook up switching functions with rapid Fourier decay by thinking of the Fourier transform as a boundary value along the imaginary axis of the Laplace transform, $\mathcal{F}[f](\omega) = \lim_{\epsilon \rightarrow 0^+} \mathcal{L}[f](\epsilon+\ii \omega)$.

Let $f \in C^\infty(\mathbb{R})$ be defined by
\begin{align}
f(\tau) =
\begin{cases} 
 (4 \pi)^{-1/2} (\kappa \tau)^{-3/2}e^{-1/(4 {\kappa \tau})}, & \tau>0\\ 
 0, & \tau \leq 0. 
\end{cases}
\end{align}
where $\kappa$ is a positive real constant with dimensions $[\kappa] = \left[\tau^{-1}\right]$.

Given that $f$ decays as $O\left(\tau^{-3/2}\right)$, it is integrable and its Fourier transform exists. The Laplace transform is given by $\mathcal{L}[f](\omega) = \kappa^{-1} \ee^{-(\omega/\kappa)^{1/2}}$ for any $\omega$ with positive real part, taking the branch cut along the negative real axis. The Fourier transform follows from dominated convergence and is given by 
\begin{equation}
\hat{f}(\omega) = \kappa^{-1} \ee^{-(\ii\omega/\kappa)^{1/2}} = \kappa^{-1} \exp \left(-|\omega/(2 \kappa)|^{1/2}(1 + \ii \text{sgn}(\omega) \right).
\end{equation} 

It is bounded from above by $\left|\hat{f}(\omega)\right| \leq \kappa^{-1} \exp \left(-|\omega/(2 \kappa)|^{1/2}\right)$. The function $\chi \in C_0^\infty(\mathbb{R})$ defined by $\chi(\tau) = f(\tau) f(1/\kappa-\tau)$, has a Fourier transform given by the convolution theorem,
\begin{equation}
\hat{\chi}(\omega) = \frac{1}{2 \pi} \int_{-\infty}^\infty \, d \omega' \hat{f}(\omega - \omega') \ee^{-\ii \omega'/\kappa} \hat{f}(-\omega').
\end{equation}

We now verify that $\chi \in F$. It follows from the bound of $\hat{f}$ that
\begin{align}
|\hat{\chi}(\omega)| & \leq \frac{1}{2 \pi \kappa^2} \int_{-\infty}^\infty d\omega' \exp \left(-|(\omega-\omega')/(2 \kappa)|^{1/2} -|\omega'/(2 \kappa)|^{1/2}\right) \nonumber \\
& \leq \frac{|\omega/\kappa|}{2 \pi \kappa} \int_{-\infty}^\infty du \exp \left[-\left|\omega/(2\kappa) \right|^{1/2} \left(|1-u| + |u| \right) \right],
\end{align}
where we have assumed $\omega \neq 0$ without loss of generality and changed variables to $\omega' = \omega u$. The integral can be now split into the subintervals $\mathbb{R} = (-\infty,0) \cup [0,1] \cup (1,\infty)$. The integrals in the intervals $(-\infty,0)$ and $(1,\infty)$ are equal by symmetry and we have that
\begin{subequations}
\begin{align}
I_{(0,1)} & \doteq \int_{0}^1 \! du \, \exp \left[-\left|\omega/(2\kappa) \right|^{1/2} \left(|1-u| + |u| \right) \right] \nonumber \\
& \leq \int_{0}^1 \! du \,\ee^{-\left|\omega/(2\kappa) \right|^{1/2}} = \ee^{-\left|\omega/(2\kappa) \right|^{1/2}}
\label{Ex1}
\end{align}
and
\begin{align}
I_{(1,\infty)} & \doteq \int_{1}^\infty \! du \, \exp \left[-\left|\omega/(2\kappa) \right|^{1/2} \left(|1-u| + |u| \right) \right] \nonumber \\
& = \int_{0}^\infty \! du \, \exp \left[-\left|\omega/(2\kappa) \right|^{1/2} \left(|u| + |u+1| \right) \right] \nonumber \\
& \leq  \int_{0}^1 \! du \, \exp \left(-\left|\omega/(2\kappa) \right|^{1/2} \right) +  \int_{1}^\infty \! du \, \exp \left(-2\left|\omega/(2\kappa) \right|^{1/2} |u|  \right) \nonumber \\
& \leq \exp \left(-\left|\omega/(2\kappa) \right|^{1/2} \right) + \exp \left(-2\left|\omega/(2\kappa) \right|^{1/2} \right) \left( \frac{1}{\sqrt{|\omega/(2 \kappa)}|} + \frac{1}{|\omega/ \kappa|} \right).
\label{Ex2}
\end{align}
\end{subequations}

Putting \eqref{Ex1} and \eqref{Ex2} together we see that

\begin{align}
\left| \hat{\chi}(\omega) \right| & \leq \frac{1}{2 \pi \kappa}\left(\ee^{-|\omega/(2\kappa)|^{1/2}}(1 + 2 |\omega/(2\kappa)|^{1/2}) +3 |\omega/(2\kappa)|  \right) \ee^{-|\omega/(2\kappa)|^{1/2}} \nonumber \\
& \leq \frac{9}{2 \pi \kappa}\left(\frac{1}{3} + |\omega/\kappa|  \right)^2 \ee^{-|\omega/\kappa|^{1/2}},
\end{align}
\\
and $\chi \in F$ has strong Fourier decay as in Def. \ref{RapidFourier}.
\end{exmp}

\section{Approaching the Unruh temperature by constant interaction}
\label{sec:plateau}

We now bring our attention to a switching function of the form \eqref{chiconst} that has $\text{supp}(\chi) = [0, \Delta \tau + 2 \Delta \tau_s]$, consisting of a switch-on profile of duration $\Delta \tau_s$, a constant interaction time of duration $\Delta \tau$ and a switch-off profile of duration $\Delta \tau_s$. As we have seen above, the profile of the switching tails is controlled by a bump function, $\psi$.

We consider the case in which only the constant-interaction time is rescaled by $\lambda$. The response rescales as shown in eq. \eqref{Flambda2}, whereby the response of the detector following the Rindler orbit rescales as
\begin{equation}
\breve{\mathcal{F}}_\lambda(E) = \frac{1}{(2 \pi)^2} \int_{-\infty}^\infty \! d \omega \,\,  \frac{2}{\lambda} \left[1-\cos\left(\omega\left(\lambda \Delta \tau + \Delta \tau_s \right)\right)\right] \left| \frac{\hat{\psi}(\omega)}{\omega} \right|^2 \frac{\omega + E}{\ee^{2\pi (\omega + E) /a} -1}.
\label{Flambda2Unruh}
\end{equation}

In this section, we denote $\breve{\mathcal{F}}_\infty$ by $\mathcal{F}_{T_\text{U}}$ as a reminder that the point-wise limit in $E$ as $\lambda \to \infty$ satisfies the detailed balance condition at the Unruh temperature. Namely, using \eqref{5:KMSconstant} and \eqref{eq:Unruh-hatWstat}
\begin{align}
\mathcal{F}_{T_\text{U}}(E) \doteq \breve{\mathcal{F}}_\infty(E) =  \left( \frac{\Delta \tau}{2\pi} \left|\hat{\psi}(0) \right|^2\right) \, \frac{E}{\ee^{2\pi E /a} -1}.
\label{FT2}
\end{align}

We now show that the constant-interaction rescaled response function \eqref{Flambda2Unruh} cannot satisfy the polynomially asymptotic KMS thermality condition. In other words, there is no polynomial $\lambda(E)$ for which one can find $\mathcal{B}^{\pm}$ for Def. \ref{DefAsympKMS}.

\begin{thm}
Let $\breve{\mathcal{F}}_\lambda(E)$ be the response function given by eq. \eqref{Flambda2Unruh}. There is no scaling factor $\lambda = P(E)$, with $P: \mathbb{R}^+ \rightarrow \mathbb{R}^+$ polynomially bounded, differentiable, strictly increasing and strictly positive, for which $\breve{\mathcal{F}}_\lambda(E)$ satisfies the detailed balance condition \eqref{AsympDBC}.
\label{thm:Main2}
\end{thm}
\begin{proof}
A necessary condition for the polynomially asymptotic thermality condition to be satisfied, is that the quantity
\begin{equation}
\mathcal{K}(E) \doteq \left|\frac{\exp(2 \pi E/a)}{2\pi E/a}\left(\mathcal{F}_\lambda(E)/\lambda - \mathcal{F}_{T_\text{U}}(E)\right) \right|
\label{K}
\end{equation}
vanish as $E \rightarrow + \infty$ and $\lambda = P(E)$ grows polynomially as a function of $E$.

It follows from the closed form of $\mathcal{F}_{T_\text{U}}$ given by eq. \eqref{FT2} that $\lim_{E \rightarrow \infty} \mathcal{K}(E) = 0$ for $\lambda = P(E)$ only if
\begin{equation}
\frac{\exp(2 \pi E/a)}{2\pi E/a}\mathcal{F}_\lambda(E)/\lambda = \frac{a \Delta \tau}{(2 \pi)^2} \left| \hat{\psi}(0) \right|^2 + \textit{o}(1)
\label{FakeLimit}
\end{equation}
as $E \rightarrow \infty$. We will show that there is no polynomial function $P$ for which the expression on the left of eq. \eqref{FakeLimit} is $O(1)$. Setting $\Omega \doteq 2\pi \omega/a$, $\mathcal{E} \doteq 2 \pi E/a$ and $\Lambda \doteq  a (\lambda \Delta \tau + \Delta \tau_s)/(2 \pi)$, it follows from our assumptions that $\Lambda: \mathbb{R}^+ \rightarrow \mathbb{R}^+$ is a polynomially bounded, differentiable, strictly increasing, strictly positive function of $\mathcal{E}$ and using these variables the left hand side of eq. \eqref{FakeLimit} reads
\begin{equation}
\frac{\exp(2 \pi E/a)}{2\pi E/a} \mathcal{F}_\lambda(E)/\lambda = \frac{2 a \Delta \tau}{(2\pi)^2} \int_{-\infty}^{\infty} \! d\Omega \, \frac{1- \cos(\Lambda \Omega)}{2 \pi \Lambda - a \Delta \tau_s}  \left| \frac{\hat{\psi}(a \Omega /(2\pi))}{\Omega} \right|^2 \left( \frac{\Omega/ \mathcal{E} + 1}{\ee^{\Omega} - \ee^{-\mathcal{E}}}\right).
\label{Unbound}
\end{equation}

Now we will show that the integral
\begin{equation}
I(\mathcal{E}) \doteq \int_{-\infty}^{\infty} \! d\Omega \, \frac{1- \cos(\Lambda \Omega)}{\Lambda}  \left| \frac{\hat{\psi}(a \Omega /(2\pi))}{\Omega} \right|^2 \left( \frac{\Omega/ \mathcal{E} + 1}{\ee^{\Omega} - \ee^{-\mathcal{E}}}\right)
\label{I}
\end{equation}
is not bounded as $\mathcal{E} \to \infty$. This implies that eq. \eqref{FakeLimit} cannot hold and the polynomially asymptotic thermality condition is not satisfied. The strategy is to write the integral defined in eq. \eqref{I} in a form that satisfies the hypotheses of Lemma \ref{LemMain2}, such that the application of Corollary \ref{Cor} follows.

Defining $g \in C_0^\infty(\mathbb{R})$ by the convolution of $\psi$ with its reflection about the origin, which we denote $R\psi$,
\begin{equation}
g(\tau) \doteq (\psi \star R\psi)(\tau) = \int_{-\infty}^\infty \! dt \, \psi( \tau - t) \psi(-t),
\end{equation}
it follows from the convolution theorem that $\hat{g}(v) = \hat{\psi}(v) \hat{\psi}(-v) = \left|\hat{\psi}(v)\right|^2$ has rapid decay. From the definition of $\Lambda$, we have that $\Lambda(0) = \lim_{\mathcal{E} \rightarrow 0^+} \Lambda(\mathcal{E}) > 0$. Considering the interval $J_\mathcal{E} = \left[-\mathcal{E}-\ee^{-\mathcal{E}/2}, -  \mathcal{E}\right]$, one can obtain the estimate
\begin{align}
I(\mathcal{E}) & \geq \frac{2 \ee^\mathcal{E}}{\mathcal{E} \Lambda(\mathcal{E})} \int_{-\mathcal{E}-\ee^{-\mathcal{E}/2}}^{-\mathcal{E}} \! d\Omega \, \frac{\sin^2(\Lambda(\mathcal{E}) \Omega/2)}{\Omega^2}  \left( \frac{\Omega + \mathcal{E}}{\ee^{\Omega+\mathcal{E}} - 1}\right) \hat{g}\left( \frac{a \Omega}{2 \pi}\right) \nonumber \\
& \geq \frac{2 \ee^{\mathcal{E}/2} Q(\mathcal{E})}{\Lambda(\mathcal{E})\mathcal{E}\left(\mathcal{E} + \ee^{-\mathcal{E}/2}\right)^2} \inf_{\Omega \in J_\mathcal{E}} \hat{g}\left( \frac{a \Omega}{2 \pi} \right)
\label{ILowerBd}
\end{align}
using the fact that $y/(\ee^y-1) \geq 1$ for $y \leq 0$ and where
\begin{equation}
Q(\mathcal{E}) \doteq \inf_{\Omega \in J_\mathcal{E}} \sin^2 \left( \frac{\Lambda(\mathcal{E}) \Omega}{2} \right).
\label{Qdef}
\end{equation}

Suppose that $I_{\max}$ is a positive constant such that $I(\mathcal{E}) \leq I_{\max}$ for all $\mathcal{E}>0$. Now, let $T = \{ \mathcal{E} \in \mathbb{R}^+ : Q(\mathcal{E}) \geq \ee^{-\mathcal{E}/4} \}$. $T$ is non-empty and unbounded from above. In $T \cap J_\mathcal{E}$, \eqref{ILowerBd} implies that
\begin{equation}
\inf_{\Omega \in J_\mathcal{E}} \hat{g}\left( \frac{a \Omega}{2\pi} \right) \leq P(\mathcal{E}) \ee^{-\mathcal{E}/4},
\label{CrudeIneq}
\end{equation}
where
\begin{equation}
P(\mathcal{E}) \doteq I_{\max} \frac{1}{2} \Lambda(\mathcal{E}) \mathcal{E}(\mathcal{E}+1)^2.
\end{equation}

In $T$ we hence have
\begin{equation}
\hat{g}\left(\frac{-a \mathcal{E}}{2 \pi}\right) \leq P(\mathcal{E}) \ee^{-\mathcal{E}/4} + C \ee^{-\mathcal{E}/2} \leq (P(\mathcal{E}) + C) \ee^{-\mathcal{E}/4},
\end{equation}
where $C \doteq (a/(2\pi)) \sup_{\mathbb{R}^+} \left| \hat{g}' \right|$. Recalling that $\hat{g}$ is even, we deduce that for any $\gamma \in (0, 1/4)$, there exists $K > 0$ such that $\hat{g}(a \mathcal{E}/(2\pi)) \leq K \ee^{-\gamma |\mathcal{E}|}$ holds for all $|\mathcal{E}| \in T$.

Next, we show that the complement of $T$ in $\mathbb{R}^+$, $T^c = \{ \mathcal{E} \in \mathbb{R}^+ : Q(\mathcal{E}) < \ee^{-\mathcal{E}/4} \}$ is contained within a suitable union of intervals. Using the bound $|\sin(\theta)| \geq 2 |\theta|/\pi$, valid for $|\theta| \leq \pi/2$,  \eqref{Qdef} gives
\begin{equation}
Q(\mathcal{E})  \geq \pi^{-2} \inf_{\Omega \in J_\mathcal{E}} \min_{k \in \mathbb{N}_0} |\Lambda(\mathcal{E}) \Omega + 2 k \pi|^2.
\end{equation}

For $\Omega \in J_\mathcal{E}$, we have $|\Lambda(\mathcal{E})(\Omega+\mathcal{E})| \leq \Lambda(\mathcal{E}) \ee^{-\mathcal{E}/2}$, and we have that
\begin{equation}
\pi Q(\mathcal{E})^{1/2}  \geq \min_{k \in \mathbb{N}_0} |\Lambda(\mathcal{E}) \mathcal{E} - 2 k \pi| - \Lambda(\mathcal{E}) \ee^{-\mathcal{E}/2}.
\label{PreSetTc}
\end{equation}

From the definition of $T$ and \eqref{PreSetTc} it follows that if $\mathcal{E} \in T^c$, there exists a $k \in \mathbb{N}_0$ such that $|\Lambda(\mathcal{E}) \mathcal{E} - 2 k \pi| \leq \Lambda(\mathcal{E}) \ee^{-\mathcal{E}/2} + \pi \ee^{-\mathcal{E}/8} \leq C' \ee^{-\mathcal{E}/8}$, where $C' = \pi + \sup_{\mathcal{E} \geq 0} \Lambda(\mathcal{E}) \ee^{-3\mathcal{E}/8}$.

Now, let $\Xi(\mathcal{E}) \doteq \mathcal{E}\Lambda(\mathcal{E})$, and for each $k \in \mathbb{N}_0$, let $\mathcal{E}_k$ be the unique solution to $\Xi(\mathcal{E})k = 2 \pi k$. ($\mathcal{E}_k$ exists and is unique because $\Xi$ is strictly increasing and $\Xi(0) = 0$.) For $\mathcal{E} \in T^c$, there hence exists a $k \in \mathbb{N}_0$ such that
\begin{equation}
\left|\Xi(\mathcal{E}) - \Xi(\mathcal{E}_k) \right| \leq C' \ee^{-\mathcal{E}/8},
\end{equation}
and hence
\begin{equation}
|\mathcal{E} - \mathcal{E}_k| \leq \frac{C'}{\Lambda(0)} \ee^{-\mathcal{E}/8}
\end{equation}
using $\Xi'(\mathcal{E}) \geq \Lambda(0)$, and further
\begin{equation}
|\mathcal{E}- \mathcal{E}_k| \leq \frac{C'\ee^{C'/(8 \Lambda(0))}}{\Lambda(0)} \ee^{-\mathcal{E}_k/8}
\end{equation}
using $\mathcal{E} \geq \mathcal{E}_k - C'/\Lambda(0)$.

Collecting, we see that
\begin{equation}
T^c \subset \bigcup_{k \in \mathbb{N}_0} \left\{ \mathcal{E} \in \mathbb{R}: \, |\mathcal{E} - \mathcal{E}_k | \leq \frac{C'\ee^{C'/(8 \Lambda(0))}}{\Lambda(0)} \ee^{-\mathcal{E}_k/8} \right\}.
\end{equation}

We now apply Corollary \ref{Cor}. Since $\Lambda(\mathcal{E}) \leq D(\mathcal{E}+1)^N$ for some $D > 0$ and $N > 0$, we have that $\mathcal{E}_k + 1 \geq (2 \pi k/D)^{1/(N+1)}$ and condition \eqref{LemSumFinite} in Lemma \ref{LemMain2} holds with any $\beta > 0$. Setting $\beta \in (0, 1/8)$, we choose the set $S$ in Lemma \ref{LemMain2} so that $\alpha = 1/8$ and $C = \delta_0  = 2 C' \ee^{C'/(8 \Lambda(0))}/\Lambda(0)$.

The above observations about $T$ and $T^c$ then show that $g$ satisfies the conditions of Corollary \ref{Cor}, by which $g$ must be identically vanishing, and this contradicts the construction of $g$. Hence $I(E)$ \eqref{I} cannot be bounded. This completes the proof.
\end{proof}

By the results of Theorem \ref{thm:Main2} one sees that scaling a constant interaction time polynomially in the energy gap $E$ is not a sufficient condition for obtaining polynomial asymptotic thermality as $E \rightarrow \infty$.

In view of Theorems \ref{thm:Main1} and \ref{thm:Main2}, we see that in order to achieve asymptotic thermality in a time scale $\lambda$ which is a polynomial function of the large energy $E \rightarrow \infty$ the switching must be long and slow. If the interaction is switched on once and for all, and the field and the detector are left to interact constantly for a time $\Delta \tau$ before switching off the interaction, we find that a polynomially long constant interaction time does not suffice to thermalise the response as $E \to \infty$. Thus, we conclude that the  details of the thermalisation at the level of the response function must be in the switch-on and switch-off.

\chapter{Conclusions}
\label{ch:conc}

In this work, we have analysed the effects of quantum fields as perceived by local observers in curved (and flat) spacetimes, with emphasis in the Hawking and Unruh effects in time-dependent situations, extensively using techniques from asymptotic analysis. 

In the case of the Hawking effect, we have concentrated in $1+1$ dimensions and considered a massless, minimally-coupled, free Klein-Gordon field, which is, in turn, conformally coupled. The conformal coupling enables us to compute interesting situations analytically. Because the Wightman function of a massless $1+1$ scalar field is ambiguous, in Chapter~\ref{ch:DeCo} we have employed a derivative-coupling detector that is insensitive to the ambiguity, and for which the massive to massless limit is continuous. We verified that this detector responds appropriately to thermal phenomena, such as the Unruh effect. Moreover, the response depends only on the twice-differentiated two-point function in $1+1$ dimensions, which diverges at small distances as strongly as the undifferentiated $3+1$ two-point function, rendering the $1+1$ detector model into an attractive tool for obtaining intuition of $3+1$ situations. 

We have exploited the conformal symmetry of the theory in $1+1$ dimensions to work out explicitly the field quantisation in lower dimensional black hole spacetimes in Chapter~\ref{ch:BH}. First, we have studied a receding mirror spacetime, which models the collapse of a star whereby a black hole is formed at late times; second, we studied the $1+1$ Schwarzschild spacetime and, third, a class of generalised $1+1$ (non-extremal) Reissner-Nordstr\"om black holes. A spacetime belonging to this class could include, for example, Reissner-Nordstr\"om black hole with quantum gravity corrections near the singularity as in \cite{Gambini:2013hna, Gambini:2014qta}.

In the receding mirror spacetime, we have computed the transition rate of a detector that is inertial in the asymptotic past, at early and late times along its trajectory. At early times, the rate responds as if a Dirichlet wall were present in the Minkowski half-space, whereas at late times the right-moving modes carry a frequency shift, due to the motion of the mirror, that satisfies, to leading order, the detailed balance form of the KMS condition, with a temperature that depends on the parameters of the mirror trajectory. The interpolating behaviour of the rate was completed numerically. We have verified that the particle detection is related to an outgoing positive energy flux moving away from the mirror.

In the Schwarzschild black hole, we have studied the responses of geodesic detectors that interact with fields in the Unruh and in the Hartle-Hawking-Israel states. We have computed the asymptotic rates near the infinity and near the black hole singularity. Near the infinity, the Hawking temperature is imprinted on the detector, as it interacts with the right-moving modes of a field in the Unruh state. This is the Hawking effect. If the state is the HHI vacuum, the detector is thermalised at the Hawking temperature. Near the singularity, the transition rate of the detector diverges in a non-integrable fashion proporitonally to $1/\tau$ as $\tau \to 0$, meaning that the response diverges logarithmically. We have computed numerically the transition rate of a freely-falling detector from infinity and demonstrated that the rate loses its thermal character as the detector closes into the black hole from infinity. This is consistent with the fact that the detector is only stationary with respect to the field modes coming from the black hole in the asymptotic infinity region. For a detector coming from the white hole region, interacting for a finite amount of time, our numerics confirm that the transition rate is not thermal. A final comment is that, for regular states across the black hole horizon, we have not detected anything significant occurring at the black hole horizon. We consider this relevant in the light of the information paradox discussion \cite{Almheiri:2012rt}.

In our class of generalised Reissner-Nordstr\"om spacetimes, we have studied the quantum effects near the future Cauchy horizon, as perceived by a geodesically infalling observer. First, we have computed the rate of a detector that falls across the horizon and towards the left part of the Cauchy horizon, coupled to a field in the HHI or in the Unruh state. We have found that the transition rate of such a detector diverges as $1/(\tau-\tau_h)$ in their proper time, as they approach the horizon at proper time $\tau_h$. Because the rate of divergence of the transition rate is non-integrable, the response of the detector also diverges. Then, we analysed the local energy density, in the HHI state, experienced by such geodesic infalling observer as they approach the future horizon. We estimated, with the aid of conformal techniques, that, in the near-horizon region, the local energy density too diverges inverse-linearly.

We stress that $1+1$ dimensional models provide insight to the full $3+1$ situations, but several technical points need to be discussed carefully. As a way to exemplify this, we consider the sharp transition rate of a derivative-coupling detector in $1+1$ dimensions interacting with a field in the Rindler state. As an inertial detector leaves the Rindler wedge and crosses to the full Minkowski spacetime, the transition rate diverges faster for a $1+1$ derivative-coupling detector than for a $3+1$ non-derivative Unruh-DeWitt detector. The transition probability diverges in the former case, while it is finite in the latter. 

In Chapter \ref{ch:Unruh}, we have studied the emergence of thermality for (ensembles of) stationary detectors that interact for a finite amount of time with a field in a KMS state and provided general results concerning the non-uniformity in the long-time and large-energy limits. In particular, we have shown that, for any thermal state (under reasonable conditions), as the energy becomes large, the response of the detector fails to satisfy the detailed balance form of the KMS condition asymptotically. This leaves one with the question of how long one needs to wait to register the KMS temperature up to a large energy scale, $E$, with particle detectors. In the case of the Unruh effect, we have examined the emergence of the Unruh temperature, as registered by a linearly uniformly accelerated detector that interacts with a $3+1$ massless Klein-Gordon field in the Minkowski vacuum state for a finite (but long) time. We have shown that the thermalisation time-scale depends crucially on the details of the switching of the interaction between the detector and the field. If the interaction switching scales adiabatically to late times, then the waiting time is polynomial in the large energy gap. However, we have proved that this situation does not need to hold in general. The waiting time may need to be exponential if the switching is not adiabatic. 

We think that the analysis on Chapter \ref{ch:Unruh} is relevant in the field of analogue gravity. The estimates that we have obtained should be interpreted in the following way: An experimentalist seeking to verify the Unruh effect up to an energy $E$ in the spectrum registered by an ensemble of detectors (or a multi-state detector) would need to switch the detector-field interaction smoothly and slowly for a long time, and only measure the state of the detector after a time that is polynomial in the energy scale that they seek to test has elapsed. The spectrum measured beyond the energy $E$ will not be thermally occupied, and, thus, should be discarded when attempting to infer the (asymptotic) temperature at which such ensemble of detectors thermalises, the Unruh temperature. The intuition is that higher-energy occupation modes of an ensemble of detectors take more time to populate than lower energy ones. The key point for an experimentalist seeking to measure the Unruh effect is not simply to conclude that the detector has been excited by its interactions with the field, but to verify that the response of the detector is thermal at the Unruh temperature. Here, we have prescribed up to what energy scale this can be verified in terms of the interaction time of the detector.

The work that we have presented in this thesis has several natural extensions. The $1+1$ derivative-coupling detector model that we have used and the techniques that we have presented can be exploited extensively in several interesting situations in $1+1$ dimensions, \textit{e.g.}, in cosmological settings. Another situation of interest is that of a derivative-coupling detector in an \textit{incoming mirror spacetime}, which produces a negative flux that can be detected. An analysis in $1+1$ dimensions with a massless scalar can be carried out in the spirit of \cite{Davies:2002bg}, where it is found that negative fluxes enhance the de-excitation response of a detector. Another interesting scenario is to analyse the response of detectors coupled to squeezed states, which can have negative energy densities \cite{Ford:1992ii}.  From such an investigation, a connection with quantum inequalities in conformal field theory can be drawn \cite{Fewster:2004nj}.  

An issue that we have left out of this work and that deserves attention is that of back-reaction in suitable $1+1$ gravity theories with local dynamics. Finally, it would be interesting to explore the issues of thermalisation time scales, which we have presented for the Unruh effect, in the context of black hole radiation.

\appendix

\chapter{Supplement to Chapter 3}
\label{app:DeCo}

\section{Massive field in Minkowski half space}
\label{sec:app3:HalfSpace}

We compute the transition rate $\dot{\mathcal{F}}_m = \dot{\mathcal{F}}^\text{reg}_m + \dot{\mathcal{F}}^\text{dist}_m$, where

\begin{subequations}
\begin{align}
\dot{\mathcal{F}}^\text{reg}_m(E) & \doteq -E \Theta(-E) + 2 \int_0^{\infty} \! ds \, \text{Re} \left[ \ee^{-\ii E s} \left( \frac{m^2}{2 \pi} K_0'' \left( \ii m s \right) + \frac{1}{2 \pi s^2} \right)\right] \label{app3:Freg} \\
\dot{\mathcal{F}}^\text{dist}_m(E) & \doteq \frac{\eta}{2 \pi} \int_{-\infty}^\infty \! ds \, \ee^{-\ii E s}  \frac{d^2}{ds^2} K_0 \left( m \sqrt{4d^2 - (s - \ii \epsilon)^2} \right) \label{app3:Fdist}
\end{align}
\end{subequations}

We start by computing eq. \eqref{app3:Freg}. Throughout our calculation, we assume that $E \neq -m$. We start by rewriting \eqref{app3:Freg} in the form of \eqref{DeCoStat}
\begin{equation}
\dot{\mathcal{F}}^\text{reg}_m(E) = \frac{m^2}{2\pi} \int_{-\infty}^\infty \! dz \, \ee^{-\ii E z} K_0''(\ii m (z - \ii \epsilon)).
\label{app3:Freg2}
\end{equation}

The branch of $K_0''$ is as explained in the main text in terms of derivatives of Hankel functions. See \eqref{2:K0WightmanTimelike}. The integral on the right hand side of \eqref{app3:Freg2} exists as a Riemann integral because of the asymptotic behaviour of $K_0$, $K_0(\ii z) = O\left( |z|^{-1/2}\right)$ as $|z| \to \infty$.
Integrating by parts twice, we obtain
\begin{equation}
\dot{\mathcal{F}}^\text{reg}_m(E) = \frac{E^2}{2\pi} \int_\gamma \! dz \, \ee^{-\ii E z} K_0(\ii m z),
\end{equation}
where the path $\gamma$ extends along the real line, except that it dips below the real line near $z = 0$. Moreover, because the singularity at $z = 0$ is logarithmic, we can perform the integration across $z = 0$, and using the analytic continuation formulas \eqref{2:K0WightmanTimelike},
\begin{equation}
\dot{\mathcal{F}}^\text{reg}_m(E) = \frac{E^2}{2} \text{Im} \int_0^\infty \! ds \, \ee^{-\ii E s} H_0^{(2)}(ms).
\label{app3:FregLee}
\end{equation}

The strategy to integrate \eqref{app3:FregLee} is to use the standard integral (6.611.7) of \cite{Gradshteyn:2007}
\begin{equation}
\int_0^\infty \! ds \, \ee^{-\alpha s} H_0^{(2)}(ms) = \frac{1}{\sqrt{\alpha^2+1}} \left[1 + \frac{2 \ii}{\pi} \ln \left( \alpha + \sqrt{1+ \alpha^2} \right) \right]
\end{equation}
for real $\alpha > 0$. One then takes the imaginary part of the analytic continuation of $\alpha$ to $\ii E$, cf. \eqref{app3:FregLee}, and one obtains
\begin{align}
\dot{\mathcal{F}}^\text{reg}_m(E) = \Theta(-E-m) \frac{E^2}{\left(E^2 - m^2\right)^{1/2}}.
\label{app3:MassiveReg}
\end{align}

We refer to \cite{Hodgkinson:2013tsa} for details. 

We proceed to integrate the expression defined by $\dot{\mathcal{F}}_m^\text{dist}$. Integrating by parts twice the right hand side of eq. \eqref{app3:Fdist} and rearranging the integration limits,
\begin{equation}
\dot{\mathcal{F}}^\text{dist}_m(E) = \frac{\eta E^2}{2 \pi} \text{Re} \int_{0}^\infty \! ds \, \ee^{-\ii E s}  K_0 \left( m \sqrt{4d^2 - (s - \ii \epsilon)^2} \right).
\label{app3:MassiveDistInt}
\end{equation}

It follows from complex analytic integration techniques that we can write \eqref{app3:MassiveDistInt} as
\begin{align}
\dot{\mathcal{F}}^\text{dist}_m(E) & = \frac{\Theta(-E) \eta E^2}{2 \pi} \text{Re} \left[\int_\gamma \! dz \, \ee^{-\ii E z}  K_0 \left( m \sqrt{4d^2 - z^2} \right) \right],
\label{app3:MassiveDistInt2}
\end{align}
where the path $\gamma$ dips into the negative imaginary part of the complex plane near $z = 4d$. By deformation contour techniques and the analytic continuation of the Bessel function, it follows that eq. \eqref{app3:MassiveDistInt2} can be expressed as
\begin{align}
\dot{\mathcal{F}}^\text{dist}_m(E) & = -\frac{\Theta(-E) \eta E^2}{4} \nonumber \\ 
& \times \text{Im} \int_{2d}^\infty \! ds \, \ee^{-\ii E s}  \left[H_0^{(1)} \left( m \sqrt{s^2 -4d^2} \right) + H_0^{(2)} \left( m \sqrt{s^2 -4d^2} \right)\right].
\label{app3:MassiveDistInt3}
\end{align}

Using the identity \cite{NIST}
\begin{equation}
I_\nu(z) = \frac{1}{2} \ee^{\mp \ii \pi \nu/2} \left(H_\nu^{(1)} \left( z \ee^{\pm \ii \pi/2} \right) + H_\nu^{(2)} \left( z \ee^{\pm \ii \pi/2} \right) \right),
\end{equation}
where $I_\nu$ is a modified Bessel function of the first kind. The integral on the right hand side of \eqref{app3:MassiveDistInt3} is standard (see \cite{Gradshteyn:2007} formula 6.611.1) if we replace $\ii E$ by the real parameter $\alpha > 0$. The analysis proceeds as before, and we find
\begin{equation}
\dot{\mathcal{F}}^\text{dist}_m(E) = \Theta(-E-m) \frac{\eta E^2 \cos\left(2d \sqrt{E^2 - m^2} \right)}{\sqrt{E^2-m^2}}.
\label{app3:MassiveDist}
\end{equation}

Once more, we refer to \cite{Hodgkinson:2013tsa} for details. 

Putting the results \eqref{app3:MassiveReg} and \eqref{app3:MassiveDist} together, we obtain eq. \eqref{3:Fm}.

\section{Massless field in Minkowski half space}
\label{sec:app3:HalfSpacem0}

We wish to compute the integral
\begin{equation}
I(E,d) \doteq \int_{-\infty}^\infty \! ds \, \ee^{-\ii E s} \frac{d^2}{ds^2}\ln \left[ \mu \sqrt{4d^2 - (s - \ii \epsilon)^2} \right].
\label{app:I0}
\end{equation}

This can be achieved by elementary complex integration methods. Taking the derivatives, and using Jordan's lemma, we can write \eqref{app:I0} as
\begin{align}
I(E,d) & = \Theta(-E) \int_{\gamma_-} \! dz \, \ee^{-\ii E z} \frac{(z-\ii \epsilon)^2 + 4d^2}{\left[(z-\ii \epsilon)^2 - 4d^2 \right]^2} \nonumber \\
& + \Theta(+E) \int_{\gamma_+} \! dz \, \ee^{-\ii E z} \frac{(z-\ii \epsilon)^2 + 4d^2}{\left[(z-\ii \epsilon)^2 - 4d^2 \right]^2},
\end{align}
where $\gamma_-$ is the complex path from $-\infty$ to $+\infty$ along the real axis and closes along an arc in the positive imaginary quadrants, and $\gamma_-$ is the complex path from $-\infty$ to $+\infty$ along the real axis that then closes as an arc in the negative imaginary quadrants. Writing the denominator as $\left[(z-\ii \epsilon + 2d)(z - \ii \epsilon - 2d) \right]^2$, it is clear that the contribution along $\gamma_+$ is zero. The contribution along $\gamma_-$ is
\begin{align}
I(E,d) & = 2 \pi \ii \Theta(-E) \left(\text{Res}_{\ii \epsilon - 2d} + \text{Res}_{\ii \epsilon + 2d}\right)
\end{align}
where the residues are given by
\begin{subequations}
\begin{align}
\text{Res}_{\ii \epsilon + 2d} =  \frac{-\ii E \ee^{-\ii 2d E}}{2}, \\
\text{Res}_{\ii \epsilon - 2d} =  \frac{-\ii E \ee^{+\ii 2d E}}{2}.
\end{align}
\end{subequations}

This leads to eq. \eqref{3:F0}.

\chapter{Supplement to Chapter 4}
\label{app:BH}

\section{The receding mirror spacetime}

\subsection{Static trajectory with respect to the mirror in the distant past}

We estimate the behaviour of the transition rate in the distant past and the distant future.

\subsubsection*{Distant past, $\tau \to -\infty$}
We start out by estimating the behaviour of $\dot{\mathcal{F}}_0$, defined by eq. \eqref{4:StaticRecMirF0}, as $\tau \to -\infty$. We introduce the parameter $h \doteq \left(1 + \ee^{\kappa(d-\tau)}\right)^{-1} \to 0^+$ and write
\begin{equation}
\dot{\mathcal{F}}_0(E,\tau) = -E \Theta (-E) + \frac{1}{2 \pi} \int_0^\infty \! ds \, \cos(Es) \left(\frac{1}{X(s)} + \frac{1}{s^2} \right)
\label{App4:F0statpast}
\end{equation}
where we have defined $X$ by
\begin{align}
X(s) \doteq - \frac{[1 - h(1 - \ee^{-\kappa s})] {\left\{\kappa s + \ln[1 - h(1 - \ee^{-\kappa s})]\right\}}^2}{\kappa^2 {(1-h)}^2}. 
\label{eq:app:X-def}
\end{align}

The key point is that an expansion of $X$ at small $h$ can be made uniformly in $s$ in such a way that the estimate can be inserted inside the integral in \eqref{App4:F0statpast}. One way to achieve this \cite{Hodgkinson:2013tsa} is to write
\begin{equation}
\frac{1}{X(s)} + \frac{1}{s^2}  = \frac{-X(s) -s^2}{s^4} \left(1 + \frac{-X(s)-s^2}{s^2} \right)^{-1}.
\label{eq:app:F0int-rearr}
\end{equation}

A Taylor expansion of the numerator of \eqref{eq:app:X-def} to quartic order in $h\left(1-\ee^{-\kappa s}\right)$ shows that the second factor in \eqref{eq:app:F0int-rearr} is of the form $1 + O(h)$, uniformly in~$s$, and yields for the first factor in \eqref{eq:app:F0int-rearr} an estimate that can be applied under the integral over $s$ and whose leading term is proportional to~$h$.  We have that 
\begin{align}
\dot{\mathcal{F}}_0 (E,\tau)
&= 
- E \Theta(-E)
+ \frac{h}{2\pi \kappa} \int_0^\infty ds \, \cos(E s) \, \frac{\ee^{-\kappa s}(2+\kappa s) +\kappa s -2}{s^3} + O(h^2) \nonumber \\
& = - E \Theta(-E) + O(h).
\label{eq:app:F0past0-final}
\end{align}

The second line comes after verifying that the integrand above is $O(1)$ as $s \to 0$.

The estimation of $\dot{\mathcal{F}}_1$, defined by \eqref{4:StaticRecMirF1}, is similar. The asymptotic behaviour can be expressed in terms of the sine and cosine integrals, $\text{si}$ and $\text{Ci}$ respectively in the notation of \cite{NIST},
\begin{subequations}
\begin{align}
\text{si}(z) & \doteq - \int_z^\infty \! dt \, \frac{\sin t}{t}, \\
\text{Ci}(z) & \doteq - \int_z^\infty \! dt \, \frac{\cos t}{t}.
\end{align}
\end{subequations}

In terms of these integral formulas, we find.
\begin{align}
\dot{\mathcal{F}}_1 (E,\tau) = \frac{1}{4 \pi d} + \frac{|E|}{2\pi} \, \left[ \cos(2 d E) \text{si}(2 d |E|) - \sin(2 d |E|) \text{Ci} (2d |E|) \right] + O(h).
\label{eq:app:F0past1-final}
\end{align}

We proceed to estimate $\dot{\mathcal{F}}_2$, defined by \eqref{4:StaticRecMirF2}. Integrating by parts reduces the integral to a form that can be evaluated exactly in terms of cosine and sine integrals. We find 
\begin{align}
\dot{\mathcal{F}}_2 (E,\tau) &= \frac{1-h}{2 \pi} \left\{ -\frac{1}{B} + |E| \left[\sin (B |E|) \text{Ci} (B |E|) - \cos( B E) \text{si}(B |E|) \right] \right. \nonumber \\
& + 2\pi E \cos( B E)\Theta(-E) \Bigg\} , 
\label{eq:app:F0pf2-exact}
\end{align}
where $B \doteq 2d - \kappa^{-1}\ln(1-h)$. A~small $h$ expansion in \eqref{eq:app:F0pf2-exact}
gives 
\begin{align}
\dot{\mathcal{F}}_2 (E,\tau) &= - \frac{1}{4 \pi d} + \frac{|E|}{2\pi} \, \left[\sin (2d |E|) \text{Ci} (2d |E|) - \cos( 2d E) \text{si}(2d |E|) \right] \nonumber \\
& + E \cos( B E)\Theta(-E) + O(h). 
\label{eq:app:F0past2-final}
\end{align}


Combining 
\eqref{eq:app:F0past0-final}, 
\eqref{eq:app:F0past1-final}
and 
\eqref{eq:app:F0past2-final}, we have that
\begin{align}
\dot{\mathcal{F}}(E,\tau) & = -E \left[1 - \cos(2dE) \right] \Theta(-E) + O\left(\ee^{\kappa \tau} \right), & \text{ as } \tau \rightarrow - \infty, \label{4App:StaticRecMirPast}
\end{align}
which is eq. \eqref{4:StaticRecMirPast}.

\subsubsection*{Distant future, $\tau \to \infty$}

We start by estimating  $\dot{\mathcal{F}}_0$, defined by eq. \eqref{4:StaticRecMirF0}. We introduce the parameter $f\doteq 1/(1+\ee^{\kappa(\tau-d)})$, such that $f \to 0^+$ as $\tau \to \infty$.  We add and subtract $\kappa^2 \cos(E s) [8\pi\sinh^2(\kappa s/2)]^{-1}$ inside the integrand, and obtain
\begin{align}
\dot{\mathcal{F}}_0 (E,\tau) & = - E \Theta(-E) + \frac{1}{2\pi} \int_0^\infty ds \, \cos(E s) \! 
\left(  \frac{1}{s^2} - \frac{\kappa^2}{4\sinh^2(\kappa s/2)} \right) \nonumber \\
& + \frac{\kappa^2}{2\pi} \int_0^\infty ds \, \cos(E s) \! \left(  \frac{1}{4\sinh^2(\kappa s/2)} 
\right. \nonumber \\
& \left. - \frac{f^2 \, \ee^{\kappa s}} {[1 + f(\ee^{\kappa s} -1)] {\left\{\ln[1 + f(\ee^{\kappa s} -1)]\right\}}^2} \right). 
\label{eq:app:F0future1}
\end{align}

In the last term in \eqref{eq:app:F0future1}, the integrand goes to zero pointwise as $f\to0$. A monotone convergence argument shows that the integral is $o(1)$ as $f\to0$. The combination of half of the first term with the second term is handled by the same techniques that we have encountered in Chapter \ref{ch:DeCo}, in the context of the Unruh effect and the thermal heat bath, with the replacement $a \to \kappa$. Hence, we have that
\begin{align}
\dot{\mathcal{F}}_0 (E,\tau) = -\frac{E}{2}  \Theta (-E) + \frac{E}{2 \, (\ee^{2\pi E/\kappa}-1)}
+ o(1), \hspace{3ex} \text{as $f\to0$}.
\label{eq:app:F0future0-final}
\end{align}

By monotone convergence arguments, we find that $\dot{\mathcal{F}}_1 (E,\tau) = o(1)$ and $\dot{\mathcal{F}}_2 (E,\tau) = O(f)$. Putting these results together,
\begin{align}
\dot{\mathcal{F}}(E,\tau) & = -\frac{E}{2}  \Theta(-E) + \frac{E}{2 \left(\ee^{2 \pi E/\kappa} -1\right)} + \textit{o}(1), \hspace{3ex} \text{as } \tau \rightarrow + \infty. \label{4App:StaticRecMirFuture}
\end{align}
which is eq. \eqref{4:StaticRecMirFuture}.

\subsection{Travelling towards mirror in the distant past}

We estimate the behaviour of the transition rate in the distant past and the distant future.

\begin{subequations}
\begin{align}
\dot{\mathcal{F}}_0(E,\tau) & = -E \Theta (-E) + \frac{1}{2 \pi} \int_0^\infty \! ds \, \cos(Es) \left( \frac{\ee^{2\lambda} p'\left( \ee^\lambda \tau \right) p'\left(\ee^\lambda (\tau-s) \right)}{\left[p\left(\ee^\lambda \tau\right) - p \left(\ee^\lambda (\tau-s)\right) \right]^2} + \frac{1}{s^2}\right) \label{4app:BoostRecMirF0} \\
\dot{\mathcal{F}}_1(E,\tau) & = \frac{1}{2 \pi} \int_0^\infty \! ds \,  \frac{\cos(Es)  p'\left(\ee^\lambda (\tau-s) \right)}{\left[\ee^{-\lambda} \tau + d - p \left(\ee^\lambda (\tau-s) \right)\right]^2} \label{4app:BoostRecMirF1} \\
\dot{\mathcal{F}}_2(E,\tau) & = \frac{1}{2 \pi} \int_0^\infty \! ds \, \text{Re} \left( \frac{\ee^{-\ii E s} p'\left( \ee^\lambda \tau \right)}{\left[p\left(\ee^\lambda \tau \right) - \ee^{-\lambda} (\tau - s) - \ii \epsilon \right]^2} \right) \label{4app:BoostRecMirF2}.
\end{align}
\end{subequations}

\subsubsection*{Distant past, $\tau \to -\infty$}

We start by estimating $\dot{\mathcal{F}}_0$, defined by \eqref{4app:BoostRecMirF0}. We proceed as in \eqref{App4:F0statpast}--\eqref{eq:app:F0past0-final} and obtain 
\begin{align}
\dot{\mathcal{F}}_0 (E,\tau) &= - E \Theta(-E) + O\left(\ee^{\ee^\lambda\kappa\tau}\right).
\label{eq:app:asdrift:F0past0-final}
\end{align}

We proceed to estimate $\dot{\mathcal{F}}_1$, defined by \eqref{4app:BoostRecMirF1}. We introduce $g \doteq \ee^{\kappa \tau \ee^\lambda} \to 0^+$ as $\tau \to -\infty$, whereby we write \eqref{4app:BoostRecMirF1} as
\begin{align}
\dot{\mathcal{F}}_1 ( E,\tau) = \frac{1}{2\pi} \int_0^\infty \frac{\cos( E s) \, ds}{\left(1 + g \ee^{-\kappa s \ee^\lambda} \right) \left[ s \ee^\lambda - 2\tau \sinh\lambda + \kappa^{-1} \ln \! \left(1 + g \ee^{-\kappa s \ee^\lambda} \right) \right]^2}.
\label{eq:app:asdrift:F0past1-1}
\end{align}

Whenever $\tau<0$, we may bound the absolute value of $\dot{\mathcal{F}}_1 ( E,\tau)$ 
by the replacements $\cos( E s) \to1$ and $g\to0$ in \eqref{eq:app:asdrift:F0past1-1}. The integral can be evaluated in terms of sine and cosine integrals and one finds that $\dot{\mathcal{F}}_1 ( E,\tau) = O(\tau^{-1})$.

Finally, for estimating $\mathcal{F}_2$, defined by \eqref{4app:BoostRecMirF2}, the integral can be performed exactly, as in the static case. One finds that
\begin{align} 
\dot{\mathcal{F}}_2 (E,\tau) &=  \frac{(1-h) \, \ee^{2\lambda}}{2 \pi} \left\{ -\frac{1}{C} 
+ |E| \left[\sin (C |E|) \text{Ci} (C |E|) - \cos( C  E) \text{si}(C | E|) \right] \right. \nonumber \\
& + 2\pi  E \cos( C  E)\Theta(- E) \Bigg\}
\ , 
\label{eq:app:drift-F0pf2-exact}
\end{align}
where $h \doteq g/(1+g)$ and $C \doteq - (\ee^{2\lambda} -1)\tau - \kappa^{-1} \ee^\lambda \ln(1-h)$. As $\tau \to -\infty$, we have $C \to \infty$, and using formulas (6.2.17) and (6.12.3) in \cite{NIST} gives 
\begin{align}
\dot{\mathcal{F}}_2 (E,\tau) = \ee^{2\lambda} \, E \cos(2 \tau \sinh\lambda \,  \ee^\lambda  E)
\Theta(- E) + O(\tau^{-3}). 
\label{eq:app:asdrift:F0past2-final}
\end{align}

Combining, we obtain that
\begin{align}
\dot{\mathcal{F}}(E,\tau) & = -E \left[1 - \ee^{2\lambda} \cos\left(2\tau \sinh \lambda \ee^\lambda E \right) \right] \Theta(-E) + O\left(\tau^{-1} \right), & \text{ as } \tau \rightarrow - \infty, \label{4App:BoostRecMirPast} 
\end{align}
in the distant past, which is eq. \eqref{4:BoostRecMirPast}.

\subsubsection*{Distant future, $\tau \to \infty$}

We carry out the estimates as $\tau \to \infty$. In the case of $\dot{\mathcal{F}}_0$, defined by \eqref{4app:BoostRecMirF0}, we proceed as in the static case. Proceeding as in eq.  \eqref{eq:app:F0future1}, we obtain 
\begin{align}
\dot{\mathcal{F}}_0 ( E ,\tau) = -\frac{ E }{2} \Theta (- E ) + \frac{ E }{2 \, (\ee^{2\pi\ee^{-\lambda} E /\kappa}-1)} + o(1) \hspace{3ex} \text{as $\tau \rightarrow \infty$}.
\label{eq:app:asdrift:F0future0-final}
\end{align}

To estimate $\dot{\mathcal{F}}_1$ \eqref{4app:BoostRecMirF1} we substitute $s = \tau+r$ and write
\begin{align}
\dot{\mathcal{F}}_1 ( E ,\tau) = \frac{\kappa^2}{2\pi} \int_{-\tau}^\infty \frac{\cos[ E (\tau+r)] \, dr}{\left(1 + \ee^{-\kappa r \ee^\lambda} \right) \left[
\kappa \tau \ee^{-\lambda} + \ln \! \left(1 + \ee^{\kappa r \ee^\lambda} \right) \right]^2}.  
\label{eq:app:asdrift:F0future1-1}
\end{align}

For $\tau > 0$, the absolute value of~\eqref{eq:app:asdrift:F0future1-1} is bounded by the replacement $\cos[ E (\tau+r)] \to 1$ in the integrand. We then extend the integration interval to the full real axis in~$r$. Elementary estimates then show that the contribution from $-\infty < r < 0$ is $O(\tau^{-2})$ and the contribution from $0 < r < \infty$ is $O(\tau^{-1})$. Hence $\dot{\mathcal{F}}_1 ( E ,\tau) = O(\tau^{-1})$. 

Finally the asymptotic behaviour of $\dot{\mathcal{F}}_2$ \eqref{4:BoostRecMirF2} is obtained by expanding the closed-form expression \eqref{eq:app:drift-F0pf2-exact}. We obtain $\dot{\mathcal{F}}_2 ( E ,\tau) = O\left(\ee^{-\ee^\lambda \kappa\tau}\right)$. 

Combining, we have that
\begin{align}
\dot{\mathcal{F}}(E,\tau) & = -\frac{E}{2}  \Theta(-E) + \frac{E}{2 \left(\ee^{2 \pi \ee^{-\lambda} E/\kappa} -1\right)} + \textit{o}(1), & \text{ as } \tau \rightarrow + \infty, \label{4App:BoostRecMirFuture}
\end{align}
which is eq. \eqref{4:BoostRecMirFuture}.

\section{The Schwarzschild black hole}

\subsection{Static detector}

Using the static transition rate formula \eqref{DeCoStat}, eq. \eqref{4:SchwA} and the equations of motion, we have that
\begin{subequations}
\begin{align}
\dot{\mathcal{F}}_\text{H}(E) & = -\frac{1}{2\pi (2\pi T_\text{loc})^2}\int_{-\infty}^{\infty} \! ds \, \ee^{-\ii E s} \left(\frac{\ee^{-2\pi s T_\text{loc}}}{\left(\ee^{-2\pi T_\text{loc} (s-\ii \epsilon)} -1 \right)^2} \right), \label{4App:StatSchwHHI}\\
\dot{\mathcal{F}}_\text{U}(E) & = -\frac{1}{4\pi}\int_{-\infty}^{\infty} \! ds \, \ee^{-\ii E s} \left[\frac{\ee^{-2 \pi s T_\text{loc}}/(2\pi T_\text{loc})^2}{\left(\ee^{-2\pi T_\text{loc} (s-\ii \epsilon)} -1 \right)^2} + \frac{1}{(\epsilon + \ii s)^2}\right], \label{4App:StatSchwUnruh}
\end{align}
\label{4App:StatSchwF}
\end{subequations}
\\
where we have defined the constant $T_\text{loc} \doteq 1/\left[(8\pi M)(1-R/(2M))^{1/2}\right]$, which shall play the role of the local temperature registered by the detector.

In the HHI vacuum, the left and right moving modes give identical contributions. The fraction in eq. \eqref{4App:StatSchwHHI} can be rendered in the form of eq. \eqref{thm:DecoUnruh}, yielding
\begin{equation}
\dot{\mathcal{F}}_\text{H}(E) = \frac{E}{\ee^{E/T_{\text{loc}}}-1},
\end{equation}
which is eq. \eqref{4:StaticHHI}.

In the Unruh vacuum, the right-movers produce a thermal spectrum, while the left-movers, coming from infinity, yield a Minkowski vacuum-like term. We obtain
\begin{equation}
\dot{\mathcal{F}}_\text{U}(E) = -\frac{E}{2}\Theta(-E) + \frac{E}{2\left(\ee^{E/T_{\text{loc}}}-1\right)},
\end{equation}
which is eq. \eqref{4:StaticUnruh}.

\subsection{Inertial detector near the infinity}

The transition rate formula in the asymptotic past \eqref{FdotAsymp}, together with \eqref{4:SchwA} along the infalling geodesic reads
\begin{subequations}
\begin{align}
\dot{\mathcal{F}}_\text{H}(E,\tau) & = -E \Theta(-E) + 2 \int_0^{\infty} \! ds \, \cos(Es) \left[ \mathcal{A}_U(\tau,\tau-s) + \mathcal{A}_V(\tau,\tau-s) + \frac{1}{2 \pi s^2} \right], \\
\dot{\mathcal{F}}_\text{U}(E,\tau) & = -E \Theta(-E) + 2 \int_0^{\infty} \! ds \, \cos(Es) \left[\mathcal{A}_U(\tau,\tau-s) + \mathcal{A}_v(\tau,\tau-s) + \frac{1}{2 \pi s^2} \right],
\end{align}
\label{4App:SchwFH-FU}
\end{subequations}
\\
with
\begin{subequations}
\begin{align}
\mathcal{A}_u(\tau,\tau') \doteq - \frac{1}{4\pi} \frac{\dot{u}\left(\tau'\right) \dot{u}\left(\tau''\right)}{\left[u\left( \tau' \right) - u\left(  \tau'' \right) \right]^2}, \\
\mathcal{A}_U(\tau,\tau') \doteq - \frac{1}{4\pi} \frac{\dot{U}\left(\tau'\right) \dot{U}\left(\tau''\right)}{\left[U\left( \tau' \right) - U\left(  \tau'' \right) \right]^2},\\
\mathcal{A}_V(\tau,\tau') \doteq - \frac{1}{4\pi} \frac{\dot{V}\left(\tau'\right) \dot{V}\left(\tau''\right)}{\left[V\left( \tau' \right) - V\left(  \tau'' \right) \right]^2}
\end{align}
\label{4App:AsSchw}
\end{subequations}
\\
along the trajectories defined by the integrals of
\begin{subequations}
\begin{align}
\dot{u} = \dot{t} - \dot{r}^* = \frac{\varepsilon + (\varepsilon - 1 +2M/r)^{1/2}}{1-r/2M} = \frac{1}{\varepsilon - \left(\varepsilon^2 - 1 +2M/r \right)^{1/2}}, \\
\dot{v} = \dot{t} + \dot{r}^* = \frac{\varepsilon - (\varepsilon - 1 +2M/r)^{1/2}}{1-r/2M} = \frac{1}{\varepsilon + \left(\varepsilon^2 - 1 +2M/r \right)^{1/2}}.
\end{align}
\label{4App:Schw-udot-vdot}
\end{subequations}

For example, along the trajectories \eqref{4App:Schw-udot-vdot}, $\mathcal{A}_U$ reads
\begin{equation}
\mathcal{A}_U(\tau,\tau') = - \frac{\kappa^2}{8 \pi} \frac{\dot{u}\left( \tau' \right) \dot{u}\left( \tau'' \right)}{\sinh^2 \left[\kappa\left( u\left( \tau' \right) - u\left( \tau'' \right) \right)/2 \right]}.
\label{4App:AUExample}
\end{equation}

The first observation that we make is that the point-wise asymptotic past limits of eq. \eqref{4App:AsSchw}, at fixed positive $s$, are given by
\begin{subequations}
\begin{align}
\lim_{\tau \to - \infty} \mathcal{A}_u(\tau,\tau-s) & = - \frac{1}{4 \pi s^2}, \\
\lim_{\tau \to - \infty} \mathcal{A}_U(\tau,\tau-s) & = -\frac{\ee^{2\lambda}}{4 \pi (8M)^2 \sinh^2\left(\ee^\lambda s/(8M) \right)}, \\
\lim_{\tau \to - \infty} \mathcal{A}_V(\tau,\tau-s) & = -\frac{\ee^{-2\lambda}}{4 \pi (8M)^2 \sinh^2\left(\ee^{-\lambda} s/(8M) \right)}.
\end{align}
\label{4App:AsSchwLim}
\end{subequations}

If we can replace eq. \eqref{4App:AsSchwLim} in \eqref{4App:SchwFH-FU}, we obtain the left and right-moving thermal spectra in the HHI vacuum and the left-moving Minkowski vacuum and the right-moving thermal spectrum in the Unruh vacuum, to leading order at early times. What is left is to provide the monotone convergence argument to replace the limits inside the integral.

Let us start by letting $q$ denote $\dot{u}$ or $\dot{v}$. Letting $\eta = -1$ for $u$ and $\eta = +1$ for $v$, we write, $q$ and its time derivative, $\dot{q}$, as
\begin{align}
q & = \frac{1}{\varepsilon + \eta \left(\varepsilon^2 - 1 +2M/r \right)^{1/2}} \\
\dot{q} & = -\frac{\eta M}{r^2 \left(\varepsilon + \eta \left(\varepsilon^2 - 1 +2M/r \right)^{1/2}\right)^2} = -\frac{\eta M q^2}{r^2}
\end{align}

Then, in view of eq. \eqref{4App:AsSchw} and \eqref{4App:AUExample}, we neeed to show that the expressions
\begin{subequations}
\label{eq:app:monot1}
\begin{align}
& \frac{\int_{\tau-s}^\tau q(\tau') \, d\tau'}{\sqrt{q(\tau) q(\tau - s)}} 
\ , 
\\[1ex]
& \frac{\sinh \! \left(\frac{1}{8M} \int_{\tau-s}^\tau q(\tau') \, d\tau' \right)}{\sqrt{q(\tau) q(\tau - s)}} 
\end{align}
\end{subequations}
\\
are monotone in $\tau$ for all $s > 0$ when $\tau$ is sufficiently large and negative. Differentiating \eqref{eq:app:monot1} with respect to~$\tau$, it suffices to show that each of the expressions
\begin{subequations}
\label{eq:app:monot2}
\begin{align}
& \int_{\tau-s}^\tau q(\tau') \, d\tau'
- 2[q(\tau) - q(\tau-s)] \left(\frac{\dot q (\tau)}{q(\tau)} 
+ \frac{\dot q (\tau-s)}{q(\tau-s)} \right)^{-1}
\ , 
\label{eq:app:monot2a}
\\[1ex]
& 
\tanh \! \left(\frac{1}{8M} \int_{\tau-s}^\tau q(\tau') \, d\tau'
\right)
- \frac{1}{4M}[q(\tau) - q(\tau-s)] 
\left(\frac{\dot q (\tau)}{q(\tau)} + \frac{\dot q (\tau-s)}{q(\tau-s)} \right)^{-1}
\ , 
\label{eq:app:monot2b}
\end{align}
\end{subequations}
\\
has a fixed sign for all $s>0$ when $\tau$ is sufficiently large and negative.  
Introducing in \eqref{eq:app:monot2} 
a new integration variable by $p' = 
\sqrt{\varepsilon^2 - 1 + 2M/r(\tau')}\,$, we see that it suffices to show that 
each of the functions 
\begin{subequations}
\label{eq:app:monot3}
\begin{align}
f_1 (p) & = \frac12 \int_p^{p_f} \frac{dp'}{(\varepsilon + \eta p'){[{p'}^2 - \varepsilon^2 + 1]}^2} \nonumber \\
& - \frac{p_f - p}{(\varepsilon + \eta p){[p_f^2 - \varepsilon^2 + 1]}^2 + (\varepsilon + \eta p_f){[p^2 - \varepsilon^2 + 1]}^2},
\label{eq:app:monot3a}
\\
f_2 (p) & = \tanh \! \left( 
\frac12 \int_p^{p_f} \frac{dp'}{(\varepsilon + \eta p'){[{p'}^2 - \varepsilon^2 + 1]}^2} 
\right) \nonumber \\
& - \frac{p_f - p}{(\varepsilon + \eta p){[p_f^2 - \varepsilon^2 + 1]}^2 + (\varepsilon + \eta p_f){[p^2 - \varepsilon^2 + 1]}^2}
\ , 
\label{eq:app:monot3b}
\end{align}
\end{subequations}
\\
defined on the domain $\sqrt{\varepsilon^2 -1} < p < p_f$, where $p_f \in \left( \sqrt{\varepsilon^2 -1}\, , \varepsilon\right)$ is a parameter, has a fixed sign when $p_f$ is sufficiently close to $\sqrt{\varepsilon^2 -1}$. 

Consider~$f_1$. $f_1'$ is a rational function whose sign can be analysed by elementary methods,
with the outcome that $f_1'$ is negative when $p_f$ is sufficiently close to $\sqrt{\varepsilon^2-1}$. 
Hence $f_1$ is positive when $p_f$ is sufficiently close to $\sqrt{\varepsilon^2-1}$. 

Consider then~$f_2$. When $p_f$ is sufficiently close to $\sqrt{\varepsilon^2-1}$, an elementary analysis shows that the second term in \eqref{eq:app:monot3b} is negative and strictly increasing, and there is a $p_1 \in \bigl( \sqrt{\varepsilon^2 -1}\, , p_f\bigr)$ such that this term takes the value $-1$ at $p = p_1$. With $p_f$ this close to $\sqrt{\varepsilon^2-1}$, it follows that $f_2$ is negative for $p\le p_1$, whereas for $p_1 < p < p_f$ $f_2$ has the same sign as
\begin{align}
f_3(p) & = \frac12 \int_p^{p_f} \frac{dp'}{(\varepsilon + \eta p'){[{p'}^2 - \varepsilon^2 + 1]}^2} \nonumber \\
& - \arctanh \! \left( \frac{p_f - p}{(\varepsilon + \eta p){[p_f^2 - \varepsilon^2 + 1]}^2 + (\varepsilon + \eta
  p_f){[p^2 - \varepsilon^2 + 1]}^2} \right). 
\label{eq:app:monot4b}
\end{align}

The function $f_3$ can be analysed by the same methods as~$f_1$, with the outcome that $f_3$ is negative when $p_f$ is  sufficiently close to $\sqrt{\varepsilon^2-1}$.  Collecting, we see that $f_2$ is negative when $p_f$ is sufficiently close to $\sqrt{\varepsilon^2-1}$. 

This completes the monotone convergence argument and we find that, in the HHI vacuum
\begin{equation}
\dot{\mathcal{F}}_\text{H}(E) = \frac{E}{2\left(\ee^{E/T_-}-1\right)} + \frac{E}{2\left(\ee^{E/T_+}-1\right)} + \textit{o}(1),
\end{equation}
which is eq. \eqref{4:InftyHHI}, whereas in the Unruh vacuum
\begin{equation}
\dot{\mathcal{F}}_\text{U}(E) = -\frac{E}{2}\Theta(-E) + \frac{E}{2\left(\ee^{E/T_+}-1\right)} + \textit{o}(1),
\end{equation}
which is eq. \eqref{4:InftyUnruh}.

\subsection{Inertial detector near the singularity}

We consider the trajectories described by \eqref{4:SchwGeodesicsE>1}, \eqref{4:SchwGeodesicsE=1}, \eqref{4:SchwGeodesicsE<1} and \eqref{4:SchwGeodesicsWH} in the HHI and Unruh vacua in Schwarzschild spacetime. We take the switch-off to occur at time $\tau$ in region II and the switch-on at time $\tau_i$, which can be pushed to the asymptotic past when $\varepsilon \geq 1$. We call $\tau_\text{sing}$ the proper time at which the trajectory hits the black hole singularity. Let $\tau_1$ be a constant such that the detector is somewhere in region II at proper time $\tau_1$. As $\tau \to \tau_\text{sing}$, we have
\begin{subequations}
\begin{align}
\dot{\mathcal{F}}_\text{H}(E,\tau) & = G_U(E,\tau,\tau_1) + G_V(E,\tau,\tau_1) + O(1), \\
\dot{\mathcal{F}}_\text{U}(E,\tau) & = G_U(E,\tau,\tau_1) + G_v(E,\tau,\tau_1) + O(1)
\end{align}
\label{4App:SchwFH-FUSing}
\end{subequations}
\\
where
\begin{subequations}
\begin{align}
G_v(E,\tau,\tau_1) & \doteq 2 \int_{\tau_1}^{\tau} \! d\tau' \, \cos\left[E\left(\tau-\tau'\right)\right] \left[ \mathcal{A}_v\left(\tau,\tau'\right) + \frac{1}{4 \pi \left(\tau-\tau'\right)^2} \right], \\
G_U(E,\tau,\tau_1) & \doteq 2 \int_{\tau_1}^{\tau} \! d\tau' \, \cos\left[E\left(\tau-\tau'\right)\right] \left[ \mathcal{A}_U\left(\tau,\tau'\right) + \frac{1}{4 \pi \left(\tau-\tau'\right)^2} \right], \\
G_V(E,\tau,\tau_1) & \doteq 2 \int_{\tau_1}^{\tau} \! d\tau' \, \cos\left[E\left(\tau-\tau'\right)\right] \left[ \mathcal{A}_V\left(\tau,\tau'\right) + \frac{1}{4 \pi \left(\tau-\tau'\right)^2} \right],
\label{4:APPSchwGsingV}
\end{align}
\label{4App:SchwGsing}
\end{subequations}
\\
with $\mathcal{A}_v$, $\mathcal{A}_U$ and $\mathcal{A}_V$ given by eq. \eqref{4App:AsSchw} along the trajectories defined in region II by the integrals of
\begin{subequations}
\begin{align}
\dot{r} & = - \left(\varepsilon^2 - 1 +2M/r \right)^{1/2}, \label{4App:SchwSing1Estr} \\
\dot{v} & = \frac{1}{\varepsilon + \left(\varepsilon^2 - 1 +2M/r \right)^{1/2}} = \frac{1}{\varepsilon - \dot{r}}. \label{4App:SchwSing1Estv}
\end{align}
\label{4App:SchwSing1Est}
\end{subequations}

From eq. \eqref{4App:SchwSing1Estr}, we have that, in the near-singularity regime, the proper time and the Schwarzschild globally-defined coordinate $r$ are related by the equation $d\tau = dr \left[ -(r/(2M))^{1/2} + O\left(r^{3/2}\right)\right]$. This relation integrates to
\begin{equation}
\tau_\text{sing}-\tau = \frac{(2/M)^{1/2}}{3} r^{3/2} + O\left(r^{5/2}\right).
\label{4App:r-tau-sing}
\end{equation}

It follows from eq. \eqref{4App:SchwSing1Estv} that
\begin{equation}
\ddot{v} = - \frac{M}{r^2 \left(\varepsilon + \sqrt{\varepsilon^2 - 1 +2M/r} \right)^2} = -\frac{M}{r^2}\dot{v}(\tau).
\end{equation}

In Appendix D of \cite{Juarez-Aubry:2014jba}, it is explicitly shown how to obtain the asymptotic behaviour for $G_v$, given by
\begin{equation}
G_v(E,\tau,\tau_1) = \frac{1}{16 \pi M} \left[\left(\frac{2M}{r(\tau)} \right)^{3/2} - \varepsilon \left(\frac{2M}{r(\tau)} \right) + \frac{1 + \varepsilon^2}{2}\left(\frac{2M}{r(\tau)} \right)^{1/2} \right] + O(1).
\label{4App:Gvfinal}
\end{equation}
We take the oportunity to illustrate how this is achieved for $G_V$. We seek to isolate the leading behaviour by integrating eq. \eqref{4:APPSchwGsingV} by parts
\begin{align}
G_V(E,\tau,\tau_1) & = \left[ \lim_{\tau' \to \tau} \frac{1}{2\pi} \cos\left[E\left(\tau-\tau'\right)\right] \left( -\frac{\dot{V}(\tau)}{V(\tau) - V\left(\tau'\right)} + \frac{1}{\tau-\tau'} \right)  + O(1)\right] \nonumber \\
& - \frac{E}{2\pi} \int_{\tau_1}^{\tau} \! d\tau' \, \sin\left[E\left(\tau-\tau'\right)\right] \left(-\frac{\dot{V}(\tau)}{V(\tau)-V\left(\tau'\right)} + \frac{1}{\tau-\tau'}\right),
\label{4App:GV}
\end{align}
where the $O(1)$ term is the boundary evaluation at $\tau_1$. The limit of the term in square brackets is
\begin{align}
\lim_{\tau' \to \tau} \frac{1}{2\pi} \cos\left[E\left(\tau-\tau'\right)\right] \left( -\frac{\dot{V}(\tau)}{V(\tau) - V\left(\tau'\right)} + \frac{1}{\tau-\tau'} \right) = -\frac{\ddot{V}(\tau)}{4\pi \dot{V}(\tau)} \\
 = -\frac{1}{4 \pi} \left(\frac{\ddot{v}(\tau)}{\dot{v}(\tau)} + \frac{\dot{v}(\tau)}{4M} \right) = -\frac{1}{4\pi}\left(-\frac{M}{r^2} + \frac{1}{4M} \right) \dot{v}
\end{align}

From eq. \eqref{4App:SchwSing1Estv}, the $\dot{v}$ is a quantity $\dot{v} = O\left( r^{1/2} \right)$, and we have that the limit is of order $O\left( r^{3/2} \right)$. From here, it follows that the integral term in \eqref{4App:GV} is $O(1)$ and making an expansion of $\dot{v}$ as a function of $r$, we have that $G_V - G_v = O(1)$,
\begin{equation}
G_V(E,\tau,\tau_1) = \frac{1}{16 \pi M} \left[\left(\frac{2M}{r(\tau)} \right)^{3/2} - \varepsilon \left(\frac{2M}{r(\tau)} \right) + \frac{1 + \varepsilon^2}{2}\left(\frac{2M}{r(\tau)} \right)^{1/2} \right] + O(1).
\label{4App:GVfinal}
\end{equation}

Finally, the estimation of $G_U$ is identical to that of $G_V$ with the replacement $\varepsilon \to - \varepsilon$,
\begin{equation}
G_U(E,\tau,\tau_1) = \frac{1}{16 \pi M} \left[\left(\frac{2M}{r(\tau)} \right)^{3/2} + \varepsilon \left(\frac{2M}{r(\tau)} \right) + \frac{1 + \varepsilon^2}{2}\left(\frac{2M}{r(\tau)} \right)^{1/2} \right] + O(1).
\label{4App:GUfinal}
\end{equation}

Putting our results together \eqref{4App:Gvfinal}, \eqref{4App:GVfinal} and \eqref{4App:GUfinal}, we find that for both the HHI and Unruh vacua,
\begin{equation}
\dot{\mathcal{F}}(E,\tau) = \frac{1}{8\pi M} \left[\left(\frac{2M}{r(\tau)}\right)^{3/2} + \frac{1+E^2}{2}\left(\frac{2M}{r(\tau)}\right)^{1/2} \right] + O(1),
\end{equation}
which is eq. \eqref{4:NearSing}, and it follows from \eqref{4App:r-tau-sing} that the transition rate diverges to leading term as $1/[6\pi(\tau_{\text{sing}}-\tau)]$.

\section{The generalised Reissner-Nordstr\"om black hole}

We start this section by writing the geodesic equations in different coordinate systems in region II. Starting from eq. \eqref{4:RNeqmotion2} and the definition of the coordinates $(\tilde{u}, \tilde{v})$

\begin{subequations}
\begin{align}
\dot{\tilde{u}} = \dot{\tilde{r}}^* - \dot{\tilde{t}} = \frac{\left(\varepsilon^2-F(r) \right)^{1/2} + \varepsilon}{F(r)} = -\frac{1}{\left(\varepsilon^2 - F(r) \right)^{1/2} - \varepsilon} \\
\dot{\tilde{v}} = \dot{\tilde{r}}^* + \dot{\tilde{t}} = \frac{\left(\varepsilon^2-F(r) \right)^{1/2} - \varepsilon}{F(r)} = - \frac{1}{\left(\varepsilon^2 - F(r) \right)^{1/2} + \varepsilon}
\end{align}
\label{4:RNeqmotionNull} 
\end{subequations}
\\
where $\dot{\tilde{r}}^* = \dot{\tilde{r}}/F(r)$. The Kruskal coordinates $(U,V)$ in region II are
\begin{subequations}
\begin{align}
\dot{U} = \partial_\tau\left[ \left(-U_-\right)^{\kappa_+/\kappa_-}\right] = -\kappa_+ \dot{\tilde{u}} \, U, \\
\dot{V} = \partial_\tau\left[ \left(-V_-\right)^{\kappa_+/\kappa_-}\right] = -\kappa_+ \dot{\tilde{v}} \, V.
\end{align}
\label{4:RNeqmotionKruskal} 
\end{subequations}

The derivatives of \eqref{4:RNeqmotionNull} and \eqref{4:RNeqmotionKruskal} will be useful for our asymptotic analysis:
\begin{subequations}
\begin{align}
\ddot{\tilde{u}} & = \frac{F'(r)}{2} \dot{\tilde{u}}^2, \\
\ddot{\tilde{v}} & = \frac{F'(r)}{2} \dot{\tilde{v}}^2 \\
\ddot{U} & =  -\kappa_+ \dot{\tilde{u}}^2 \left(\frac{F'(r)}{2} - \kappa_+ \right) U \\
\ddot{V} & = -\kappa_+ \dot{\tilde{v}}^2 \left(\frac{F'(r)}{2} - \kappa_+ \right) V
\end{align}
\label{4:RNeqmotionddot} 
\end{subequations}

\subsection{Boundary terms in eq. (4.56)}

The boundary terms in eq. \eqref{4:HHIlim} are given by
\begin{subequations}
\begin{align}
B_1(E, \tau, \tau_0) & \doteq - 2 \cos\left(E \Delta \tau\right)\partial_\tau \mathcal{W}_\text{H}(\tau,\tau_0), \label{4App:B1}\\
B_2(E, \tau) & \doteq \lim_{\tau' \rightarrow \tau }2 \cos\left(E(\tau-\tau')\right) \left[\partial_\tau \mathcal{W}_\text{H}(\tau,\tau')+\frac{1}{2 \pi (\tau-\tau')} \right]. \label{4App:B2}
\end{align}
\end{subequations}

Let us start by estimating the order of $B_1$. We have that
\begin{equation}
B_1(E,\tau,\tau_0) = \frac{1}{2\pi} \cos (E\Delta \tau) \left[\frac{\dot{U}(\tau)}{U(\tau)-U(\tau_0)} +  \frac{\dot{V}(\tau)}{V(\tau)-V(\tau_0)} \right].
\end{equation}

Using eq. \eqref{4:RNeqmotionNull} and \eqref{4:RNeqmotionKruskal}, we see that the $V$-contribution of $B_1$ is $O(1)$ close to the Cauchy horizon. Then, the asymptotic behaviour of eq. \eqref{4App:B1} near $\mathcal{C}^\text{FL}$ is
\begin{align}
B_1(E,\tau,\tau_0) & = \frac{\cos (E\Delta \tau)}{2\pi}  \frac{\dot{U}(\tau)}{U(\tau)}\left[ 1 + \frac{U(\tau_0)}{U(\tau)} + O\left( \frac{1}{U(\tau)^2} \right) \right] + O(1).
\label{4App:B1Exp}
\end{align}

The factor multiplying the brackets is
\begin{align}
& \frac{\cos (E\Delta \tau)}{2\pi}   \frac{\dot{U}(\tau)}{U(\tau)} = -\frac{\cos (E\Delta \tau) \kappa_+}{2\pi} \dot{\tilde{u}}(\tau) \nonumber \\
 & = -\frac{\cos (E\Delta \tau) \kappa_+}{2\pi} \left[\frac{2 E}{ f(r_-)(r(\tau)-r_-)} + \left(-\frac{1}{2 E} - \frac{2 E f'(r_-)}{f(r_-)^2} \right) + O(r(\tau)-r_-)\right] \nonumber \\
 & = -\frac{\cos (E\Delta \tau) \kappa_+}{2\pi} \left( \frac{1}{ \kappa_-(\tau_h-\tau)} + O(1)\right).
\end{align}
where we have used the equation of motion \eqref{4:RNr-eq} in the last line. The linear correction term in the expansion inside the brackets \eqref{4App:B1Exp} can be estimated as
\begin{align}
\frac{U(\tau)}{U(\tau_0)} & = \exp \left[- \kappa_+\int_{\tau_0}^\tau \! d\tau' \, \dot{\tilde{u}}(\tau') \right] = \exp \left[\kappa_+\int_{\tau_0}^\tau \! d\tau' \, \left( -\frac{1}{\kappa_-(\tau_h-\tau)} + O(1) \right) \right] \nonumber \\
& = \exp \left[\ln \left( \left(\frac{\tau_h-\tau}{\tau_h-\tau_0} \right)^{\kappa_+/\kappa_-} \right) + O(1) \right] =  C(\tau,\tau_0) \left(\frac{\tau_h - \tau}{\tau_h-\tau_0} \right)^{\kappa_+/\kappa_-},
\label{4App:Uestimate}
\end{align}
where $C(\tau, \tau_0)>0$ is a factor of order $O(1)$. We conclude that
\begin{equation}
B_1(E,\tau,\tau_0)  = -\frac{\cos (E\Delta \tau) \kappa_+/\kappa_-}{2\pi}   \left( 1 + O\left(\tau_h-\tau \right)^{-\kappa_+/\kappa_-} \right) \frac{1}{ \tau_h-\tau} + O(1).
\label{4App:B1Final}
\end{equation}

The second boundary term, given by \eqref{4App:B2}, is
\begin{align}
B_2(E, \tau) & = -\frac{1}{2\pi} \lim_{\tau' \rightarrow \tau } \left[\frac{\dot{U}(\tau)}{U(\tau) - U(\tau')} - \frac{1}{\tau-\tau'} \right] + O(1) = -\frac{1}{4\pi} \frac{\ddot{U}(\tau)}{ \dot{U}(\tau)} + O(1),
\end{align}
and can be estimated by the same methods to yield
\begin{equation}
B_2(E, \tau) = \frac{\kappa_+/\kappa_- -1}{4 \pi} \frac{1}{\tau_h-\tau} + O(1).
\label{4App:B2Final}
\end{equation}

Putting eq. \eqref{4App:B1Final} and \eqref{4App:B2Final} together, we obtain eq. \eqref{4:HHIlim2}.

\subsection{Integral term in eq. (4.57)}

We start by splitting the integration limits in eq. \eqref{4:RNIint} as
\begin{align}
I(E, \tau, \tau_0) & = \int_{0}^{\tau_h-\tau_0} \! d s \, \frac{ \sin \left(E s\right)}{1-U(\tau-s)/U(\tau)} + \int_{\tau_h-\tau_0}^{\tau-\tau_0} \! d s \, \frac{ \sin \left(E s\right)}{1-U(\tau-s)/U(\tau)}.
\label{4App:RNsplit}
\end{align}

The second integral above is subleading because
\begin{align}
\left| \int_{\tau_h-\tau_0}^{\tau-\tau_0} \! d s \, \frac{ \sin \left(E s\right)}{1-U(\tau-s)/U(\tau)} \right| & \leq \left| \int_{\tau_h-\tau_0}^{\tau-\tau_0} \! d s \,  \frac{ 1}{1-U(\tau_0)/U(\tau)}  \right| \nonumber \\
& \leq \frac{\tau-\tau_h}{1-U(\tau_0)/U(\tau)} = O(\tau_h-\tau).
\end{align}

The first term on the right hand side of eq. \eqref{4App:RNsplit} can be estimated by writing $U(\tau) = (\tau_h - \tau)^{-A} H(\tau_h-\tau)$ where $A>0$ and $H$ is positive and smooth everywhere in our integration interval, including $\tau = 0$. In our case, $A = -\kappa_+/\kappa_-$. We now make use of the dominated convergence theorem in order to estimate the integral. Let $0< \epsilon \doteq \tau_h - \tau$, we are interested in the limit $\epsilon \rightarrow 0^+$ of
\begin{equation}
\int_{0}^{\tau_h-\tau_0} \! d s \, \frac{ \sin \left(E s\right)}{1-\cfrac{(\epsilon+s)^{-A}H(\epsilon+s)}{\epsilon^{-A}H(\epsilon)}}.
\label{app:int1}
\end{equation}

An integrable function that dominates this integrand can by found with the aid of the following
\begin{lem}
Let $M>0$, let $H: [0,M] \to \mathbb{R}^+$ be~$C^1$, and let $A>0$. 
Suppose $\partial_x [x^{-A} H(x)] < 0 $ for $x \in (0,M]$.
Then there exists a constant $B \in (0,A)$ such that 
$\partial_x [x^{-B} H(x)] < 0 $ for $x \in (0,M]$.
\label{Lem:RN}
\end{lem}
\begin{proof}
Consider on $[0,M]$ the auxiliary function 
$g(x) =  x H'(x)/H(x)$. 
As $g$ is continuous, 
it attains on $[0,M]$ a maximum value, which we denote by~$C$. 
Since 
$\partial_x [x^{-A} H(x)] < 0 $ for $x \in (0,M]$, 
it follows that $g(x) < A $ for $x \in (0,M]$, 
and since $g(0)=0$ and $A>0$, it further follows that 
$g(x) < A $ for $x \in [0,M]$. Hence $C < A$. From 
$g(0)=0$ it follows that $C\ge0$. 
Define now $B \doteq \frac12(A+C)$. Then $0 \le C < B < A$. Since 
$g(x) < B $ for $x \in [0,M]$, it follows that 
$\partial_x [x^{-B} H(x)] < 0 $ for $x \in (0,M]$. 
\end{proof}

As $U(\tau)$ is a strictly decreasing function of~$\tau$, 
Lemma B.1 provides a $B\in(0,A)$ such that 
$(\epsilon+s)^{-B} H(\epsilon+s)$ 
is a strictly deacreasing function of $s$ over the domain 
of integration in \eqref{app:int1}. We may hence rearrange \eqref{4App:RNsplit} and \eqref{app:int1} as 
\begin{align}
I(E,\tau,\tau_0) = \int_{0}^{\tau_h-\tau_0} \! d \tau' \,  \frac{\sin \left(E s\right)}{1 - \cfrac{(\epsilon+s)^{-B} H(\epsilon+s)}{\epsilon^{-B} H(\epsilon)} \cfrac{(\epsilon+s)^{-A+B}}{\epsilon^{-A+B}}} + O(\tau_h-\tau).
\label{4App:PreDom}
\end{align}

Fixing a positive $\epsilon'$, the absolute value of the integrand in \eqref{4App:PreDom} is bounded for $\epsilon \le \epsilon'$ by 
\begin{align}
& \left| 
\sin(E s) 
\left( 
1 - \frac{(\epsilon+s)^{-B} H(\epsilon+s)}{\epsilon^{-B} H(\epsilon)}
\frac{\epsilon^{A-B}}{(\epsilon+s)^{A-B}}
\right)^{-1} 
\right| 
\notag
\\
& \ \ 
\le 
|E| s 
\left( 
1 - \frac{\epsilon^{A-B}}{(\epsilon+s)^{A-B}}
\right)^{-1} 
\notag 
\\
& \ \ = 
|E| s 
\left( 
1 - \frac{1}{\bigl[1+(s/\epsilon)\bigr]^{A-B}}
\right)^{-1} 
\notag 
\\
& \ \ \le
|E| s
\left( 
1 - \frac{1}{\bigl[1+(s/\epsilon')\bigr]^{A-B}}
\right)^{-1} 
\ , 
\end{align}
which bound is independent of $\epsilon$ and integrable over the 
domain in \eqref{4App:PreDom}. 
By dominated convergence we may hence take the 
limit $\epsilon\to0$ in (B.64) under the integral, obtaining 
\begin{align}
I(E, \tau, \tau_0) & = \int_{0}^{\tau_h-\tau_0} \! d s \, \sin \left(E s\right) +\mathit{o}(1) = - \frac{1 - \cos(E \Delta \tau)}{E} + \mathit{o}(1),
\end{align}
which is eq. \eqref{4:RNintResult}.

\subsection{The $\varepsilon = 0$ geodesic.}

The $\varepsilon = 0$ trajectory is given by the orbit of
\begin{subequations}
\begin{align}
\dot{\tilde{u}} = \dot{\tilde{v}} = - \frac{1}{(-F(r))^{1/2}} = \frac{1}{\dot{r}}.
\end{align}
\end{subequations}

The transition rate of a detector \eqref{4:HHIlim} along this trajectory can be estimated near the horizon using the near-horizon relation $(r - r_-)^{1/2} + O\left((r-r_-)^{3/2}\right)= (-\kappa_-/2)^{1/2}(\tau_h-\tau)$, leading to
\begin{align}
\dot{\tilde{u}}(\tau) = \dot{\tilde{v}}(\tau) = \frac{1}{\dot{r}(\tau)} = \frac{1}{\kappa_-(\tau_h-\tau)} + O\left((\tau_h-\tau)^3\right).
\end{align}

The same techniques as above are applicable and we find that
\begin{equation}
\dot{F}_\text{H}(E,\tau,\tau_0) = -\frac{1}{2 \pi} \left(1 +\frac{\kappa_+}{\kappa_-} +\textit{o}(1) \right) \frac{1}{\tau - \tau_h}.
\label{4App:FdotUp}
\end{equation}

\subsection{Local energy density near the Cauchy horizon}

We compute the local energy density along the worldline of an infalling observer approaching $\mathcal{C}^\text{FL}$, to leading order, in the near-horizon regime.

For the moment, we leave the test function, $\chi$, unspecified, with the requirement that $\chi$ is compact and its support lies in the past of the future Cauchy horizon. We start by writing eq. \eqref{4:RNWorldlineTab} explicitly along the worldline in terms of the conformal factor and its proper time derivatives succinctly  as
\begin{align}
\langle \rho \rangle^\text{ren}[\chi  \dot{\gamma} \otimes  \dot{\gamma}] & = - \frac{1}{12 \pi} \int_{\supp(\chi)} \! d \tau' \, \chi(\tau') \left[ \left(\frac{\ddot{U}(\tau')}{\dot{U}(\tau')} + \frac{\ddot{V}(\tau')}{\dot{V}(\tau')} \right) \frac{\dot{\Omega}(\tau')}{\Omega(\tau')} \right. \nonumber \\
& \left. +  4 \left(\frac{\dot{\Omega}(\tau')}{\Omega(\tau')}\right)^2 - 2 \frac{\ddot{\Omega}(\tau')}{\Omega(\tau')}\right] + O(1),
\end{align}
where
\begin{subequations}
\begin{align}
\dot{\Omega}/\Omega & =  \left(\frac{\kappa_+ + F'(r)/2}{F(r)}\right)\dot{r}, \\
\ddot{\Omega}/\Omega & =  \dot{r}^2 \left(\frac{\kappa_+ + F'(r)/2}{F(r)}\right)^2 - \frac{F'(r)}{2} \left(\frac{\kappa_+ + F'(r)/2}{F(r)}\right) \nonumber \\
& -  \dot{r}^2 \left(  \frac{(\kappa_+ + F'(r)/2)F'(r) }{F(r)^2} - \frac{F''(r)}{2F(r)} \right).
\end{align}
\label{Omega}
\end{subequations}

In view of eq. \eqref{Omega}, we can define the $O(1)$ smooth function $G$ given by

\begin{align}
G(\tau') & \doteq F[r(\tau')]^2 \left[ \left(\frac{\ddot{U}(\tau')}{\dot{U}(\tau')} + \frac{\ddot{V}(\tau')}{\dot{V}(\tau')} \right) \frac{\dot{\Omega}(\tau')}{\Omega(\tau')} +  4 \left(\frac{\dot{\Omega}(\tau')}{\Omega(\tau')}\right)^2 - 2 \frac{\ddot{\Omega}(\tau')}{\Omega(\tau')}\right] \nonumber \\
& = 4 \varepsilon^2 ( \kappa_+ + \kappa_-)^2 + O(\tau_h-\tau).
\end{align}

Defining the functionals
\begin{subequations}
\begin{align}
G_0[\chi] & \doteq \int_{\supp(\chi)} \! d \tau' \, \frac{\chi(\tau')}{F^2[r(\tau')]}, \label{4App:G0} \\
G_1[\chi] & \doteq \int_{\supp(\chi)} \! d \tau' \, \chi(\tau') \frac{\tau_h -\tau'}{F^2[r(\tau')]}, \label{4App:G1}
\end{align}
\end{subequations}
\\
we can write the local energy density as
\begin{equation}
\langle \rho \rangle^\text{ren}[\chi  \dot{\gamma} \otimes  \dot{\gamma}] = -\frac{\varepsilon^2 ( \kappa_+ + \kappa_-)^2}{3 \pi} G_0[\chi] + -\frac{\dot{G}(\tau_h)}{12\pi} G_1[\chi] + O(1).
\label{4App:rhoExpansion}
\end{equation}

The leading behaviour comes from the integral defining $G_0[\chi]$, while a subleading divergent behaviour arises from $G_1[\chi]$. We now specify the switching function and compute the leading behaviour.

\subsubsection{Sharp switching function}

Let us consider the sharp function $\chi_0(\tau') = \Theta(\tau-\tau')\Theta(\tau'-\tau_0)$. We start out by writing out the integral in eq. \eqref{4App:G0} in terms of the $r$ coordinate,
\begin{align}
G_0[\chi_0] & = \int_{r_0}^{r} \frac{d r'}{\dot{r}(r')}\, \frac{ f(r')^{-2}}{(r'-r_-)^2}.
\end{align}

The leading singular behaviour can be obtained via integration by parts,
\begin{align}
G_0[\chi_0] & = \left[ - \frac{1}{f(r')^2 \dot{r}(r')} \frac{1}{r'-r_-}  \right]_{r_0}^r + \int_{r_0}^r \frac{dr'}{r'-r_-} \frac{d}{dr} \left( \frac{1}{f(r')^2 \dot{r}(r')} \right) \nonumber \\
& = \frac{1}{4 \varepsilon (\kappa_-)^2} \left[1 + O((r-r_-)\ln\left(1-\frac{r_-}{r}\right) \right] \frac{1}{r-r_-}.
\end{align}

The contribution for $G_1$ is logarithmic and, hence, doesn't need to be estimated to leading order. In terms of the proper time difference, $\tau_h-\tau$, one obtains that eq. \eqref{4App:rhoExpansion} yields
\begin{equation}
\langle \rho \rangle^\text{ren}[\chi_0  \dot{\gamma} \otimes  \dot{\gamma}] = -\frac{(1 + \kappa_+/\kappa_-)^2}{12 \pi }\left[1 + O\left((\tau_h-\tau) \ln \left(1 - \frac{\tau}{\tau_h} \right)\right)\right]\frac{1}{\tau_h - \tau}.
\end{equation}
which is eq. \eqref{4:RNSharpEnergy}.

\subsubsection{Smooth switching function}

In the case of the smooth test function, we have that
\begin{equation}
\langle \rho \rangle^\text{ren}[\chi  \dot{\gamma} \otimes  \dot{\gamma}] = \langle \rho \rangle^\text{ren}[\chi_0  \dot{\gamma} \otimes  \dot{\gamma}] + \langle \rho \rangle^\text{ren}[\chi_{\text{off}} \, \dot{\gamma} \otimes  \dot{\gamma}] + O(1).
\end{equation}
and we wish to estimate $\langle \rho \rangle^\text{ren}[\chi  \dot{\gamma} \otimes  \dot{\gamma}]$ as the switch-off time $\tau$ comes close to the horizon-crossing time, $\tau_h$.

The calculation boils down to estimating the quantity
\begin{align}
G_0[\chi_{\text{off}}] & = \int_\tau^{\tau_h} \! d \tau' \, \frac{\chi_{\text{off}}(\tau')}{F^2[r(\tau')]} = \int_{\tau}^{\tau_h} \! d \tau' \, \frac{h_1((\tau_h - \tau')/(\tau_h-\tau))}{F^2[r(\tau')]}.
\end{align}

Integration by parts yields
\begin{align}
G_0[\chi_{\text{off}}] & = \left[ - \frac{h_1\left( \frac{\tau_h-\tau(r')}{\tau_h-\tau}\right)}{f(r')^2 \dot{r}(r')} \frac{1}{r'-r_-}  \right]_{r}^{r_-} + \int_{r}^{r_-} \frac{dr'}{r'-r_-} \frac{d}{dr} \left( \frac{h_1\left( \frac{\tau_h-\tau(r')}{\tau_h-\tau}\right)}{f(r')^2 \dot{r}(r')} \right).
\end{align}

The boundary term contributes as $-G_0[\chi_0] + O\left((\tau_h-\tau) \ln \left(1 - \frac{\tau}{\tau_h} \right)\right) (\tau_h-\tau)^{-1}$, while the integral term gives rise to a linear divergence stemming from the derivation acting on the argument of $h_1$,

\begin{align}
G_0[\chi_0] + G_0[\chi_{\text{off}}] & =  \frac{1}{\tau_h-\tau} \left[ (\tau_h-\tau) \int_\tau^{\tau_h} \frac{d\tau'}{r(\tau')-r_-} \frac{d}{d \tau'} \left(\frac{h_1\left( \frac{\tau_h-\tau')}{\tau_h-\tau}\right)}{f[r(\tau')]^2 \dot{r}[r(\tau')]}\right) \right.
\nonumber \\
& \left. + O\left((\tau_h-\tau) \ln \left(1 - \frac{\tau}{\tau_h} \right)\right) \right] \nonumber \\
& = \frac{1}{4 \varepsilon^2 (\kappa_-)^2}\frac{1}{\tau_h-\tau} \left[\int_\tau^{\tau_h}  d\tau'  \frac{h_1'\left( \frac{\tau_h-\tau'}{\tau_h-\tau}\right)}{\tau_h-\tau'} \right. \nonumber \\
& \left. + O\left((\tau_h-\tau) \ln \left(1 - \frac{\tau}{\tau_h} \right)\right) \right].
\end{align}

The integral factor can be expressed more succinctly in terms of the integration variable $x = (\tau_h-\tau')/(\tau_h-\tau)$ and one finds that, to leading order
\begin{equation}
\rho[\chi] = -\frac{(1 + \kappa_+/\kappa_-)^2}{12 \pi } \left[\int_0^1 \! dx \, \frac{h_1'(x)}{x} + O\left((\tau_h-\tau) \ln \left(1 - \frac{\tau}{\tau_h} \right)\right)\right]\frac{1}{\tau_h - \tau},
\end{equation}
which is eq. \eqref{4:RNSmoothEnergy}.

\subsection{The $1+1$ Rindler horizon}

In this appendix we show that a detector following the trajectory $(t,x) = (\tau, \tau_h)$, as it interacts with a massless scalar in the Rindler state has an instantaneous transition rate given by eq. \eqref{4:1+1} in the vicinity of the future Rindler horizon. 

The pullback of the Wightman function along the trajectory is $\mathcal{W}_\text{R}(\mathsf{x},\mathsf{x}') = -(1/(4\pi)) \ln[(\epsilon + \ii \Delta u_R)(\epsilon + \ii \Delta v_R)]$, where the trajectory in Rindler coordinates is given by
\begin{subequations}
\begin{align}
u_R(\tau) & = -\frac{1}{a} \ln[a(\tau_h - \tau)], \\
v_R(\tau) & = -\frac{1}{a} \ln[a(\tau_h + \tau)].
\end{align}
\label{4App:rindler-eqmotion}
\end{subequations}

An integration by parts of the transition rate formula leads to
\begin{align}
\dot{\mathcal{F}}_\text{1+1}(E, \tau, \tau_0) & = - 2 \cos\left(E \Delta \tau\right)\partial_\tau \mathcal{W}_\text{R}(\tau,\tau_0)\nonumber \\
& + \lim_{\tau' \rightarrow \tau }2 \cos\left(E(\tau-\tau')\right) \left[\partial_\tau \mathcal{W}_\text{R}(\tau,\tau')+\frac{1}{2 \pi (\tau-\tau')} \right] \nonumber \\
& - 2 \int_{\tau_0}^\tau \! d \tau' \, E \sin \left(E(\tau-\tau')\right) \left[\partial_\tau \mathcal{W}_\text{R}(\tau,\tau')+\frac{1}{2 \pi (\tau-\tau')} \right] + O(1),
\label{4App:rinlim}
\end{align}

The leading contribution to the transition rate in eq. \eqref{4App:rinlim} comes from the left-moving sector. The boundary terms in eq. \eqref{4App:rinlim} yield
\begin{align}
\dot{\mathcal{F}}_\text{1+1}(E, \tau, \tau_0) & = \left[-\frac{1}{4 \pi} + O\left(\frac{1}{\ln(1-\tau/\tau_h)} \right) \right] \frac{1}{\tau_h - \tau} +  \nonumber \\
& + \left[\frac{E \dot{u}_R(\tau)}{2 \pi} \int_{\tau_0}^\tau \! d \tau' \, \frac{\sin \left(E(\tau-\tau')\right)}{u_R(\tau)-u_R(\tau')} + O(1) \right].
\label{4App:rinlim2}
\end{align}

Eq. \eqref{4App:rinlim2} strongly suggests that the divergent behaviour is linear, as stated in chapter \ref{ch:BH}, section \ref{sec:4Rindler}. In order to corroborate that this is the case, one needs only estimate the second term on the right hand side of eq. \eqref{4App:rinlim2} and verify that it is indeed subleading. The factor multiplying the integral is $\dot{u}_R(\tau) = (a(\tau_h-\tau))^{-1}$. We now show that the contribution from the integral 
\begin{equation}
I(E,\tau,\tau_0) \doteq a^{-1} \int_{\tau_0}^\tau \! d \tau' \, \frac{\sin(E(\tau-\tau')}{u_R(\tau)-u_R(\tau')} = \int_{\tau_0}^\tau \! d \tau' \, \frac{\sin(E(\tau-\tau')}{\ln\left(\frac{\tau_h-\tau'}{\tau_h-\tau} \right)}.
\label{4App:Irin}
\end{equation}
is at most $\mathit{o}(1)$. Changing variables to $p = (\tau-\tau')/\tau_h$ and defining $\epsilon \doteq 1-\tau/\tau_h$, $\alpha \doteq 1 - \tau_0/\tau_h$ and $\beta \doteq E \tau_h$, with $0< \epsilon$, $0 < \alpha < 2$ and $\beta \in \mathbb{R}$, we can write
\begin{equation}
I(E,\tau,\tau_0) = \tau_h \int_0^{\alpha-\epsilon} \! d p \, \frac{\sin(\beta p)}{\ln(1+ p/\epsilon)},
\end{equation}
and the upper limit of the integral above can be extended to $\alpha$ because
\begin{align}
\left| \int_{\alpha-\epsilon}^{\alpha} \! dp \, \frac{\sin(\beta p)}{\ln(1+p/\epsilon)} \right| & \leq  \int_{\alpha-\epsilon}^{\alpha} \! dp \, \frac{1}{\ln(1+p/\epsilon)} \leq \frac{\epsilon}{\ln[1+(\alpha - \epsilon)/\epsilon]} \nonumber \\
& = \frac{\epsilon}{\ln(\alpha/\epsilon)} = O(\epsilon/\ln(\epsilon)).
\end{align}

Next, we bound the integral \eqref{4App:Irin} as
\begin{align}
\left| \tau_h \int_0^{\alpha} \! d p \, \frac{\sin(\beta p)}{\ln(1+ p/\epsilon)}\right| \leq \tau_h |\beta| \int_0^{\alpha} \! d p \, \frac{p}{\ln(1+ p/\epsilon)} 
& = \tau_h |\beta| \epsilon^2 \int_0^{\alpha/\epsilon}  \! d r \, \frac{r}{\ln(1+ r)},
\end{align}
where in the last line we have performed the change of variables $r = p/\epsilon$. This integral can now be written in terms of the exponential integral Ei and the logarithmic integral li, whose asymptotic behaviour is known,
\begin{align}
I & \leq \tau_h |\beta| \epsilon^2 \text{Ei}[2 \ln(1 + \alpha/\epsilon)] - \ln 2 - \text{li}(1 + \alpha/\epsilon) + O(\epsilon/\ln(\epsilon)) \nonumber \\
& = - \frac{\alpha^2 |\beta| \tau_h}{2} \frac{1}{\ln (\epsilon/\alpha)} + O\left( \frac{1}{[\ln (\epsilon/\alpha)]^2} \right).
\label{4App:rinresult}
\end{align}

The estimate in eq. \eqref{4App:rinresult} leads to
\begin{align}
\dot{\mathcal{F}}_\text{1+1}(E, \tau, \tau_0) & = \left[-\frac{1}{4 \pi} + O\left(\frac{1}{\ln(1-\tau/\tau_h)} \right) \right] \frac{1}{\tau_h - \tau},
\end{align}
which is the transition rate near the $1+1$ Rindler horizon displayed in eq. \eqref{4:1+1}.

\chapter{Supplement to Chapter 5}
\label{app:Unruh}

This appendix contains a series of technical results that are used throughout the Chapter \ref{ch:Unruh}.

\section{A collection of lemmas}

In this section, we collect a series of useful lemmas.

\begin{lem}
Let $\psi \in C_0^\infty(\mathbb{R})$ and $\chi \in C_0^\infty(\mathbb{R})$ be related by eq. \eqref{chiconst}. Then the Fourier transforms $\hat{\psi} = \mathcal{F}[\psi]$ and $\hat{\chi} = \mathcal{F}[\chi]$ are related by $\hat{\chi}(\omega) = (\ii/\omega)\left(\ee^{-\ii\omega(\Delta\tau + \Delta \tau_s)} -1 \right) \hat{\psi}(\omega)$.
\label{LemChi}
\end{lem}
\begin{proof}
If $g \in C_0^\infty(\mathbb{R})$, then also $g' \in C_0^\infty(\mathbb{R})$. It follows from integration by parts that $\mathcal{F}[g'](\omega) = \int d \tau \ee^{-\ii \omega \tau} g'(\tau) = \ii \omega \mathcal{F}[g](\omega)$. Setting $f = g'$, it follows from $g(\tau) = \int_{-\infty}^\tau \! d\tau' \, g'(\tau')$ that
\begin{equation}
\mathcal{F}\left[\int_{-\infty}^\cdot \! d\tau' \, f(\tau') \right](\omega) = -\frac{\ii}{\omega} \mathcal{F}[f](\omega).
\end{equation}

Setting $\psi_r(\tau) \doteq \psi(\tau-r)$, a direct computation yields the desired result,

\begin{equation}
\hat{\chi}(\omega) = -\frac{\ii}{\omega} \mathcal{F}[\psi - \psi_r](\omega) = -\frac{\ii}{\omega} \int \! d\tau \, \ee^{-\ii \omega \tau} \left(\psi(\tau) - \ee^{-\ii \omega r} \psi(\tau) \right) = \frac{\ii}{\omega} \left(\ee^{-\ii\omega r}-1\right) \hat{\psi}(\omega).
\end{equation}
\end{proof}

\begin{lem}
Let $f_s: \mathbb{R}^2 \rightarrow \mathbb{R}$ be
\begin{equation}
f_s(u,v) \doteq \frac{u + v}{\ee^{u+v} - 1} + \frac{u - v}{\ee^{u - v} - 1} - \frac{2 u}{\ee^{u} - 1},
\label{fs}
\end{equation}
where the values at $u= v$ and $u = -v$ are understood in the limiting sense, $f_s$ is even in each of its arguments and $0 \leq f_s(u,v) \leq f_s(0,v)$.
\label{Lem1}
\end{lem}
\begin{proof}
Evenness in each of the arguments is a direct computation. To prove that $0 \leq f(u,v) \leq f(0,v)$, it suffices to assume $u > 0$, $v > 0$.

Firstly, notice that $f_s(0,v) = 2[(v/2) \coth(v/2)-1] >0$. Secondly, we have that $\lim_{u \rightarrow \infty} f_s(u,v) = 0$. It hence suffices to show that $\partial_u f_s(u,v) \neq 0$. Defining
\begin{subequations}
\begin{align}
f_1(u,v) & \doteq -\frac{2}{\ee^u-1} + \frac{1}{\ee^{u-v}-1} \left(1 + \frac{v}{1-\ee^{-(u-v)}} \right) + \frac{1}{\ee^{u+v}-1} \left(1- \frac{v}{1-\ee^{-(u+v)}} \right), \\
f_2(u,v) & \doteq \frac{2\ee^u}{(\ee^u-1)^2} -  \frac{\ee^{u-v}}{(\ee^{u-v}-1)^2} - \frac{\ee^{u+v}}{(\ee^{u+v}-1)^2},
\end{align}
\end{subequations}
\\
one can write that $\partial_u f_s(u,v) = f_1(u,v) + u f_2(u,v)$. At fixed $v$, we define the function $f_0:\mathbb{R}\setminus{u_*} \rightarrow \mathbb{R}$ by $f_0(u) \doteq f_1(u,v)/f_2(u,v)+u$, where $u_{*} = \text{arccosh}[(-1+\sqrt{5+4 \cosh v})/2]$ is a root of $f_2$. Verifying that $\partial_u f_s(u,v) \neq 0$  is equivalent to verifying that $f_0(u,v) \neq 0$. We start by writing $f_0$ in the form
\begin{equation}
f_0(u) = \sinh(u) \frac{-\cosh(u) + \cosh(v) - v \coth(v/2) (\cosh(u)-1)}{\cosh^2(u) + \cosh(u) -  \cosh(v) - 1} + u,
\end{equation}
in the intervals $(0,u_{*})$ and $(u_{*}, \infty)$. First, we see that $\lim_{u \to 0^+} f_0(u,v)  = 0$. Second, computing the derivative of $f_0$,
\begin{align}
\partial_u f_0(u) & = \frac{\left(\cosh(v) - \cosh(u)\right) (\cosh(u)-1)}{\left[-1 + \cosh(u) +  \cosh^2(u) - \cosh(v)\right]^2} \Big( v [1 + 2 \cosh(u)] \coth(v/2) \nonumber \\
& -(2 + 2 \cosh(u) + \cosh^2(u) + \cosh(v)) \Big) ,
\label{App5:Lem2Aux}
\end{align}
we can verify that for $u \in (0,u_{*})$, $\partial_u f_0(u)< 0$. To see this, one first verifies that the term inside the big parentheses in \eqref{App5:Lem2Aux} is always negative.\footnote{To this end, first verify that the term in parenthesis is negative at $u = 0$ and as $u \to \infty$. Then, verify that the derivative vanishes only at $\cosh(u) = v \coth(v/2)-1$, and that the term in parenthesis evaluated at this (maximum) point is a negative function of $v$.} Then, one realises that the numerator multiplying the big parentheses is positive if $u < v$ and negative for $u > v$. Thus, the derivative \eqref{App5:Lem2Aux} is negative for $u < v$ and positive for $u > v$. Because $u_{*} < v$, we conclude that in the interval $(0,u_{*})$, $f_0(u) < 0$.

For the interval $u \in (u_{*}, \infty)$, we first notice that, at fixed $v$, $\lim_{u \to (u_*)^+} f_0(u,v)  >0$ and $\lim_{u \to \infty} f_0(u,v)  = \infty$. Then, by the sign analysis of eq. \eqref{App5:Lem2Aux} that we have explained above, $u = v$ is a minimum of $f_0(u)$ in the interval $(u_{*}, \infty)$. Hence, $f_0(u)\geq 0$ in this interval.

All that is left from the analysis above is to study the point $u = v$ at the level of $\partial_u f_s(u,v)$ to verify that it has fixed sign. One finds that $\partial_u f_s(u,v)|_{u = v} < 0$, which concludes the analysis.
\end{proof}

\begin{lem}
The response function \eqref{Flambda}, with $\widehat{\W}$ defined as in eq. \eqref{eq:Unruh-hatWstat}, satisfies $\mathcal{F}_{T_\text{U}}(E) \leq \mathcal{F}_\lambda(E)/\lambda \leq \mathcal{F}_{T_\text{U}}(E) + (24 \pi a)^{-1}\lambda^{-2} \left|\left| \omega \hat{\chi}(\omega) \right|\right|^2_{L^2}$, where $||\cdot||_{L^2}$ is the $L^2$ norm with respect to the Lebesgue measure.
\label{Lem2}
\end{lem}
\begin{proof}
Symmetrising the integrals in \eqref{Flambda} and \eqref{FT}, we have that
\begin{align}
\mathcal{F}_\lambda(E)/\lambda - \mathcal{F}_{T_\text{U}}(E) & = \frac{a}{(2 \pi)^3} \int_{0}^{\infty} \! d\omega \, \left|\hat{\chi}(\omega)\right|^2 f_s(2 \pi E/a,2 \pi\omega/(\lambda a)),
\label{Fsym}
\end{align}
\\
where $f_s$ is defined as in eq. \eqref{fs}. Using Lemma \ref{Lem1} it follows that
\begin{equation}
0 \leq \mathcal{F}_\lambda(E)/\lambda - \mathcal{F}_{T_\text{U}}(E) \leq \frac{a}{(2 \pi)^3} \int_{0}^{\infty} \! d\omega \, \left|\hat{\chi}(\omega)\right|^2 f_s(0,2 \pi\omega/(\lambda a)).
\end{equation}

One can verify that $f_s(0,v) = 2(v/2) \coth(v/2)-2 \leq (2/3)(v/2)^2$, which implies that
\begin{equation}
0 \leq \mathcal{F}_\lambda(E)/\lambda - \mathcal{F}_{T_\text{U}}(E) \leq \frac{1}{12 \pi a \lambda^2} \int_{0}^{\infty} \! d\omega \, \omega^2 \left|\hat{\chi}(\omega)\right|^2 = \frac{1}{24 \pi a \lambda^2}\left|\left| \omega \hat{\chi}(\omega) \right|\right|^2_{L^2}.
\end{equation}
\end{proof}

\begin{lem}
Let $h: \mathbb{R}^{+}\times\mathbb{R}^{+} \rightarrow \mathbb{R}: (u,v) \mapsto h(u,v) = (u - v)/(\cosh(u)-\cosh(v))$, where $h$ at $u=v$ is understood in the limiting sense. $h$ is strictly decreasing in each of its arguments.
\label{LemBound1}
\end{lem}
\begin{proof}
It is immediate to see that $h$ is symmetric in the two arguments, $h(u,v) = h(v,u)$. Fix $u >0$, then $\lim_{v \rightarrow 0} h(u,v) = u/(\cosh(u)-1) > 0$. Also, $\lim_{v \rightarrow \infty} h(u,v) = 0$. It is left to show that the partial derivative in $v$ is non-positive. The derivative is given by
\begin{equation}
\partial_v h(u,v) = \frac{\cosh(v) -\cosh(u) +(u-v) \sinh(v) }{(\cosh(u) - \cosh(v))^2},
\label{DLemBound1}
\end{equation}
understood at $u = v$ in the limiting sense. An elementary analysis shows that the numerator in \eqref{DLemBound1} is negative for $u \neq v$, and examination of the $u \to v$ limit in \eqref{DLemBound1} (for example, by L'H\^opital's rule) shows that $\partial_v h(u,v)|_{u = v} < 0$. 
\end{proof}

\begin{lem}
Let $\left(\epsilon_k\right)_{k \in \mathbb{N}}$ be a strictly positive sequence such that
\begin{equation}
\sum_{k = 1}^\infty \ee^{-\beta \epsilon_k} < \infty
\label{LemSumFinite}
\end{equation}
with $\beta >0$ and let $\delta_k = C \, \ee^{-\alpha \epsilon_k}$ for some $C>0$ and $\alpha > \beta$. Define $\epsilon_{-k} = -\epsilon_k$ and $\delta_k = \delta_{-k}$ for $k \in \mathbb{N}$. If $F: \mathbb{R} \rightarrow \mathbb{C}$ is a locally integrable polynomially bounded function with support ${\rm supp}(F) \subset S = \cup_{k \in \mathbb{Z}\setminus 0} \left\{\epsilon: \left|\epsilon - \epsilon_k \right| < \delta_k \right\} \cup (-\delta_0, \delta_0)$, for some $\delta_0 > 0$, then the inverse Fourier transform of $F$, 
\begin{equation}
\mathcal{F}^{-1}[F](z) = \frac{1}{2 \pi} \int_{-\infty}^\infty \! d\epsilon \, \ee^{\ii z \epsilon} F(\epsilon),
\label{LemFourierExists}
\end{equation}
is analytic in the strip $|\text{Im}(z)| < \alpha - \beta$.
\label{LemMain2}
\end{lem}

\begin{proof}
Since $F$ is polynomially bounded, there exist $D > 0$ and $N>0$ such that $|F(\epsilon)| \leq D (1 + |\epsilon|^N)$. We therefore have the estimates
\begin{align}
2 \pi \left| \mathcal{F}^{-1}[F](z) \right| & \leq \int_{-\infty}^\infty \! d\epsilon \, \left| \ee^{\ii \epsilon z} F(\epsilon)\right| \leq 2 D \int_{0}^\infty \! d\epsilon \, \ee^{|\text{Im}(z)| |\epsilon|} \left(1+|\epsilon|^N \right) \nonumber \\
& \leq 2 D \sum_{k = 1}^\infty (2 \delta_k ) \ee^{|\text{Im}(z)| (\epsilon_k + \delta_k)} \left(1+(\epsilon_k + \delta_k)^N \right) \nonumber \\
& \leq 4 C D \sum_{k = 1}^\infty \left(1+(\epsilon_k + \delta_k)^N \right) \ee^{|\text{Im}(z)| (\epsilon_k + \delta_k)-\alpha \epsilon_k}.
\label{boundFourier}
\end{align}

From \eqref{LemSumFinite} it follows that $\epsilon_k \to \infty$ as $k\to\infty$, and hence 
$\delta_k \to 0$ as $k\to\infty$. If $|\text{Im}(z)| < \alpha - \beta$, 
for sufficiently large $k$ we may hence estimate the terms in \eqref{boundFourier} by 
\begin{align}
& \left(1 + {(\epsilon_k+\delta_k)}^N\right)
\ee^{|\text{Im}(z)|(\epsilon_k+\delta_k) - \alpha\epsilon_k}
\notag
\\
& \ \ \le \left(1 + {(\epsilon_k+1)}^N\right)
\ee^{|\text{Im}(z)|(\epsilon_k+1) - \alpha\epsilon_k}
\notag
\\
& \ \ = \ee^{|\text{Im}(z)|} \left(1 + {(\epsilon_k+1)}^N\right)
\ee^{ - (\alpha - \beta - |\text{Im}(z)|)\epsilon_k} 
\, 
\ee^{- \beta\epsilon_k}
\notag
\\
& \ \ \le 
\ee^{\alpha - \beta} \, 
\ee^{- \beta\epsilon_k}
\  . 
\end{align}
This shows that $\mathcal{F}^{-1}[F](z)$ exists in the strip 
$|\text{Im}(z)| < \alpha - \beta$. 

To show that $\mathcal{F}^{-1}[F](z)$ is analytic in the strip 
$|\text{Im}(z)| < \alpha - \beta$, we may use the inequality 
\begin{align}
\left|\frac{\ee^{\ii (z+h)\epsilon} 
- \ee^{\ii z\epsilon}}{h} 
- \ii \epsilon \ee^{\ii z\epsilon}\right| 
\le 2 |\epsilon| \ee^{(|\text{Im}(z)| + |h| )\epsilon}
\ , 
\end{align}
valid for $h\in\mathbb{C}\setminus\{0\}$, 
together with estimates similar to those above, to provide 
a dominated convergence argument that justifies differentiating 
\eqref{LemFourierExists} under the integral sign, with the outcome 
\begin{align}
\frac{d}{dz} \mathcal{F}^{-1}[F](z) 
= 
\frac{\ii}{2\pi} \int_{-\infty}^\infty d\epsilon \, 
\epsilon \, \ee^{\ii z \epsilon} F(\epsilon)
\ . 
\end{align}
\end{proof}

\begin{cor}
{\rm (Follows from Lemma \ref{LemMain2}.)} Suppose $g \in C_0^\infty(\mathbb{R})$ obeys $\left| \hat{g}(\epsilon) \right| \leq K \ee^{-\gamma |\epsilon|}$ for some $K, \gamma > 0$ and all $\epsilon \in \mathbb{R} \setminus S$, with $S$ the union of intervals as in Lemma \ref{LemMain2}. Then $g = 0$.
\label{Cor}
\end{cor}
\begin{proof}
Let $\varphi$ be the characteristic function of $S$. Then $F = \varphi \hat{g}$ satisfies the conditions of Lemma \ref{LemMain2} and has an inverse Fourier transform $\mathcal{F}^{-1}[\varphi \hat{g}]$ which is analytic on a strip, say $S_1$, around the real axis. As, $|(1-\varphi(\epsilon))\hat{g}(\epsilon)| \leq K \ee^{-\gamma |\epsilon|}$ for all $\epsilon \in \mathbb{R}$, $\mathcal{F}^{-1}[(1-\varphi)\hat{g}]$ exists and is analytic on a strip $S_2$ around the real axis. Thus, $g$ extends from the real axis to an analytic function in the strip $S_1 \cap S_2$. But since $g$ has compact support on the real axis, $g$ has to vanish everywhere.
\end{proof}

\section{Bounds for (5.33) in Theorem 5.10}

We provide the following bounds:

\subsection{Bound for $\mathcal{G}_{(0,\kappa)}$}

We estimate
\begin{equation}
\mathcal{G}_{(0,\kappa)}(E) = \frac{a}{(2 \pi)^4 (E/a)}  \int_0^\kappa \! d\omega \, \left| \hat{\chi}(\omega) \right|^2 g(2\pi E/a, 2 \pi \omega/(\lambda a)),
\label{G0K}
\end{equation}
from above, with
\begin{equation}
g(u,v) = \frac{u + v}{\ee^{v} - \ee^{-u}} + \frac{u - v}{\ee^{- v} - \ee^{-u}} - \frac{2 u}{1 - \ee^{-u}}.
\label{app5:g}
\end{equation}
with the time scale $\lambda$ scaling as a function of $E>0$ by $\lambda(E) = \alpha(2\pi E/a)^{1+p}$.

We bound $g\left((u,v/\left[\alpha u^{1+p}\right]\right)$, where we have set $u = 2\pi E/a$ and $v = 2\pi \omega/a$:

\begin{align}
g\left(u,v/\left[\alpha u^{1+p}\right]\right) & \leq \frac{u + v/(\alpha u^{1+p})}{1-\ee^{-u}}  +\frac{u - v/(\alpha u^{1+p})}{\ee^{-v/(\alpha u^{1+p})}-\ee^{-u}} - \frac{2u}{1-\ee^{-u}} \nonumber \\
& = \left( u - \frac{v}{\alpha u^{1+p}}\right) \left(\frac{1}{\ee^{-v/(\alpha u^{1+p})}-\ee^{-u}} -\frac{1}{1-\ee^{-u}}\right) \nonumber \\
& \leq \left( u - \frac{v}{\alpha u^{1+p}}\right) \left(\frac{1}{\ee^{-(2 \pi \kappa/a)/ (\alpha u^{1+p})}-\ee^{-u}} -\frac{1}{1-\ee^{-u}}\right),
\end{align}
where in the last equality $E$ is assumed so large that $u > v/(\alpha u^{1+p})$.

On the last line, the term in the second parentheses is equal to 
\begin{align}
\frac{2\pi\kappa/a}{\alpha u^{1+p}} 
\left[ 1 + O \! \left(\frac{2\pi\kappa/a}{\alpha u^{1+p}}\right)\right]
\ , 
\end{align}
and the last line is hence 
\begin{align}
\frac{2\pi\kappa/a}{\alpha u^{p}}
+ O \! \left(u^{-(1+2p)}\right)
\ .  
\end{align}
Using this estimate in \eqref{G0K}, and extending 
the integration range in $\omega$ to be $(0,\infty)$, we find 
\begin{equation}
\mathcal{G}_{(0,\kappa)}(E) \leq \frac{\kappa/\alpha}{(2 \pi)^3 (E/a)} \left( \left(\frac{2\pi E}{a}\right)^{-p} + O\left(\frac{2 \pi E}{a}\right)^{-(1+2p)} \right)  \left|\left|\hat{\chi} \right|\right|_{L^2}.
\label{boundpiece1}
\end{equation}

\subsection{Bound for $\mathcal{G}_{(\kappa, E \lambda)}$}

We provide a bound of
\begin{equation}
\mathcal{G}_{(\kappa, E \lambda)}(E) = \frac{a}{(2 \pi)^4 (E/a)}  \int_\kappa^{E \lambda} \! d\omega \, \left| \hat{\chi}(\omega) \right|^2 g(2\pi E/a, 2 \pi \omega/(\lambda a)),\end{equation}
from above, with $g$ given by \eqref{app5:g}.

We set again $u = 2\pi E/a$ and $v = 2\pi \omega/a$ and write the last factor under the integral as $g\left(u,v/\left(\alpha u^{1+p}\right)\right)$. In this case, we have that $v \in (2 \pi \kappa/a, \alpha(2 \pi E/a)^{2+p})$. We start by writing $g$ as
\begin{align}
g(u,v/(\alpha u^{1+p})) & =  \left( u - \frac{v}{\alpha u^{1+p}}\right) \left(\frac{1}{\text{e}^{-v/(\alpha u^{1+p})}-\text{e}^{- u}} -\frac{1}{\text{e}^{v/(\alpha u^{1+p})} -\text{e}^{- u}}\right) \nonumber \\
& + 2u \left(\frac{1}{\ee^{v/(\alpha u^{1+p})} -  \ee^{-u}} - \frac{1}{1- \ee^{-u}} \right).
\label{gprebound}
\end{align}

The second term is strictly decreasing in $v$, thus bounded from above by zero. Hence we have that.

\begin{align}
g(u,v/(\alpha u^{1+p})) & \leq  \left( u - \frac{v}{\alpha u^{1+p}}\right) \left(\frac{1}{\text{e}^{-v/(\alpha u^{1+p})}-\text{e}^{- u}} -\frac{1}{\text{e}^{v/(\alpha u^{1+p})} -\text{e}^{- u}}\right) \nonumber \\
& =  \left( u - \frac{v}{\alpha u^{1+p}}\right) \frac{\text{e}^u \sinh (v/(\alpha u^{1+p}))}{\cosh u-\cosh (v/(\alpha u^{1+p}))}.
\label{gbound}
\end{align}

Following Lemma \ref{LemBound1}, we have that
\begin{align}
\mathcal{G}_{(\kappa, E \lambda)} & \leq  \frac{a \, \ee^{2 \pi E/a}}{(2 \pi)^3\left[\cosh(2 \pi E/a)-1 \right]} \int_\kappa^{E \lambda} \! d\omega \, \left| \hat{\chi}(\omega) \right|^2 \sinh \left( \frac{2 \pi \omega}{\lambda a} \right)
\label{Gest1}
\end{align}

Notice that the prefactor of the integral is of order $O(1)$. This is easy to see as $\ee^{2 \pi E/a}(\cosh(2 \pi E/a)-1)^{-1} \leq 2/[1-2\exp(-2\pi E/a)]$. 

We split the integration in the subintervals $(\kappa, \omega_0)\cup(\omega_0, E \lambda)$, where $\omega_0 \doteq (a/2\pi) \alpha (2\pi E/a)^{\left(1 + p + q^{-1}\right)/2 }$. We have that
\begin{align}
\int_\kappa^{\omega_0} \! d\omega \, & \left| \hat{\chi}(\omega) \right|^2 \sinh \left( \frac{2 \pi \omega}{\lambda a} \right)   \leq || \hat{\chi}||_{L^2} \sinh \left( \frac{2 \pi \omega_0}{\lambda a} \right)  \nonumber \\ & = || \hat{\chi}||_{L^2} \sinh \left( \left(\frac{2\pi E}{a}\right)^{-\left(1+p-q^{-1}\right)/2} \right)  \nonumber \\
& = ||\hat{\chi}||_{L^2} \left(\frac{2\pi E}{a}\right)^{-\left(1+p-q^{-1}\right)/2} \left[1 + O\left(\left(\frac{2\pi E}{a}\right)^{-\left(1+p-q^{-1}\right)}\right) \right],
\end{align} 
where we have extended the integration limits to obtain the $L^2$ norm of $\hat{\chi}$. 

We now proceed to estimate the contribution of the $(\omega_0, E \lambda)$ integration interval to \eqref{Gest1}. At this stage, we use the assumption that $\chi \in F$ is of strong Fourier decay, whereby
\begin{equation}
|\hat{\chi}(\omega)|  \leq \kappa^{-1} C (B+|\omega/\kappa|)^r \exp(-A |\omega/\kappa|^q).
\end{equation}

Thus, we have that for $\omega \in (\omega_0, E\lambda)$
\begin{align}
\left| \hat{\chi}(\omega) \right|^2 & \sinh \left( \frac{2 \pi \omega}{\lambda a} \right)  \leq \frac{C^2}{\kappa^{2}} \left(B+ \frac{\omega}{\kappa} \right)^{2r} \exp\left(-2A \left(\frac{\omega}{\kappa}\right)^q + \frac{2 \pi \omega}{a \lambda} \right) \nonumber \\
& \leq \frac{C^2}{\kappa^{2}} \left(B+ \frac{\omega}{\kappa} \right)^{2r} \exp \left( -2A \left(\frac{\omega_0}{\kappa}\right)^q + \frac{2 \pi E}{a} \right) \nonumber \\
& = \frac{C^2}{\kappa^{2}} \left(B+ \frac{\omega}{\kappa} \right)^{2r} \exp \left( -2A \left(\frac{a \alpha}{2 \pi \kappa}\right)^q \left(\frac{2 \pi E}{a} \right)^{\left(1 + q( 1 + p) \right)/2} + \frac{2 \pi E}{a} \right).
\end{align}

The first term in the exponent is negative and dominates at large positive $E$ because $q(1 + p) > 1$ and $2 A[(a \alpha)/(2 \pi \kappa)]^{q} > 1$ by our assumptions. Thus the factor coming from the exponent dies off faster than any polynomial at large $E$. Furthermore, because the integral
\begin{equation}
\int_{\omega_0}^{E\lambda} \! d\omega \, \left(B+ \frac{\omega}{\kappa} \right)^{2r}
\end{equation}
has only polynomial growth, the overall contribution in the $(\omega_0, E \lambda)$ integration interval dies off faster than any polynomial. 

Collecting  the estimates,
\begin{equation}
\mathcal{G}_{(\kappa, E \lambda)} \leq  \frac{2 a ||\hat{\chi}||_{L^2} }{(2 \pi)^3} \left(\frac{2\pi E}{a}\right)^{-\left(1+p-q^{-1}\right)/2} \left[1 + O\left(\left(\frac{2\pi E}{a}\right)^{-\left(1+p-q^{-1}\right)}\right) \right].
\label{boundpiece2}
\end{equation}

\subsection{Bound for $\mathcal{G}_{(E \lambda,\infty)}$}

We bound
\begin{equation}
\mathcal{G}_{(E \lambda,\infty)}(E) = \frac{a}{(2 \pi)^4 (E/a)}  \int_{E \lambda}^\infty \! d\omega \, \left| \hat{\chi}(\omega) \right|^2 g(2\pi E/a, 2 \pi \omega/(\lambda a)).\end{equation}
from above, with $g$ given by \eqref{app5:g}. We set again $u = 2\pi E/a$ and $v = 2\pi \omega/a$ and write the last factor under the integral as $g\left(u,v/\left(\alpha u^{1+p}\right)\right)$. In this case, we have that $v \in (E \lambda, \infty)$. 

The steps \eqref{gprebound} - \eqref{gbound} are applicable also for $\mathcal{G}_{(E \lambda,\infty)}$, and we use the bound \eqref{gbound} to write
\begin{equation}
\mathcal{G}_{(E \lambda, \infty)} \leq  \frac{1}{(2 \pi)^3 (E/a)} \int_{E \lambda}^\infty \! d \omega \, \left|\hat{\chi}(\omega)\right|^2 \left( E - \frac{\omega}{\lambda}\right) \frac{\text{e}^{2 \pi E/a} \sinh (2 \pi \omega /(a\lambda))}{\cosh (2 \pi E/a) -\cosh (2 \pi \omega/(a \lambda) )}.
\label{app5:G3}
\end{equation}

Set $1 + p < s < q(2 +p) + p$. There is a guarantee that such $s$ exists because, by our hypotheses, $q(2+p) > 1$. We split the integration interval as $(E \lambda, \alpha E  (2 \pi E/a)^s) \cup (\alpha E  (2 \pi E/a)^s, \infty)$ and consider $E$ so large that $2 \pi E/a > 1$. For the first interval, we have by Lemma \ref{LemBound1} that
\begin{align}
\mathcal{G}_{(E \lambda, \alpha E  (2 \pi E/a)^s)} & \leq \frac{a}{(2 \pi)^4 (E/a)} \frac{\ee^{2\pi E/a}}{\sinh(2 \pi E/a)} \int_{E \lambda}^{\alpha E  (2 \pi E/a)^s} \! d \omega \, \left|\hat{\chi}(\omega)\right|^2 \sinh (2 \pi \omega /(a\lambda)).
\end{align}

The factor multiplying the integral above is $O(1/E)$ and the integrand can be controlled in this interval as
\begin{align}
\left|\hat{\chi}(\omega)\right|^2 & \sinh \left(\frac{2 \pi \omega}{a \lambda} \right) \leq \frac{C^2}{\kappa^{2}} \left(B+ \frac{\omega}{\kappa} \right)^{2r} \exp \left[ -2A \left(\frac{\omega}{\kappa}\right)^q + \frac{2 \pi \omega}{a \lambda} \right] \nonumber \\
& \leq \frac{C^2}{\kappa^{2}} \left(B+ \frac{\omega}{\kappa} \right)^{2r} \exp \left[ -2A \left(\frac{E \lambda}{\kappa}\right)^q + \frac{2 \pi (\alpha E (2 \pi E/a)^s)}{a \lambda} \right] \nonumber \\
& = \frac{C^2}{\kappa^{2}} \left(B+ \frac{\omega}{\kappa} \right)^{2r} \exp \left[ -2A \left(\frac{a \alpha}{2 \pi \kappa}\right)^q \left(\frac{2 \pi E}{a} \right)^{q(2+p)} + \left(\frac{2 \pi E}{a} \right)^{s-p} \right].
\label{supfastG31}
\end{align}

We see from the last line of \eqref{supfastG31} that the contribution in the integration interval $(E \lambda, \alpha E  (2 \pi E/a)^s)$ to \eqref{app5:G3} decays faster than any polynomial as $E \to \infty$, because the first term in the exponent is dominating because $q(2 + p) > s - p$ and the parameters satisfy $2A(a \alpha/(2\pi \kappa))^q > 1$.

For the subinterval $(\alpha E  (2 \pi E/a)^s, \infty)$, we have
\begin{align}
\frac{ \left( E - \omega/\lambda\right) \sinh (2 \pi \omega /(a\lambda))}{\cosh (2 \pi E/a) -\cosh (2 \pi \omega/(a \lambda) )} & \leq  \frac{ \omega}{\lambda}\text{tanh}\left(\frac{2\pi \omega}{a\lambda}\right)\left(1-\cfrac{\cosh(2 \pi E/a)}{\cosh(2 \pi \omega /(a\lambda))}\right)^{-1} \nonumber \\
& \leq \frac{ \omega}{\lambda} \left(1-\cfrac{\cosh(2 \pi E/a)}{\cosh(2 \pi E/a)^{s-p})}\right)^{-1}.
\label{boundshyp}
\end{align}

Considering that $1+p < s$, the denominator in eq. \eqref{boundshyp} is of order unity. Therefore, we see that this piece contributes linearly in the integration variable $\omega$. We are left to estimate $\left|\chi(\omega)\right|^2 \ee^{2\pi E/a}$. The decay properties of $\hat{\chi}$ yield that, for $\omega \in(\alpha E(2 \pi E/a)^s,\infty)$,
\begin{align}
\left|\hat{\chi}(\omega)\right|^2 \ee^{2\pi E/a}  & \leq \frac{C^2}{\kappa^{2}} \left(B+ \frac{\omega}{\kappa} \right)^{2r} \exp \left[ -2A \left(\frac{\omega}{\kappa}\right)^q + \frac{2 \pi E}{a} \right] \nonumber \\
&  \leq \frac{C^2}{\kappa^{2}} \left(B+ \frac{\omega}{\kappa} \right)^{2r} \exp\left[-2A\left(\frac{\omega}{\kappa}\right)^q \left(1 - \frac{\pi E}{A a}\left(\frac{\omega}{\kappa}\right)^{-q} \right)\right].
\label{boundpiece3}
\end{align}
At large $E$, the factor $1- (\pi E/(A\, a))(\omega/\kappa)^{-q} > 0$ and the factor inside the exponent is overall negative at large $E$. Thus, the integrand is rapidly suppressed and, collecting our results, $\mathcal{G}_{(E \lambda, \alpha E  (2 \pi E/a)^s)}$ dies off faster than any polynomial in $E$.

\bibliographystyle{ieeetr}
\bibliography{BenitoBib}

\end{document}